\documentclass[11pt,leqno]{article}
\usepackage{amssymb,amsfonts}%amsLatex
\usepackage{amsmath,amsthm,amsxtra}
\usepackage[all]{xy}
\usepackage{mathabx}
\usepackage{color}
\usepackage{verbatim}
\usepackage[linktocpage=true,colorlinks=true, linkcolor=blue, citecolor=red, urlcolor=green]{hyperref}

\usepackage{enumerate}

\newcommand\bigcheck[1]{#1 \raise1ex\hbox{$\hspace{-1ex}{}^\vee$}}
\newcommand\sucheck[1]{#1 \raise0.5ex\hbox{$\hspace{-1ex}{}^\vee$}}

\newcommand{\wcheck}{\!\!\!\!\!\!\widecheck{}\,\,\,\,\,\,}
\newcommand{\wwcheck}{\!\!\!\!\!\!\!\!\!\!\!\!\!\!\!\!\!\!\widecheck{}\,\,\,\,\,\,\,\,\,\,\,\,\,\,\,\,\,\,}

\makeatletter
\@addtoreset{equation}{section}
\makeatother

\newtheorem{theorem}{Theorem}[section]
\newtheorem{lemma}[theorem]{Lemma}
\newtheorem{corollary}[theorem]{Corollary}
\newtheorem{proposition}[theorem]{Proposition}
\newtheorem*{lemma*}{Lemma}

\theoremstyle{definition}
\newtheorem{definition}[theorem]{Definition}

\theoremstyle{remark}
\newtheorem{remark}[theorem]{Remark}
\newtheorem{example}[theorem]{Example}

\newcommand{\mc}[1]{{\mathcal #1}}
\newcommand{\mf}[1]{{\mathfrak #1}}
\newcommand{\mb}[1]{{\mathbb #1}}

\newcommand\tint{{\textstyle\int}}

\newcommand{\id}{{1 \mskip -5mu {\rm I}}}

\renewcommand{\tilde}{\widetilde}

\newcommand{\ad}{\mathop{\rm ad }}

\newcommand{\Mat}{\mathop{\rm Mat }}

\newcommand{\ord}{\mathop{\rm ord }}

\renewcommand{\ker}{\mathop{\rm Ker }}
\newcommand{\im}{\mathop{\rm Im }}

\newcommand{\Span}{\mathop{\rm Span }}

\newcommand{\ass}[1]{\stackrel{#1}{\longleftrightarrow}}

\newcommand{\dord}{\mathop{\rm dord }}

\definecolor{light}{gray}{.9}

\begin{document}

%%%%%%%%%%%%%%%%%%%%%%%%%%%%%%%%%%%%%%%%%%%%%%%%%%%%%%%%%%%%%%%%%%%%%%%%%%%%%%%%%%%%%%%%%%%%%%%%%%%%%%%%%%%%%%
%%%%%%%%%%%%%%% TITLE %%%%%%%%%%%%%%%%%%%%%%%%%%%%%%%%%%%%%%%%%%%%%%%%%%%%%%%%%%%%%%%%%%%%%%%%%%%%%%%%%%%%%%%
%%%%%%%%%%%%%%%%%%%%%%%%%%%%%%%%%%%%%%%%%%%%%%%%%%%%%%%%%%%%%%%%%%%%%%%%%%%%%%%%%%%%%%%%%%%%%%%%%%%%%%%%%%%%%%

\title{Non-local Poisson structures and applications to the theory 
of integrable systems}

\author{
Alberto De Sole
\thanks{Dipartimento di Matematica, Universit\`a di Roma ``La Sapienza'',
00185 Roma, Italy ~~
desole@mat.uniroma1.it ~~~~
Supported in part by Department of Mathematics, M.I.T.},~~
Victor G. Kac
\thanks{Department of Mathematics, M.I.T.,
Cambridge, MA 02139, USA.~~
kac@math.mit.edu~~~
Supported in part by an NSF grant, and the Simons Fellowship
}~~
}

\date{}

\maketitle

\vspace{4pt}

\begin{center}
\emph{Dedicated to Minoru Wakimoto on his 70-th birthday.}
\end{center}

\vspace{2pt}

%%%%%%%%%%%%%%%%%%%%%%%%%%%%%%%%%%%%%%%%%%%%%%%%%%%%%%%%%%%%%%%%%%%%%%%%%%%%%%%%%%%%%%%%%%%%%%%%%%%%%%%%%%%%%%
%%%%%%%%%%%%%%% Abstract %%%%%%%%%%%%%%%%%%%%%%%%%%%%%%%%%%%%%%%%%%%%%%%%%%%%%%%%%%%%%%%%%%%%%%%%%%%%%%%%%%%%%%%
%%%%%%%%%%%%%%%%%%%%%%%%%%%%%%%%%%%%%%%%%%%%%%%%%%%%%%%%%%%%%%%%%%%%%%%%%%%%%%%%%%%%%%%%%%%%%%%%%%%%%%%%%%%%%%

\begin{abstract}
\noindent 
We develop a rigorous theory of non-local Poisson structures,
built on the notion of a non-local Poisson vertex algebra.
As an application, we find conditions that guarantee applicability of the Lenard-Magri scheme
of integrability to a pair of compatible non-local Poisson structures.
We apply this scheme to several such pairs,
proving thereby integrability of various evolution equations, as well as hyperbolic equations.
\end{abstract}

\medskip

\noindent
\emph{Keywords and phrases:}
non-local Poisson vertex algebra,
non-local Poisson structure,
rational matrix pseudodifferential operators,
Lenard-Magri scheme of integrability,
bi-Hamiltonian integrable hierarchies.

\medskip

\noindent
\emph{Mathematics Subject Classification (2010):}
37K10 (Primary) 35Q53, 17B80, 17B69, 37K30, 17B63 (Secondary)

\vfill\eject

\tableofcontents

\vfill\eject

%%%%%%%%%%%%%%%%%%%%%%%%%%%%%%%%%%%%%%%%%%%%%%%%%%%%%%%%%%%%%%%%%%%%%%%%%%%%%%%%%%%%%%%%%%%%%%%%%%%%%%%%%%%%%%
%%%%%%%%%%%%%%% Sect 1 %%%%%%%%%%%%%%%%%%%%%%%%%%%%%%%%%%%%%%%%%%%%%%%%%%%%%%%%%%%%%%%%%%%%%%%%%%%%%%%%%%%%%%%
%%%%%%%%%%%%%%%%%%%%%%%%%%%%%%%%%%%%%%%%%%%%%%%%%%%%%%%%%%%%%%%%%%%%%%%%%%%%%%%%%%%%%%%%%%%%%%%%%%%%%%%%%%%%%%

\section{Introduction}
\label{sec:intro}

Local Poisson brackets play a fundamental role in the theory of integrable systems.
Recall that a local Poisson bracket is defined by (see e.g. \cite{TF86}):
\begin{equation}\label{intro:eq1.1}
\{u_i(x),u_j(y)\}=H_{ij}\big(u(y),u'(y),\dots;\partial_y\big)\delta(x-y)\,,
\end{equation}
where $u=(u_1,\dots,u_\ell)$ is a vector valued function on a 1-dimensional manifold $M$,
$\delta(x-y)$ is the $\delta$-function: $\tint_Mf(y)\delta(x-y)dy=f(x)$,
and $H(\partial)=\big(H_{ij}(\partial)\big)_{i,j=1}^\ell$
is an $\ell\times\ell$ matrix differential operator,
whose coefficients are functions in $u,u',\dots,u^{(k)}$.
One requires, in addition, that \eqref{intro:eq1.1} ``satisfies the Lie algebra axioms''.

One of the ways to formulate the latter condition is as follows.
Let $\mc V$ be an algebra of differential polynomials in $u_1,\dots,u_\ell$,
i.e. the algebra of polynomials in $u_i^{(n)},\,i\in I=\{1,\dots,\ell\},\,n\in\mb Z_+$,
with $u_i^{(0)}=u_i$ and the derivation $\partial$, defined by $\partial u_i^{(n)}=u_i^{(n+1)}$,
or its algebra of differential functions extension.
The bracket \eqref{intro:eq1.1} extends, by the Leibniz rule and bilinearity,
to arbitrary $f,g\in\mc V$:
\begin{equation}\label{intro:eq1.2}
\{f(x),g(y)\}=
\sum_{i,j\in I}\sum_{m,n\in\mb Z_+}
\frac{\partial f(x)}{\partial u_i^{(m)}} \frac{\partial g(y)}{\partial u_j^{(n)}}
\partial_x^m \partial_y^n\{u_i(x),u_j(y)\}\,.
\end{equation}
Applying integration by parts, we get the following bracket on $\mc V/\partial\mc V$:
\begin{equation}\label{intro:eq1.3}
\{\tint f,\tint g\}=
\int \frac{\delta g}{\delta u}\cdot H(\partial)\frac{\delta f}{\delta u}\,,
\end{equation}
where $\tint$ is the canonical quotient map $\mc V\to\mc V/\partial\mc V$
and $\frac{\delta f}{\delta u}$ is the vector of variational derivatives
$\frac{\delta f}{\delta u_i}=\sum_{n\in\mb Z_+}(-\partial)^n\frac{\partial f}{\partial u_i^{(n)}}$.
Then one requires that the bracket \eqref{intro:eq1.3}
satisfies the Lie algebra axioms.
(The skewsymmetry of this bracket is equivalent to the skewadjointness of $H(\partial)$,
but the Jacobi identity is a complicated system of non-linear PDE on its coefficients.)
In this case the matrix differential operator $H(\partial)$ is called a \emph{Poisson structure}.
(Sometimes in literature, including our previous papers, this is called
a Hamiltonian structure, or a Hamiltonian operator,
but the name Poisson structure seems to be more appropriate.)

Given an element $\tint h\in\mc V/\partial\mc V$, called a \emph{Hamiltonian functional},
the \emph{Hamiltonian equation} associated to $H(\partial)$ is the following evolution equation:
\begin{equation}\label{intro:eq1.4}
\frac{du}{dt}=H(\partial)\frac{\delta h}{\delta u}\,.
\end{equation}
For example, taking $H(\partial)=\partial$ and $h=\frac12(u^3+cuu'')$,
we obtain the KdV equation: $\frac{du}{dt}=3uu'+cu'''$.

Equation \eqref{intro:eq1.4} is called \emph{integrable} if $\tint h$
is contained in an infinite dimensional abelian subalgebra $A$ of the Lie algebra $\mc V/\partial\mc V$
with bracket \eqref{intro:eq1.3}.
Picking a basis $\{\tint h_n\}_{n\in\mb Z_+}$ of $A$,
we obtain a hierarchy of compatible integrable Hamiltonian equations:
$$
\frac{du}{dt_n}=H(\partial)\frac{\delta h_n}{\delta u}\,\,,\,\,\,\,n\in\mb Z_+\,.
$$

An alternative approach, proposed in \cite{BDSK09},
is to apply the Fourier transform $F(x,y)\mapsto\tint_Mdxe^{\lambda(x-y)}F(x,y)$
to both sides of \eqref{intro:eq1.2}
to obtain the following ``Master formula'' \cite{DSK06}:
\begin{equation}\label{intro:eq1.5}
\{f_\lambda g\}=
\sum_{i,j\in I}\sum_{m,n\in\mb Z_+}
\frac{\partial g}{\partial u_j^{(n)}}
(\lambda+\partial)^n
H_{ji}(\lambda+\partial)
(-\lambda-\partial)^m
\frac{\partial f}{\partial u_i^{(m)}}\,. 
\end{equation}
For an arbitrary $\ell\times\ell$ matrix differential operator $H(\partial)$ 
this $\lambda$-\emph{bracket}
is polynomial in $\lambda$, i.e. it takes values in $\mc V[\partial]$,
satisfies the \emph{left} and \emph{right Leibniz rules}:
\begin{equation}\label{intro:eq1.6}
\{f_\lambda gh\}=
g\{f_\lambda h\}+h\{f_\lambda g\}
\,\,,\,\,\,\,
\{fg_\lambda h\}=
\{f_{\lambda+\partial} g\}_\to h+\{f_{\lambda+\partial} h\}_\to g\,,
\end{equation}
where the arrow means that $\lambda+\partial$ should be moved to the right,
and the  \emph{sesquilinearity} axioms:
\begin{equation}\label{intro:eq1.7}
\{\partial f_\lambda g\}=-\lambda\{f_\lambda g\}
\,\,,\,\,\,\,
\{f_\lambda\partial g\}=(\lambda+\partial)\{f_\lambda g\}\,.
\end{equation}
It is proved in \cite{BDSK09} that the requirement that \eqref{intro:eq1.3} satisfies the Lie algebra axioms
is equivalent to the following two properties of \eqref{intro:eq1.5}:
\begin{equation}\label{intro:eq1.8}
\{g_\lambda f\}=-\{f_{-\lambda-\partial} g\}\,,
\end{equation}
\begin{equation}\label{intro:eq1.9}
\{f_\lambda \{g_\mu h\}\}=\{g_\mu\{f_\lambda h\}\}
+\{\{f_\lambda g\}_{\lambda+\mu}h\}\,.
\end{equation}
A differential algebra $\mc V$,
endowed with a polynomial $\lambda$-bracket, satisfying axioms \eqref{intro:eq1.6}--\eqref{intro:eq1.9},
is called a \emph{Poisson vertex algebra} (PVA).

It was demonstrated in \cite{BDSK09} that the PVA approach greatly simplifies the theory
of integrable Hamiltonian PDE, based on local Poisson brackets.
For example, equation \eqref{intro:eq1.4} becomes, in terms of the $\lambda$-bracket associated to $H$:
$$
\frac{du}{dt}=\{h_\lambda u\}\big|_{\lambda=0}\,,
$$
and the Lie bracket \eqref{intro:eq1.3} becomes
$$
\{\tint f,\tint g\}=\tint \{f_\lambda g\}\big|_{\lambda=0}\,.
$$

However the majority of important integrable equations, including the non-linear Schroedinger equation,
is Hamiltonian with respect to a non-local Poisson bracket.
It has been an open problem to develop a rigorous theory of such brackets.
The purpose of the present paper is to demonstrate that the adequate (and in fact indispensable)
tool for understanding non-local Poisson brackets is the theory of ``non-local'' PVA.

We define a \emph{non-local} $\lambda$-\emph{bracket} on the differential algebra $\mc V$
to take its values in $\mc V((\lambda^{-1}))$,
formal Laurent series in $\lambda^{-1}$ with coefficients in $\mc V$,
and to satisfy properties \eqref{intro:eq1.6} and \eqref{intro:eq1.7}.
The main example is the $\lambda$-bracket given by the Master Formula \eqref{intro:eq1.5},
where $H(\partial)$ is a matrix pseudodifferential operator.
The only problem with this definition is the interpretation of the operator $\frac1{\lambda+\partial}$;
this is defined by the geometric progression
$$
\frac1{\lambda+\partial}=\sum_{n\in\mb Z_+}(-1)^n\lambda^{-n-1}\partial^n\,.
$$

Property \eqref{intro:eq1.8} of the $\lambda$-bracket is interpreted in the same way,
but the interpretation of property \eqref{intro:eq1.9} is more subtle.
Indeed, in general, we have $\{f_\lambda\{g_\mu h\}\}\in\mc V((\lambda^{-1}))((\mu^{-1}))$,
but $\{g_\mu\{f_\lambda h\}\}\in\mc V((\mu^{-1}))((\lambda^{-1}))$,
and $\{\{f_\lambda g\}_{\lambda+\mu} h\}\in\mc V(((\lambda+\mu)^{-1}))((\lambda^{-1}))$,
so that all three terms of \eqref{intro:eq1.9} lie in different spaces.
Our key idea is to consider the space
$$
\mc V_{\lambda,\mu}=\mc V[[\lambda^{-1},\mu^{-1},(\lambda+\mu)^{-1}]][\lambda,\mu]\,,
$$
which is canonically embedded in all three of the above spaces.
We say that a $\lambda$-bracket is \emph{admissible} if 
$$
\{f_\lambda\{g_\mu h\}\}\in\mc V_{\lambda,\mu}
\,\,\,\,
\text{ for all } f,g,h\in\mc V\,.
$$
It is immediate to see that then the other two terms of \eqref{intro:eq1.9} 
lie in $\mc V_{\lambda,\mu}$ as well,
hence \eqref{intro:eq1.9} is an identity in $\mc V_{\lambda,\mu}$.

We call the differential algebra $\mc V$, endowed with a non-local $\lambda$-bracket,
a \emph{non-local PVA}, if it satisfies \eqref{intro:eq1.8},
is admissible, and satisfies \eqref{intro:eq1.9}.

For an arbitrary pseudodifferential operator $H(\partial)$ the $\lambda$-bracket \eqref{intro:eq1.5}
is not admissible, but it is admissible for any \emph{rational} pseudodifferential operator,
i.e. such that the entries of the matrix $H(\partial)$ are
contained in the subalgebra generated by differential operators and their inverses.
We show that, as in the local case (see \cite{BDSK09}),
this $\lambda$-bracket satisfies conditions \eqref{intro:eq1.8} and \eqref{intro:eq1.9}
if and only if \eqref{intro:eq1.8} holds for any pair $u_i,u_j$,
and \eqref{intro:eq1.9} holds for any triple $u_i,u_j,u_k$.
Also, \eqref{intro:eq1.8} is equivalent to skewadjointness of $H(\partial)$.

The simplest example of a non-local PVA corresponds 
to the skewadjoint operator $H(\partial)=\partial^{-1}$.
Then 
$$
\{u_\lambda u\}=\lambda^{-1}\,,
$$
and equation \eqref{intro:eq1.9} trivially holds for the triple $u,u,u$.
Note that \eqref{intro:eq1.1} in this case reads:
$\{u(x),u(y)\}=\partial_y^{-1}\delta(x-y)$,
which is quite difficult to work with
(cf. \cite{MN01}).

The next example corresponds to Sokolov's operator \cite{Sok84}
$H(\partial)=u'\partial^{-1}\circ u'$.
The corresponding $\lambda$-bracket is
$$
\{u_\lambda u\}=u'\frac1{\lambda+\partial}u'\,.
$$
The verification of \eqref{intro:eq1.9} for the triple $u,u,u$ is straightforward.

We say that 
a rational pseudodifferential operator $H(\partial)$ is a \emph{Poisson structure} on $\mc V$
if the $\lambda$-bracket \eqref{intro:eq1.5} endows $\mc V$ with a structure of a non-local PVA
(in other words $H(\partial)$ should be skewadjoint and \eqref{intro:eq1.9} should hold
for any triple $u_i,u_j,u_k$).

Fix a ``minimal fractional decomposition'' $H=AB^{-1}$.
This means that $A,B$ are differential operators over $\mc V$,
such that $\ker A\cap\ker B=0$
in any algebra of differential functions extension of $\mc V$.
It is shown in \cite{CDSK12b} that such a decomposition always exists and that the above property
is equivalent to the property that any common right factor of $A$ and $B$
is invertible over the field of fractions $\mc K$ of $\mc V$.
Then the basic notions of the theory of integrable systems are defined as follows.
A \emph{Hamiltonian functional} (for $H=AB^{-1}$)
is an element $\tint h\in\mc V/\partial\mc V$ such that $\frac{\delta\tint h}{\delta u}=B(\partial)F$
for some $F\in\mc K^\ell$.
Then the element $P=A(\partial)F$ is called an associated \emph{Hamiltonian vector field},
and we write $\tint h\ass{H}P$ or $P\ass{H}\tint h$.
Denote by $\mc F(H)\subset\mc V/\partial\mc V$ the subspace of all Hamiltonian functionals,
and by $\mc H(H)\subset\mc V^\ell$ the subspace of all Hamiltonian vector fields
(they are independent of the choice of the minimal fractional decomposition for $H$):
$$
\mc F(H)=\Big(\frac{\delta}{\delta u}\Big)^{-1}\Big(\im B\Big)\subset\mc V/\partial\mc V
\,\,,\,\,\,\,
\mc H(H)=A\Big(B^{-1}\Big(\im\frac{\delta}{\delta u}\Big)\Big)\subset\mc V^\ell\,.
$$
Then it is easy to show that $\mc F(H)$ is a Lie algebra with respect to the well-defined
bracket \eqref{intro:eq1.3}, and $\mc H(H)$ is a subalgebra of the Lie algebra $\mc V^\ell$
with bracket $[P,Q]=D_Q(\partial)P-D_P(\partial)Q$, where $D_P(\partial)$ is the Frechet derivative.

A \emph{Hamiltonian equation}, corresponding to the Poisson structure $H$
and a Hamiltonian functional $\tint h\in\mc F(H)$,
with an associated Hamiltonian vector field $P\in\mc H(H)$,
is the following evolution equation:
\begin{equation}\label{intro:eq1.10}
\frac{du}{dt}=P\,.
\end{equation}
Note that \eqref{intro:eq1.10} coincides with \eqref{intro:eq1.4} in the local case.
The Hamiltonian equation \eqref{intro:eq1.10} is called \emph{integrable} if there exist
linearly independent infinite sequences
$\tint h_n\in\mc F(H)$ and $P_n\in\mc H(H)$, $n\in\mb Z_+$,
such that $\tint h_0=\tint h$, $P_0=P$,
$P_n$ is associated to $\tint h_n$,
and $\{\tint h_m,\tint h_n\}=0,\,[P_m,P_n]=0$
for all $m,n\in\mb Z_+$.
In this case we have a hierarchy of compatible integrable equations
$$
\frac{du}{dt_n}=P_n
\,\,,\,\,\,\, n\in\mb Z_+\,.
$$
(The $P_n$'s are called the generalized symmetries of equation \eqref{intro:eq1.10}
and the $h_n$'s are its conserved densities.)

Having given rigorous definitions of the basic notions of the theory of Hamiltonian equations
with non-local Poisson structures,
we proceed to establish some basic results of the theory.

The first result is Theorem \ref{20111021:thm},
which states that if $H$ and $K$ are compatible non-local Poisson structures and $K$ is invertible
(as a pseudodifferential operator), then the sequence
of rational pseudodifferential operators
$H^{[0]}=K, H^{[n]}=(H K^{-1})^{n-1} H,\,n\geq1$,
is a compatible family of non-local Poisson structures.
(As usual \cite{Mag78,Mag80} a collection of non-local Poisson structures
is called compatible if any their finite linear combination is again a non-local Poisson structure.)
This result was first stated in \cite{Mag80}
and its partial proof was given in \cite{FF81}
(of course, without having rigorous definitions).

Next, we give a rigorous definition of a non-local symplectic structure and prove 
(the ``well-known'' fact) that, if $S$ is invertible (as a pseudodifferential operator),
then it is a non-local symplectic structure if and only if $S^{-1}$ is a non-local Poisson structure
(Theorem \ref{20111012:thm}).
Since we completely described (local) symplectic structures in \cite{BDSK09},
this result provides a large collection of non-local Poisson structures.
We also establish a connection between Dirac structures (see \cite{Dor93} and \cite{BDSK09})
with non-local Poisson structures
(Theorems \ref{20111020:thm} and \ref{20120126:prop2}).

After that we discuss the Lenard-Margi scheme of integrability
for a pair of compatible non-local Poisson structures $H$, $K$,
similar to that discussed in \cite{Mag78,Mag80,Dor93,BDSK09} in the local case,
and give sufficient conditions when this scheme works
(Theorem \ref{20130123:thm} and Corollary \ref{20130123:cor}).
This means that there exists an infinite sequence of Hamiltonian functionals 
$\tint h_n,\,n\in\mb Z_+$,
and Hamiltonian vector fields $P_n,\,n\in\mb Z_+$,
such that we have
\begin{equation}\label{1.1}
\tint 0\ass{H}P_0\ass{K}\tint h_0\ass{H} P_1\ass{K}\tint h_1\ass{H}\dots\,,
\end{equation}
and the spans of the $\tint h_n$'s and of the $P_n$'s are infinite dimensional.
%(the association relation $\tint f\ass{H}P$ is recalled in Definition \ref{20120124:def}).
%
This produces integrable Hamiltonian equations $u_{t_n}=P_n$.

Let us also mention that the Lenard-Magri scheme in the weakly non-local case
(in the sense of \cite{MN01}) was studied in \cite{Wan09}.

% other approaches

Now we compare briefly our approach to integrability with other algebraic approaches.
Probably the earliest approach is the Lax pair presentation (see the book \cite{Dic03}).
The main difficulty of this approach is to establish linear independence of the 
integrals of motion.
Another popular approach
is based on a recursion operator (see books \cite{Olv93} and \cite{Bla98}),
which is applied to a conserved density
or a generalized symmetry to produce a new one.
Unfortunately, since $R$ is non-local (even for the KdV)
this approach is not rigorous and often leads to wrong conclusions
(as demonstrated, for example, in \cite{SW01}).
A more recent approach, due to Dorfman \cite{Dor93}
is based on the notion of a Dirac structure.
This theory, along with its further developments in \cite{BDSK09}
and the present paper, is a basis of our non-local bi-Hamiltonian approach.
In fact, our approach overcomes the main difficulty, that of constructing 
a Dirac structure in Dorfman's approach,
and that of proving linear independence of integrals of motion
in the Lax pair approach.
Also, its advantage as compared to the recursion operator approach is that it is rigorous.

%%% INTRO FROM NON-LOCAL2

The applications of this theory 
to concrete examples are studied in Sections \ref{secb:3}, \ref{secb:4} and \ref{secb:5}.
In Section \ref{secb:3} we consider three compatible scalar non-local Poisson structures:
$$
L_1=\partial\,,\,\, %(GFZ)
L_2=\partial^{-1}\,,\,\, %(Toda)
L_3=u'\partial^{-1}\circ u'
\,(\text{ Sokolov } \cite{Sok84})\,, %(Sokolov)
$$
and take for a compatible pair $(H,K)$ two arbitrary linear combinations of these three structures:
$H=\sum_ia_iL_i$, $K=\sum_ib_iL_i$.
We study in detail for which values of the coefficients $a_i$ and $b_i$
the corresponding Lenard-Magri scheme is integrable.

Furthermore we study when the infinite sequence \eqref{1.1} can be extended to the left.
The most interesting cases are those when the sequence is ``blocked''
at some step $P_{-n},\,n>0$, to the left.
This leads to some interesting integrable hyperbolic equations.
As a result, we prove integrability of two such equations
\begin{equation}\label{20121020:eq8-intro}
u_{tx} =
e^{u}-\alpha e^{-u}
+\epsilon(e^{u}-\alpha e^{-u})_{xx}
\,,
\end{equation}
where $\alpha$ and $\epsilon$ are $0$ or $1$, and
\begin{equation}\label{20121020:eq10-intro}
u_{tx} =
u+(u^3)_{xx}
\,.
\end{equation}
Of course, in the case when $\epsilon=0$, equation \eqref{20121020:eq8-intro} is the Liouville
(respectively sinh-Gordon) equation if $\alpha=0$ (resp. $\alpha=1$).
For $\epsilon=1$ equation \eqref{20121020:eq8-intro} 
was studied in \cite{Fok95}.
Equation \eqref{20121020:eq10-intro}, studied in \cite{SW02},
is called the short pulse equation.
Its integrability was proved in \cite{SS04}

In Section \ref{secb:4}
we study, in a similar way, two linear combinations of the compatible non-local Poisson structures
$$
L_1=
u'\partial^{-1}\circ u' 
\,\,,\,\,\,\,
L_2=
\partial^{-1}\circ u'\partial^{-1}\circ u'\partial^{-1}
\,\,(\text{ Dorfman } \cite{Dor93})
\,.
$$
As a result we (re)prove integrability of the Schwarz KdV (also called the degenerate
Krichever-Novikov) equation
$$
u_t=u_{xxx}-\frac32\frac{u_{xx}^2}{u_x}\,,
$$
and also, moving to the left, establish integrability of the following equation
$$
\bigg(\frac1{u_x}
\Big(\frac{u_{tx}}{u_x}\Big)_x\bigg)_x
=
\frac1{u_x}\Big(
\frac{c_0+c_1u+c_2u^2}{u_x}
\Big)_x
\,,
$$
where $c_0,c_1,c_2$ are arbitrary constants. 

Finally, in Section \ref{secb:5}
we study, in a similar way, three two-component non-local Poisson structures
that are used in the study of the non-linear Schroedinger equation (NLS), 
see \cite{Mag80,TF86,Dor93,BDSK09}.
As a result, we establish integrability of the following generalization 
of NLS:
$$
i\psi_t=\psi_{xx}+\alpha\psi|\psi|^2+i\beta(\psi|\psi|^2)_x\,,
$$
where $\alpha$ and $\beta$ are arbitrary constants
(NLS corresponds to $\beta=0$).
This equation has been studied in the papers \cite{CLL79} and \cite{WKI79}
(see also \cite{KN78} and \cite{CC87}).

In conclusion of the introduction we would like to comment on our definition of integrability.
The existence of infinitely many linearly independent 
integrals of motion in involution $\tint h_n$,
and of infinitely many linearly independent commuting higher symmetries $P_n$,
is only a necessary condition of integrability.
In Section \ref{sec:7.1b} we introduce the notion of \emph{complete integrability}
which, in our opinion, is the right necessary and sufficient condition of integrability.
This condition requires that the orthocomplement to the
span $\Xi$ of the variational derivatives 
of the conserved densities $\xi_n=\frac{\delta h_n}{\delta u},\,n\in\mb Z_+$,
lies in the span $\Pi$ of the commuting generalized symmetries $P_n,\,n\in\mb Z_+$,
and the orthocomplement to $\Pi$ lies in $\Xi$.
This definition is a straightforward generalization of Liouville integrability
of finite dimensional Hamiltonian systems.
We intend to study this notion in a forthcoming publication.

% AKNOWLEDGEMENTS

Throughout the paper, unless otherwise specified,
all vector spaces are considered over a field $\mb F$ of characteristic zero.

We wish to thank Pavel Etingof and Andrea Maffei for (always) useful discussions.
We are greatly indebted to Alexander Mikhailov and Vladimir Sokolov 
for very useful correspondence and discussions.
We also wish to thank Takayuki Tsuchida for pointing out,
right after the paper appeared in the arXiv,
various references where some of the equations that we consider were
previously studied.
The present paper was partially written during the first author's several visits to MIT,
the second author's several visits to the Center for Mathematics and Theoretical Physics (CMTP)
in Rome,
and both authors' visits to IHP and IHES, which we thank for their warm hospitality.

%%%%%%%%%%%%%%%%%%%%%%%%%%%%%%%%%%%%%%%%%%%%%%%%%%%%%%%%%%%%%%%%%%%%%%%%%%%%%%%%%%%%%%%%%%%%%%%%%%%%%%%%%%%%%%
%%%%%%%%%%%%%%% Sect 2 %%%%%%%%%%%%%%%%%%%%%%%%%%%%%%%%%%%%%%%%%%%%%%%%%%%%%%%%%%%%%%%%%%%%%%%%%%%%%%%%%%%%%%%
%%%%%%%%%%%%%%%%%%%%%%%%%%%%%%%%%%%%%%%%%%%%%%%%%%%%%%%%%%%%%%%%%%%%%%%%%%%%%%%%%%%%%%%%%%%%%%%%%%%%%%%%%%%%%%

\section{Rational matrix pseudodifferential operators}
\label{sec:2}

\subsection{The space $V_{\lambda,\mu}$}
\label{sec:2.1}

Throughout the paper we shall use the following standard notation.
Given a vector space $V$, we denote by $V[\lambda]$ the space of polynomials in $\lambda$ with coefficients in $V$,
by $V[[\lambda^{-1}]]$ the space of formal power series in $\lambda^{-1}$ with coefficients in $V$,
and by $V((\lambda^{-1}))=V[[\lambda^{-1}]][\lambda]$ the space of formal Laurent series in $\lambda^{-1}$ 
with coefficients in $V$.

We have the obvious identifications $V[\lambda,\mu]=V[\lambda][\mu]=V[\mu][\lambda]$
and $V[[\lambda^{-1},\mu^{-1}]]=V[[\lambda^{-1}]][[\mu^{-1}]]=V[[\mu^{-1}]][[\lambda^{-1}]]$.
However the space $V((\lambda^{-1}))((\mu^{-1})$ does not coincide 
with $V((\mu^{-1}))((\lambda^{-1}))$.
Both spaces contain naturally the subspace $V[[\lambda^{-1},\mu^{-1}]][\lambda,\mu]$.
In fact, this subspace is their intersection in the ambient space $V[[\lambda^{\pm1},\mu^{\pm1}]]$
of all infinite series of the form $\sum_{m,n\in\mb Z}a_{m,n}\lambda^m\mu^n$.

The most important for this paper will be the space
$$
V_{\lambda,\mu}:=V[[\lambda^{-1},\mu^{-1},(\lambda+\mu)^{-1}]][\lambda,\mu]\,,
$$
namely, the quotient of the $\mb F[\lambda,\mu,\nu]$-module
$V[[\lambda^{-1},\mu^{-1},\nu^{-1}]][\lambda,\mu,\nu]$
by the submodule 
$(\nu-\lambda-\mu)V[[\lambda^{-1},\mu^{-1},\nu^{-1}]][\lambda,\mu,\nu]$.
By definition, the space $V_{\lambda,\mu}$ consists of elements which can be written (NOT uniquely) in the form
\begin{equation}\label{20111006:eq1}
A=\sum_{m=-\infty}^M\sum_{n=-\infty}^N\sum_{p=-\infty}^P a_{m,n,p}\lambda^m\mu^n(\lambda+\mu)^p\,,
\end{equation}
for some $M,N,P\in\mb Z$ (in fact, we can always choose $P\leq 0$), and $a_{m,n,p}\in V$.

In the space $V[[\lambda^{-1},\mu^{-1},\nu^{-1}]][\lambda,\mu,\nu]$
we have a natural notion of degree, by letting $\deg(\lambda)=\deg(\mu)=\deg(\nu)=1$.
Every element $A\in V[[\lambda^{-1},\mu^{-1},\nu^{-1}]][\lambda,\mu,\nu]$
decomposes as a sum
$A=\sum_{d=-\infty}^NA^{(d)}$
(possibly infinite), where $A^{(d)}$ is a finite linear combination of monomials of degree $d$.
Since $\nu-\lambda-\mu$ is homogenous (of degree 1),
this induces a well-defined notion of degree on the quotient space $V_{\lambda,\mu}$,
and we denote by $V_{\lambda,\mu}^d$, for $d\in\mb Z$,
the span of elements of degree $d$ in $V_{\lambda,\mu}$.
If $A\in\mc V_{\lambda,\mu}$ has the form \eqref{20111006:eq1}, then it decomposes as $A=\sum_{d=-\infty}^{M+N+P}A^{(d)}$,
where $A^{(d)}\in V_{\lambda,\mu}^d$ is given by
$$
A^{(d)}=\sum_{\substack{m\leq M,n\leq N,p\leq P \\ (m+n+p=d)}} a_{m,n,p}\lambda^m\mu^n(\lambda+\mu)^p\,.
$$
The coefficients $a_{m,n,p}\in V$ are still not uniquely defined,
but now the sum in $A^{(d)}$ is finite (since $d-2K\leq m,n,p\leq K:=\max(M,N,P)$).
Hence, we have the following equality
$$
V^d_{\lambda,\mu}=V[\lambda^{\pm1},\mu^{\pm1},(\lambda+\mu)^{-1}]^d\,,
$$
where, as before, the superscript $d$ denotes the subspace consisting of polynomials
in $\lambda^{\pm1},\mu^{\pm1},(\lambda+\mu)^{-1}$, of degree $d$.
\begin{lemma}\label{20110919:lem1}
The following is a basis of the space $V^d_{\lambda,\mu}$ over $V$:
$$
\lambda^{d-i}\mu^i,\,i\in\mb Z
\,\,\,\,;\,\,\,\,\,\,\,\,
\lambda^{d+i}(\lambda+\mu)^{-i},\,i\in\mb Z_{>0}=\{1,2,\dots\}\,,
$$
in the sense that any element of the space $V^d_{\lambda,\mu}$
can be written uniquely as a finite linear combination 
of the above elements with coefficients in $V$.
\end{lemma}
\begin{proof}
First, it suffices to prove the claim for $d=0$.
In this case, letting $t=\mu/\lambda$, the elements of $V^0_{\lambda,\mu}$
are rational functions in $t$ with poles at 0 and -1.
But any such rational functions can be uniquely written,
by partial fractions decomposition, as a linear combination of $t^i$, with $i\in\mb Z$,
and of $(1+t)^{-i}$, with $i\in\mb Z_{>0}$.
\end{proof}
\begin{remark}
One has natural embeddings of $V_{\lambda,\mu}$
in all the vector spaces 
$V((\lambda^{-1}))((\mu^{-1}))$, $V((\mu^{-1}))((\lambda^{-1}))$,
$V((\lambda^{-1}))(((\lambda+\mu)^{-1}))$, $V((\mu^{-1}))(((\lambda+\mu)^{-1}))$,
$V(((\lambda+\mu)^{-1}))((\lambda^{-1}))$, $V(((\lambda+\mu)^{-1}))((\mu^{-1}))$,
defined by expanding one of the variables $\lambda,\mu$ or $\nu=\lambda+\mu$
in terms of the other two.
For example, we have the embedding
\begin{equation}\label{20110919:eq1}
\iota_{\mu,\lambda}:\,V_{\lambda,\mu}\hookrightarrow V((\lambda^{-1}))((\mu^{-1}))\,,
\end{equation}
obtained by expanding all negative powers of $\lambda+\mu$ in the region $|\mu|>|\lambda|$:
\begin{equation}\label{20110919:eq1b}
\iota_{\mu,\lambda}(\lambda+\mu)^{-n-1}
=
\sum_{k=0}^\infty\binom{-n-1}k \lambda^k\mu^{-n-k-1}\,.
\end{equation}
Similarly in all other cases.
Note that, even though $V_{\lambda,\mu}$ is naturally embedded in both spaces
$V((\lambda^{-1}))((\mu^{-1}))$ and $V((\mu^{-1}))((\lambda^{-1}))$,
it is not contained in their intersection $V[[\lambda^{-1},\mu^{-1}]][\lambda,\mu]$.
\end{remark}

\begin{lemma}\label{20111006:lem}
If $V$ is an algebra, then $V_{\lambda,\mu}$ is also an algebra, with the obvious product.
Namely, if $A(\lambda,\mu),B(\lambda,\mu)\in V_{\lambda,\mu}$, then $A(\lambda,\mu)B(\lambda,\mu)\in V_{\lambda,\mu}$.
More generally, 
if $S,T:\,V\to V$ are endomorphisms of $V$ (viewed as a vector space), then
$$
A(\lambda+S,\mu+T)B(\lambda,\mu)\in V_{\lambda,\mu}\,,
$$
where we expand the negative powers of $\lambda+S$ and $\mu+T$ in non-negative powers of $S$ and $T$, 
acting on the coefficients of $B$.
\end{lemma}
\begin{proof}
We expand $A$ and $B$ as in \eqref{20111006:eq1}:
$$
\begin{array}{l}
\displaystyle{
A(\lambda,\mu)=\sum_{m=-\infty}^M\sum_{n=-\infty}^N\sum_{p=-\infty}^P a_{m,n,p}\lambda^m\mu^n(\lambda+\mu)^p\,,
} \\
\displaystyle{
B(\lambda,\mu)=\sum_{m'=-\infty}^{M'}\sum_{n'=-\infty}^{N'}\sum_{p'=-\infty}^{P'} b_{m',n',p'}\lambda^{m'}\mu^{n'}(\lambda+\mu)^{p'}\,.
}
\end{array}
$$
Using the binomial expansion, we then get
$$
A(\lambda+S,\mu+T)B(\lambda,\mu)
=
\sum_{\bar{m}=-\infty}^{M+M'}\sum_{\bar{n}=-\infty}^{N+N'}\sum_{\bar{p}=-\infty}^{P+P'} c_{\bar{m},\bar{n},\bar{p}}
\lambda^{\bar{m}}\mu^{\bar{n}}(\lambda+\mu)^{\bar{p}}\,,
$$
where
$$
\begin{array}{r}
\displaystyle{
c_{\bar{m},\bar{n},\bar{p}}=
\sum_{\substack{m\leq M,m'\leq M',i\geq0 \\ (m+m'-i=\bar{m})}}
\,
\sum_{\substack{n\leq N,n'\leq N',j\geq0 \\ (n+n'-j=\bar{n})}}
\sum_{\substack{p\leq P,p'\leq P',k\geq0 \\ (p+p'-k=\bar{p})}}
} \\
\displaystyle{
\binom{m}{i} \binom{n}{j} \binom{p}{k}
a_{m,n,p} \big(S^{i} T^{j} (S+T)^{k} b_{m',n',p'}\big)\,.
}
\end{array}
$$
To conclude, we just observe that each sum in the RHS is finite,
since, for example, in the first sum we have $i=m+m'-\bar{m}$, $\bar{m}-M'\leq m\leq M$ and $\bar{m}-M\leq m'\leq M'$.
\end{proof}
\begin{lemma}\label{20120131:lem1}
Let $V$ be a vector space and let $U\subset V$ be a subspace.
Then we have:
$$
\begin{array}{l}
\big\{A\in V_{\lambda,\mu}\,\big|\,\iota_{\mu,\lambda}A\in U((\lambda^{-1}))((\mu^{-1}))\big\}
\\
=
\big\{A\in V_{\lambda,\mu}\,\big|\,\iota_{\lambda,\mu}A\in U((\mu^{-1}))((\lambda^{-1}))\big\}
\\
=
\big\{A\in V_{\lambda,\mu}\,\big|\,\iota_{\lambda+\mu,\lambda}A
\in U((\lambda^{-1}))(((\lambda+\mu)^{-1}))\big\}
\\
=
\big\{A\in V_{\lambda,\mu}\,\big|\,\iota_{\lambda+\mu,\mu}A
\in U((\mu^{-1}))(((\lambda+\mu)^{-1}))\big\}
\\
=
\big\{A\in V_{\lambda,\mu}\,\big|\,\iota_{\lambda,\lambda+\mu}A
\in U(((\lambda+\mu)^{-1}))((\lambda^{-1}))\big\}
\\
=
\big\{A\in V_{\lambda,\mu}\,\big|\,\iota_{\mu,\lambda+\mu}A
\in U(((\lambda+\mu)^{-1}))((\mu^{-1}))\big\}
= U_{\lambda,\mu}\,.
\end{array}
$$
\end{lemma}
\begin{proof}
We only need to prove that 
$\big\{A\in V_{\lambda,\mu}\,\big|\,\iota_{\mu,\lambda}A
\in U((\lambda^{-1}))((\mu^{-1}))\big\}\subset U_{\lambda,\mu}$.
Indeed, the opposite inclusion is obvious,
and the argument for the other equalities is the same.

Let $A\in V^d_{\lambda,\mu}$ be such that its expansion 
$\iota_{\mu,\lambda}A\in V((\lambda^{-1}))((\mu^{-1}))$ has coefficients in $U$.
We want to prove that $A$ lies in $U_{\lambda,\mu}$.
By Lemma \ref{20110919:lem1}, $A$ can be written uniquely as
$$
A=\sum_{i=-M}^Nv_i\lambda^{d+i}\mu^{-i}+\sum_{j=1}^N w_j\lambda^{d+j}(\lambda+\mu)^{-j}
\,\,,\,\,\,\,
\text{ with } v_i,w_j\in V
\,.
$$
Its expansion in $V((\lambda^{-1}))((\mu^{-1}))$ is
$$
\iota_{\mu,\lambda}A=
\sum_{i=-M}^Nv_i\lambda^{d+i}\mu^{-i}
+\sum_{j=1}^N\sum_{k=0}^\infty\binom{-j}{k}w_j\lambda^{d+j+k}\mu^{-j-k}\,.
$$
Since, by assumption, $\iota_{\mu,\lambda}A\in U((\lambda^{-1}))((\mu^{-1}))$,
we have
$$
\begin{array}{ll}
v_i\in U &\quad\text{ for }\quad -M\leq i\leq -1 \,,\\
\displaystyle{
v_i+\sum_{j=1}^i \binom{-j}{i-j}w_j
\in U 
} &\quad\text{ for }\quad 0\leq i\leq N  \,,\\
\displaystyle{
\sum_{j=1}^N \binom{-j}{i-j}w_j
\in U 
} &\quad\text{ for }\quad i>N \,.
\end{array}
$$
From the first condition above we have that $v_i$ lies in $U$ for $i<0$.
Using the third condition, we want to deduce that $w_j$ lies in $U$ for all $1\leq j\leq N$.
It then follows, from the second condition, that $v_i$ lies in $U$ for $i\geq0$ as well,
proving the claim.

For $i>N$ and $1\leq j\leq N$ we have $\binom{-j}{i-j}=(-1)^{i-j}\binom{i-1}{j-1}$.
Hence, we will be able to deduce that $w_j$ lies in $U$ for every $j$,
once we prove that the following matrix
$$
P=\left(\,(-1)^{i+j}\binom{i-1}{j-1}\,\right)_{\substack{N+1\leq i<\infty \\ 1\leq j\leq N }}\,,
$$
has rank $N$.
Since the sign $(-1)^{i+j}$ does not change the rank of the above matrix,
it sufficies to prove that the matrices
$$
T_h=\left(\,\binom{i-1}{j-1}\,\right)_{\substack{h+1\leq i\leq h+N \\ 1\leq j\leq N }}\,,
$$
are non-degenerate for every $h\geq0$.
This is clear since the matrix $T_0$ is upper triangular with $1$'s on the diagonal, 
and, by the Tartaglia-Pascal triangle, $T_h$ and $T_{h+1}$ have the same determinant.
\end{proof}

\subsection{Rational pseudodifferential operators}
\label{sec:2.2a}

For the rest of this section, 
let $\mc A$ be a differential algebra, i.e. a unital commutative associative algebra
with a derivation $\partial$, and assume that $\mc A$ is a domain.
For $a\in\mc A$, we denote $a'=\partial(a)$ and $a^{(n)}=\partial^n(a)$, for a non negative integer $n$.
We denote by $\mc K$ the field of fractions of $\mc A$.
Then of course we can extend $\partial$ to a derivation of $\mc K$
making it a differential field.

Recall that a \emph{pseudodifferential operator} over $\mc A$
is an expression of the form
\begin{equation}\label{20111003:eq1}
A=A(\partial)
=\sum_{n=-\infty}^N a_n \partial^n
\,\,,\,\,\,\, a_n\in\mc A\,.
\end{equation}
If $a_N\neq0$, one says that $A$ has \emph{order} $N$.
Pseudodifferential operators form a unital associative algebra, 
denoted by $\mc A((\partial^{-1}))$,
with product $\circ$ defined by letting
\begin{equation}\label{20111130:eq1}
\partial^n\circ a=\sum_{j\in\mb Z_+}\binom nj a^{(j)}\partial^{n-j}
\,\,,\,\,\,\, n\in\mb Z,\, a\in\mc A\,.
\end{equation}
We will often omit $\circ$ if no confusion may arise.

Clearly, $\mc K((\partial^{-1}))$ is a skewfield,
and it is the skewfield of fractions of $\mc A((\partial^{-1}))$.
If $A\in\mc A((\partial^{-1}))$ is a non-zero pseudodifferential operator
of order $N$ as in \eqref{20111003:eq1},
its inverse $A^{-1}\in\mc K((\partial^{-1}))$ is computed as follows.
We write
$$
A
=a_N\Big(1+\sum_{n=-\infty}^{-1} a_N^{-1}a_{n+N} \partial^n\Big)\partial^N\,,
$$
and expanding by geometric progression, we get
\begin{equation}\label{20111130:eq2}
A^{-1}
=
\partial^{-N}\circ \sum_{k=0}^\infty\Big(-\sum_{n=-\infty}^{-1} a_N^{-1}a_{n+N} \partial^n\Big)^k\circ a_N^{-1}\,,
\end{equation}
which is well defined as a pseudodifferential operator in $\mc K((\partial^{-1}))$,
since, by formula \eqref{20111130:eq1},
the powers of $\partial$ are bounded above by $-N$,
and the coefficient of each power of $\partial$ is a finite sum.

The \emph{symbol} of the pseudodifferential operator $A(\partial)$ in \eqref{20111003:eq1}
is the formal Laurent series
$A(\lambda)=\sum_{n=-\infty}^N a_n \lambda^n\,\in\mc A((\lambda^{-1}))$,
where $\lambda$ is an indeterminate commuting with $\mc A$.
We thus get a bijective map $\mc A((\partial^{-1}))\to\mc A((\lambda^{-1}))$
(which is not an algebra homomorphism).
A closed formula for the associative product in $\mc A((\partial^{-1}))$
in terms of the corresponding symbols is the following:
\begin{equation}\label{20111003:eq2}
(A\circ B)(\lambda)=A(\lambda+\partial)B(\lambda)\,.
\end{equation}
Here and further on, we always expand an expression as $(\lambda+\partial)^{n},\,n\in\mb Z$, 
in non-negative powers of $\partial$:
\begin{equation}\label{20111004:eq1}
(\lambda+\partial)^{n}=\sum_{j=0}^\infty\binom nj \lambda^{n-j}\partial^j\,.
\end{equation}
Therefore, the RHS of \eqref{20111003:eq2} means
$\sum_{m,n=-\infty}^N\sum_{j=0}^\infty \binom{m}{j}a_m b_n^{(j)} \lambda^{m+n-j}$.
%
%\begin{remark}\label{20111003:rem}
%It is clear that, if $A(\partial)\in\mc M((\partial^{-1}))^\times$,
%then its leading coefficient $A_N\in\mc M$ is either invertible or a zero divisor.
%%
%As an example of the latter case,
%let $\mc M=\mb F\oplus\mb F$, with $\partial$ acting as zero.
%Then
%$A(\partial)=(0,1)+(1,0)\partial$ is invertible
%and its inverse is $(0,1)+(1,0)\partial^{-1}$.
%\end{remark}

The algebra $\mc A((\partial^{-1}))$
contains the algebra of \emph{differential operators} $\mc A[\partial]$ as a subalgebra.
\begin{definition}\label{20110926:def}
The field $\mc K(\partial)$ of \emph{rational pseudodifferential operators}
is the smallest subskewfield of $\mc K((\partial^{-1}))$ containing $\mc A[\partial]$.
We denote $\mc A(\partial)=\mc K(\partial)\cap\mc A((\partial^{-1}))$,
the subalgebra of \emph{rational pseudodifferential operators with coefficients in} $\mc A$.
\end{definition}
The following Proposition (see \cite[Prop.3.4]{CDSK12a}) describes explicitly 
the skewfield $\mc K(\partial)$ of rational pseudodifferential operators.
\begin{proposition}\label{20111003:thm2}
Let $\mc A$ be a differential domain, and let $\mc K$ be its field of fractions.
\begin{enumerate}[(a)]
\item
%The skewfield $\mc A(\partial)$ of rational pseudodifferential operators over $\mc A$ is
%$$
%\begin{array}{l}
%\vphantom{\Big(}
%\mc A(\partial)
%=\big\{A(\partial)S^{-1}(\partial)\,\big|\, A(\partial),S(\partial)\in\mc A[\partial],\,S(\partial)\neq0\big\}
%\\
%\vphantom{\Big(}
%=\big\{S^{-1}(\partial)A(\partial)\,\big|\, A(\partial),S(\partial)\in\mc A[\partial],\,S(\partial)\neq0\big\}\,.
%\end{array}
%$$
%In other words,
Every rational pseudodifferential operator $L\in\mc K(\partial)$
can be written as a right (resp. left) fraction
$L=AS^{-1}$ (resp. $L=S^{-1}A$)
for some $A,S\in\mc A[\partial]$ with $S\neq0$.
\item
Let $L=AS^{-1}$ (resp. $L=S^{-1}A$), with $A,S\in\mc A[\partial]$, $S\neq0$, 
be a decomposition of $L\in\mc K(\partial)$
such that $S$ has minimal possible order.
Then any other decomposition $L=A_1S_1^{-1}$
(resp. $L=S_1^{-1}A_1$),
with $A_1,S_1\in\mc A[\partial]$,
we have $A_1=AK$, $S_1=SK$ 
(resp. $A_1=KA$, $A_1=KS$), 
for some $K\in\mc K[\partial]$.
\end{enumerate}
\end{proposition}
%
%\begin{remark}\label{20111219:rem2}
%If $\mc A$ is not a domain, one can define the algebra $\mc A(\partial)$ of rational %pseudodifferential operators
%over $\mc A$ as the smallest subalgebra of $\text{Fr}\mc A((\partial^{-1}))$,
%where $\text{Fr}\mc A$ is the algebra of fractions of $\mc A$,
%containing $\mc A[\partial]$ and the inverses of all elements $S(\partial)\in\mc A[\partial]$
%which are invertible in $\text{Fr}\mc A((\partial^{-1}))$.
%\end{remark}

\subsection{Rational matrix pseudodifferential operators}
\label{sec:2.3}

\begin{definition}\label{20111013:def}
A matrix pseudodifferential operator $A\in\Mat_{\ell\times\ell}\mc A((\partial^{-1}))$
is called \emph{rational with coefficients in $\mc A$}  
if its entries are rational pseudodifferential operators with coefficients in $\mc A$.
In other words, 
the algebra of rational matrix pseudodifferential operators with coefficients in $\mc A$ 
is $\Mat_{\ell\times\ell}\mc A(\partial)$.
\end{definition}
Let $M=\big(A_{ij}B_{ij}^{-1}\big)_{i,j\in I}$
be a rational matrix pseudodifferential operator with coefficients in $\mc A$, 
with $A_{ij},B_{ij}\in\mc A[\partial]$.
By the Ore condition (see e.g. \cite{CDSK12a}), 
we can find a common right multiple $B\in\mc A[\partial]$
of all operators $B_{ij}$,
i.e. for every $i,j$ we can factor $B=B_{ij}C_{ij}$ for some $C_{ij}\in\mc A[\partial]$.
Hence, $A_{ij}B_{ij}^{-1}=\tilde A_{ij}B^{-1}$, where $\tilde A_{ij}=A_{ij}C_{ij}$.
Then, the matrix $M$ can be represented as a ratio of two matrices:
$M=\tilde A (B\id)^{-1}$.
Hence,
$$
\Mat{}_{\ell\times\ell}\mc A(\partial)
=
\left\{A(B\id)^{-1}\,\left|\,
\begin{array}{c}
A\in\Mat{}_{\ell\times\ell}\mc A[\partial],\,B\in\mc A[\partial],\\
A_{ij}B^{-1}\in\mc A((\partial^{-1}))\,\,\forall i,j
\end{array}
\right.\right\}\,.
$$
However, in general this is not a representation of the rational matrix $M$ in ``minimal terms''
(see Definition \ref{def:minimal-fraction} below).

We recall now some linear algebra over the skewfield $\mc K((\partial^{-1}))$
and, in particular, the notion of the Dieudonn\'e determinant
(see \cite{Art57} for an overview over an arbitrary skewfield).

An \emph{elementary row operation} of an $\ell\times\ell$ matrix pseudodifferential operator 
$A$ is either a permutation of two rows of it,
or the operation $\mc T(i,j;P)$, where $1\leq i\neq j\leq m$
and $P\in\mc K((\partial^{-1}))$,
which replaces the $j$-th row by itself minus $i$-th row
multiplied on the left by $P$.
Using the usual Gauss elimination, we can get the (well known) analogues
of standard linear algebra theorems for matrix pseudodifferential operators.
In particular, 
any matrix pseudodifferential operator $A\in\Mat_{m\times\ell}\mc K((\partial^{-1}))$
can be brought by elementary row operations to a row echelon form.

The \emph{Dieudonn\'e determinant} of a $A\in\Mat_{\ell\times\ell}\mc K((\partial^{-1}))$
has the form $\det A=c\xi^d$, where $c\in\mc A$, $\xi$ is an indeterminate, 
and $d\in\mb Z$.
It is defined by the following properties:
$\det A$ changes sign if we permute two rows of $A$,
and it is unchanged under any elementary row operation $\mc T(i,j;P)$ defined above,
for aribtrary $i\neq j$ and a pseudodifferential operator $P\in\mc K((\partial^{-1}))$;
moreover, if $A$ is upper triangular,
with diagonal entries $A_{ii}$ of order $n_i$ and leading coefficoent $a_i$,
then 
$$
\det A=\Big(\prod_i a_i\Big) \xi^{\sum_in_i}\,.
$$
It was proved in \cite{Die43} (for any skewfield) that the Dieudonn\'e determinant is well defined
and $\det(A B)=(\det A)(\det B)$
for every $\ell\times\ell$ matrix pseudodifferential operators 
$A,B\in\Mat_{\ell\times\ell}\mc K((\partial^{-1}))$.

The Dieudonn\'e determinant gives a way to characterize invertible matrix pseudodifferential operators,
thanks to the following well known fact (see e.g. \cite[Prop.4.3]{CDSK12a}):
\begin{proposition}\label{20111005:prop2}
Let $\mc D$ be a subskewfield of the skewfield $\mc K\!((\partial^{-1}))$,\
and let $A\in\!\Mat_{\ell\times\ell}\mc D$.
Then $A$ is invertible in $\Mat_{\ell\times\ell}\mc D$ if and only if $\det A\neq0$.
\end{proposition}
\begin{corollary}\label{20111005:prop3}
Let $A\in\Mat_{\ell\times\ell}\mc K((\partial^{-1}))$
be a matrix with $\det A\neq0$.
Then $A$ is a rational matrix if and only if $A^{-1}$ is.
\end{corollary}
\begin{proof}
It is a special case of Proposition \ref{20111005:prop2} when $\mc D$
is the subskewfield $\mc K(\partial)\subset\mc K((\partial^{-1}))$ of rational pseudodifferential operators.
\end{proof}
\begin{remark}\label{20111216:rem}
It is proved in \cite{CDSK12a} that,
if $A\in\Mat_{\ell\times\ell}\mc A((\partial^{-1}))$
then we have $\det A=c\xi^d$, with $c\in\bar{\mc A}$, 
where $\bar{\mc A}$ is the integral closure of $\mc A$.
Furthermore, if $c$ is an invertible element of $\bar{\mc A}$,
then the inverse matrix $A^{-1}$
lies in $\Mat_{\ell\times\ell}\bar{\mc A}((\partial^{-1}))$.
\end{remark}

\begin{definition}\label{def:non-deg}
Let $\mc A$ be a differential domain.
An $\ell\times\ell$-matrix pseudodifferential operator $A\in\Mat_{\ell\times\ell}\mc A((\partial^{-1}))$
is called \emph{non-degenerate}
if it has non-zero Dieudonn\`e determinant,
or, equivalently, if it is invertible in the ring $\Mat_{\ell\times\ell}\mc K((\partial^{-1}))$
of pseudodifferential operators with coefficients in the differential field of fractions $\mc K$ of $\mc A$.
\end{definition}

\begin{definition}[see \cite{CDSK12b}]\label{def:minimal-fraction}
%\label{20120201:rem1}
Let $H\in\Mat_{\ell\times\ell}\mc K(\partial)$ be a rational matrix pseudodifferential 
operator with coefficients in the differential field $\mc K$.
A fractional decomposition $H=A B^{-1}$,
with $A,B\in\Mat_{\ell\times\ell}\mc K[\partial]$ and $B$ non-degenerate,
is called \emph{minimal} if $\deg_\xi \det B$ is minimal
(recall that it is a non-negative integer).
\end{definition}
\begin{proposition}[\cite{CDSK12b}]\label{prop:minimal-fraction}
\begin{enumerate}[(a)]
\item
A fractional decomposition $H\!\!=\!\!A B^{-1}$
of a rational matrix pseudodifferential operator $H\in\Mat_{\ell\times\ell}\mc K(\partial)$
is minimal if and only if
\begin{equation}\label{20120124:eq3}
\ker A\cap\ker B=0\,,
\end{equation}
in any differential field extension of $\mc K$.
\item
The minimal fractional decomposition of $H$ exists and is
unique up to multiplication of $A$ and $B$ on the right 
by a matrix differential operator $D$
which is invertible in the algebra $\Mat_{\ell\times\ell}\mc K[\partial]$.
Any other fractional decomposition of $H$ is obtained by multiplying $A$ and $B$
on the right by a non-degenerate matrix differential operator.
\end{enumerate}
\end{proposition}
%
%\begin{remark}\label{rem:minimal-fraction}
%In the case $\ell=1$ the fractional decomposition $H=A B^{-1}\in\mc K(\partial)$,
%is minimal if and only if the the differential operators $A,B\in\mc K[\partial]$
%have no right common divisor of order greater than 0
%(i.e. the right greatest common divisor of $A$ and $B$ is 1).
%\end{remark}
%
\begin{remark}\label{rem:minimal-fraction2}
Let $\mc A$ be a differential domain, and let $\mc K$ be its field of fractions.
A fractional decomposition $H=AB^{-1}$ of $H\in\Mat_{\ell\times\ell}\mc A[\partial]$ over $\mc K$
can be turned into a fractional decomposition over $\mc A$ by clearing the denominators 
of $A$ and $B$.
Hence, a minimal fractional decomposition $H=AB^{-1}$ over $\mc A$,
in the sense that it has minimal possible $\deg_\xi\det B$ among all
fractional decompositions of $H$ with $A,B\in\Mat_{\ell\times\ell}\mc A[\partial]$,
is also minimal over $\mc K$.
\end{remark}

%%%%%%%%%%%%%%%%%%%%%%%%%%%%%%%%%%%%%%%%%%%%%%%%%%%%%%%%%%%%%%%%%%%%%%%%%%%%%%%%%%%%%%%%%%%%%%%%%%%%%%%%%%%%%%
%%%%%%%%%%%%%%% Sect 3 %%%%%%%%%%%%%%%%%%%%%%%%%%%%%%%%%%%%%%%%%%%%%%%%%%%%%%%%%%%%%%%%%%%%%%%%%%%%%%%%%%%%%%%
%%%%%%%%%%%%%%%%%%%%%%%%%%%%%%%%%%%%%%%%%%%%%%%%%%%%%%%%%%%%%%%%%%%%%%%%%%%%%%%%%%%%%%%%%%%%%%%%%%%%%%%%%%%%%%

\section{Non-local Poisson vertex algebras}
\label{sec:3}

\subsection{Non-local $\lambda$-brackets and non-local Lie conformal algebras}
\label{sec:3.1}

Let $R$ be a module over the algebra of polynomials $\mb F[\partial]$.
\begin{definition}\label{20110919:def1}
A \emph{non-local} $\lambda$-\emph{bracket} on $R$ is a linear map
$\{\cdot\,_\lambda\,\cdot\}:\,R\otimes R\to R((\lambda^{-1}))$
satisfying the following \emph{sesquilinearity} conditions ($a,b\in R$):
\begin{equation}\label{20110921:eq1}
\{\partial a_\lambda b\}=-\lambda\{a_\lambda b\}
\,\,,\,\,\,\,
\{a_\lambda\partial b\}=(\lambda+\partial)\{a_\lambda b\}\,.
\end{equation}
The non-local $\lambda$-bracket $\{\cdot\,_\lambda\,\cdot\}$ is said to be \emph{skewsymmetric} 
(respectively \emph{symmetric})
if ($a,b\in R$)
\begin{equation}\label{20110921:eq2}
\{b_\lambda a\}=-\{a_{-\lambda-\partial}b\}
\quad \Big(\text{ resp. } =\{a_{-\lambda-\partial}b\}\Big)
\,.
\end{equation}
\end{definition}
The RHS of the skewsymmetry condition should be interpreted as follows:
if $\{a_\lambda b\}=\sum_{n=-\infty}^Nc_n\lambda^n$, then
$$
\begin{array}{c}
\displaystyle{
\{a_{-\lambda-\partial} b\}
=
\sum_{n=-\infty}^N(-\lambda-\partial)^n c_n
=
\sum_{n=-\infty}^N\sum_{k=0}^\infty\binom{n}{k}(-1)^n(\partial^k c_n)\lambda^{n-k} 
}\\
\displaystyle{
= \sum_{m=-\infty}^N\Big(\sum_{k=0}^{N-m}\binom{m+k}{k}(-1)^{m+k}(\partial^k c_{m+k})\Big)\lambda^m
\,.
}
\end{array}
$$
In other words, we move $-\lambda-\partial$ to the left and
we expand in non negative powers of $\partial$ as in \eqref{20111004:eq1}.

In general we have $\{a_\lambda\{b_\mu c\}\}\in R((\lambda^{-1}))((\mu^{-1}))$
for an arbitrary $\lambda$-bracket $\{\cdot\,_\lambda\,\cdot\}$.
Recall from Section \ref{sec:2.1} that $R_{\lambda,\mu}$ can be considered
as a subspace of $R((\lambda^{-1}))((\mu^{-1}))$ via the embedding $\iota_{\mu,\lambda}$.
\begin{definition}\label{20110919:def2}
The non-local $\lambda$-bracket $\{\cdot\,_\lambda\,\cdot\}$ on $R$
is called \emph{admissible} if
$$
\{a_\lambda\{b_\mu c\}\}\in R_{\lambda,\mu}
\qquad\forall a,b,c\in R\,.
$$
\end{definition}
\begin{remark}\label{20110919:rem}
If $\{\cdot\,_\lambda\,\cdot\}$ is a skewsymmetric admissible $\lambda$-bracket on $R$,
then $\{b_\mu\{a_\lambda c\}\}\in R_{\lambda,\mu}$ and $\{\{a_\lambda b\}_{\lambda+\mu} c\}\in R_{\lambda,\mu}$
for all $a,b,c\in R$.
Indeed, the first claim is obvious since $R_{\lambda,\mu}=R_{\mu,\lambda}$.
For the second claim, by skewsymmetry 
$\{\{a_\lambda b\}_{\lambda+\mu} c\}=-\{c_{-\lambda-\mu-\partial}\{a_\lambda b\}\}$,
and by the admissibility assumption $\{c_\nu\{a_\lambda b\}\}\in R_{\lambda,\nu}$.
To conclude it suffices to note that when replacing $\nu$ by $-\lambda-\mu-\partial$
in an element of $R_{\lambda,\nu}=R[[\lambda^{-1},\nu^{-1},(\lambda+\nu)^{-1}]][\lambda,\nu]$,
we have that $\nu^{-1}$ is expanded in negative powers of $\lambda+\mu$
and $(\lambda+\nu)^{-1}$ is expanded in negative powers of $\mu$.
As a result, we get an element of $R[[\lambda^{-1},\mu^{-1},(\lambda+\mu)^{-1}]][\lambda,\mu]=R_{\lambda,\mu}$.
\end{remark}
\begin{definition}\label{20110921:def1}
A \emph{non-local Lie conformal algebra} is an $\mb F[\partial]$-module $R$
endowed with an admissible skewsymmetric $\lambda$-bracket
$\{\cdot\,_\lambda\,\cdot\}:\,R\otimes R\to R((\lambda^{-1}))$
satisfying the Jacobi identity (in $R_{\lambda,\mu}$):
\begin{equation}\label{20110922:eq3}
\{a_\lambda\{b_\mu c\}\}-\{b_\mu\{a_\lambda c\}\}=\{\{a_\lambda b\}_{\lambda+\mu} c\}
\,\,\,\,\text{ for every } a,b,c\in R\,.
\end{equation}
\end{definition}
\begin{example}\label{20110921:ex1}
Let $R=\big(\mb F[\partial]\otimes V\big)\oplus\mb FC$,
where $V$ is a vector space with a symmetric bilinear form $(\cdot\,|\,\cdot)$.
Define the (non-local) $\lambda$-bracket on $R$ by
letting $C$ be a central element, defining
$$
\{a_\lambda b\}=(a|b)C\lambda^{-1}
\,\,\,\,\text{ for } a,b\in V\,,
$$
and extending it to a $\lambda$-bracket on $R$ by sesquilinearity.
Skewsymmetry for this $\lambda$-bracket holds since, by assumption, $(\cdot\,|\,\cdot)$ is symmetric.
Moreover, since any triple $\lambda$-bracket is zero,
the $\lambda$-bracket is obviously admissible and it satisfies the Jacobi identity.
Hence, we have a non-local Lie conformal algebra.
\end{example}

\subsection{Non-local Poisson vertex algebras}
\label{sec:3.2}

Let $\mc V$ be a differential algebra, i.e. a unital commutative associative algebra
with a derivation $\partial:\,\mc V\to\mc V$.
As before, we assume that $\mc V$ is a domain and denote by $\mc K$ its field of fractions.
\begin{definition}\label{20110921:def2}
\begin{enumerate}[(a)]
\item 
A \emph{non-local} $\lambda$-\emph{bracket} on the differential algebra $\mc V$ is a linear map
$\{\cdot\,_\lambda\,\cdot\}:\,\mc V\otimes \mc V\to \mc V((\lambda^{-1}))$
satisfying the sesquilinearity conditions \eqref{20110921:eq1}
and the following left and right \emph{Leibniz rules}:
\begin{equation}\label{20110921:eq3}
\begin{array}{l}
\{a_\lambda bc\}=b\{a_\lambda c\}+c\{a_\lambda b\}\,, \\
\{ab_\lambda c\}=\{a_{\lambda+\partial}c\}_\to b+\{b_{\lambda+\partial} c\}_\to a\,.
\end{array}
\end{equation}
Here and further an expression $\{a_{\lambda+\partial}b\}_\to c$ is interpreted as follows:
if $\{a_{\lambda}b\}=\sum_{n=-\infty}^Nc_n\lambda^n$, 
then $\{a_{\lambda+\partial}b\}_\to c=\sum_{n=-\infty}^Nc_n(\lambda+\partial)^nc$,
where we expand $(\lambda+\partial)^nc$ in non-negative powers of $\partial$ as in \eqref{20111004:eq1}.
\item
The conditions of (\emph{skew})\emph{symmetry}, \emph{admissibility} and \emph{Jacobi identity} 
for a non-local $\lambda$-bracket $\{\cdot\,_\lambda\,\cdot\}$ on $\mc V$
are the same as in Definitions \ref{20110919:def1}, \ref{20110919:def2} and \ref{20110921:def1} respectively.
\item
A \emph{non-local Poisson vertex algebra} is a differential algebra $\mc V$
endowed with a \emph{non-local Poisson} $\lambda$-\emph{bracket},
i.e. a skewsymetric admissible non-local $\lambda$-bracket,
satisfying the Jacobi identity.
\end{enumerate}
\end{definition}
\begin{example}[cf. Example \ref{20110921:ex1}]\label{20110921:ex2}
Let $\mc V=\mb F[u_i^{(n)}\,|\,i=1,\dots,\ell,n\in\mb Z_+]$ 
be the algebra of diffenertial polynoamials in $\ell$ differential variables $u_1,\dots,u_\ell$.
Let $C=\big(c_{ij}\big)_{i,j=1}^\ell$ be an $\ell\times\ell$ symmetric matrix 
with coefficients in $\mb F$.
The following formula defines a structure of a non-local Poisson vertex algebra on $\mc V$:
$$
\{P_\lambda Q\}
=
\sum_{m,n\in\mb Z_+}\sum_{i,j\in\mb Z_+} c_{ij}
\frac{\partial Q}{\partial u_j^{(n)}} (-1)^m 
(\lambda+\partial)^{m+n-1} 
\frac{\partial P}{\partial u_i^{(m)}}\,.
$$
For example, $\{{u_i}_\lambda {u_j}\}=c_{ij}\lambda^{-1}$ but, for 
$P,Q\in\mb F[u_1,\dots,u_\ell]\subset\mc V$,
we get an infinite formal Laurent series in $\lambda^{-1}$:
$$
\begin{array}{l}
\displaystyle{
\{P_\lambda Q\}
=
\sum_{i,j=1}^\ell c_{ij} \frac{\partial Q}{\partial u_j} 
(\lambda+\partial)^{-1} \frac{\partial P}{\partial u_i}
} \\
\displaystyle{
=\sum_{i,j=1}^\ell \sum_{n=0}^\infty (-1)^n  \frac{\partial Q}{\partial u_j}
\Big(\partial^n \frac{\partial P}{\partial u_i}\Big) \lambda^{-n-1}
\in\mc V((\lambda^{-1}))\,.
}
\end{array}
$$
We will prove that this is indeed a non-local Poisson $\lambda$-bracket 
in the next section,
where we will discuss a general construction of non-local Poisson vertex algebras,
which will include this example as a special case
(see Theorem \ref{20110923:prop}).
\end{example}

\begin{proposition}\label{20111219:prop}
Let $\{\cdot\,_\lambda\,\cdot\}$ be a non-local Poisson vertex algebra structure on the differential domain $\mc V$.
Then there is a unique way to extend it to a non-local Poisson vertex algebra structure
on the differential field of fractions $\mc K$,
and it can be computed using the following formulas ($a,b\in\mc K\backslash\{0\}$):
\begin{equation}\label{20111219:eq1}
\{a_\lambda b^{-1}\}=-b^{-2}\{a_\lambda b\}
\,\,,\,\,\,\,
\{a^{-1}_\lambda b\}=-\{a_{\lambda+\partial} b\}_\to a^{-2}\,.
\end{equation}
\end{proposition}
\begin{proof}
It is straightforward to check that formulas \eqref{20111219:eq1} define
a non-local $\lambda$-bracket on the field of fraction $\mc K$,
satysfying all the axioms of non-local Poisson vertex algebra.
In particular, admissibility of the $\lambda$-bracket can be derived from Lemma \ref{20111006:lem}.
The details of the proof are left to the reader.
\end{proof}
%
%\begin{remark}\label{20111219:rem}
%When $\mc V$ is an algebra of differential functions with a non-local Poisson vertex algebra structure
%given by the Matser Formula \eqref{20110922:eq1},
%then Proposition \ref{20111219:prop} follows from Theorem \ref{20110923:prop}
%and the fact that the field of fractions $\mc K$ of $\mc V$ is also an algebra of differential functions
%in the same variables $u_i,\,i=1,\dots,\ell$.
%\end{remark}

Thanks to Proposition \ref{20111219:prop}
we can extend, uniquely, a non-local Poisson vertex algebra $\lambda$-bracket
on $\mc V$ to its field of fractions $\mc K$.
The following results are useful to prove admissibility of a non-local $\lambda$-bracket.
\begin{lemma}\label{20111012:lem}
Let $\mc V$ be a differential algebra, endowed 
with a non-local $\lambda$-bracket $\{\cdot\,_\lambda\,\}$.
Assume that $\mc V$ is a domain, and let $\mc K$ be its field of fractions.
Let $S=\big(S_{ij}\big)_{i,j\in I}\in\Mat_{\ell\times\ell}\big(\mc K((\partial^{-1}))\big)$
be an invertible $\ell\times\ell$ matrix pseudodifferential operator with coefficients in $\mc K$.
Letting $S_{ij}=\sum_{n=-\infty}^N s_{ij;n}\partial^n$,
the following identities hold for every $a\in\mc K$ and $i,j\in I$:
\begin{equation}\label{20111012:eq2a}
\begin{array}{r}
\displaystyle{
\big\{a_\lambda (S^{-1})_{ij}(\mu)\big\}
=
-\sum_{r,t=1}^\ell\sum_{n=-\infty}^N
\iota_{\mu,\lambda}(S^{-1})_{ir}(\lambda+\mu+\partial)
} \\
\displaystyle{
\{a_\lambda s_{rt;n}\} (\mu+\partial)^n (S^{-1})_{tj}(\mu)
\,\in\mc K((\lambda^{-1}))((\mu^{-1}))
\,,
}
\end{array}
\end{equation}
and
\begin{equation}\label{20111012:eq2b}
\begin{array}{l}
\displaystyle{
\big\{(S^{-1})_{ij}(\lambda) _{\lambda+\mu} a\big\}
=
-\sum_{r,t=1}^\ell\sum_{n=-\infty}^N
\{{s_{rt;n}}_{\lambda+\mu+\partial}a\}_\to
} \\
\displaystyle{
\Big((\lambda+\partial)^n (S^{-1})_{tj}(\lambda)\Big)
\iota_{\lambda,\lambda+\mu}({S^*}^{-1})_{ri}(\mu) 
\,\in\mc K(((\lambda+\mu)^{-1}))((\lambda^{-1}))
\,,
}
\end{array}
\end{equation}
where $\iota_{\mu,\lambda}$ and $\iota_{\lambda,\lambda+\mu}$ 
are as in \eqref{20110919:eq1b}.
In equation \eqref{20111012:eq2b} $S^*$ denotes the adjoint 
of the matrix pseudodifferential operator $S$
(its inverse being $(S^{-1})^*$).
\end{lemma}
\begin{proof}
The identity $S\circ S^{-1}=1$ becomes, in terms of symbols,
$$
\sum_{t=1}^\ell S_{r,t}(\mu+\partial)(S^{-1})_{tj}(\mu)=\delta_{rj}\,.
$$
Taking $\lambda$-bracket with $a$, we have, by sesquilinearity and the (left) Leibniz rule,
$$
\begin{array}{l}
\displaystyle{
0=\sum_{t=1}^\ell \big\{a_\lambda S_{rt}(\mu+\partial)(S^{-1})_{tj}(\mu)\big\}
} \\
\displaystyle{
=
\sum_{t=1}^\ell\sum_{n=-\infty}^N 
\big\{a_\lambda s_{rt;n} (\mu+\partial)^n (S^{-1})_{tj}(\mu)\big\}
} \\
\displaystyle{
=
\sum_{t=1}^\ell\sum_{n=-\infty}^N \{a_\lambda s_{rt;n}\} (\mu+\partial)^n (S^{-1})_{tj}(\mu)
} \\
\displaystyle{
+\sum_{t=1}^\ell \iota_{\mu,\lambda} S_{rt}(\lambda+\mu+\partial) \big\{a_\lambda (S^{-1})_{tj}(\mu)\big\}
\,.
}
\end{array}
$$
Note that $\iota_{\mu,\lambda} S(\lambda+\mu+\partial)$
is invertible in
$\Mat_{\ell\times\ell}\big(\mc K[\partial]((\lambda^{-1}))((\mu^{-1}))\big)$,
its inverse being $\iota_{\mu,\lambda} S^{-1}(\lambda+\mu+\partial)$.
We then apply $\iota_{\mu,\lambda}(S^{-1})_{ir}(\lambda+\mu+\partial)$ on the left to both sides of the above equation
and we sum over $r=1,\dots,\ell$, to get
$$
\begin{array}{l}
\displaystyle{
\sum_{t=1}^\ell \delta_{it} \big\{a_\lambda (S^{-1})_{tj}(\mu)\big\}
} \\
\displaystyle{
=-
\sum_{r=1}^\ell
\sum_{t=1}^\ell\sum_{n=-\infty}^N \iota_{\mu,\lambda}(S^{-1})_{ir}(\lambda+\mu+\partial) \{a_\lambda s_{rt;n}\} (\mu+\partial)^n (S^{-1})_{tj}(\mu)\,,
}
\end{array}$$
proving equation \eqref{20111012:eq2a}.

Similarly, for the second equation we have, by the right Leibniz rule,
$$
\begin{array}{l}
\displaystyle{
0=\sum_{t=1}^\ell \big\{S_{rt}(\lambda+\partial)(S^{-1})_{tj}(\lambda)\,_{\lambda+\mu}a\big\}
} \\
\displaystyle{
=
\sum_{t=1}^\ell \sum_{n=-\infty}^N \big\{s_{rt;n}(\lambda+\partial)^n(S^{-1})_{tj}(\lambda)_{\lambda+\mu}a\big\}
} \\
\displaystyle{
=
\sum_{t=1}^\ell \sum_{n=-\infty}^N \big\{{s_{rt;n}}_{\lambda+\mu+\partial}a\big\}_\to (\lambda+\partial)^n(S^{-1})_{tj}(\lambda)
} \\
\displaystyle{
+ \sum_{t=1}^\ell \sum_{n=-\infty}^N \big\{(S^{-1})_{tj}(\lambda)_{\lambda+\mu+\partial}a\big\}_\to 
\iota_{\lambda,\lambda+\mu}(\lambda-\lambda-\mu-\partial)^ns_{rt;n}
} \\
\displaystyle{
=
\sum_{t=1}^\ell \sum_{n=-\infty}^N \big\{{s_{rt;n}}_{\lambda+\mu+\partial}a\big\}_\to (\lambda+\partial)^n(S^{-1})_{tj}(\lambda)
} \\
\displaystyle{
+ \sum_{t=1}^\ell \big\{(S^{-1})_{tj}(\lambda)_{\lambda+\mu+\partial}a\big\}_\to \iota_{\lambda,\lambda+\mu} S^*_{tr}(\mu)
\,.
}
\end{array}
$$
We next replace in the above equation $\mu$ (placed at the right) by $\mu+\partial$,
and we apply the resulting differential operator to $\iota_{\lambda,\lambda+\mu}({S^*}^{-1})_{ri}(\mu)$.
As a result we get, after summing over $r=1,\dots,\ell$,
$$
\begin{array}{l}
\displaystyle{
\sum_{t=1}^\ell 
\big\{(S^{-1})_{tj}(\lambda)_{\lambda+\mu+\partial}a\big\}_\to \delta_{ti}
} \\
\displaystyle{
=
-\sum_{r=1}^\ell
\sum_{t=1}^\ell 
\sum_{n=0}^N 
\big\{{s_{rt;n}}_{\lambda+\mu+\partial}a\big\}_\to 
\Big((\lambda+\partial)^n(S^{-1})_{tj}(\lambda)\Big)
\iota_{\lambda,\lambda+\mu}({S^*}^{-1})_{ri}(\mu)
\,,
}
\end{array}
$$
proving equation \eqref{20111012:eq2b}.
\end{proof}
\begin{corollary}\label{20111014:cor}
Let $\mc V$ be a differential algebra, endowed 
with a non-local $\lambda$-bracket $\{\cdot\,_\lambda\,\cdot\}$.
Assume that $\mc V$ is a domain, and let $\mc K$ be its field of fractions.
Let $S=\big(S_{ij}\big)_{i,j\in I}\in\Mat_{\ell\times\ell}\big(\mc K[\partial]\big)$
be non-degenerate (cf. Definition \ref{def:non-deg}).
Then the following identities hold for every $a\in\mc K$ and $i,j\in I$:
\begin{equation}\label{20111012:eq2c}
\begin{array}{l}
\displaystyle{
\big\{a_\lambda (S^{-1})_{ij}(\mu)\big\}
} \\
\displaystyle{
=
-\sum_{r,t=1}^\ell\sum_{n=0}^N
(S^{-1})_{ir}(\lambda+\mu+\partial) \{a_\lambda s_{rt;n}\} (\mu+\partial)^n (S^{-1})_{tj}(\mu)
\,\in\mc K_{\lambda,\mu}
\,,
}
\end{array}
\end{equation}
and
\begin{equation}\label{20111012:eq2d}
\begin{array}{l}
\displaystyle{
\big\{(S^{-1})_{ij}(\lambda) _{\lambda+\mu} a\big\}
} \\
\displaystyle{
=
-\sum_{r,t=1}^\ell\sum_{n=0}^N
\{{s_{rt;n}}_{\lambda+\mu+\partial}a\}_\to
\Big((\lambda+\partial)^n (S^{-1})_{tj}(\lambda)\Big)
({S^*}^{-1})_{ri}(\mu) 
\,\in\mc K_{\lambda,\mu}
\,,
}
\end{array}
\end{equation}
where $S_{ij}=\sum_{n=0}^N s_{ij;n}\partial^n$.
\end{corollary}
\begin{proof}
It is immediate from equations \eqref{20111012:eq2a} and \eqref{20111012:eq2b}.
\end{proof}
\begin{corollary}\label{20111007:prop}
Let $\mc V$ be a differential algebra, endowed 
with a non-local $\lambda$-bracket $\{\cdot\,_\lambda\,\cdot\}$.
Assume that $\mc V$ is a domain, and let $\mc K$ be its field of fractions.
Let $A\in\mc V(\partial)=\mc K(\partial)\cap\mc V((\partial^{-1}))$ 
be a rational pseudodifferential operator with coefficients in $\mc V$.
Then 
$\{a_\lambda A(\mu)\}$
and $\{A(\lambda)_{\lambda+\mu} a\}$
lie in $\mc V_{\lambda,\mu}$
for every $a\in\mc V$.
In particular, the $\lambda$-bracket is admissible.
\end{corollary}
\begin{proof}
First, note that if the pseudodifferential operators $A,B\in\mc K((\partial^{-1}))$ 
satisfy the conditions
$$
\{a_\lambda A(\mu)\}
\,,\{A(\lambda)_{\lambda+\mu} a\}
\{a_\lambda B(\mu)\}
\,,\{B(\lambda)_{\lambda+\mu} a\}
\,\,
\in\mc K_{\lambda,\mu}\,,
$$
for every $a\in\mc K$,
so does $A B$.
Indeed, by the Leibniz rule,
$$
\begin{array}{l}
\displaystyle{
\{a_\lambda (A B)(\mu)\}
=
\{a_\lambda A(\mu+\partial) B(\mu)\}
} \\
\displaystyle{
=
\{a_\lambda A(\mu+\partial)\}_{\to} B(\mu)
+ A(\lambda+\mu+\partial) \{a_\lambda B(\mu)\}\,,
}
\end{array}
$$
and both terms in the RHS lie in $\mc K_{\lambda,\mu}$ by the assumption on $A$ and $B$,
thanks to Lemma \ref{20111006:lem}.
Similarly, by the right Leibniz rule,
$$
\begin{array}{l}
\displaystyle{
\{ (A B)(\lambda) _{\lambda+\mu} a \}
=
\{ A(\lambda+\partial) B(\lambda) _{\lambda+\mu} a \}
} \\
\displaystyle{
=
\{ B(\lambda) _{\lambda+\mu+\partial} a \}_\to 
\iota_{\lambda,\lambda+\mu}A^*(\mu)
+
\{ A(\lambda+\partial) _{\lambda+\mu+\partial} a \}_\to B(\lambda)
\,,
}
\end{array}
$$
and both terms in the RHS lie in $\mc K_{\lambda,\mu}$ 
(rather in the image of $\mc K_{\lambda,\mu}$ in $\mc K(((\lambda+\mu)^{-1}))((\lambda^{-1}))$ 
via $\iota_{\lambda,\lambda+\mu}$) by Lemma \ref{20111006:lem}.
By Corollary \ref{20111014:cor}
we have that, if $S\in\mc V[\partial]$, then 
$\{a_\lambda S^{-1}(\mu)\}$ and $\{S^{-1}(\lambda)_{\lambda+\mu} a\}$
lie in $\mc K_{\lambda,\mu}$ for all $a\in\mc K$.
Hence, 
by Definition \ref{20110926:def} and the above observations,
we get that, if $A\in\mc V(\partial)=\mc K(\partial)\cap\mc V((\partial^{-1}))$,
then 
$\{a_\lambda A(\mu)\}$ and $\{A(\lambda)_{\lambda+\mu} a\}$
lie in $\mc K_{\lambda,\mu}$
for all $a\in\mc K$.
On the other hand, if $a\in\mc V$, we clearly have
$\{a_\lambda A(\mu)\}\in\mc V((\lambda^{-1}))((\mu^{-1}))$ 
and $\{A(\lambda)_{\lambda+\mu} a\}\in\mc V(((\lambda+\mu)^{-1}))((\lambda^{-1}))$.
The claim follows from Lemma \ref{20120131:lem1}
applied to $V=\mc K$ and $U=\mc V$.
\end{proof}
\begin{remark}\label{20111104:rem}
In the case when $S\in\mc V(\partial)$ 
is a rational pseudodifferential operator with coefficients in $\mc V$,
thanks to Corollary \ref{20111007:prop},
we can drop $\iota_{\mu,\lambda}$ and $\iota_{\lambda,\lambda+\mu}$ respectively from
equations \eqref{20111012:eq2a} and \eqref{20111012:eq2b},
which hold in the space $\mc V_{\lambda,\mu}$.
\end{remark}

%%%%%%%%%%%%%%%%%%%%%%%%%%%%%%%%%%%%%%%%%%%%%%%%%%%%%%%%%%%%%%%%%%%%%%%%%%%%%%%%%%%%%%%%%%%%%%%%%%%%%%%%%%%%%%
%%%%%%%%%%%%%%% Sect 4 %%%%%%%%%%%%%%%%%%%%%%%%%%%%%%%%%%%%%%%%%%%%%%%%%%%%%%%%%%%%%%%%%%%%%%%%%%%%%%%%%%%%%%%
%%%%%%%%%%%%%%%%%%%%%%%%%%%%%%%%%%%%%%%%%%%%%%%%%%%%%%%%%%%%%%%%%%%%%%%%%%%%%%%%%%%%%%%%%%%%%%%%%%%%%%%%%%%%%%

\section{Non-local Poisson structures}
\label{sec:4}

\subsection{Algebras of differential functions}
\label{sec:4.1}

Let $R_\ell=\mb F [u_i^{(n)}\, |\, i \in I,n \in \mb Z_+]$ be the algebra of differential polynomials
in the $\ell$ variables $u_i,\,i\in I=\{1,\dots,\ell\}$,
with the derivation $\partial$ defined by $\partial (u_i^{(n)}) = u^{(n+1)}_i$.
The partial derivatives $\frac{\partial}{\partial u_i^{(n)}}$
are commuting derivations of $R_\ell$, and they satisfy the following 
commutation relations with $\partial$:
\begin{equation}\label{eq:0.4}
\left[  \frac{\partial}{\partial u_i^{(n)}}, \partial \right] = \frac{\partial}{\partial u_i^{(n-1)}}
\,\,\,\,
\text{ (the RHS is $0$ if $n=0$) }\,.
\end{equation}

Recall from \cite{BDSK09} that an \emph{algebra of differential functions} 
is a differential algebra extension $\mc V$ of $R_\ell$,
endowed with commuting derivations
$$
\frac{\partial}{\partial u_i^{(n)}}:\,\mc V\to\mc V
\,\,,\,\,\,\,
i \in I ,\,n \in \mb Z_+\,,
$$
extending the usual partial derivatives on $R_\ell$,
such that only a finite number of
$\frac{\partial f}{\partial u_i^{(n)}}$ are non-zero for each $f\in \mc V$, 
and such that the commutation rules \eqref{eq:0.4} hold on $\mc V$.

Given an algebra of differential functions $\mc V$,
its  localization by a multiplicative subset is again an algebra of differential functions.
Also, we can add to $\mc V$ solutions of algebraic equations over $\mc V$,
or functions of the form $F(\varphi_1,\dots,\varphi_n)$,
where $F$ is an infinitely differentiable function in $n$ variables
and $\varphi_1,\dots,\varphi_n$ lie in $\mc V$,
to obtain again an algebra of differential functions.
(Note, though, that in general we cannot add to $\mc V$ solutions of linear differential equations.
For example, a solution of the equation $f'=fu'$ is $f=e^u$, which can be added,
while a non-zero solution of the equation $f'=fu$ can never be added,
due to simple differential order considerations.)
We will use these facts in the examples further on.

\begin{remark}\label{20130111:rem}
An algebra of differential functions $\mc V$ can be equivalently defined as follows.
It is a commutative associative algebra extension of $R_\ell$,
endowed with commuting derivations
$\frac{\partial}{\partial x}$ and $\frac{\partial}{\partial u_i^{(n)}},\,i\in I,n\in\mb Z_+$,
such that $\frac{\partial}{\partial x}$ acts trivially on $R_\ell$,
$\frac{\partial}{\partial u_i^{(n)}}$ extends the usual action of partial derivatives on $R_\ell$,
and for every $f\in\mc V$ we have $\frac{\partial f}{\partial u_i^{(n)}}=0$
for all but finitely many choices of indices $(i,n)$.
In this case, 
the total derivative $\partial:\,\mc V\to\mc V$
defined by the formula
$$
\partial f=\sum_{i\in I,n\in\mb Z_+}\frac{\partial f}{\partial u_i^{(n)}}u_i^{(n+1)}+\frac{\partial f}{\partial x}\,,
$$
satisfies equation \eqref{eq:0.4}
\end{remark}

Note that if the algebra of differential functions $\mc V$ is a domain,
then its field of fractions $\mc K$ is again an algebra of differential functions
in the same variables $u_1,\dots,u_\ell$,
with the maps $\frac{\partial}{\partial u_i^{(n)}}:\,\mc K\to\mc K$ defined in the obvious way.
When $\mc V=\mc K$ we call it a \emph{field of differential functions}.

We denote by $\mc C=\big\{c\in\mc V\,\big|\,\partial c=0\big\}\subset\mc V$ the subalgebra of \emph{constants},
and by
$$
\mc F=\Big\{f\in\mc V\,\Big|\,\frac{\partial f}{\partial u_i^{(n)}}=0
\,\,\text{ for all }\, i\in I,n\in\mb Z_+\Big\}\subset\mc V
$$
the subalgebra of \emph{quasiconstants}. It is easy to see that $\mc C\subset\mc F$.

Given $f\in\mc V$ which is not a quasiconstant, we say that is has \emph{differential order} $N$ if
$\frac{\partial f}{\partial u_i^{(N)}}\neq0$ for some $i\in I$,
and $\frac{\partial f}{\partial u_j^{(n)}}=0$ for every $j\in I$ and $n>N$.
We also set the differential order of a quasiconstant element equal to $-\infty$.
We let $\mc V_N$ be the subalgebra of elements of differential order at most $N$.
This gives an increasing sequence of subalgebras 
\begin{equation}\label{20130111:eq1}
\mc C\subset\mc F=\mc V_{-\infty}\subset\mc V_0\subset\mc V_1\subset\dots\subset\mc V\,,
\end{equation}
such that $\partial\mc V_N\subset\mc V_{N+1}$.
Clearly, if $\mc V$ is a field of differential functions,
then this is a tower of field extensions.

It is easy to show, using \eqref{eq:0.4}, that
\begin{equation}\label{20120907:eq1}
\partial\mc V\cap\mc V_N=\partial\mc V_{N-1}
\,\text{ for } N\geq 1\,,
\,\,\text{ and }\,\,
\partial\mc V\cap\mc V_0=\partial\mc F
\,.
\end{equation}

\subsection{Normal extensions}
\label{sec:4.1a}

We refine the filtration \eqref{20130111:eq1}
to a filtration 
$\mc V_m=\mc V_{m,0}\subset\mc V_{m,1}\subset\dots\subset\mc V_{m,\ell}=\mc V_{m+1}$, 
where
\begin{equation}\label{20130111:eq2}
V_{m,i}=\Big\{f\in\mc V_m\,\Big|\,
\frac{\partial f}{\partial u_j^{(m)}}=0 \text{ for all } j>i
\Big\}
\subset\mc V_m\,.
\end{equation}
Clearly, each subspace $\mc V_{m,i}$ is preserved 
by all partial derivatives $\frac{\partial}{\partial u_j^{(n)}}$ for $(n,j)\leq(m,i)$
(in lexicographic order),
and it is annihilated by $\frac{\partial}{\partial u_j^{(n)}}$ for $(n,j)>(m,i)$.
\begin{definition}\label{20130111:def}
The algebra of differential functions $\mc V$ is said to be \emph{normal}
if the map $\frac{\partial}{\partial u_i^{(m)}}:\,\mc V_{m,i}\to\mc V_{m,i}$
is surjective for every $i\in I,m\in\mb Z_+$.
\end{definition}

\begin{lemma}\label{20130112:lem}
Any algebra of differential function $\mc V$ can be extended to a normal one,
which can be taken to be a domain provided that $\mc V$ is.
\end{lemma}
\begin{proof}
Given an algebra of differential functions $\mc V$
and an element $a\in\mc V_{m,i}$ which is not in the image 
of $\frac{\partial}{\partial u_i^{(m)}}$,
one can construct an algebra of differential functions $\tilde{\mc V}$ extension of $\mc V$
with an element $A\in\tilde{\mc V}_{m,i}$ such that $\frac{\partial A}{\partial u_i^{(m)}}=a$.
For example, we can take the algebra of polynomials in infinitely many variables
$$
\tilde{\mc V}=\mc V\bigg[\frac{\partial^{k+k_{0,1}+\dots+k_{m,i-1}}A}
{\partial x^{k}\partial {u_1^{(0)}}^{k_{0,1}}\dots\partial {u_{i-1}^{(m)}}^{k_{m,i-1}}}
\,\,\bigg|\,\,k,k_{0,1},\dots,k_{m,i-1}\in\mb Z_+\bigg]\,,
$$
and define on it a structure of algebra of differential functions by
letting
$\frac{\partial}{\partial x}$ and $\frac{\partial}{\partial u_j^{(n)}}$ for $(n,j)\leq(m,i-1)$
act on 
$\frac{\partial^{k+k_{0,1}+\dots+k_{m,i-1}}A}
{\partial x^{k}\partial {u_1^{(0)}}^{k_{0,1}}\dots\partial {u_{i-1}^{(m)}}^{k_{m,i-1}}}$
in the obvious way (suggested by the notation used to denote the new variables),
letting
$\frac{\partial}{\partial u_j^{(n)}}$ for $(n,j)>(m,i)$
act as zero on 
$\frac{\partial^{k+k_{0,1}+\dots+k_{m,i-1}}A}
{\partial x^{k}\partial {u_1^{(0)}}^{k_{0,1}}\dots\partial {u_{i-1}^{(m)}}^{k_{m,i-1}}}$,
and letting
$$
\frac{\partial}{\partial u_i^{(m)}}
\Big(
\frac{\partial^{k+k_{0,1}+\dots+k_{m,i-1}}A}
{\partial x^{k}\partial {u_1^{(0)}}^{k_{0,1}}\dots\partial {u_{i-1}^{(m)}}^{k_{m,i-1}}}
\Big)
=
\frac{\partial^{k+k_{0,1}+\dots+k_{m,i-1}}a}
{\partial x^{k}\partial {u_1^{(0)}}^{k_{0,1}}\dots\partial {u_{i-1}^{(m)}}^{k_{m,i-1}}}
\,.
$$
The lemma follows by standard arguments using Zorn's Lemma.
\end{proof}

%
%\pecetta{
%If $a\in\mc V_{m,i}$, is it possible to extend $\mc V$ by adding $A=\tint du_i^{(m)}a$
%in a way that $\mc V_{m,i-1}$ does not change?
%(the above construction does not work).
%}

\begin{example}\label{20130112:ex1}
The algebra $\mc F\big[u_i^{(n)},\,i\in I,n\in\mb Z_+\big]$
of differential polynomials over a differential algebra $\mc F$
is a normal algebra of differential functions in the variables $u_i,\,\in I$ 
(we can always integrate polynomials).
\end{example}

\begin{example}\label{20130112:ex2}
The algebra of differential functions 
$\mc V=\mc F[u^{\pm1},u^{(n)},\,n\geq1]$,
is not normal, since $u^{-1}$ is not in the image of $\frac{\partial}{\partial u}$.
A normal extension of it is
$\tilde{\mc V}=\mc F[u^{\pm1},u^{(n)},\,n\geq1,\log u]$.
Indeed a preimage via $\frac{\partial}{\partial u}$ of $u^m(\log u)^n$, 
$m\in\mb Z\backslash\{-1\},n\in\mb Z_+$,
is obtained by integration by parts
$$
\tint du\,u^m(\log u)^n=\frac{1}{m+1}u^{m+1}(\log u)^n-\frac{n}{m+1}\tint du\,u^m(\log u)^{n-1}\,,
$$
and a preimage via $\frac{\partial}{\partial u}$ of $u^{-1}(\log u)^n$ is $\frac1{n+1}(\log u)^{n+1}$.
Similarly, $\mc F[u^{(n)},\,n\in\mb Z_+,{u^{(s)}}^{-1},\log u^{(s)}]$
is normal for every $s$.
\end{example}

\begin{example}\label{20130112:ex3}
Other examples of normal algebras of differential functions are
$\mc F[u^{(n)},\,n\in\mb Z_+,e^{\pm u}]$,
since we can always integrate by parts $P(u)e^{nu},\,n\in\mb Z$, 
for a polynomial $P(u)$.
\end{example}

\subsection{Variational complex}
\label{sec:4.1.a.5}

For $f\in\mc V$, as usual we denote by $\tint f$
the image of $f$ in the quotient space $\mc V/\partial\mc V$.
Recall that, by \eqref{eq:0.4}, we have a well-defined variational derivative
$\frac{\delta}{\delta u}:\,\mc V/\partial\mc V\to\mc V^{\oplus\ell}$,
given by
$$
\frac{\delta\tint f}{\delta u_i}=\sum_{n\in\mb Z_+}(-\partial)^n\frac{\partial f}{\partial u_i^{(n)}},\,i\in I\,.
$$

Given a set $J$ and an element $X\in\mc V^J$,
we define the \emph{Frechet derivative} of $X$ as the differential operator
$D_X(\partial):\,\mc V^\ell\to\mc V^J$ given by
\begin{equation}\label{20111020:eq1}
\big(D_X(\partial)P\big)_j
=\sum_{n\in\mb Z_+}\sum_{i\in I}\frac{\partial X_j}{\partial u_i^{(n)}} \partial^n P_i
\,\,,\,\,\,\,j\in J
\,.
\end{equation}
Its adjoint operator is the map $D_X^*(\partial):\,\mc V^{\oplus J}\to\mc V^{\oplus\ell}$ given by:
\begin{equation}\label{20111020:eq2}
\big(D_X^*(\partial)Y\big)_i
=\sum_{n\in\mb Z_+}\sum_{j\in J}(-\partial)^n\Big(\frac{\partial X_j}{\partial u_i^{(n)}} Y_j\Big)
\,\,,\,\,\,\,i\in I\,.
\end{equation}
Here and further, for a possibly infinite set $J$, $\mc V^{\oplus J}$
denotes the space of column vectors in $\mc V^J$ with only finitely many non-zero entries.
(Though in this paper we do not consider any example with infinite $\ell$,
we still distinguish $\mc V^\ell$ and $\mc V^{\oplus\ell}$ as a book-keeping device.)

The following identity can be checked directly and it will be useful later:
\begin{equation}\label{20120405:eq1}
\tint X\cdot D_Y(\partial)P+\tint Y\cdot D_X(\partial)P
=\tint P\cdot \frac{\delta}{\delta u}(X\cdot Y)\,,
\end{equation}
for all $X\in\mc V^J,\,Y\in\mc V^{\oplus J},\,P\in\mc V^\ell$.

The above notions are linked naturally in the variational complex:
$$
0\to\mc F/\frac{\partial}{\partial x}\mc F\to\mc V/\partial\mc V
\stackrel{\frac{\delta}{\delta u}}{\longrightarrow}
\mc V^{\oplus\ell}
\stackrel{\delta}{\longrightarrow}
\Sigma_\ell
\to\dots
$$
where $\Sigma_\ell$ is the space of skewadjoint $\ell\times\ell$ 
matrix differential operators over $\mc V$,
and $\delta(F)=D_F(\partial)-D^*_F(\partial)$,
for $F\in\mc V^{\oplus\ell}$.
The construction of the whole complex can be found in \cite{DSK09},
but we will not need it here.
In \cite{BDSK09} it is proved that the variational complex
is exact, provided that the algebra of differential functions $\mc V$ is normal.
In particular, if $\mc V$ is normal we have 
that $\ker\big(\frac\delta{\delta u}\big)=\mc F+\partial\mc V$,
and that $F\in\mc V^{\oplus\ell}$ is \emph{closed}, i.e. its Frechet derivative $D_F(\partial)$ is selfadjoint,
if and only if it is \emph{exact}, i.e. $F\in\frac{\delta}{\delta u}\mc V^{\oplus\ell}$.

\subsection{Differential orders}
\label{sec:4.1.b}

Given an arbitrary $k\times \ell$-matrix $A$ with entries in $\mc V$,
we define its differential order, denoted by $\dord(A)$,
as the maximal differential order of all its entries.

Given a matrix differential operator $D=\sum_{i=0}^n A_i\partial^i\in\Mat_{k\times\ell}\mc V[\partial]$,
we define its \emph{differential order} as
\begin{equation}\label{20120910:eq1}
\dord(D)=\max\{\dord(A_1),\dots,\dord(A_n)\}\,,
\end{equation}
which should not be confused with its \emph{order},
defined as 
\begin{equation}\label{20120910:eq2}
|D|=n \,\,\text{ if }\, A_n\neq0\,.
\end{equation}
(Note that the notion of order carries over to matrix pseudofferential operators,
while the differential order is not defined in general.)
\begin{lemma}\label{20120910:lem1}
Let $D\in\Mat_{k\times\ell}\mc V[\partial]$
be a matrix  differential operator over $\mc V$
and let $F\in\mc V^\ell$.
Then:
\begin{enumerate}[(a)]
\item
$\dord(DF)\leq\max\{\dord(D),\dord(F)+|D|\}$.
\item
If $D$ has non-degenerate leading coefficient
(meaning that its determinant is not a zero divisor in $\mc V$)
and it satisfies $\dord(F)+|D|>\dord(D)$, then $\dord(DF)=\dord(F)+|D|$.
\item
If $D$ has non-degenerate leading coefficient
and it satisfies $\dord(DF)>\dord(D)$, then $\dord(DF)=\dord(F)+|D|$.
\end{enumerate}
\end{lemma}
\begin{proof}
Let $D=\sum_{s=0}^n A_s\partial^s$.
Clearly, for $f\in\mc V$ and $s\in\mb Z_+$, we have $\dord(f^{(i)})=\dord(f)+i$.
Hence, If  $h>\max\{\dord(D),\dord(F)+|D|\}$, we have
$$
\frac{\partial}{\partial u^{(h)}}(DF)_i=
\sum_{j=1}^\ell\sum_{s=0}^n \frac{\partial}{\partial u^{(h)}} (A_s)_{ij}F^{(s)}_j=0\,,
$$
for every $i=1,\dots,k$, proving part (a).
Furthermore, 
if $|D|=n$ and $\dord(F)+n>\dord|D|$, we can use \eqref{eq:0.4} to get
$$
\begin{array}{l}
\displaystyle{
\frac{\partial}{\partial u^{(\dord(F)+n)}}(DF)_i
=\sum_{j=1}^\ell\sum_{s=0}^n \frac{\partial}{\partial u^{(\dord(F)+n)}} (A_s)_{ij} F^{(s)}_j
} \\
\displaystyle{
=\sum_{j=1}^\ell (A_n)_{ij} \frac{\partial F^{(n)}_j}{\partial u^{(\dord(F)+n)}}
=\sum_{j=1}^\ell (A_n)_{ij} \frac{\partial F_j}{\partial u^{(\dord(F))}}
\,.
}
\end{array}
$$
Since, by assumption,
the leading coefficient $A_n\in\Mat_{\ell\times\ell}\mc V$ of $D$ is non-degenerate,
the RHS above is non-zero for some $i$.
Hence, 
$\dord(DF)=\dord(F)+n$, proving part (b).
Part (c) follows from parts (a) and (b).
\end{proof}

\subsection{Construction of non-local Poisson structures}
\label{sec:4.2}

Let $\mc V$ be an algebra of differential functions in $u_1,\dots,u_\ell$.
Assume that $\mc V$ is a domain, and let $\mc K$ be the corresponding
field of fractions, which is a field of differential functions.
Let $H=\big(H_{ij}\big)_{i,j\in I}\in\Mat_{\ell\times\ell}\mc V((\partial^{-1}))$
be an $\ell\times\ell$ matrix pseudodifferential operator over $\mc V$,
namely
$$
H_{ij}=\sum_{n=-\infty}^NH_{ij;n}\partial^n
\,\,\in\mc V((\partial^{-1}))
\,\,,\,\,\,\,
i,j\in I\,.
$$
We associate to this matrix $H$ a non-local $\lambda$-bracket on $\mc V$
given by the following \emph{Master Formula} (cf. \cite{DSK06})
\begin{equation}\label{20110922:eq1}
\{f_\lambda g\}_H
=
\sum_{\substack{i,j\in I \\ m,n\in\mb Z_+}} 
\frac{\partial g}{\partial u_j^{(n)}}
(\lambda+\partial)^n
H_{ji}(\lambda+\partial)
(-\lambda-\partial)^m
\frac{\partial f}{\partial u_i^{(m)}}
\,\in\mc V((\lambda^{-1}))
\,.
\end{equation}
In particular
\begin{equation}\label{20110923:eq1}
\{{u_i}_\lambda{u_j}\}_H
=
H_{ji}(\lambda)
\,\,,\,\,\,i,j\in I\,.
\end{equation}

The following result gives a way to check if a matrix pseudodifferential operator 
$H\in\Mat_{\ell\times\ell}\mc V((\partial^{-1}))$
defines a structure of non-local Poisson vertex algebra on $\mc V$.
The analogous statement in the local case was proved in \cite{BDSK09}.
\begin{theorem}\label{20110923:prop}
Let $\mc V$ be an algebra of differential functions, which is a domain,
ane let $\mc K$ be its field of fractions.
Let $H\in\Mat_{\ell\times\ell}\mc V((\partial^{-1}))$.
Then:
\begin{enumerate}[(a)]
\item
Formula \eqref{20110922:eq1} gives a well-defined non-local $\lambda$-bracket on $\mc V$.
\item
This non-local $\lambda$-bracket is skewsymmetric if and only if $H$
is a skew-adjoint matrix pseudodifferential operator.
\item
If $H=\big(H_{ij}\big)_{i,j\in I}\in\Mat_{\ell\times\ell}\mc V(\partial)$
is a rational matrix pseudodifferential operator with coefficients in $\mc V$,
then the corresponding non-local $\lambda$-bracket 
$\{\cdot\,_\lambda\,\cdot\}_H:\,\mc V\times\mc V\to\mc V((\lambda^{-1}))$
(given by equation \eqref{20110922:eq1}) is admissible.
\item
Let $H=\big(H_{ij}\big)_{i,j\in I}\in\Mat_{\ell\times\ell}\mc V(\partial)$
be a skewadjoint rational matrix pseudodifferential operator with coefficients in $\mc V$.
Then the non-local $\lambda$-bracket $\{\cdot\,_\lambda\,\cdot\}_H$ defined by \eqref{20110922:eq1}
is a Poisson non-local $\lambda$-bracket, i.e. it satisfies the Jacobi identity \eqref{20110922:eq3},
if and only if the Jacobi identity holds on generators ($i,j,k\in I$):
\begin{equation}\label{20110922:eq4}
\{{u_i}_\lambda\{{u_j}_\mu {u_k}\}_H\}_H-\{{u_j}_\mu\{{u_i}_\lambda {u_k}\}_H\}_H
-\{{\{{u_i}_\lambda {u_j}\}_H}_{\lambda+\mu} {u_k}\}_H=0\,,
%\,\,\,\,\text{ for every } i,j,k\in I\,,
\end{equation}
where the equality holds in the space $\mc V_{\lambda,\mu}$.
\end{enumerate}
\end{theorem}
\begin{proof}
For the proofs of (a), (b) and (d) one does the same computations 
as in the proof of \cite[Thm.1.15]{BDSK09} for the local case.
So, we only prove part (c).
Let $a,f,g\in\mc V$. By the Master Formula \eqref{20110922:eq1}
and the left Leibniz rule, we have
$$
\begin{array}{l}
\displaystyle{
\{a_\lambda\{f_\mu g\}_H\}_H
=
\sum_{\substack{i,j\in I \\ m,n\in\mb Z_+}} 
\big\{
a_\lambda
\frac{\partial g}{\partial u_j^{(n)}}
(\mu+\partial)^n
H_{ji}(\mu+\partial)
(-\mu-\partial)^m
\frac{\partial f}{\partial u_i^{(m)}}
\big\}_H
} \\
\displaystyle{
=
\sum_{\substack{i,j\in I \\ m,n\in\mb Z_+}} 
\big\{
a_\lambda
\frac{\partial g}{\partial u_j^{(n)}}
\big\}_H
(\mu+\partial)^n
H_{ji}(\mu+\partial)
(-\mu-\partial)^m
\frac{\partial f}{\partial u_i^{(m)}}
} \\
\displaystyle{
+
\sum_{\substack{i,j\in I \\ m,n\in\mb Z_+}} 
\frac{\partial g}{\partial u_j^{(n)}}
(\lambda+\mu+\partial)^n
{\{
a_\lambda
H_{ji}(\mu+\partial)
\}_H}_\to
(-\mu-\partial)^m
\frac{\partial f}{\partial u_i^{(m)}}
} \\
\displaystyle{
+
\sum_{\substack{i,j\in I \\ m,n\in\mb Z_+}} 
\frac{\partial g}{\partial u_j^{(n)}}
(\lambda+\mu+\partial)^n
H_{ji}(\lambda+\mu+\partial)
(-\lambda-\mu-\partial)^m
\big\{
a_\lambda
\frac{\partial f}{\partial u_i^{(m)}}
\big\}_H\,.
}
\end{array}
$$
All sums in the above equations are finite.
Therefore, all three terms in the RHS lie in $\mc V_{\lambda,\mu}$,
thanks to Corollary \ref{20111007:prop} and Lemma \ref{20111006:lem}.
\end{proof}
%
%\begin{remark}\label{20111219:rem1}
%If $\mc V$ is an algebra, not necessarily a field, of differential functions,
%the Master Formula \eqref{20110922:eq1} still defines a non-local $\lambda$-bracket on $\mc V$.
%%
%To prove admissibility of this $\lambda$-bracket, one needs to generalize the notion of rational
%pseudodifferential operators to this case (see Remark \ref{20111219:rem2}).
%\end{remark}
%
\begin{definition}\label{20111007:def}
Let $\mc V$ be an algebra of differential functions.
A \emph{non-local Poisson structure} on $\mc V$
is  a skewadjoint rational matrix pseudodifferential operator with coefficients in $\mc V$,
$H=\big(H_{ij}\big)_{i,j\in I}\in\Mat_{\ell\times\ell}\mc V(\partial)$,
satisfying equation \eqref{20110922:eq4} for every $i,j,k\in I$.
\end{definition}
\begin{remark}\label{20111007:rem2}
It is easy to show that, if $L\in\mc K(\partial)$ is a rational pseudodifferential operator,
then it can be expanded as
\begin{equation}\label{20110922:eq2}
L=\sum_{s=1}^\infty
\sum_{n=0}^N
\sum_{\substack{p_1,\dots,p_s\in\mc V_M \\ (\text{finite sum})}}
p_1\partial^{-1}\circ p_2\partial^{-1}\circ\dots\partial^{-1}\circ p_s\partial^n
\,,
\end{equation}
for some fixed $M,N\in\mb Z_+$.
To see this, write $L=AS^{-1}$, where $A,S\in\mc V[\partial]$
and $S=\sum_{n=0}^Ns_n\partial^n$ has non-zero leading coefficient $s_N$,
and expand $S^{-1}$ using geometric progression:
\begin{equation}\label{20111007:eq1}
S^{-1}=
\partial^{-N}\sum_{i=0}^\infty\Big(-s_N^{-1}s_{N-1}\partial^{-1}-\dots-s_N^{-1}s_0\partial^{-N}\Big)^i
\circ s_N^{-1}\,.
\end{equation}
On the other hand, it is not hard to see that if $L$ admits an expansion as in \eqref{20110922:eq2},
then $\{a_\lambda L(\mu)\}_H\in K_{\lambda,\mu}$ for every $a\in\mc K$
and every matrix pseudodifferential operator $H$.
As a consequence, if all the entries of a matrix pseudodifferential operator $H$
admit an expansion as in \eqref{20110922:eq2}, 
then the corresponding $\lambda$-bracket $\{\cdot\,_\lambda\,\cdot\}_H$ on $\mc K$
is admissible.
\end{remark}
\begin{remark}\label{20111019:def}
It is claimed in the literature (without a proof) \cite{DN89}
that, in order to show that a skewadjoint operator $H$ defines a (local) Poisson structure,
it suffices to check the Jacobi identity 
for the Lie bracket $\{\cdot\,,\,\cdot\}_H=\{\cdot\,_\lambda\,\cdot\}_H\big|_{\lambda=0}$
in $\mc V/\partial\mc V$ on triples of elements of the form $\tint f u_i$, where $f\in\mc F$ is a quasiconstant.
This is indeed true, provided that the algebra of quasiconstants $\mc F$ is ``big enough'', by the following argument.
By a straightforward computation, using the Master Formula, we get
$$
\begin{array}{l}
\displaystyle{
\{\tint fu_i,\{\tint gu_j,\tint hu_k\}_H\}_H
-\{\tint gu_j,\{\tint fu_i,\tint hu_k\}_H\}_H
} \\
\displaystyle{
-\{\{\tint fu_i,\tint gu_j\}_H,\tint hu_k\}_H
=
\tint h
\Big(
\{{u_i}_\lambda\{{u_j}_\mu{u_k}\}_H\}_H
} \\
\displaystyle{
-\{{u_j}_\mu\{{u_i}_\lambda{u_k}\}_H\}_H
-\{{\{{u_i}_\lambda{u_j}\}_H}_{\lambda+\mu}{u_k}\}_H
\Big)
\big(|_{\lambda=\partial}f\big)
\big(|_{\mu=\partial}g\big)\,.
}
\end{array}
$$
Clearly, this is zero for all $f,g,h\in\mc F$ and all $i,j,k\in I$
if and only if $H$ is a Poisson structure,
provided that the algebra $\mc F$ satisfies the following non-degeneracy conditions:
\begin{enumerate}[(i)]
\item 
if $\tint ha=0$ for some $a\in\mc V$ and all $h\in\mc F$, then $a=0$,
\item
if $P(\partial)f=0$ for some differential operator $P\in\mc V[\partial]$ 
and for all $f\in\mc F$, then $P=0$.
\end{enumerate}
Obviously, $\mc F$ fulfills these conditions if it contains the algebra of polynomials $\mb F[x]$.
Often in the literature this criterion is used also for non-local Poisson structures,
which does not seem to have much sense, since in the non-local case $\mc V/\partial\mc V$
does not have a Lie algebra structure.
\end{remark}

%%%
\subsection{Examples}
\label{sec:4.3}

\begin{example}\label{20111010:ex1}
Let $\mc V$ be any algebra of differential functions in $\ell$ differential variables, 
with subalgebra of quasiconstants $\mc F\subset\mc V$.
Any skewadjoint rational matrix pseudodifferential operator 
with quasiconstant coefficients,
$H=\big(H_{ij}(\partial)\big)_{ij\in I}\in\Mat_{\ell\times\ell}\mc F(\partial)$,
is a Poisson structure.
Indeed, by askewadjointness of $H$ the $\lambda$-bracket $\{\cdot\,_\lambda\,\cdot\}_H$
is skewsymmetric, and by the Master Formula, all triple $\lambda$-brackets are zero.
Note that, if $H\in\Mat_{\ell\times\ell}\mc F((\partial^{-1}))$ is skewadjoint, 
even if it is not a rational matrix,
the corresponding $\lambda$-bracket $\{\cdot\,_\lambda\,\cdot\}_H$ 
is still admissible,
hence it defines a non-local Poisson vertex algebra on $\mc V$.

In the special case when $H_{ij}(\lambda)=c_{ij}\lambda^{-1}$,
and $C=(c_{ij})_{i,j=1}^\ell$ is a symmetric matrix with constant coefficients,
we recover the non-local Poisson vertex algebras from Example \ref{20110921:ex2}.
When $C$ if a symmetrized Cartan matrix or extended Cartan matrix of a simple Lie algebra,
we get the Poisson structure for a Toda lattice (see \cite{Fr98}).
\end{example}
\begin{example}\label{20110922:ex1}
The following three operators form a compatible family of non-local Poisson structures
(i.e. any their linear combination is a non-local Poisson structure)
on the algebra $R_1=\mb F[u,u',u'',\dots]$
of differential polynomials in one variable:
\begin{enumerate}[(i)]
\item
$K_{1}=\partial$ (GFZ Hamiltonina structure),
\item
$K_{-1}=\partial^{-1}$ (Toda non-local Poisson structure),
\item
$H=u'\partial^{-1}\circ u'$ (Sokolov non-local Hamitonian structure),
\end{enumerate}
First, any linear combination over $\mc C$ of $K_1$ and $K_{-1}$
is a non-local Poisson structure, as discussed in Example \ref{20111010:ex1}.
Next, it is easy to show (cf. \cite[Example 3.14]{BDSK09})
that $H^{-1}$ is a symplectic structure over the field of fractions $\mc K_1=\text{Frac} R_1$,
known as the Sokolov symplectic structure, \cite{Sok84}.
Hence, by Theorem \ref{20111012:thm} below,
we deduce that $H$ is a non-local Poisson structure.
To conclude that $K_1,K_{-1},H$ form a compatible family,
it suffices to check that
\begin{equation}\label{20130617:eq1}
\{u_\lambda H(\mu)\}_{K_{\pm1}}
-\{u_\mu H(\lambda)\}_{K_{\pm1}}
=\{H(\lambda)_{\lambda+\mu}u\}_{K_{\pm1}}\,,
\end{equation}
where $H(\lambda)=u'(\partial+\lambda)^{-1}u'\in\mc V((\lambda^{-1}))$.
This is straightforward,
but we shall perform the computation in order to demonstrate how it works.
We have
$$
\begin{array}{l}
\displaystyle{
\vphantom{\Big(}
\{u_\lambda H(\mu)\}_{K_{\pm1}}
=\{u_\lambda u'(\partial+\mu)^{-1}u'\}_{K_{\pm1}}
} \\
\displaystyle{
\vphantom{\Big(}
=\Big((\partial+\lambda)\{u_\lambda u\}_{K_{\pm1}}\Big)(\partial+\mu)^{-1}u'
+u'(\partial+\lambda+\mu)^{-1}(\partial+\lambda)\{u_\lambda u\}_{K_{\pm1}}
} \\
\displaystyle{
\vphantom{\Big(}
=\lambda^{1\pm1}(\partial+\mu)^{-1}u'
+u'(\lambda+\mu)^{-1}\lambda^{1\pm1}
\,.}
\end{array}
$$
In the second identity we used the Leibniz rule and sesquilinearity,
and in the last identity we used the definition of the $K_{\pm1}$-$\lambda$-bracket.
Hence, the LHS of \eqref{20130617:eq1} equals
\begin{equation}\label{20130617:eq2}
\lambda^{1\pm1}(\partial+\mu)^{-1}u'
-\mu^{1\pm1}(\partial+\lambda)^{-1}u'
+u'(\lambda+\mu)^{-1}(\lambda^{1\pm1}-\mu^{1\pm1})
\,.
\end{equation}
Similarly, for the RHS of \eqref{20130617:eq1} we have
\begin{equation}\label{20130617:eq3}
\begin{array}{l}
\displaystyle{
\vphantom{\Big(}
\{H(\lambda)_{\lambda+\mu} u\}_{K_{\pm1}}
=\{u'(\partial+\lambda)^{-1}u'_{\lambda+\mu}u\}_{K_{\pm1}}
} \\
\displaystyle{
\vphantom{\Big(}
=
-{\{u_{\lambda+\mu+\partial} u\}_{K_{\pm1}}}_\to
(\lambda+\mu+\partial)
\Big((\partial+\lambda)^{-1}u'
+(-\partial-\mu)^{-1}u'\Big)
} \\
\displaystyle{
\vphantom{\Big(}
=(\lambda+\mu+\partial)^{1\pm1}
\big(
-(\partial+\lambda)^{-1}u'
+(\partial+\mu)^{-1}u'
\big)
\,.}
\end{array}
\end{equation}
It is then immediate to check that \eqref{20130617:eq2} and the RHS of \eqref{20130617:eq3} 
are equal.
\end{example}
\begin{example}\label{20110922:ex2}
Dorfman non-local Poisson structure 
on the algebra of differential polynomials
$R_1=\mb F[u,u',u'',\dots]$ is:
$$
H=\partial^{-1}\circ u'\partial^{-1}\circ u'\partial^{-1}\,.
$$
One easily shows (cf. \cite[Example 3.14]{BDSK09})
that $H^{-1}$ is a symplectic structure over the field of fractions $\mc K_1=\text{Frac} R_1$, 
known as Dorfman symplectic structre, \cite{Dor93},
hence $H$ is indeed a non-local Poisson structure.
Furthermore, one can show, by a lengthy calculation,
that Sokolov's and Dorfman's non-local Poisson structures
are compatible.
\end{example}
\begin{example}[cf. \cite{Dor93}]\label{20110922:ex3}
Another triple of compatible non-local Poisson structures
on $R_1=\mb F[u,u',u'',\dots]$ is:
\begin{enumerate}[(i)]
\item
$K_{1}=\partial$ (GFZ Poisson structure),
\item
$K_{-1}=\partial^{-1}$ (Toda non-local Poisson structure),
\item
$H=\partial^{-1}\circ u'+u'\partial^{-1}$ (potential Virasoro-Magri non-local Poisson structure).
\end{enumerate}
\end{example}
\begin{example}[cf. \cite{Mag80}]\label{20110922:ex4}
There is yet another triple of compatible non-local Poisson structures
on $R_1=\mb F[u,u',u'',\dots]$:
\begin{enumerate}[(i)]
\item
$K_1=\partial$ (GFZ Poisson structure),
\item
$K_3=\partial^3$,
\item
$H=\partial\circ u\partial^{-1}\circ u\partial$ 
(modified Virasoro-Magri non-local Poisson structure).
\end{enumerate}
\end{example}
\begin{example}[cf. \cite{Mag78,Mag80}]\label{20110922:ex5}
The following is a triple of compatible non-local Poisson structures
on $R_2=\mb F[u,v,u',v',\dots]$:
\begin{enumerate}[(i)]
\item
$K_1=\partial\id$ (GFZ Poisson structure),
\item
$K=\left(\begin{array}{cc} 0 & -1 \\ 1 & 0 \end{array}\right)$,
\item
$H=\left(\begin{array}{cc} 
v\partial^{-1}\circ v & -v\partial^{-1}\circ u \\
-u\partial^{-1}\circ v & u\partial^{-1}\circ u
\end{array}\right)$ (NLS non-local Poisson structure).
\end{enumerate}
\end{example}

%%%%%%%%%%%%%%%%%%%%%%%%%%%%%%%%%%%%%%%%%%%%%%%%%%%%%%%%%%%%%%%%%%%%%%%%%%%%%%%%%%%%%%%%%%%%%%%%%%%%%%%%%%%%%%
%%%%%%%%%%%%%%% Sect 6 %%%%%%%%%%%%%%%%%%%%%%%%%%%%%%%%%%%%%%%%%%%%%%%%%%%%%%%%%%%%%%%%%%%%%%%%%%%%%%%%%%%%%%%
%%%%%%%%%%%%%%%%%%%%%%%%%%%%%%%%%%%%%%%%%%%%%%%%%%%%%%%%%%%%%%%%%%%%%%%%%%%%%%%%%%%%%%%%%%%%%%%%%%%%%%%%%%%%%%

\section{Constructing families of compatible non-local \\ Poisson structures}
\label{sec:6}

As in the previous sections, let $\mc V$ be an algebra of differential functions
in the variables $u_1,\dots,u_\ell$,
we assume that $\mc V$ is a domain, and we let $\mc K$ be its field of fractions.
As in the local case, two non-local Poisson vertex algebra $\lambda$-brackets on $\mc V$
%$\{\cdot\,_\lambda\,\cdot\}_1,\,\{\cdot\,_\lambda\,\cdot\}_2$ 
(respectively two non-local Poisson structures) 
%$H_1(\partial),\,H_2(\partial)$
are said to be \emph{compatible} if any their linear combination
is again a non-local Poisson vertex algebra structure
(resp. a non-local Poisson structure).
Such a pair is called a \emph{bi-Poisson structure}.
More generally, a collection of non-local 
Poisson structures $\{H^\alpha\}_{\alpha\in\mc A}$ on $\mc V$,
is called \emph{compatible} if any their (finite) linear combination
is a non-local Poisson structure over $\mc V$.

Recalling the Jacobi identity \eqref{20110922:eq4}, we introduce the following notation.
Given rational $\ell\times\ell$-matrix pseudodifferential operators 
$K,H\in\Mat_{\ell\times\ell}\mc V(\partial)$,
we let $J(H,K)=J^1(H,K)-J^2(H,K)-J^3(H,K)$,
where $J^\alpha(H,K)=\big(J^\alpha_{ijk}(H,K)(\lambda,\mu)\big)_{i,j,k\in I}$, for $\alpha=1,2,3$,
are the arrays with the following entries in $\mc V_{\lambda,\mu}$:
\begin{equation}\label{20111118:eq1}
\begin{array}{rcl}
J^1(H,K)_{ijk}(\lambda,\mu)&=&\{{u_i}_\lambda\{{u_j}_\mu{u_k}\}_H\}_K\,,\\
J^2(H,K)_{ijk}(\lambda,\mu)&=&\{{u_j}_\mu\{{u_i}_\lambda{u_k}\}_H\}_K\,,\\
J^3(H,K)_{ijk}(\lambda,\mu)&=&\{{\{{u_i}_\lambda{u_j}\}_H}_{\lambda+\mu}{u_k}\}_K\,.
\end{array}
\end{equation}
Consider a collection $\{H^\alpha\}_{\alpha\in\mc A}$
of skewadjoint rational non-local matrix pseudodifferential operators.
By definition, $H^\alpha$ is a Poisson structure
if and only if 
$J(H^\alpha,H^\alpha)=0$.
It is easy to see that the $H^\alpha$'s form a compatible family of Poisson structures
if and only if each pair is compatible, i.e.
\begin{equation}\label{20111116:eq3}
J(H^\alpha,H^\beta)+J(H^\beta,H^\alpha)=0
\,\,,\,\,\,\,\forall \alpha,\beta\in\mc A\,.
\end{equation}

\begin{theorem}\label{20111021:thm}
Let $H,\,K\in\Mat_{\ell\times\ell}\mc V(\partial)$
be compatible non-local Poisson structures over
the the algebra of differential functions $\mc V$, which is a domain.
Assume that $K$ is an invertible element of the algebra $\Mat_{\ell\times\ell}\mc V(\partial)$.
Then the following sequence of rational matrix pseudodifferential operators
with coefficients in $\mc V$:
$$
H^{[0]}=K
\,\,,\,\,\,\,
H^{[n]} :=
\big(H K^{-1}\big)^{n-1} H
\,\in\Mat{}_{\ell\times\ell}\mc V(\partial)
\,\,,\,\,\,\,n\geq1\,,
$$
form a compatible family of non-local Poisson structures over $\mc V$.
\end{theorem}
\begin{remark}\label{20111114:rem}
It is stated in \cite{FF81} that $H^{[n]},\,n\geq0$, are non-local Poisson structures,
but the prove there is given only under the additional
assupmtion that $H$ is invertible as well.
In this case the proof becomes much easier since $H^{[n]}$ is invertible, 
therefore one needs to prove that $(H^{[n]})^{-1}$ is a symplectic structure.
\end{remark}
Following the idea in \cite{TT11}, we will reduce the proof of Theorem \ref{20111021:thm}
to the following special case of it:
\begin{lemma}\label{20111116:lem1}
Let $\tilde H,\,K\in\Mat_{\ell\times\ell}\mc V(\partial)$
be compatible non-local Poisson structures over $\mc V$,
and assume that $K$ is an invertible element 
of the algebra $\Mat_{\ell\times\ell}\mc V(\partial)$.
Then 
%the rational matrix pseudodifferential operator
$\tilde H(\partial) K^{-1}(\partial)\tilde H(\partial)\in\Mat{}_{\ell\times\ell}\mc V(\partial)$
is a non-local Poisson structure over $\mc V$.
\end{lemma}
\begin{proof}
To simplify notation, in this proof we denote $\tilde H$ by $H$,
and we let $R=H K^{-1}$
so that $R^*=K^{-1} H$.
Let $H^{[2]}=H K^{-1} H
\Big(=R H=H R^*\Big)$,
and let $\{\cdot\,_\lambda\,\cdot\}_2=\{\cdot\,_\lambda\,\cdot\}_{H^{[2]}}$
be the non-local $\lambda$-bracket on $\mc V$ associated
to $H^{[2]}\in\Mat_{\ell\times\ell}\mc V(\partial)$
via \eqref{20110922:eq1}.
We need to prove the Jacobi identity, i.e. using the notation in \eqref{20111118:eq1},
that $J(H^{[2]},H^{[2]})=0$.

We need to compute all three terms $J^\alpha=J^\alpha(H^{[2]},H^{[2]})_{ijk}(\lambda,\mu)$, for $\alpha=1,2,3$,
of the Jacobi identity.
First, if $f\in\mc V$ and $i\in I$, we have, in $\mc V((\lambda^{-1}))$,
\begin{eqnarray}
\label{20111103:eq1a}
\{{u_i}_\lambda f\}_2
&=&
\displaystyle{
\sum_{s\in I}{\{{u_s}_{\lambda+\partial}f\}_H}_\to R^*_{si}(\lambda)\,,
} \\
\label{20111103:eq1b}
\{{u_j}_\mu f\}_2
&=&
\displaystyle{
\sum_{t\in I}{\{{u_t}_{\mu+\partial}f\}_H}_\to R^*_{tj}(\mu)\,,
} \\
\label{20111103:eq2}
\{f_{\lambda+\mu}{u_k}\}_2
&=&
\displaystyle{
\sum_{r\in I}R_{kr}(\lambda+\mu+\partial)\, \{f_{\lambda+\mu}u_r\}_H\,.
}
\end{eqnarray}
Both the above equations follow immediately from the Master formula \eqref{20110922:eq1}
and the definition of $H^{[2]}$.
The following identities are proved in a similar way, using that $K\circ K^{-1}=\id$,
\begin{eqnarray}
\label{20111121:eq1a}
\{{u_i}_\lambda f\}_H
&=&
\displaystyle{
\sum_{s\in I}{\{{u_s}_{\lambda+\partial}f\}_K}_\to R^*_{si}(\lambda)\,,
} \\
\label{20111121:eq1b}
\{{u_j}_\mu f\}_H
&=&
\displaystyle{
\sum_{t\in I}{\{{u_t}_{\mu+\partial}f\}_K}_\to R^*_{tj}(\mu)\,,
} \\
\label{20111121:eq2}
\{f_{\lambda+\mu}{u_k}\}_H
&=&
\displaystyle{
\sum_{r\in I}R_{kr}(\lambda+\mu+\partial)\, \{f_{\lambda+\mu}u_r\}_K\,.
}
\end{eqnarray}
Next, it is not hard to chek, using the left and right Leibniz rules
and Lemma \ref{20111012:lem}, that,
given an admissible non-local $\lambda$-bracket $\{\cdot\,_\lambda\,\cdot\}$ on $\mc V$,
the following identities hold in $\mc V_{\lambda,\mu}$, for every $i,j,k\in I$:
\begin{eqnarray}
\label{20111103:eq3a}
&&\{{u_i}_\lambda H^{[2]}_{kj}(\mu)\} 
=
\displaystyle{
\sum_{t\in I}
\{{u_i}_\lambda \{{u_t}_y{u_k}\}_H\} \Big(\Big|_{y=\mu+\partial}R^*_{tj}(\mu)\Big)
} \\
&&\displaystyle{
- \sum_{r,t\in I}
R_{kr}(\lambda\!+\!\mu\!+\!\partial)\,
\{{u_i}_\lambda \{{u_t}_y{u_r}\}_K\} \Big(\Big|_{y=\mu+\partial}\!\!\!R^*_{tj}(\mu)\Big)
}\nonumber \\
&&\displaystyle{
+ \sum_{r\in I}
R_{kr}(\lambda+\mu+\partial)\,
\{{u_i}_\lambda \{{u_j}_\mu{u_r}\}_H\}
\,,
}\nonumber \\
\label{20111103:eq3b}
&&\{{u_j}_\mu H^{[2]}_{ki}(\lambda)\} 
=
\displaystyle{
\sum_{s\in I}
\{{u_j}_\mu \{{u_s}_x{u_k}\}_H\} \Big(\Big|_{x=\lambda+\partial}R^*_{si}(\lambda)\Big)
} \\
&&\displaystyle{
- \sum_{r,s\in I}
R_{kr}(\lambda\!+\!\mu\!+\!\partial)\,
\{{u_j}_\mu \{{u_s}_x{u_r}\}_K\} \Big(\Big|_{x=\lambda+\partial}\!\!\!R^*_{si}(\lambda)\Big)
}\nonumber \\
&&\displaystyle{
+ \sum_{r\in I}
R_{kr}(\lambda+\mu+\partial)\,
\{{u_j}_\mu \{{u_i}_\lambda{u_r}\}_H\}
\,,
}\nonumber \\
\label{20111103:eq3c}
&&\{H^{[2]}_{ji}(\lambda)_{\lambda+\mu}{u_k}\} 
=
\displaystyle{
\sum_{s\in I}
\{{\{{u_s}_x{u_j}\}_H}_{\lambda+\mu+\partial}{u_k}\}_\to 
\Big(\Big|_{x=\lambda+\partial}\!\!\!R^*_{si}(\lambda)\Big)
} \\
&&\displaystyle{
- \sum_{s,t\in I}
\{{\{{u_s}_x{u_t}\}_K}_{\lambda+\mu+\partial}{u_k}\}_\to 
\Big(\Big|_{x=\lambda+\partial}\!\!\!\!R^*_{si}(\lambda)\!\Big)
\Big(\Big|_{y=\mu+\partial}\!\!\!\!R^*_{tj}(\mu)\!\Big)
}\nonumber \\
&&\displaystyle{
+ \sum_{t\in I}
\{{\{{u_i}_\lambda{u_t}\}_H}_{\lambda+\mu+\partial}{u_k}\}_\to 
R^*_{tj}(\mu)
\,.
}\nonumber 
\end{eqnarray}
Here and further we use the following notation:
given an element 
$$
P(\lambda,\mu)=\sum_{m,n,p=-\infty}^N p_{m,n,p}\lambda^m\mu^n(\lambda+\mu)^p
\in\mc V_{\lambda,\mu}\,,
$$ 
and $f,g\in\mc V$, we let
\begin{equation}\label{20111018:eq5}
\begin{array}{l}
\displaystyle{
P(x,y)\Big(\Big|_{x=\lambda+\partial}f\Big)
\Big(\Big|_{y=\mu+\partial}g\Big)
} \\
\displaystyle{
=
\sum_{m,n,p=-\infty}^N
p_{m,n,p}(\lambda+\mu+\partial)^p
\big((\lambda+\partial)^mf\big)\big((\mu+\partial)^ng\big)
\,\,\in\mc V_{\lambda,\mu}
\,.
}
\end{array}
\end{equation}
In equation \eqref{20111103:eq3c} we used the assumption that $H$ and $K$ are skewadjoint.
Combining equations \eqref{20111103:eq1a} and \eqref{20111103:eq3a},
equations \eqref{20111103:eq1b} and \eqref{20111103:eq3b},
and equations \eqref{20111103:eq2} and \eqref{20111103:eq3c},
we get, respectively,
\begin{eqnarray}
\label{20111121:eq3a}
&&\displaystyle{
J^1 = \{{u_i}_\lambda \{{u_j}_\mu{u_k}\}_2\}_2
} \\
&&\displaystyle{
=\sum_{s,t\in I}
\{{u_s}_x \{{u_t}_y{u_k}\}_H\}_H 
\Big(\Big|_{x=\lambda+\partial}R^*_{si}(\lambda)\Big)
\Big(\Big|_{y=\mu+\partial}R^*_{tj}(\mu)\Big)
}\nonumber \\
&&\displaystyle{
- \sum_{r,s,t\in I}
R_{kr}(\lambda\!+\!\mu\!+\!\partial)\,
\{{u_s}_x \{{u_t}_y{u_r}\}_K\}_H 
%}\nonumber \\
%&&\displaystyle{
%\,\,\,\,\,\,\,\,\,\,\,\,\,\,\,\,\,\,\,\,\,\,\,\,\,\,\,\,\,\,\,\,\,\,\,\,\,\,\,\,\,\,\,\,\,\,\,\,\,\,\,\,\,\,
%\,\,\,\,\,\,
\Big(\Big|_{x=\lambda+\partial}R^*_{si}(\lambda)\Big)
\Big(\Big|_{y=\mu+\partial}R^*_{tj}(\mu)\Big)
}\nonumber \\
&&\displaystyle{
+ \sum_{r,s\in I}
R_{kr}(\lambda+\mu+\partial)\,
\{{u_s}_x \{{u_j}_\mu{u_r}\}_H\}_H
\Big(\Big|_{x=\lambda+\partial}R^*_{si}(\lambda)\Big)
\,,
}\nonumber \\
\label{20111121:eq3b}
&&\displaystyle{
J^2 = \{{u_j}_\mu \{{u_i}_\lambda{u_k}\}_2\}_2
} \\
&&\displaystyle{
=\sum_{s,t\in I}
\{{u_j}_y \{{u_s}_x{u_k}\}_H\}_H 
\Big(\Big|_{x=\lambda+\partial}R^*_{si}(\lambda)\Big)
\Big(\Big|_{y=\mu+\partial}R^*_{tj}(\mu)\Big)
}\nonumber \\
&&\displaystyle{
- \sum_{r,s,t\in I}
R_{kr}(\lambda\!+\!\mu\!+\!\partial)\,
\{{u_t}_y \{{u_s}_x{u_r}\}_K\}_H 
%}\nonumber \\
%&&\displaystyle{
%\,\,\,\,\,\,\,\,\,\,\,\,\,\,\,\,\,\,\,\,\,\,\,\,\,\,\,\,\,\,\,\,\,\,\,\,\,\,\,\,\,\,\,\,\,\,\,\,\,\,\,\,\,\,
%\,\,\,\,\,\,
\Big(\Big|_{x=\lambda+\partial}R^*_{si}(\lambda)\Big)
\Big(\Big|_{y=\mu+\partial}R^*_{tj}(\mu)\Big)
}\nonumber \\
&&\displaystyle{
+ \sum_{r,t\in I}
R_{kr}(\lambda+\mu+\partial)\,
\{{u_t}_y \{{u_i}_\lambda{u_r}\}_H\}_H
\Big(\Big|_{y=\mu+\partial}R^*_{tj}(\mu)\Big)
\,,
}\nonumber \\
\label{20111121:eq3c}
&&\displaystyle{
J^3 = \{{\{{u_i}_\lambda{u_j}\}_2}_{\lambda+\mu}{u_k}\}_2
} \\
&&\displaystyle{
=\sum_{r,s\in I}
R_{kr}(\lambda+\mu+\partial)\,
{\{{\{{u_s}_x{u_j}\}_H}_{\lambda+\mu+\partial}{u_r}\}_H}_\to 
\Big(\Big|_{x=\lambda+\partial}\!\!\!R^*_{si}(\lambda)\Big)
}\nonumber \\
&&\displaystyle{
- \sum_{r,s,t\in I}
R_{kr}(\lambda+\mu+\partial)\,
{\{{\{{u_s}_x{u_t}\}_K}_{\lambda+\mu+\partial}{u_r}\}_H}_\to 
%}\nonumber \\
%&&\displaystyle{
%\,\,\,\,\,\,\,\,\,\,\,\,\,\,\,\,\,\,\,\,\,\,\,\,\,\,\,\,\,\,\,\,\,\,\,\,\,\,\,\,\,\,\,\,\,\,\,\,\,\,\,\,\,\,
%\,\,\,\,\,\,\,\,\,\,\,\,\,\,\,\,\,\,\,\,\,
\Big(\Big|_{x=\lambda+\partial}R^*_{si}(\lambda)\!\Big)
R^*_{tj}(\mu)
}\nonumber \\
&&\displaystyle{
+ \sum_{r,t\in I}
R_{kr}(\lambda+\mu+\partial)\,
{\{{\{{u_i}_\lambda{u_t}\}_H}_{\lambda+\mu+\partial}{u_r}\}_H}_\to 
R^*_{tj}(\mu)
\,.
}\nonumber 
\end{eqnarray}
We need to prove that $J^1-J^2-J^3=0$.
The first term of the RHS of \eqref{20111121:eq3a} combined with the first term of the RHS of \eqref{20111121:eq3b}
gives, by the Jacobi identity for $H$ and by equation \eqref{20111121:eq2},
\begin{equation}\label{20111121:eq4a}
\begin{array}{l}
\displaystyle{
\sum_{s,t\in I}
\!\!
\Big(\{{u_s}_x \{{u_t}_y{u_k}\}_H\}_H - \{{u_j}_y \{{u_s}_x{u_k}\}_H\}_H\Big)\!\!
\Big(\Big|_{x=\lambda+\partial}R^*_{si}(\lambda)\Big)\!\!
\Big(\Big|_{y=\mu+\partial}R^*_{tj}(\mu)\Big)
} \\
\displaystyle{
=
\sum_{s,t\in I}
\{{\{{u_s}_x{u_t}\}_H}_{x+y}{u_k}\}_H\}
\Big(\Big|_{x=\lambda+\partial}R^*_{si}(\lambda)\Big)
\Big(\Big|_{y=\mu+\partial}R^*_{tj}(\mu)\Big)
} \\
\displaystyle{
=
\sum_{r,s,t\in I}
\!\!
R_{kr}(\lambda+\mu+\partial)
\{{\{{u_s}_x{u_t}\}_H}_{x+y}{u_r}\}_K
\Big(\Big|_{x=\lambda+\partial}R^*_{si}(\lambda)\Big)
\!\!\Big(\Big|_{y=\mu+\partial}R^*_{tj}(\mu)\Big).
}
\end{array}
\end{equation}
Similarly,
the third term of the RHS of \eqref{20111121:eq3a} combined with the first term of the RHS of \eqref{20111121:eq3c}
gives, by the Jacobi identity for $H$ and by equation \eqref{20111121:eq1b},
\begin{equation}\label{20111121:eq4b}
\begin{array}{l}
\displaystyle{
\sum_{r,s\in I}
R_{kr}(\lambda+\mu+\partial)\,
\Big(\{{u_s}_x \{{u_j}_\mu{u_r}\}_H\}_H
-{\{{\{{u_s}_x{u_j}\}_H}_{\lambda+\mu+\partial}{u_r}\}_H}_\to\Big)
} \\
\displaystyle{
\Big(\Big|_{x=\lambda+\partial}R^*_{si}(\lambda)\Big)
=
\sum_{r,s\in I}
R_{kr}(\lambda+\mu+\partial)\,
\{{u_j}_\mu \{{u_s}_x{u_r}\}_H\}_H
\Big(\Big|_{x=\lambda+\partial}R^*_{si}(\lambda)\Big)
} \\
\displaystyle{
=
\sum_{r,s,t\in I}
R_{kr}(\lambda+\mu+\partial)\,
\{{u_t}_y \{{u_s}_x{u_r}\}_H\}_K
\Big(\Big|_{x=\lambda+\partial}R^*_{si}(\lambda)\Big)
\Big(\Big|_{y=\mu+\partial}R^*_{tj}(\mu)\Big)\,.
}
\end{array}
\end{equation}
In the same way,
the third term of the RHS of \eqref{20111121:eq3b} combined with the third term of the RHS of \eqref{20111121:eq3c}
gives, by the Jacobi identity for $H$ and by equation \eqref{20111121:eq1a},
\begin{equation}\label{20111121:eq4c}
\begin{array}{c}
\displaystyle{
-\sum_{r,t\in I}
R_{kr}(\lambda+\mu+\partial)\,
\Big(
\{{u_t}_y \{{u_i}_\lambda{u_r}\}_H\}_H
+{\{{\{{u_i}_\lambda{u_t}\}_H}_{\lambda+y}{u_r}\}_H} 
\Big)
} \\
\displaystyle{
\Big(\Big|_{y=\mu+\partial}R^*_{tj}(\mu)\Big)
=
-\sum_{r,t\in I}
R_{kr}(\lambda+\mu+\partial)\,
\{{u_i}_\lambda \{{u_t}_y{u_r}\}_H\}_H
} \\
\displaystyle{
\Big(\Big|_{y=\mu+\partial}R^*_{tj}(\mu)\Big)
=
-\sum_{r,s,t\in I}
R_{kr}(\lambda+\mu+\partial)\,
\{{u_s}_\lambda \{{u_t}_y{u_r}\}_H\}_K
} \\
\displaystyle{
\Big(\Big|_{x=\lambda+\partial}R^*_{si}(\lambda)\Big)
\Big(\Big|_{y=\mu+\partial}R^*_{tj}(\mu)\Big)\,.
}
\end{array}
\end{equation}
Finally, combining the second term in the RHS of \eqref{20111121:eq3a}, \eqref{20111121:eq3b} and \eqref{20111121:eq3c},
together with the RHS of equations \eqref{20111121:eq4a}, \eqref{20111121:eq4b} and \eqref{20111121:eq4c},
we get
$$
\begin{array}{c}
\displaystyle{
J^1-J^2-J^3=
\sum_{r,s,t\in I}
R_{kr}(\lambda+\mu+\partial)\,
\Big(
-\{{u_s}_x \{{u_t}_y{u_r}\}_K\}_H 
} \\
\displaystyle{
+\{{u_t}_y \{{u_s}_x{u_r}\}_K\}_H 
+{\{{\{{u_s}_x{u_t}\}_K}_{x+y}{u_r}\}_H}_\to 
+\{{\{{u_s}_x{u_t}\}_H}_{x+y}{u_r}\}_K
} \\
\displaystyle{
+\{{u_t}_y \{{u_s}_x{u_r}\}_H\}_K
-\{{u_s}_\lambda \{{u_t}_y{u_r}\}_H\}_K
\Big)
\Big(\Big|_{x=\lambda+\partial}R^*_{si}(\lambda)\Big)
\Big(\Big|_{y=\mu+\partial}R^*_{tj}(\mu)\Big)\,,
}
\end{array}
$$
which is zero since, by assumption, $H$ and $K$ are compatible.
\end{proof}
\begin{remark}\label{20111122:rem}
The proof of Lemma \ref{20111116:lem1} does not use the assumption that $K$
is a Poisson structure.
\end{remark}
\begin{proof}[Proof of Theorem \ref{20111021:thm}]
We prove, by induction on $n\geq1$, that the rational matrix pseudodifferential operators
$$
H^{[0]}=K,\,H^{[1]}=H,\dots,H^{[n]}
\,\in\Mat{}_{\ell\times\ell}\mc V(\partial)\,,
$$
form a compatible family of non-local Poisson structures over $\mc V$.
For $n=1$, this holds by assumption.
Assuming by induction that the statement holds for $n\geq1$,
we will prove that it holds for $n+1$.
Namely, 
%\begin{enumerate}[(i)]
%\item
%$H^{[n+1]}(\partial)$ is a Poisson structure,
%\item
%for every $m=0,\dots,n$, the Poisson structures $H^{[m]}(\partial)$ and $H^{[n+1]}(\partial)$ are compatible.
%\end{enumerate}
%%
%Equivalently, 
thanks to the observations at the beginning of the section,
we need to prove that
\begin{enumerate}[(i)]
\item
$J(H^{[n+1]},H^{[n+1]})=0$,
\item
$J(H^{[m]},H^{[n+1]})+J(H^{[n+1]},H^{[m]})=0$
for every $m=0,\dots,n$.
\end{enumerate}
By the inductive assumption, $\tilde H=\sum_{i=0}^nx_iH^{[i]}$
is a Poisson structure for every $x_0,\dots,x_n\in\mb F$.
Hence, by Lemma \ref{20111116:lem1},
we get the following Poisson structure for every point $(x_0,\dots,x_n)\in\mb F^{n+1}$:
$$
\tilde H K^{-1}\tilde H
=
\sum_{i,j=0}^nx_ix_j H^{[i]} K^{-1} H^{[j]}
=
\sum_{p=0}^{2n}Q_p(x_0,\dots,x_n) H^{[p]}\,,
$$
where, for $p=0,\dots,2n$,
%$Q_p$ denotes the following homogeneous polynomial
%of degree 2 in the variables $x_0,\dots,x_n$:
\begin{equation}\label{20111116:eq4}
Q_p(x_0,\dots,x_n)=
\sum_{\substack{i,j=0 \\ (i+j=p)}}^n x_ix_j\,.
\end{equation}
We thus get
$$
\begin{array}{l}
\displaystyle{
0=J(\tilde H K^{-1}\tilde H,\tilde H K^{-1}\tilde H)
=
\sum_{p=0}^{2n}Q_p^2(x_0,\dots,x_n) J(H^{[p]},H^{[p]})
} \\
\displaystyle{
+\sum_{\substack{p,q=0 \\ (p<q)}}^{2n}Q_p(x_0,\dots,x_n) Q_q(x_0,\dots,x_n) 
\big(J(H^{[p]},H^{[q]})+J(H^{[q]},H^{[p]})\big)\,,
}
\end{array}
$$
for every $(x_0,\dots,x_n)\in\mb F^{n+1}$.
Note that, by the inductive assumption, $J(H^{[p]},H^{[p]})=0$
for every $0\leq p\leq n$
and $J(H^{[p]},H^{[q]})+J(H^{[q]},H^{[p]})=0$ for every $0\leq p<q\leq n$.
Hence the above equation gives
\begin{equation}\label{20111116:eq5}
\begin{array}{l}
\displaystyle{
\sum_{p=n+1}^{2n} Q_p^2(x_0,\dots,x_n) J(H^{[p]},H^{[p]})
} \\
\displaystyle{
+\sum_{p=0}^n\sum_{q=n+1}^{2n} Q_p(x_0,\dots,x_n) Q_q(x_0,\dots,x_n) 
\big(J(H^{[p]},H^{[q]})+J(H^{[q]},H^{[p]})\big)
} \\
\displaystyle{
+\!\sum_{\substack{p,q=n+1 \\ (p<q)}}^{2n} Q_p(x_0,\dots,x_n) Q_q(x_0,\dots,x_n) 
\big(J(H^{[p]},H^{[q]})+J(H^{[q]},H^{[p]})\big)=0
}
\end{array}
\end{equation}
for every $(x_0,\dots,x_n)\in\mb F^{n+1}$.
Next, we introduce a grading in the algebra of polynomials in $x_0,\dots,x_n$,
letting $\deg(x_i)=i$.
Then $Q_p(x_0,\dots,x_n)$ is homogeneous of degree $p$.
By looking at the terms of degree $d=2n+2$ in equation \eqref{20111116:eq5},
we get
\begin{equation}\label{20111116:eq6}
\begin{array}{l}
\displaystyle{
Q_{n+1}^2(x_0,\dots,x_n) J(H^{[n+1]},H^{[n+1]})
+\sum_{p=2}^n
Q_p(x_0,\dots,x_n) 
} \\
\displaystyle{
Q_{2n+2-p}(x_0,\dots,x_n) 
\big(J(H^{[p]},H^{[2n+2-p]})+J(H^{[2n+2-p]},H^{[p]})\big)=0\,,
}
\end{array}
\end{equation}
while, 
by looking at the terms of degree $d=m+n+1$ with $m\in\{0,\dots,n\}$
in equation \eqref{20111116:eq5}, we get
\begin{equation}\label{20111116:eq7}
\begin{array}{l}
\displaystyle{
\sum_{p=0}^m
Q_p(x_0,\dots,x_n) Q_{m+n+1-p}(x_0,\dots,x_n) 
} \\
\displaystyle{
\big(J(H^{[p]},H^{[m+n+1-p]})+J(H^{[m+n+1-p]},H^{[p]})\big)=0\,,
}
\end{array}
\end{equation}
for every $(x_0,\dots,x_n)\in\mb F^{n+1}$.
To conclude the proof, we only need to show that equations \eqref{20111116:eq6} and \eqref{20111116:eq7}
imply respectively relations (i) and (ii) above.
This is a consequence of the following lemma.
\begin{lemma}\label{20111116:lem2}
\begin{enumerate}[(a)]
\item
For every $n\geq1$,
\begin{equation}\label{20111117:eq1}
Q_{n+1}^2(x_0,\dots,x_n)\not\in
\Span{}_{\mb F}\Big\{Q_p(x_0,\dots,x_n) Q_{2n+2-p}(x_0,\dots,x_n)\Big\}_{2\leq p\leq n}\,.
\end{equation}
\item
For every $n\geq1$ and $m\in\{0,\dots,n\}$,
\begin{equation}\label{20111117:eq2}
Q_m Q_{n+1}
\not\in\Span{}_{\mb F}\Big\{Q_p Q_{m+n+1-p}\Big\}_{0\leq p\leq m-1}\,.
\end{equation}
\end{enumerate}
\end{lemma}
\begin{proof}
Note that,
$$
\begin{array}{l}
Q_p(0,\dots,0,x_k,\dots,x_n) \\
\displaystyle{
=
\sum_{\substack{i,j=k \\ (i+j=p)}}^nx_ix_j=
\left\{\begin{array}{ll}
0 & \text{ if } p<2k\,,\\
x_k^2 & \text{ if } p=2k\,,\\
\displaystyle{
2x_kx_{p-k}+\dots } & \text{ if } p>2k\,.
\end{array}\right.
}
\end{array}
$$

We prove part (a) separately in the cases when $n$ is even and odd.
If $n=2k-1$ is odd, letting $x_0=\dots=x_{k-1}=0$ we have
$Q_{n+1}=x_k^2\neq0$, and $Q_p=0$ for all $p=2,\dots,n=2k-1$.
This implies \eqref{20111117:eq1} for odd $n$.
If $n=2$, we have $Q_2=2x_0x_2+x_1^2,\,Q_3=2x_1x_2,\,Q_4=x_2^2$,
hence $Q_3^2\not\in\mb FQ_2Q_4$.
If $n=2k$ with $k\geq2$, letting $x_0=\dots=x_{k-1}=0$ we have
$Q_p=0$ for all $p=2,\dots,n-1$,
$Q_n=x_k^2$,
$Q_{n+1}=2x_kx_{k+1}$, 
$Q_{n+2}=2x_kx_{k+2}+x_{k+1}^2$.
Since $Q_{n+1}^2=4x_k^2x_{k+1}^2$ is not a multiple of $Q_nQ_{n+2}=2x_k^3x_{k+2}+x_k^2x_{k+1}^2$,
\eqref{20111117:eq1} holds for even $n$.

Similarly, we prove part (b) separately in the cases when $m$ is even and odd.
If $m=2k$ is even, letting $x_0=\dots=x_{k-1}=0$ we have
$Q_{m}Q_{n+1}=x_k^2Q_{n+1}\neq0$, and $Q_p=0$ for all $p=2,\dots,m-1$.
Hence \eqref{20111117:eq2} holds for even $m$.
For $m=1\leq n$, we have $Q_0=x_0^2,\,Q_1=2x_0x_1,\,Q_{n+1}=2x_1x_n+\dots$.
Therefore, $Q_0Q_{n+2}$ is divisible by $x_0^2$, while $Q_1Q_{n+1}=2x_0x_1(2x_1x_n+\dots)$ is not.
%This proves \eqref{20111117:eq2}.
%
Finally, if $m=2k+1$ is odd, with $k\geq1$, letting $x_0=\dots=x_{k-1}=0$ we have
$Q_p=0$ for all $p=2,\dots,m-2$,
$Q_{m-1}=x_k^2$,
$Q_{m}=2x_kx_{k+1}$, 
$Q_{n+1}=2x_kx_{n+1-k}+2x_{k+1}x_{n-k}+\dots$.
Hence, $Q_{m-1}Q_{n+2}=x_k^2Q_{n+2}$ is divisible by $x_k^2$,
while $Q_mQ_{n+1}=4x_k^2x_{k+1}x_{n+1-k}+4x_kx_{k+1}^2x_{n-k}+\dots$ is not,
proving \eqref{20111117:eq2} for odd $m$.
\end{proof}
\end{proof}
\begin{example}\label{20120224:vic}
Let $K=\partial^3,\,H=\partial^2\circ\frac1u\partial\circ\frac1u\partial^2$. These are compatible Hamiltonian
structures (see \cite{DSKW10}).
Hence, by Theorem \ref{20111021:thm},
$$
H^{[n]}=(H K^{-1})^{n-1} H
=\partial^2\circ (\frac1u\circ \partial)^{2n}\circ\partial
\,\,,\,\,\,\,n\in\mb Z_+\,,
$$
are compatible Poisson structures.
This was proved in \cite{DSKW10} by direct verification,
and deduced from Theorem \ref{20111021:thm} in \cite{TT11}.
\end{example}

%%%%%%%%%%%%%%%%%%%%%%%%%%%%%%%%%%%%%%%%%%%%%%%%%%%%%%%%%%%%%%%%%%%%%%%%%%%%%%%%%%%%%%%%%%%%%%%%%%%%%%%%%%%%%%
%%%%%%%%%%%%%%% Sect 5 %%%%%%%%%%%%%%%%%%%%%%%%%%%%%%%%%%%%%%%%%%%%%%%%%%%%%%%%%%%%%%%%%%%%%%%%%%%%%%%%%%%%%%%
%%%%%%%%%%%%%%%%%%%%%%%%%%%%%%%%%%%%%%%%%%%%%%%%%%%%%%%%%%%%%%%%%%%%%%%%%%%%%%%%%%%%%%%%%%%%%%%%%%%%%%%%%%%%%%

\section{Symplectic structures and Dirac structures in terms of non-local Poisson structures}
\label{sec:5}

\subsection{Simplectic structure as inverse of a non-local Poisson structure}
\label{sec:5.1}

As in the previous sections,
let $\mc V$ be an algebra of differential functions in the variables $u_1,\dots,u_\ell$,
which is a domain, and let $\mc K$ be its field of fractions.

Recall that (see e.g. \cite{BDSK09}) a (local) \emph{symplectic structure} on $\mc V$ 
is an $\ell\times\ell$ matrix differential operator
$S=\big(S_{ij}(\partial)\big)_{i,j\in I}\in\Mat_{\ell\times\ell}\mc V[\partial]$
which is skewadjoint and satisfies the
following \emph{symplectic identity}:
\begin{equation}\label{20111012:eq1}
\sum_{n\in\mb Z_+}\Big(
\frac{\partial S_{ki}(\mu)}{\partial u_j^{(n)}} \lambda^n
-\frac{\partial S_{kj}(\lambda)}{\partial u_i^{(n)}} \mu^n
+(-\lambda-\mu-\partial)^n \frac{\partial S_{ij}(\lambda)}{\partial u_k^{(n)}}
\Big)=0\,.
\end{equation}

We can write the symplectic identity \eqref{20111012:eq1} in terms 
of the \emph{Beltrami} $\lambda$-\emph{bracket}
$\langle\cdot\,_\lambda\,\cdot\rangle:\,\mc V\times\mc V\to\mc V[\lambda]$,
introduced in \cite{BDSK09}.
It is defined as the symmetric $\lambda$-bracket such that $\langle {u_i}_\lambda u_j\rangle=\delta_{ij}$,
and extended by the Master Formula \eqref{20110922:eq1}:
$$
\langle f_\lambda g\rangle
=
\sum_{\substack{ i\in I \\ m,n\in\mb Z_+}}(-1)^m
\frac{\partial g}{\partial u_i^{(n)}}(\lambda+\partial)^{m+n}\frac{\partial f}{u_i^{(m)}}\,.
$$
Then, the symplectic identity \eqref{20111012:eq1} becomes
\begin{equation}\label{20111017:eq1}
\langle {u_j}_\lambda\{{u_i}_\mu{u_k}\}_S\rangle
-\langle {u_i}_\mu\{{u_j}_\lambda{u_k}\}_S\rangle
+\langle {\{{u_j}_\lambda {u_i}\}_S}_{\lambda+\mu}{u_k}\rangle=0
\,,
\end{equation}
where, recalling \eqref{20110922:eq1}, we let $\{{u_j}_\lambda{u_i}\}_S=S_{ij}(\lambda)$.

Note that, 
if $S\in\Mat_{\ell\times\ell}\mc V(\partial)$ 
is a rational matrix pseudodifferential operator with coefficients in $\mc V$,
then, by Corollary \ref{20111007:prop}, all three terms in the LHS of equation \eqref{20111017:eq1}
lie in $\mc V_{\lambda,\mu}$. Hence, equation \eqref{20111017:eq1} still makes sense
(as an equation in $\mc V_{\lambda,\mu}$).
\begin{definition}\label{20111017:def}
A \emph{non-local symplectic structure} on $\mc V$
is a skewadjoint rational matrix pseudodifferential operator 
$S=\big(S_{ij}(\partial)\big)_{i,j\in I}\in\Mat_{\ell\times\ell}\mc V(\partial)$
with coefficients in $\mc V$,
satisftying equation \eqref{20111017:eq1} in $\mc V_{\lambda,\mu}$ for all $i,j,k\in I$.
\end{definition}
\begin{theorem}\label{20111012:thm}
Let $S\in\Mat_{\ell\times\ell}\mc V(\partial)$ be a skewadjoint
rational matrix pseudodifferential operator with coefficients in 
the algebra  of differential functions $\mc V$.
Assume that $S$ is an invertible element of the algebra $\Mat_{\ell\times\ell}\mc V(\partial)$.
Then, $S$ is a non-local symplectic structure over $\mc V$ if and only if 
$S^{-1}$ is a non-local Poisson structure over $\mc V$.
\end{theorem}
\begin{proof}
Clearly, $S$ is skewadjoint if and only if $S^{-1}$ is skewadjoint.
Hence, recalling the Definition \ref{20111007:def} of non-local Poisson structure,
we only need to show that equation \eqref{20111012:eq1} in $\mc V_{\lambda,\mu}$ 
is equivalent to the Jacobi identity \eqref{20110922:eq4}, again in $\mc V_{\lambda,\mu}$, 
for $H=S^{-1}$.
By equation \eqref{20111012:eq2a}, Remark \ref{20111104:rem},
and the Master Formula \eqref{20110922:eq1}, we have,
letting $S_{ij}(\partial)=\sum_{p=-\infty}^Ns_{ij;p}\partial^p$,
\begin{equation}\label{20111018:eq1}
\begin{array}{l}
\displaystyle{
\{{u_i}_\lambda\{{u_j}_\mu{u_k}\}_H\}_H
=
\big\{{u_i}_\lambda (S^{-1})_{kj}(\mu)\big\}_{S^{-1}}
=
}\\
\displaystyle{
-\sum_{r,t=1}^\ell\sum_{p=-\infty}^N
(S^{-1})_{kr}(\lambda+\mu+\partial) 
\{{u_i}_\lambda s_{rt;p}\}_{S^{-1}} (\mu+\partial)^p (S^{-1})_{tj}(\mu)
}\\
\displaystyle{
=
-\sum_{r,s,t\in I,\,n\in\mb Z_+}
(S^{-1})_{kr}(\lambda+\mu+\partial) 
\Big(
(\lambda+\partial)^n
(S^{-1})_{si}(\lambda)
\Big)
}\\
\displaystyle{
\,\,\,\,\,\,\,\,\,\,\,\,\,\,\,\,\,\,\,\,\,\,\,\,\,\,\,\,\,\,\,\,\,\,\,\,\,\,\,\,\,\,\,\,\,\,\,\,\,\,\,\,\,\,
\,\,\,\,\,\,\,\,\,\,\,\,\,\,\,\,\,\,\,\,\,\,\,\,\,\,\,\,\,\,\,\,\,\,\,\,\,\,\,\,\,\,\,\,\,\,\,\,\,\,\,\,\,\,
\Big(
\frac{\partial S_{rt}(\mu+\partial)}{\partial u_s^{(n)}}
(S^{-1})_{tj}(\mu)
\Big)
\,.
}
\end{array}
\end{equation}
Exchanging $i$ with $j$ and $\lambda$ with $\mu$, we get
\begin{equation}\label{20111018:eq2}
\begin{array}{c}
\displaystyle{
\{{u_j}_\mu\{{u_i}_\lambda{u_k}\}_H\}_H
=
-\sum_{r,s,t\in I,\,n\in\mb Z_+}
(S^{-1})_{kr}(\lambda+\mu+\partial) 
}\\
\displaystyle{
\Big(
(\mu+\partial)^n
(S^{-1})_{tj}(\mu)
\Big)
\Big(
\frac{\partial S_{rs}(\lambda+\partial)}{\partial u_t^{(n)}}
(S^{-1})_{si}(\lambda)
\Big)
\,.
}
\end{array}
\end{equation}
Similarly, by equation \eqref{20111012:eq2b} and Remark \ref{20111104:rem}, we have,
using the assumption that $S$ is skewadjoint,
\begin{equation}\label{20111018:eq3}
\begin{array}{l}
\displaystyle{
\{{\{{u_i}_\lambda{u_j}\}_H}_{\lambda+\mu}{u_k}\}_H
=
\big\{(S^{-1})_{ji}(\lambda)_{\lambda+\mu}{u_k}\big\}_{S^{-1}}
}\\
\displaystyle{
=
\sum_{s,t=1}^\ell\sum_{p=-\infty}^N
{\{{s_{ts;p}}_{\lambda+\mu+\partial}{u_k}\}_{S^{-1}}}_\to
(S^{-1})_{tj}(\mu) (\lambda+\partial)^p (S^{-1})_{si}(\lambda)
}\\
\displaystyle{
=
\sum_{r,s,t\in I,\,m\in\mb Z_+}
(S^{-1})_{kr}(\lambda+\mu+\partial)
(-\lambda-\mu-\partial)^m
}\\
\displaystyle{
\,\,\,\,\,\,\,\,\,\,\,\,\,\,\,\,\,\,\,\,\,\,\,\,\,\,\,\,\,\,\,\,\,\,\,\,\,\,\,\,\,\,\,\,\,\,\,\,\,\,\,\,\,\,
\,\,\,\,\,\,\,\,\,\,\,\,\,
\Big(
(S^{-1})_{tj}(\mu)
\frac{\partial S_{ts}(\lambda+\partial)}{\partial u_r^{(m)}} (S^{-1})_{si}(\lambda)
\Big)
\,.
}
\end{array}
\end{equation}
Combining equations \eqref{20111018:eq1}, \eqref{20111018:eq2} and \eqref{20111018:eq3},
we get that the LHS of the Jacobi identity \eqref{20110922:eq4} is
\begin{equation}\label{20111018:eq4}
\begin{array}{l}
\displaystyle{
\{{u_i}_\lambda\{{u_j}_\mu {u_k}\}_H\}_H-\{{u_j}_\mu\{{u_i}_\lambda {u_k}\}_H\}_H
-\{{\{{u_i}_\lambda {u_j}\}_H}_{\lambda+\mu} {u_k}\}_H
} \\
\displaystyle{
=\sum_{r,s,t\in I,\,n\in\mb Z_+}
(S^{-1})_{kr}(\lambda+\mu+\partial) 
\Bigg(
-\frac{\partial S_{rt}(y)}{\partial u_s^{(n)}}x^n
+\frac{\partial S_{rs}(x)}{\partial u_t^{(n)}}y^n
} \\
\displaystyle{
-(-x-y-\partial)^n
\frac{\partial S_{ts}(x)}{\partial u_r^{(n)}}
\Bigg)
\Big(\Big|_{x=\lambda+\partial}(S^{-1})_{si}(\lambda)\Big)
\Big(\Big|_{y=\mu+\partial}(S^{-1})_{tj}(\mu)\Big)
\,,
}
\end{array}
\end{equation}
where we used the notation introduced in \eqref{20111018:eq5}.
Clearly, the RHS of \eqref{20111018:eq4} is zero, provided that the symplectic identity \eqref{20111012:eq1} holds.
For the opposite implication,
we have, by \eqref{20111018:eq4},
$$
\begin{array}{l}
\displaystyle{
\sum_{i,j,k\in I}
S_{\gamma k}(x+y+\partial)
\Bigg(\{{u_i}_x\{{u_j}_y {u_k}\}_H\}_H-\{{u_j}_y\{{u_i}_x {u_k}\}_H\}_H
} \\
\displaystyle{
-\{{\{{u_i}_x {u_j}\}_H}_{x+y} {u_k}\}_H
\Bigg)
\Big(\Big|_{x=\lambda+\partial} S_{i\alpha}(\lambda)\Big)
\Big(\Big|_{y=\mu+\partial}S_{j\beta}(\mu)\Big)
} \\
\displaystyle{
=
\sum_{n\in\mb Z_+}
\Bigg(
-\frac{\partial S_{\gamma\beta}(\mu)}{\partial u_\alpha^{(n)}}\lambda^n
+\frac{\partial S_{\gamma\alpha}(\lambda)}{\partial u_\beta^{(n)}}\mu^n
-(-\lambda-\mu-\partial)^n
\frac{\partial S_{\beta\alpha}(\lambda)}{\partial u_\gamma^{(n)}}
\Bigg)
\,.
}
\end{array}
$$
Hence, equation \eqref{20110922:eq4} implies equation \eqref{20111012:eq1}.
\end{proof}

\subsection{Dirac structure in terms of non-local Poisson structure}
\label{sec:5.2}

Let $\mc V$ be an algebra of differential functions, which is a domain,
and let $\mc K$ be its field of fractions.

We have the usual pairing $\mc V^{\oplus\ell}\times\mc V^\ell\to\mc V/\partial\mc V$
given by $(F|P)=\tint F\cdot P$.
This pairing is non-degenerate (see e.g. \cite[Prop.1.3(a)]{BDSK09}).
We extend it to a non-degenerate  symmetric bilinear form
\begin{equation}\label{20111020:eq3}
\langle\cdot\,|\,\cdot\rangle\,:\,\,
\big(\mc V^{\oplus\ell}\oplus\mc V^\ell\big)\times\big(\mc V^{\oplus\ell}\oplus\mc V^\ell\big)\to\mc V/\partial\mc V\,,
\end{equation}
given by
$\langle F\oplus P|G\oplus Q\rangle=\tint (F\cdot Q+G\cdot P)$.

%The space $\mc V^\ell$ has the following Lie algebra bracket:
%\begin{equation}\label{20111020:eq4}
%[P,Q]=D_Q(\partial)P-D_P(\partial)Q\,.
%\end{equation}

The \emph{Courant-Dorfman product}  $\circ$ is the following product 
on the space $\mc V^{\oplus\ell}\oplus\mc V^{\ell}$:
\begin{equation}\label{20111020:eq5}
(F\oplus P)\circ(G\oplus Q)
=
\big(
D_G(\partial)P+D_P^*(\partial)G-D_F(\partial)Q+D_F^*(\partial)Q
\big)
\oplus
[P,Q]\,,
\end{equation}
where, for $P,Q\in\mc V^\ell$, we let
\begin{equation}\label{20120126:eq1}
[P,Q]
=
D_Q(\partial)P-D_P(\partial)Q\,.
\end{equation}
By definition, if $F\in\mc V^{\oplus\ell}$ is closed (cf. Section \ref{sec:4.1.a.5}),
we have $D_F(\partial)Q-D_F^*(\partial)Q=0$.
Moreover, it is straightforward to check that, for arbitrary $G\in\mc V^{\oplus\ell}$ and $P\in\mc V^\ell$,
we have
$$
D^*_G(\partial)P+D_P^*(\partial)G=\frac{\delta}{\delta u}\tint P\cdot G\,.
$$
Hence, formula \eqref{20111020:eq5}
takes a simpler form when $F$ and $G$ are closed elements of $\mc V^{\oplus\ell}$:
\begin{equation}\label{20120127:eq1}
(F\oplus P)\circ(G\oplus Q)
=
\frac{\delta}{\delta u}\Big(\int P\cdot G\Big)\oplus[P,Q]
\,.
\end{equation}

\begin{remark}\label{20111020:rem1}
All the above notions have a natural interpretation from the point of view of variational calculus.
Indeed, the space $\mc V^\ell$ is naturally identified with the Lie algebra of evolutionary vector fields $\mf g^\partial$,
and the space $\mc V^{\oplus\ell}$ is naturally identified with the space 
of variational 1-forms $\Omega^1$.
Then the contraction of variational 1-forms by evolutionary vector fields gives the inner product \eqref{20111020:eq3};
the Courant-Dorfman product corresponds to the derived bracket
$[\cdot\,,\,\cdot]_d$, where $[\cdot\,,\,\cdot]$ is the Lie superalgebra bracket 
on the space of endomorphisms of the space of all de Rham forms over $\mc V$,
and $d=\ad(\delta)$, where $\delta$ is the de Rham differential, \cite[Prop.4.2]{BDSK09}.
\end{remark}

\begin{definition}[\cite{Dor93,BDSK09}]\label{20111020:def}
A \emph{Dirac structure} is a subspace $\mc L\subset\mc V^{\oplus\ell}\oplus\mc V^{\ell}$,
which is maximal isotropic with respect to the inner product \eqref{20111020:eq3},
and which is closed under the Courant-Dorfman product \eqref{20111020:eq5}.
\end{definition}
\begin{remark}\label{20130123:rem}
If $\mc L\subset\mc V^{\oplus\ell}\oplus\mc V^\ell$ is a Dirac structure,
then the subspace 
$$
\mf g=\big\{F\oplus P\in\mc L\big|\, F \text{ is closed }\big\}\subset\mc L
$$
is a Lie algebra with respect to the Courant-Dorfman product,
and its derived Lie algebra lies in the subalgebra
$$
\mf h=\Big\{\frac{\delta f}{\delta u}\oplus P\in\mc L\Big\}\,.
$$
Indeed,
the Courant-Dorfman product \eqref{20111020:eq5} on $\mc V^{\oplus\ell}\oplus\mc V^\ell$
satisfies the left Jacobi identity, cf. \cite[Sec.4.2]{BDSK09}.
Moreover, by the isotropicity assumption on $\mc L$
we have $\tint P\cdot G=-\tint Q\cdot F$ for $F\oplus P,\,G\oplus Q\in\mc L$,
so that the Courant-Dorfman product restricted to $\mf g$,
given by formula \eqref{20120127:eq1} is also skewsymmetric.
Therefore $\mf g$ is a Lie algebra with respect to $\circ$, and,
again by formula \eqref{20120127:eq1}, 
$\mf h$ contains the derived subalgebra $\mf g\circ\mf g$.
\end{remark}
Given two $\ell\times\ell$ 
matrix differential operators $A,B\in\Mat_{\ell\times\ell}\mc V[\partial]$
consider the following subspace of $\mc V^{\oplus\ell}\oplus\mc V^{\ell}$:
\begin{equation}\label{20120109:eq1}
\mc L_{A,B}=\big\{B(\partial)X\oplus A(\partial)X\,\big|\,X\in\mc V^{\oplus\ell}\big\} \,.
%\subset\mc V^{\oplus\ell}\oplus\mc V^{\ell}
\end{equation}
\begin{proposition}\label{20120103:propa}
The subspace $\mc L_{A,B}\subset\mc V^{\oplus\ell}\oplus\mc V^{\ell}$ is isotropic
with respect to the inner product \eqref{20111020:eq3}
if and only if 
\begin{equation}\label{20120107:eq1}
A^* B+B^* A=0\,.
\end{equation}
If, moreover, $B$ is non-degenerate,
then \eqref{20120107:eq1} holds if and only if
$A B^{-1}\in\Mat_{\ell\times\ell} \mc K((\partial^{-1}))$ 
is skewadjiont,
while if $A$ is non-degenerate 
then \eqref{20120107:eq1} holds if and only if
$B A^{-1}\in\Mat_{\ell\times\ell} \mc K((\partial^{-1}))$
is skewadjiont.
\end{proposition}
\begin{proof}
For $X,Y\in\mc V^{\oplus\ell}$ we have
$$
\langle B(\partial)X\oplus A(\partial)X\,|\,B(\partial)Y\oplus A(\partial)Y\rangle
=\tint Y\cdot\big(A^*(\partial)B(\partial)+B^*(\partial)A(\partial)\big)X\,.
$$
Hence, due to non-degeeracy of the pairing $(F|P)=\tint F\cdot P$,
the space $\mc L_{A,B}$ is isotropic if and only if \eqref{20120107:eq1} holds.
The remaining statements are straightforward.
\end{proof}
\begin{example}\label{20120124:ex0}
Letting $A\in\Mat_{\ell\times\ell}\mc V$ and $B=\id_\ell\,\partial$,
condition \eqref{20120107:eq1} holds
if and only if $A$ is a symmetric matrix with entries in $\mc C\subset\mc V$ 
(the subring of constant functions).
In this case $AB^{-1}$ is a skewadjont matrix pseudodifferential operator
and 
$\mc L_{A,B}=\big\{AX\oplus\partial X\,\big|\,X\in\mc V^{\oplus\ell}\big\}$ 
is an isotropic subspace of $\mc V^{\oplus\ell}\oplus\mc V^\ell$.
It is not hard to show directly that
$\mc L_{A,B}$ is maximal isotropic if and only if 
the matrix $A$ is non-degenerate.
When $\mc V=\mc K$ is a differential field, this is a corollary of the following general result:
\end{example}
\begin{proposition}[\cite{CDSK12b}]\label{prop:max-isotrop}
Let $\mc K$ be a differential field, 
and let $H=AB^{-1}$ be a minimal fractional decomposition
of the skewadjoint rational matrix pseudodifferential operator $H\in\Mat_{\ell\times\ell}\mc K(\partial)$.
Then the subspace $\mc L_{A,B}\subset\mc K^{\oplus\ell}\oplus\mc K^\ell$
is maximal isotropic with respect to the inner product \eqref{20111020:eq3}.
\end{proposition}
\begin{proposition}\label{20120103:propb}
Suppose that $A,B\in\Mat_{\ell\times\ell}\mc V[\partial]$
satisfy equation \eqref{20120107:eq1}.
Then the following conditions are equivalent:
\begin{enumerate}[(i)]
\item
$\langle X\circ Y,Z\rangle=0$ for all $X,Y,Z\in\mc L_{A,B}$.
\item 
for every $F,G\in\mc V^{\ell}$ one has:
\begin{equation}\label{20120103:eq1}
\begin{array}{l}
A^*(\partial)D_{B(\partial)G}(\partial)A(\partial)F
+A^*(\partial)D_{A(\partial)F}^*(\partial)B(\partial)G \\
-A^*(\partial)D_{B(\partial)F}(\partial)A(\partial)G
+A^*(\partial)D_{B(\partial)F}^*(\partial)A(\partial)G \\
+B^*(\partial)D_{A(\partial)G}(\partial)A(\partial)F
-B^*(\partial)D_{A(\partial)F}(\partial)A(\partial)G\,=\,0\,.
\end{array}
\end{equation}
\item
for every $i,j,k\in I$, one has in the space $\mc V[\lambda,\mu]$:
\end{enumerate}
\begin{equation}\label{20120103:eq2}
\begin{array}{l}
\displaystyle{
\sum_{\substack{s,t\in I \\ n\in\mb Z_+}}
\!\Bigg(\!
A^*_{ks}(\lambda+\mu+\partial) \bigg(
\frac{\partial B_{sj}(\mu)}{\partial u_t^{(n)}} (\lambda+\partial)^nA_{ti}(\lambda)
-
\frac{\partial B_{si}(\lambda)}{\partial u_t^{(n)}} (\mu+\partial)^nA_{tj}(\mu)
\!\bigg)
} \\
\displaystyle{
+B^*_{ks}(\lambda+\mu+\partial) \bigg(
\frac{\partial A_{sj}(\mu)}{\partial u_t^{(n)}} 
(\lambda+\partial)^nA_{ti}(\lambda)
-
\frac{\partial A_{si}(\lambda)}{\partial u_t^{(n)}} 
(\mu+\partial)^nA_{tj}(\mu)
\bigg)
} \\
\displaystyle{
+A^*_{ks}(\lambda+\mu+\partial) 
(-\lambda-\mu-\partial)^n
\bigg(
\!
\frac{\partial A_{ti}(\lambda)}{\partial u_s^{(n)}} 
B_{tj}(\mu)
+
\frac{\partial B_{ti}(\lambda)}{\partial u_s^{(n)}} 
A_{tj}(\mu)
\bigg)
\!\Bigg)
=\,0\,.
}
\end{array}
\end{equation}
\end{proposition}
\begin{proof}
Letting 
$X=B(\partial)F\oplus A(\partial)F,\,Y=B(\partial)G\oplus A(\partial)G,\,Z=B(\partial)E\oplus A(\partial)E$,
condition (i) reads
$$
\begin{array}{l}
\tint
\big(A(\partial)E\big)
\cdot \big(D_{B(\partial)G}(\partial)A(\partial)F
+D_{A(\partial)F}^*(\partial)B(\partial)G
-D_{B(\partial)F}(\partial) A(\partial)G 
\\
+D_{B(\partial)F}^*(\partial)A(\partial)G\big)
\!+\!
\big(B(\partial)E\big)
\!\!\cdot\!\!
\big(D_{A(\partial)G}(\partial)A(\partial)F
\!-\!
D_{A(\partial)F}(\partial)A(\partial)G\big)
\!=\!0
.
\end{array}
$$
Since the above equation holds for every $E\in\mc V^{\oplus\ell}$
it reduces, integrating by parts, to equation \eqref{20120103:eq1}.
For this we use the non-degeneracy of the pairing \eqref{20111020:eq3}.
This proves that conditions (i) and (ii) are equivalent.

We next prove that conditions (ii) and (iii) are equivalent,
provided that \eqref{20120107:eq1} holds.
For $\alpha=1,\dots,6$, let \eqref{20120103:eq1}$_\alpha$ 
be the $k$-entry of the $\alpha$-th term of the LHS of \eqref{20120103:eq1}:
for example
\eqref{20120103:eq1}$_1=\big(A^*(\partial)D_{B(\partial)G}(\partial)A(\partial)F\big)_k$.
We have, by the definition of the Frechet derivative and some algebraic manipulations
(similar to those used in the proof of \cite[Prop.1.16]{BDSK09}),
$$
\begin{array}{rcl}
\eqref{20120103:eq1}_1&=&
\displaystyle{
%A^*(\partial)D_{B(\partial)G}(\partial)A(\partial)F
\sum_{i,j,s,t\in I}\sum_{n\in\mb Z_+}
\bigg(
A^*_{ks} (\partial)\Big(\frac{\partial B_{sj}(\partial)}{\partial u_t^{(n)}} G_j\Big) \partial^n A_{ti}(\partial)F_i
}\\&&
\displaystyle{
+
A^*_{ks}(\partial)B_{sj}(\partial) \frac{\partial G_j}{\partial u_t^{(n)}} \partial^n A_{ti}(\partial)F_i
\bigg)
\,,} \\
\eqref{20120103:eq1}_2&=&
\displaystyle{
%A^*(\partial)D_{A(\partial)F}^*(\partial)B(\partial)G 
\sum_{i,j,s,t\in I}\sum_{n\in\mb Z_+}
\bigg(
A^*_{ks}(\partial) (-\partial)^n 
\Big(\frac{\partial A_{ti}(\partial)}{\partial u_s^{(n)}} F_i\Big)  B_{tj}(\partial)G_j
}\\&&
\displaystyle{
+
A^*_{ks}(\partial) (-\partial)^n
\frac{\partial F_i}{\partial u_s^{(n)}} A^*_{it}(\partial) B_{tj}(\partial)G_j
\bigg)
\,,} \\
\eqref{20120103:eq1}_3&=&
\displaystyle{
%-A^*(\partial)D_{B(\partial)F}(\partial)A(\partial)G
-\sum_{i,j,s,t\in I}\sum_{n\in\mb Z_+}
\bigg(
A^*_{ks} (\partial)\Big(\frac{\partial B_{si}(\partial)}{\partial u_t^{(n)}} F_i\Big) \partial^n A_{tj}(\partial)G_j
}\\&&
\displaystyle{
+
A^*_{ks}(\partial)B_{si}(\partial) \frac{\partial F_i}{\partial u_t^{(n)}} \partial^n A_{tj}(\partial)G_j
\bigg)
\,,} %\\
\end{array}
$$
%%%
$$
\begin{array}{rcl}
\eqref{20120103:eq1}_4&=&
\displaystyle{
%A^*(\partial)D_{B(\partial)F}^*(\partial)A(\partial)G 
\sum_{i,j,s,t\in I}\sum_{n\in\mb Z_+}
\bigg(
A^*_{ks}(\partial) (-\partial)^n 
\Big(\frac{\partial B_{ti}(\partial)}{\partial u_s^{(n)}} F_i\Big)  A_{tj}(\partial)G_j
}\\&&
\displaystyle{
+
A^*_{ks}(\partial) (-\partial)^n
\frac{\partial F_i}{\partial u_s^{(n)}} B^*_{it}(\partial) A_{tj}(\partial)G_j
\bigg)
\,,} \\
\eqref{20120103:eq1}_5&=&
\displaystyle{
%B^*(\partial)D_{A(\partial)G}(\partial)A(\partial)F
\sum_{i,j,s,t\in I}\sum_{n\in\mb Z_+}
\bigg(
B^*_{ks} (\partial)\Big(\frac{\partial A_{sj}(\partial)}{\partial u_t^{(n)}} G_j\Big) \partial^n A_{ti}(\partial)F_i
}\\&&
\displaystyle{
+
B^*_{ks}(\partial)A_{sj}(\partial) \frac{\partial G_j}{\partial u_t^{(n)}} \partial^n A_{ti}(\partial)F_i
\bigg)
\,,}
\end{array}
$$
$$
\begin{array}{rcl}
\eqref{20120103:eq1}_6&=&
\displaystyle{
%-B^*(\partial)D_{A(\partial)F}(\partial)A(\partial)G
-\sum_{i,j,s,t\in I}\sum_{n\in\mb Z_+}
\bigg(
B^*_{ks} (\partial)\Big(\frac{\partial A_{si}(\partial)}{\partial u_t^{(n)}} F_i\Big) \partial^n A_{tj}(\partial)G_j
}\\&&
\displaystyle{
+
B^*_{ks}(\partial)A_{si}(\partial) \frac{\partial F_i}{\partial u_t^{(n)}} \partial^n A_{tj}(\partial)G_j
\bigg)
\,.}
\end{array}
$$
Combining the second terms in the RHS 
of \eqref{20120103:eq1}$_1$ and \eqref{20120103:eq1}$_5$
we get zero, thanks to equation \eqref{20120107:eq1}.
Similarly, we get zero if we combine the second terms in the RHS 
of \eqref{20120103:eq1}$_2$ and \eqref{20120103:eq1}$_4$,
and if we combine the second terms in the RHS 
of \eqref{20120103:eq1}$_3$ and \eqref{20120103:eq1}$_6$.
Equation \eqref{20120103:eq2} thus follows from equation \eqref{20120103:eq1}
once we replace $\partial$ acting on $F_i$ by $\lambda$,
and $\partial$ acting on $G_j$ by $\mu$.
\end{proof}
\begin{remark}
It follows by Definition \ref{20111020:def} and Proposition \ref{20120103:propb}
that a Dirac structure is a maximal isotropic 
subspace $\mc L$ of $\mc V^{\oplus\ell}\oplus\mc V^\ell$
satisfying one of the equivalent conditions (i)--(iii) above.
\end{remark}
\begin{proposition}\label{20120103:propc}
Suppose that $A,B\in\Mat_{\ell\times\ell}\mc V[\partial]$
satisfy equation \eqref{20120107:eq1}
and assume that $B$ is non-degenerate.
Suppose, moreover, that
the (skewadjoint) rational matrix pseudodifferential operator 
$H=A B^{-1}$ has coefficients in $\mc V$, i.e. $H\in\Mat_{\ell\times\ell}\mc V(\partial)$.
Consider the corresponding non-local $\lambda$-bracket $\{\cdot\,_\lambda\,\cdot\}_H$
given by the Master Formula \eqref{20110922:eq1}.
Then the Jacobi identity \eqref{20110922:eq4} on $\{\cdot\,_\lambda\,\cdot\}_H$
is equivalent to equation \eqref{20120103:eq2} on the entries of matrices 
$A$ and $B$.
\end{proposition}
\begin{proof}
Letting $A_{st}(\partial)=\sum_{m=0}^Ma_{st;m}\partial^m$
and $B_{st}(\partial)=\sum_{n=0}^M b_{ij;n}\partial^n$.
By formula \eqref{20110923:eq1} and the left Leibniz rule \eqref{20110921:eq3}
we have,
$$
\begin{array}{c}
\displaystyle{
\{{u_i}_\lambda\{{u_j}_\mu{u_k}\}_H\}_H
=
\sum_{r\in I}
\{{u_i}_\lambda A_{kr}(\mu+\partial)B^{-1}_{rj}(\mu)\}_H
} \\
\displaystyle{
=
\sum_{r\in I}\sum_{m=0}^M
\{{u_i}_\lambda a_{kr;m}\}_H (\mu+\partial)^m B^{-1}_{rj}(\mu)
} \\
\displaystyle{
+\sum_{r\in I}
A_{kr}(\lambda+\mu+\partial) \{{u_i}_\lambda B^{-1}_{rj}(\mu)\}_H 
\,.} 
\end{array}
$$
By equation \eqref{20111012:eq2c} we have
$$
\begin{array}{l}
\displaystyle{
\{{u_i}_\lambda B^{-1}_{rj}(\mu)\}_H
} \\
\displaystyle{
=
-\sum_{s,t\in I} \sum_{m=0}^M
(B^{-1})_{rs}(\lambda+\mu+\partial) \{{u_i}_\lambda b_{st;m}\}_H (\mu+\partial)^m (B^{-1})_{tj}(\mu)\,.
}
\end{array}
$$
Combining the above two equations we then get, using the Master Formula \eqref{20110922:eq1},
$$
\begin{array}{l}
\displaystyle{
\{{u_i}_\lambda\{{u_j}_\mu{u_k}\}_H\}_H
%=
%\sum_{r\in I}\sum_{m=0}^M
%\{{u_i}_\lambda a_{kr;m}\}_H (\mu+\partial)^m B^{-1}_{rj}(\mu)
%} \\
%\displaystyle{
%-\sum_{r,s,t\in I} \sum_{m=0}^M
%A_{kr}(\lambda+\mu+\partial) 
%(B^{-1})_{rs}(\lambda+\mu+\partial) \{{u_i}_\lambda b_{st;m}\}_H (\mu+\partial)^m (B^{-1})_{tj}(\mu)
} \\
\displaystyle{
=
\sum_{r,s,t\in I}\sum_{m=0}^M\sum_{n\in\mb Z_+}
\frac{\partial a_{kr;m}}{\partial u_s^{(n)}} 
\Big(
(\mu+\partial)^m B^{-1}_{rj}(\mu)
\Big)
\Big(
(\lambda+\partial)^n
A_{st}(\lambda+\partial)
B^{-1}_{ti}(\lambda)
\Big)
} \\
\displaystyle{
-\sum_{p,q,r,s,t\in I} \sum_{m=0}^M \sum_{n\in\mb Z_+}
A_{kr}(\lambda+\mu+\partial) 
(B^{-1})_{rs}(\lambda+\mu+\partial) 
} \\
\displaystyle{
\,\,\,\,\,\,\,\,\,\,\,\,\,\,\,\,\,\,\,\,\,\,\,\,\,\,\,\,\,\,\,\,\,
\frac{\partial b_{st;m}}{\partial u_p^{(n)}}
\Big((\mu+\partial)^m (B^{-1})_{tj}(\mu)\Big)
\Big((\lambda+\partial)^n
A_{pq}(\lambda+\partial) B^{-1}_{qi}(\lambda)\Big)
} \\
\displaystyle{
=
\sum_{r,s,t\in I}\sum_{n\in\mb Z_+}
\Big(
\frac{\partial A_{kr}(\mu+\partial)}{\partial u_s^{(n)}} 
B^{-1}_{rj}(\mu)
\Big)
\Big(
(\lambda+\partial)^n
A_{st}(\lambda+\partial)
B^{-1}_{ti}(\lambda)
\Big)
} \\
\displaystyle{
-\sum_{p,q,r,s,t\in I} \sum_{n\in\mb Z_+}
A_{kr}(\lambda+\mu+\partial) 
(B^{-1})_{rs}(\lambda+\mu+\partial) 
} \\
\displaystyle{
\,\,\,\,\,\,\,\,\,\,\,\,\,\,\,\,\,\,\,\,\,\,\,\,\,\,\,\,\,\,\,\,\,
\Big(
\frac{\partial B_{st}(\mu+\partial)}{\partial u_p^{(n)}}
(B^{-1})_{tj}(\mu)\Big)
\Big((\lambda+\partial)^n
A_{pq}(\lambda+\partial) B^{-1}_{qi}(\lambda)\Big)
\,.}
\end{array}
$$
Next, we apply $B^*_{k'k}(\lambda+\mu+\partial)$ to both sides of the above equation (on the left),
replace $\lambda$ by $\lambda+\partial$ acting on $B_{ii'}(\lambda)$,
replace $\mu$ by $\mu+\partial$ acting on $B_{jj'}(\mu)$,
and sum over $i,j,k\in I$.
As a result we get, using the assumption \eqref{20120107:eq1}
(see notation \eqref{20111018:eq5}),
\begin{equation}\label{20120109:eq2}
\begin{array}{c}
\displaystyle{
\sum_{i,j,k\in I}B^*_{k'k}(\lambda+\mu+\partial)
\{{u_i}_x\{{u_j}_y{u_k}\}_H\}_H
\Big(\Big|_{x=\lambda+\partial}B_{ii'}(\lambda)\Big)
\Big(\Big|_{y=\mu+\partial}B_{jj'}(\mu)\Big)
} \\
\displaystyle{
=
\sum_{k,i\in I}\sum_{n\in\mb Z_+}
\bigg(
B^*_{k'k}(\lambda+\mu+\partial)
\frac{\partial A_{kj'}(\mu)}{\partial u_i^{(n)}} 
(\lambda+\partial)^n
A_{ii'}(\lambda)
} \\
\displaystyle{
+
A^*_{k'k}(\lambda+\mu+\partial)
\frac{\partial B_{kj'}(\mu)}{\partial u_i^{(n)}}
(\lambda+\partial)^n
A_{ii'}(\lambda) 
\bigg)
\,.}
\end{array}
\end{equation}
Exchanging $i'$ with $j'$ and $\lambda$ with $\mu$ in \eqref{20120109:eq2}, we get
\begin{equation}\label{20120109:eq3}
\begin{array}{c}
\displaystyle{
\sum_{i,j,k\in I}B^*_{k'k}(\lambda+\mu+\partial)
\{{u_j}_y\{{u_i}_x{u_k}\}_H\}_H
\Big(\Big|_{x=\lambda+\partial}B_{ii'}(\lambda)\Big)
\Big(\Big|_{y=\mu+\partial}B_{jj'}(\mu)\Big)
} \\
\displaystyle{
=
\!\sum_{i,j,k\in I}B^*_{k'k}(\lambda+\mu+\partial)
\{{u_i}_x\{{u_j}_y{u_k}\}_H\}_H
\Big(\Big|_{x=\mu+\partial}B_{ij'}(\mu)\Big)
\Big(\Big|_{y=\lambda+\partial}B_{ji'}(\lambda)\Big)
} \\
\displaystyle{
=
\sum_{k,j\in I}\sum_{n\in\mb Z_+}
\bigg(
B^*_{k'k}(\lambda+\mu+\partial)
\frac{\partial A_{ki'}(\lambda)}{\partial u_j^{(n)}} 
(\mu+\partial)^n
A_{jj'}(\mu)
} \\
\displaystyle{
+
A^*_{k'k}(\lambda+\mu+\partial)
\frac{\partial B_{ki'}(\lambda)}{\partial u_j^{(n)}}
(\mu+\partial)^n
A_{jj'}(\mu) 
\bigg)
\,.}
\end{array}
\end{equation}
We are left to study the third term in the Jacobi identity.
By the right Leibniz rule \eqref{20110921:eq3} we get,
$$
\begin{array}{l}
\displaystyle{
\{{\{{u_i}_\lambda{u_j}\}_H}_{\lambda+\mu}{u_k}\}_H
=
\sum_{r\in I}
\{A_{jr}(\lambda+\partial)B^{-1}_{ri}(\lambda)_{\lambda+\mu}{u_k}\}_H
} \\
\displaystyle{
=
\sum_{r\in I}\sum_{m=0}^M
{\{{a_{jr;m}}_{\lambda+\mu+\partial}{u_k}\}_H}_\to
(\lambda+\partial)^m B^{-1}_{ri}(\lambda)
} \\
\displaystyle{
+
\sum_{r\in I}\sum_{m=0}^M
{\{B^{-1}_{ri}(\lambda)_{\lambda+\mu+\partial}{u_k}\}_H}_\to
(-\mu-\partial)^m a_{jr;m}
\,.} 
\end{array}
$$
By equation \eqref{20111012:eq2d} we have
$$
\begin{array}{l}
\displaystyle{
{\{B^{-1}_{ri}(\lambda)_{\lambda+\mu+\partial}{u_k}\}_H}_\to
} \\
\displaystyle{
=
-\sum_{s,t=1}^\ell\sum_{m=0}^M
\{{b_{st;m}}_{\lambda+\mu+\partial}{u_k}\}_\to
\circ
\Big((\lambda+\partial)^m (B^{-1})_{ti}(\lambda)\Big)
({B^*}^{-1})_{sr}(\mu+\partial) 
\,.}
\end{array}
$$
Combining the above two equations and using the Master Formula \eqref{20110922:eq1} we then get
$$
\begin{array}{l}
\displaystyle{
\{{\{{u_i}_\lambda{u_j}\}_H}_{\lambda+\mu}{u_k}\}_H
} \\
\displaystyle{
=
\sum_{\substack{r,s,t\in I \\ n\in\mb Z_+}}
A_{kt}(\lambda+\mu+\partial)B^{-1}_{ts}(\lambda+\mu+\partial)
(-\lambda-\mu-\partial)^n
\frac{\partial A_{jr}(\lambda+\partial)}{\partial u_s^{(n)}} B^{-1}_{ri}(\lambda)
} \\
\displaystyle{
-
\sum_{\substack{r,p,q,s,t\in I \\ n\in\mb Z_+}}
A_{kq}(\lambda+\mu+\partial)B^{-1}_{qp}(\lambda+\mu+\partial)
(-\lambda-\mu-\partial)^n
} \\
\displaystyle{
\,\,\,\,\,\,\,\,\,\,\,\,\,\,\,\,\,\,\,\,\,\,\,\,\,\,\,\,\,\,\,\,\,\,\,\,
\Big(
\frac{\partial B_{st}(\lambda+\partial)}{\partial u_p^{(n)}}
(B^{-1})_{ti}(\lambda)\Big)
\Big(({B^*}^{-1})_{sr}(\mu+\partial) A^*_{rj}(\mu)\Big)
\,.}
\end{array}
$$
Hence, if we apply, as before, $B^*_{k'k}(\lambda+\mu+\partial)$ on the left,
replace $\lambda$ by $\lambda+\partial$ acting on $B_{ii'}(\lambda)$,
replace $\mu$ by $\mu+\partial$ acting on $B_{jj'}(\mu)$,
and sum over $i,j,k\in I$, we get, using \eqref{20120107:eq1},
\begin{equation}\label{20120109:eq4}
\begin{array}{l}
\displaystyle{
\sum_{i,j,k\in I}B^*_{k'k}(\lambda+\mu+\partial)
\{{\{{u_i}_x{u_j}\}_H}_{x+y}{u_k}\}_H
\Big(\Big|_{x=\lambda+\partial}B_{ii'}(\lambda)\Big)
} \\
\displaystyle{
\Big(\Big|_{y=\mu+\partial}B_{jj'}(\mu)\Big)
=
-\sum_{j,k\in I}\sum_{n\in\mb Z_+}
A^*_{k'k}(\lambda+\mu+\partial)
(-\lambda-\mu-\partial)^n
} \\
\displaystyle{
\,\,\,\,\,\,\,\,\,\,\,\,\,\,\,\,\,\,\,\,\,\,\,\,\,\,\,\,\,\,\,\,\,\,\,\,
\,\,\,\,\,\,\,\,\,\,\,\,\,\,\,\,\,\,
\bigg(
\frac{\partial A_{ji'}(\lambda)}{\partial u_k^{(n)}} B_{jj'}(\mu)
+
\frac{\partial B_{ji'}(\lambda)}{\partial u_k^{(n)}}
A_{jj'}(\mu)
\bigg)
\,.}
\end{array}
\end{equation}
Combining \eqref{20120109:eq2}, \eqref{20120109:eq3} and \eqref{20120109:eq4},
we get that the expression
$$
\begin{array}{l}
\displaystyle{
\sum_{i,j,k\in I}B^*_{k'k}(\lambda+\mu+\partial)
\Big(
\{{u_i}_x\{{u_j}_y{u_k}\}_H\}_H
-\{{u_j}_y\{{u_i}_x{u_k}\}_H\}_H
} \\
\displaystyle{
-\{{\{{u_i}_x{u_j}\}_H}_{x+y}{u_k}\}_H
\Big)
\Big(\Big|_{x=\lambda+\partial}B_{ii'}(\lambda)\Big)
\Big(\Big|_{y=\mu+\partial}B_{jj'}(\mu)\Big)
}
\end{array}
$$
is the same as the LHS of \eqref{20120103:eq2}.
The claim follows from the assumption that 
the matrix $B\in\Mat_{\ell\times\ell}\mc V[\partial]$
has non-zero Dieudonn\`e determinant.
\end{proof}
\begin{theorem}\label{20111020:thm}
Let $\mc V$ be an algebra of differential functions in the variables $u_1,\dots,u_\ell$,
which is a domain, and let $\mc K$ be its field of fractions.
Let $H=A B^{-1}$, with $A,B\in\Mat_{\ell\times\ell}\mc V[\partial]$, $B$ non-degenerate,
be a minimal fractional decomposition (cf. Definition \ref{def:minimal-fraction} 
and Remark \ref{rem:minimal-fraction2}) of the rational matrix pseudodifferential 
operator $H\in\Mat_{\ell\times\ell}\mc V(\partial)$.
Then the subspace 
\begin{equation}\label{eq:dirac}
\mc L_{A,B}(\mc K)=\big\{B(\partial)X\oplus A(\partial)A\,\big|\,X\in\mc K^{\oplus\ell}\big\}
\,\subset\mc K^{\oplus\ell}\oplus\mc K^\ell\,,
\end{equation}
is a Dirac structure if and only if
$H$ is a non-local Poisson structure over $\mc V$.
\end{theorem}
\begin{proof}
It immediately follows from Remark \ref{rem:minimal-fraction2} 
and Propositions \ref{20120103:propa}, \ref{prop:max-isotrop}, 
\ref{20120103:propb} and \ref{20120103:propc}.
\end{proof}
\begin{remark}
We may define a ``generalized'' Dirac structure
as a subspace $\mc L$ of $\mc V^{\oplus\ell}\oplus\mc V^\ell$,
such that $\mc L\subset\mc L^\perp$ (i.e. $\mc L$ is isotropic),
and $\mc L\circ\mc L\subset\mc L^\perp$ (i.e. condition (i) in Proposition \ref{20120103:propb} holds),
where $\mc L^\perp$ is the orthogonal complement to $\mc L$ 
with respect to the inner product \eqref{20111020:eq3}.
Note that a Dirac structure is a special case of this when $\mc L$ is maximal isotropic.
If $A,B\in\Mat_{\ell\times\ell}\mc V[\partial]$, with $B$ non-degenerate,
then $\mc L_{A,B}$ is a generalized Dirac structure
if and only if $H=AB^{-1}$ is a non-local Poisson structure over $\mc V$
(not necessarily in its minimal fractional decomposition).
Note also that any subspace of a generalized Dirac structure is a generalized Dirac structure.
\end{remark}

%%%
\subsection{Compatible pairs of Dirac structures}
\label{sec:6.3}

The notion of compatibility of Dirac structures was introduced by Gelfand and Dorfman 
\cite{GD80}, \cite{Dor93} (see also \cite{BDSK09}).
In this paper we introduce a weaker, but more natural, 
notion of compatibility, which still can be used to implement successfully 
the Lenard-Magri scheme of integrability, and which is more closely related 
to the notion of compatibility of the corresponding non-local Poisson structures.

Given two Dirac structures $\mc L$ and $\mc L^\prime\subset\mc V^{\oplus\ell}\oplus\mc V^\ell$,
we define the relations 
\begin{equation}\label{eq:20090322_1}
\begin{array}{l}
\mc N_{\mc L,\mc L^\prime}
=
\big\{P\oplus P^\prime
%\in\mc V^\ell\oplus\mc V^\ell
\,\big|\,
F\oplus P\in\mc L,\,F\oplus P^\prime\in\mc L^\prime\,\text{ for some } F\in\mc V^{\oplus\ell}\big\}
\subset\mc V^\ell\oplus\mc V^\ell\,,
\\
\mc N_{\mc L,\mc L^\prime}\wcheck
=
\big\{F\oplus F^\prime
%\in\mc V^{\oplus\ell}\oplus\mc V^{\oplus\ell}
\,\big|\,
F\oplus P\in\mc L,\,F^\prime\oplus P\in\mc L^\prime\,\text{ for some } P\in\mc V^{\ell}\big\}
\subset\mc V^{\oplus\ell}\!\oplus\!\mc V^{\oplus\ell}.
\end{array}
\end{equation}
\begin{definition}\label{2006_NRel}
Two Dirac structures $\mc{L},\,\mc{L}^\prime\,\subset \mc V^{\oplus\ell}\oplus\mc V^\ell$ 
are said to be \emph{compatible} if 
for all $P,P^\prime,Q,Q^\prime\in\mc V^\ell,\,
F,F^\prime,F^{\prime\prime}\in\mc V^{\oplus\ell}$
such that
$$
P\oplus P^\prime,\,Q\oplus Q^\prime\,\in\mc N_{\mc L,\mc L^\prime}
\,\,\,\,\text{ and }\,\,\,\,
F\oplus F^\prime,\,F^\prime\oplus F^{\prime\prime}
\in \mc N_{\mc L,\mc L^\prime}\wcheck\,,
$$ 
we have
\begin{equation}\label{eq:20090320_1}
(F|[P,Q])-(F^\prime|[P,Q^\prime])-(F^\prime|[P^\prime,Q])
+(F^{\prime\prime}|[P^\prime,Q^\prime])\,=\,0\,,
\end{equation}
where, as before, $(F|P)=\tint F\cdot P$,
and, for $P,Q\in\mc V^\ell$, $[P,Q]$ is given by \eqref{20120126:eq1}.
\end{definition}
\begin{remark}\label{oldcompatibility}
The original notion of compatibility, introduced by Dorfman in \cite{Dor93},
is similar, except that $\mc N_{\mc L,\mc L^\prime}\wcheck$
is replaced by the ``dual'' relation
$$
\mc N^*_{\mc L,\mc L^\prime}
\,=\,
\big\{F\oplus F^\prime\in\mc V^{\oplus\ell}\oplus\mc V^{\oplus\ell}
\,\big|\,
\tint F\cdot P=\tint F^\prime\cdot P^\prime\,\text{ for all } P\oplus P^\prime
\in\mc N_{\mc L,\mc L^\prime}\big\}\,.
$$
Since $\mc L$ and $\mc L^\prime$ are isotropic, we have, for 
$F\oplus F^\prime\in\mc N_{\mc L,\mc L^\prime}\wcheck$,
and for $Q\oplus Q^\prime\in\mc N_{\mc L,\mc L^\prime}$,
$\tint F\cdot Q=-\tint G\cdot P=\tint F^\prime\cdot Q^\prime$,
where $P\in\mc V^\ell$ and $G\in\mc V^{\oplus\ell}$ are such that 
$F\oplus P,G\oplus Q\in\mc L,\,F^\prime\oplus P,G\oplus Q^\prime\in\mc L^\prime$.
Hence, $\mc N_{\mc L,\mc L^\prime}\wcheck\subset\mc N^*_{\mc L,\mc L^\prime}$.
\end{remark}

Even with the weaker notion of compatibility,
the following important theorem still holds (cf. \cite[Thm.4.13]{BDSK09}).
\begin{theorem}\label{mtst}
Let $(\mc L,\mc L^\prime)$ be a pair of compatible Dirac structures.
Let $F_0,F_1,F_2\in\mc V^{\oplus\ell}$ be such that:
\begin{enumerate}[(i)]
\item
$F_0$ and $F_1$ are closed, i.e.
$D_{F_n}^*(\partial)=D_{F_n}(\partial)$ for $n=0,1$;
\item
$F_0\oplus F_1,\,F_1\oplus F_2\,\in\mc N_{\mc L,\mc L^\prime}\wcheck$.
\end{enumerate}
Then, for all $P\oplus P^\prime,Q\oplus Q^\prime\in\mc N_{\mc L,\mc L^\prime}$, we have
\begin{equation}\label{20120405:eq2}
\tint Q^\prime\cdot \big(D_{F_2}(\partial)-D_{F_2}^*(\partial)\big)P^\prime
=0\,.
\end{equation}
\end{theorem}
\begin{proof}
By the assumption \eqref{eq:20090320_1}, we have
$$
\begin{array}{l}
\displaystyle{
0=(F_0|[P,Q])-(F_1|[P,Q^\prime])-(F_1|[P^\prime,Q])
+(F_2|[P^\prime,Q^\prime])
} \\
\displaystyle{
=\int \Big(
F_0\cdot D_Q(\partial)P -F_0\cdot D_P(\partial)Q
-F_1\cdot D_{Q^\prime}(\partial)P +F_1\cdot D_P(\partial)Q^\prime
} \\
\displaystyle{
-F_1\cdot D_Q(\partial)P^\prime +F_1\cdot D_{P^\prime}(\partial)Q
+F_2\cdot D_{Q^\prime}(\partial)P^\prime -F_2\cdot D_{P^\prime}(\partial)Q^\prime
\Big)
} \\
\displaystyle{
=\int P\cdot \frac{\delta}{\delta u} \big((F_0|Q)-(F_1|Q^\prime)\big)
-\int Q\cdot\frac{\delta}{\delta u} \big((F_0|P)-(F_1|P^\prime)\big)
} \\
\displaystyle{
-\int P^\prime\cdot \frac{\delta}{\delta u}\big((F_1|Q)- (F_2|Q^\prime)\big)
+\int Q^\prime\cdot \frac{\delta}{\delta u} \big((F_1|P)-(F_2|P^\prime)\big)
} \\
\displaystyle{
-\int Q\cdot D_{F_0}(\partial)P 
+\int P\cdot D_{F_0}(\partial)Q
+\int {Q^\prime}\cdot D_{F_1}(\partial)P 
-\int P\cdot D_{F_1}(\partial)Q^\prime
} \\
\displaystyle{
+\int Q\cdot D_{F_1}(\partial)P^\prime 
-\int {P^\prime}\cdot D_{F_1}(\partial)Q
-\int {Q^\prime}\cdot D_{F_2}(\partial)P^\prime 
+\int {P^\prime}\cdot D_{F_2}(\partial)Q^\prime
\,.
}\end{array}
$$
In the second identity we used the definition \eqref{20120126:eq1} of the Lie bracket on $\mc V^{\ell}$,
and in the last identity we used equation \eqref{20120405:eq1}.
Since, by assumption, $F_0\oplus F_1\in\mc N_{\mc L,\mc L^\prime}\wcheck$ 
and $Q\oplus Q^\prime\in\mc N_{\mc L,\mc L^\prime}$, we have (by Remark \ref{oldcompatibility})
that $(F_0|Q)=(F_1|Q^\prime)$. 
Hence the first term in the RHS above is zero,
and, by the same argument, the first four terms are zero.
The following six terms are also zero since, by assumption, $D_{F_0}(\partial)$
and $D_{F_1}(\partial)$ are selfadjoint.
In conclusion, equation \eqref{20120405:eq2} holds.
\end{proof}

%%%
\subsection{Compatible non-local Poisson structures 
and the corresponding compatible pairs of Dirac structures}
\label{sec:6.4}

In Theorem \ref{20111020:thm} we proved that to a non-local Hamiltonian 
structure $H\in\Mat_{\ell\times\ell}\mc V(\partial)$
in its minimal fractional decomposition $H=AB^{-1}$, 
with $A,B\in\Mat_{\ell\times\ell}\mc V[\partial]$,
there corresponds a Dirac structure $\mc L_{A,B}(\mc K)$ 
over the field of frations $\mc K$.
In this section we prove that to a compatible pair 
of non-local Poisson structures $H=AB^{-1},\,K=CD^{-1}$,
in their minimal fractional decompositions,
there corresponds a compatible pair 
of Dirac structures $\mc L_{A,B}(\mc K),\,\mc L_{C,D}(\mc K)$ over $\mc K$.
This is stated in the following:
\begin{theorem}\label{20120126:prop2}
Let $\mc V$ be an algebra of differential functions in $u_1,\dots,u_\ell$,
which is a domain, and let $\mc K$ be its field of fractions.
Let $H,\,K\in\Mat_{\ell\times\ell}\mc V(\partial)$ be compatible 
non-local Poisson structures over $\mc V$.
Let $H=AB^{-1},\,K=CD^{-1}$ be their minimal fractional decompositions
(cf. Definition \ref{def:minimal-fraction} and Remark \ref{rem:minimal-fraction2}).
Then $\mc L_{A,B}(\mc K)$ and $\mc L_{C,D}(\mc K)$
are compatible Dirac structures over $\mc K$.
\end{theorem}

By Theorem \ref{20110923:prop}, 
the Poisson structures $H$ and $K$ over $\mc V$ are compatible
if and only if we have the following ``mixed'' Jacobi identity on generators
($i,j,k\in I$):
\begin{equation}\label{20120405:eq3}
\begin{array}{l}
\{{u_i}_\lambda\{{u_j}_\mu {u_k}\}_H\}_K-\{{u_j}_\mu\{{u_i}_\lambda {u_k}\}_H\}_K
-\{{\{{u_i}_\lambda {u_j}\}_H}_{\lambda+\mu} {u_k}\}_K \\
+\{{u_i}_\lambda\{{u_j}_\mu {u_k}\}_K\}_H-\{{u_j}_\mu\{{u_i}_\lambda {u_k}\}_K\}_H
-\{{\{{u_i}_\lambda {u_j}\}_K}_{\lambda+\mu} {u_k}\}_H
=0\,,
\end{array}\end{equation}

In order to relate the above condition to the compatibility 
of the corresponding Dirac structures $\mc L_{A,B}$ and $\mc L_{C,D}$,
we need to compute explicitly each term of the above equation.
This is done in the following:
\begin{lemma}\label{20120405:lem1}
Suppose that the pairs $(A,B)$ and $(C,D)$, with $A,B,C,D\in\Mat_{\ell\times\ell}\mc V[\partial]$,
satisfy equation \eqref{20120107:eq1}:
\begin{equation}\label{20120405:skew}
A^* B+B^* A=0
\,\,,\,\,\,\,
C^* D+D^* C=0\,.
\end{equation}
Assume that $B$ and $D$ are non-degenerate,
and that the (skewadjoint) rational matrix pseudodifferential operators 
$H=A B^{-1}$ and $K=CD^{-1}$
have coefficients in $\mc V$, i.e. $H,K\in\Mat_{\ell\times\ell}\mc V(\partial)$.
Consider the corresponding non-local $\lambda$-brackets 
$\{\cdot\,_\lambda\,\cdot\}_H$ and $\{\cdot\,_\lambda\,\cdot\}_K$
given by the Master Formula \eqref{20110922:eq1}.
Then, in terms of notation \eqref{20111018:eq5}, 
we have the following identities for every $i',j',k'\in I$:
\begin{equation}\label{20120405:A1}
\begin{array}{c}
%A1
\displaystyle{
\sum_{i,j,k\in I}
B^*_{k'k}(\lambda+\mu+\partial)
\{{u_i}_x\{{u_j}_y {u_k}\}_H\}_K
\big(\big|_{x=\lambda+\partial}D_{ii'}(\lambda)\big)
\big(\big|_{y=\mu+\partial}B_{jj'}(\mu)\big)
} \\
\displaystyle{
=
\sum_{i,k\in I}\sum_{n\in\mb Z_+}
B^*_{k'k}(\lambda+\mu+\partial)
\frac{\partial A_{kj'}(\mu)}{\partial u_i^{(n)}}
(\lambda+\partial)^n C_{ii'}(\lambda)
} \\
\displaystyle{
+\sum_{i,k\in I}\sum_{n\in\mb Z_+}
A^*_{k'k}(\lambda+\mu+\partial)
\frac{\partial B_{kj'}(\mu)}{\partial u_i^{(n)}}
(\lambda+\partial)^n C_{ii'}(\lambda)
\,,} 
\end{array}
\end{equation}
\begin{equation}\label{20120405:A2}
\begin{array}{c}
%A2
\displaystyle{
\sum_{i,j,k\in I}
D^*_{k'k}(\lambda+\mu+\partial)
\{{u_i}_x\{{u_j}_y {u_k}\}_K\}_H
\big(\big|_{x=\lambda+\partial}B_{ii'}(\lambda)\big)
\big(\big|_{y=\mu+\partial}D_{jj'}(\mu)\big)
} \\
\displaystyle{
=
\sum_{i,k\in I}\sum_{n\in\mb Z_+}
D^*_{k'k}(\lambda+\mu+\partial)
\frac{\partial C_{kj'}(\mu)}{\partial u_i^{(n)}}
(\lambda+\partial)^n A_{ii'}(\lambda)
} \\
\displaystyle{
+\sum_{i,k\in I}\sum_{n\in\mb Z_+}
C^*_{k'k}(\lambda+\mu+\partial)
\frac{\partial D_{kj'}(\mu)}{\partial u_i^{(n)}}
(\lambda+\partial)^n A_{ii'}(\lambda)
\,,} 
\end{array}
\end{equation}
\begin{equation}\label{20120405:B1}
\begin{array}{c}
%B1
\displaystyle{
\sum_{i,j,k\in I}
B^*_{k'k}(\lambda+\mu+\partial)
\{{u_j}_y\{{u_i}_x {u_k}\}_H\}_K
\big(\big|_{x=\lambda+\partial}B_{ii'}(\lambda)\big)
\big(\big|_{y=\mu+\partial}D_{jj'}(\mu)\big)
} \\
\displaystyle{
=
\sum_{j,k\in I}\sum_{n\in\mb Z_+}
B^*_{k'k}(\lambda+\mu+\partial)
\frac{\partial A_{ki'}(\lambda)}{\partial u_j^{(n)}}
(\mu+\partial)^n C_{jj'}(\mu)
} \\
\displaystyle{
+\sum_{j,k\in I}\sum_{n\in\mb Z_+}
A^*_{k'k}(\lambda+\mu+\partial)
\frac{\partial B_{ki'}(\lambda)}{\partial u_j^{(n)}}
(\mu+\partial)^n C_{jj'}(\mu)
\,,} 
\end{array}
\end{equation}
\begin{equation}\label{20120405:B2}
\begin{array}{c}
%B2
\displaystyle{
\sum_{i,j,k\in I}
D^*_{k'k}(\lambda+\mu+\partial)
\{{u_j}_y\{{u_i}_x {u_k}\}_K\}_H
\big(\big|_{x=\lambda+\partial}D_{ii'}(\lambda)\big)
\big(\big|_{y=\mu+\partial}B_{jj'}(\mu)\big)
} \\
\displaystyle{
=
\sum_{j,k\in I}\sum_{n\in\mb Z_+}
D^*_{k'k}(\lambda+\mu+\partial)
\frac{\partial C_{ki'}(\lambda)}{\partial u_j^{(n)}}
(\mu+\partial)^n A_{jj'}(\mu)
} \\
\displaystyle{
+\sum_{j,k\in I}\sum_{n\in\mb Z_+}
C^*_{k'k}(\lambda+\mu+\partial)
\frac{\partial D_{ki'}(\lambda)}{\partial u_j^{(n)}}
(\mu+\partial)^n A_{jj'}(\mu)
\,,} 
\end{array}
\end{equation}
\begin{equation}\label{20120405:C1}
\begin{array}{c}
%C1
\displaystyle{
\sum_{i,j,k\in I}
D^*_{k'k}(\lambda+\mu+\partial)
\{{\{{u_i}_x{u_j}\}_H}_{x+y}{u_k}\}_K
\big(\big|_{x=\lambda+\partial}B_{ii'}(\lambda)\big)
\big(\big|_{y=\mu+\partial}B_{jj'}(\mu)\big)
} \\
\displaystyle{
=
-\sum_{j,k\in I}\sum_{n\in\mb Z_+}
C^*_{k'k}(\lambda+\mu+\partial)
(-\lambda-\mu-\partial)^n
\frac{\partial A_{ji'}(\lambda)}{\partial u_k^{(n)}}
B_{jj'}(\mu)
} \\
\displaystyle{
-\sum_{j,k\in I}\sum_{n\in\mb Z_+}
C^*_{k'k}(\lambda+\mu+\partial)
(-\lambda-\mu-\partial)^n
\frac{\partial B_{ji'}(\lambda)}{\partial u_k^{(n)}}
A_{jj'}(\mu)
\,,} 
\end{array}
\end{equation}
\begin{equation}\label{20120405:C2}
\begin{array}{c}
%C2
\displaystyle{
\sum_{i,j,k\in I}
B^*_{k'k}(\lambda+\mu+\partial)
\{{\{{u_i}_x{u_j}\}_K}_{x+y}{u_k}\}_H
\big(\big|_{x=\lambda+\partial}D_{ii'}(\lambda)\big)
\big(\big|_{y=\mu+\partial}D_{jj'}(\mu)\big)
} \\
\displaystyle{
=
-\sum_{j,k\in I}\sum_{n\in\mb Z_+}
A^*_{k'k}(\lambda+\mu+\partial)
(-\lambda-\mu-\partial)^n
\frac{\partial C_{ji'}(\lambda)}{\partial u_k^{(n)}}
D_{jj'}(\mu)
} \\
\displaystyle{
-\sum_{j,k\in I}\sum_{n\in\mb Z_+}
A^*_{k'k}(\lambda+\mu+\partial)
(-\lambda-\mu-\partial)^n
\frac{\partial D_{ji'}(\lambda)}{\partial u_k^{(n)}}
C_{jj'}(\mu)
\,.} 
\end{array}
\end{equation}
\end{lemma}
\begin{proof}
For equation \eqref{20120405:A1}, 
we can use the Leibniz rule and equation \eqref{20111012:eq2c} to get
$$
\begin{array}{l}
\displaystyle{
\{{u_i}_x\{{u_j}_y {u_k}\}_H\}_K
=\sum_{r\in I}\{{u_i}_x A_{kr}(y+\partial)(B^{-1})_{rj}(y)\}_K
} \\
\displaystyle{
=\sum_{r\in I}\sum_{m\in\mb Z_+}\{{u_i}_x a_{kr;m}\}_K
(y+\partial)^m(B^{-1})_{rj}(y)
} \\
\displaystyle{
-\sum_{\substack{r,p,q\in I \\ m\in\mb Z_+}}\!\!\!
A_{kr}(x+y+\partial)
(B^{-1})_{rp}(x+y+\partial)
\{{u_i}_\lambda b_{pq;m}\}_K
(y+\partial)^m(B^{-1})_{qj}(y)\,.
}
\end{array}
$$
We can then use the Master Formula \eqref{20110922:eq1} to get
$$
\begin{array}{l}
\displaystyle{
\{{u_i}_x\{{u_j}_y {u_k}\}_H\}_K
%} \\
%\displaystyle{
=\sum_{r,s\in I}\sum_{n\in\mb Z_+}
\Big(\frac{\partial A_{kr}(y+\partial)}{\partial u_s^{(n)}} (B^{-1})_{rj}(y)\Big)
(x+\partial)^n K_{si}(x)
} \\
\displaystyle{
-\sum_{\substack{r,p,q,s\in I \\ n\in\mb Z_+}}\!\!\!
A_{kr}(x+y+\partial)
(B^{-1})_{rp}(x+y+\partial)
\Big(\frac{\partial B_{pq}(y+\partial)}{\partial u_s^{(n)}} (B^{-1})_{qj}(y)\Big)
} \\
\displaystyle{
\,\,\,\,\,\,\,\,\,\,\,\,\,\,\,\,\,\,\,\,\,\,\,\,\,\,\,\,\,\,\,\,\,\,\,\,\,\,\,\,\,\,\,\,\,\,\,\,\,\,\,\,\,\,\,\,\,\,\,\,\,\,\,\,\,\,\,\,\,\,\,\,\,\,\,\,\,\,\,\,\,
\,\,\,\,\,\,\,\,\,\,\,\,\,\,\,\,\,\,\,\,\,\,\,\,\,\,\,\,\,\,\,\,\,\,\,\,\,
\times
(x+\partial)^n K_{si}(x)
\,.
}
\end{array}
$$
If we now replace $x$ with $\lambda+\partial$ acting on $D_{ii'}(\lambda)$
and $y$ by $\mu+\partial$ acting on $B_{jj'}(\mu)$,
and we apply $B^*_{k'k}(\lambda+\mu+\partial)$, acting from the left,
to both sides of the above equation, 
we get, after using the assuption \eqref{20120405:skew},
that equation \eqref{20120405:A1} holds.
Equation \eqref{20120405:A2} is obtained from \eqref{20120405:A1}
by exchanging the roles of $H$ and $K$.
Equation \eqref{20120405:B1} is obtained from \eqref{20120405:A1}
by exchanging $\lambda$ with $\mu$ and $i$ and $i'$ with $j$ and $j'$ respectively,
and equation \eqref{20120405:B2} is obtained
from \eqref{20120405:B1} by exchaing the roles of $H$ and $K$.
Finally, equations \eqref{20120405:C1} and \eqref{20120405:C2} 
can be derived with a similar computation,
which involves the right Leibniz rule (instead of the left)
and equation \eqref{20111012:eq2d} (instead of \eqref{20111012:eq2c}).
\end{proof}

Let us next describe the relations \eqref{eq:20090322_1}
associated to Dirac structures $\mc L_{A,B}(\mc K)$ and $\mc L_{C,D}(\mc K)$
defined in \eqref{eq:dirac}.
We have
\begin{equation}\label{20120405:eq5}
\begin{array}{l}
\mc N_{\mc L_{A,B}(\mc K),\mc L_{C,D}(\mc K)}
=
\big\{
A(\partial)X\oplus C(\partial)X^\prime
\,\big|\,
X,X^\prime\in\mc K^{\oplus\ell}\,,\,\,B(\partial)X=D(\partial)X^\prime
\big\}
\,,
\\
\mc N_{\mc L_{A,B}(\mc K),\mc L_{C,D}(\mc K)}\wwcheck
=
\big\{
B(\partial)Z\oplus D(\partial)Z^\prime
\,\big|\,
Z,Z^\prime\in\mc K^{\oplus\ell}\,,\,\,A(\partial)Z=C(\partial)Z^\prime
\big\}
\,.
\end{array}
\end{equation}
Hence, by Definition \ref{2006_NRel},
the Dirac structures $\mc L_{A,B}$ and $\mc L_{C,D}$
are compatible if and only if,
for every $X,X^\prime,Y,Y^\prime,Z,Z^\prime,W,W^\prime\in\mc V^{\oplus\ell}$ 
such that
\begin{equation}\label{20120405:eq6}
\begin{array}{c}
\displaystyle{
\vphantom{\Big(}
B(\partial)X=D(\partial)X^\prime
\,\,,\,\,\,\,
B(\partial)Y=D(\partial)Y^\prime
\,\,,\,\,\,\,
B(\partial)W=D(\partial)Z^\prime
\,,}\\
\displaystyle{
\vphantom{\Big(}
A(\partial)Z=C(\partial)Z^\prime
\,\,,\,\,\,\,
A(\partial)W=C(\partial)W^\prime
\,,}
\end{array}
\end{equation}
we have the following identity:
\begin{equation}\label{20120405:eq7}
\begin{array}{l}
\displaystyle{
\vphantom{\Big(}
\big(B(\partial)Z\big|[A(\partial)X,A(\partial)Y]\big)
-\big(D(\partial)Z^\prime\big|[A(\partial)X,C(\partial)Y^\prime]\big)
}\\
\displaystyle{
\vphantom{\Big(}
-\big(B(\partial)W\big|[C(\partial)X^\prime,A(\partial)Y]\big)
+\big(D(\partial)W^\prime\big|[C(\partial)X^\prime,C(\partial)Y^\prime]\big)
=0
\,.}
\end{array}\end{equation}

\begin{lemma}\label{20120405:lem2}
Suppose that $H=AB^{-1}$ and $K=CD^{-1}$ are non-local Poisson structures,
and that conditions \eqref{20120405:eq6} hold. 
Then equation \eqref{20120405:eq7} is equivalent to the following equation:
\begin{equation}\label{20120405:eq16}
\begin{array}{l}
\displaystyle{
\vphantom{\Big(}
-\tint (A(\partial)Y)\cdot D_{B(\partial)X}(\partial)C(\partial)Z^\prime
+\tint (A(\partial)Y)\cdot D^*_{B(\partial)X}(\partial) C(\partial)Z^\prime
} \\
\displaystyle{
\vphantom{\Big(}
+\tint (B(\partial)Y)\cdot D_{C(\partial)Z^\prime}(\partial)A(\partial)X
-\tint (B(\partial)Y)\cdot D_{A(\partial)X}(\partial) C(\partial)Z^\prime
} \\
\displaystyle{
\vphantom{\Big(}
+\tint (C(\partial)Y^\prime)\cdot D_{D(\partial)Z^\prime}(\partial)A(\partial)X
+\tint (C(\partial)Y^\prime)\cdot D^*_{A(\partial)X}(\partial) D(\partial)Z^\prime
} \\
\displaystyle{
\vphantom{\Big(}
+\tint (A(\partial)Y)\cdot D_{B(\partial)W}(\partial)C(\partial)X^\prime
+\tint (A(\partial)Y)\cdot D^*_{C(\partial)X^\prime}(\partial) B(\partial)W
} \\
\displaystyle{
\vphantom{\Big(}
-\tint (C(\partial)Y^\prime)\cdot D_{D(\partial)X^\prime}(\partial)A(\partial)W
+\tint (C(\partial)Y^\prime)\cdot D^*_{D(\partial)X^\prime}(\partial) A(\partial)W
} \\
\displaystyle{
\vphantom{\Big(}
+\tint (D(\partial)Y^\prime)\cdot D_{A(\partial)W}(\partial)C(\partial)X^\prime
-\tint (D(\partial)Y^\prime)\cdot D_{C(\partial)X^\prime}(\partial) A(\partial)W
=0\,.
}\end{array}
\end{equation}
\end{lemma}
\begin{proof}
By \eqref{20120126:eq1} and \eqref{20120405:eq1}, we have 
\begin{equation}\label{20120405:eq8}
\begin{array}{l}
\displaystyle{
\vphantom{\Big(}
\big(B(\partial)Z\big|[A(\partial)X,A(\partial)Y]\big)
=\tint (B(\partial)Z)\cdot D_{A(\partial)Y}(\partial)A(\partial)X
} \\
\displaystyle{
\vphantom{\Big(}
-\tint (B(\partial)Z)\cdot D_{A(\partial)X}(\partial)A(\partial)Y
=\tint (A(\partial)X)\cdot\frac{\delta}{\delta u}(B(\partial)Z|A(\partial)Y)
} \\
\displaystyle{
\vphantom{\Big(}
-\tint (A(\partial)Y)\cdot D_{B(\partial)Z}(\partial)A(\partial)X
-\tint (A(\partial)Y)\cdot D^*_{A(\partial)X}(\partial) B(\partial)Z
\,.
}\end{array}
\end{equation}
Similarly, we have
\begin{equation}\label{20120405:eq9}
\begin{array}{l}
\displaystyle{
\vphantom{\Big(}
\big(D(\partial)Z^\prime\big|[A(\partial)X,C(\partial)Y^\prime]\big)
=\tint (A(\partial)X)\cdot\frac{\delta}{\delta u}(D(\partial)Z^\prime|C(\partial)Y^\prime)
} \\
\displaystyle{
\vphantom{\Big(}
-\tint (C(\partial)Y^\prime)\cdot D_{D(\partial)Z^\prime}(\partial)A(\partial)X
-\tint (C(\partial)Y^\prime)\cdot D^*_{A(\partial)X}(\partial) D(\partial)Z^\prime
\,,
}\end{array}
\end{equation}
\begin{equation}\label{20120405:eq10}
\begin{array}{l}
\displaystyle{
\vphantom{\Big(}
\big(B(\partial)W\big|[C(\partial)X^\prime,A(\partial)Y]\big)
=\tint (C(\partial)X^\prime)\cdot\frac{\delta}{\delta u}(B(\partial)W|A(\partial)Y)
} \\
\displaystyle{
\vphantom{\Big(}
-\tint (A(\partial)Y)\cdot D_{B(\partial)W}(\partial)C(\partial)X^\prime
-\tint (A(\partial)Y)\cdot D^*_{C(\partial)X^\prime}(\partial) B(\partial)W
\,,
}\end{array}
\end{equation}
and
\begin{equation}\label{20120405:eq11}
\begin{array}{l}
\displaystyle{
\vphantom{\Big(}
\big(D(\partial)W^\prime\big|[C(\partial)X^\prime,C(\partial)Y^\prime]\big)
=\tint (C(\partial)X^\prime)\cdot\frac{\delta}{\delta u}(D(\partial)W^\prime|C(\partial)Y^\prime)
} \\
\displaystyle{
\vphantom{\Big(}
-\tint (C(\partial)Y^\prime)\cdot D_{D(\partial)W^\prime}(\partial)C(\partial)X^\prime
-\tint (C(\partial)Y^\prime)\cdot D^*_{C(\partial)X^\prime}(\partial) D(\partial)W^\prime
\,.
}\end{array}
\end{equation}
By the skewadnointness of $H$ and $K$, which translates to \eqref{20120405:skew}, 
and by conditions \eqref{20120405:eq6}, we have
$$
\begin{array}{l}
(B(\partial)Z|A(\partial)Y)
=-(A(\partial)Z|B(\partial)Y) \\
=-(C(\partial)Z^\prime|D(\partial)Y^\prime)
=(D(\partial)Z^\prime|C(\partial)Y^\prime)
\,,
\end{array}
$$
hence the first terms in the RHS of \eqref{20120405:eq8} and \eqref{20120405:eq9}
cancel.
Similarly for the first terms in the RHS of \eqref{20120405:eq10} and \eqref{20120405:eq11}.
Therefore, combining equations \eqref{20120405:eq8}--\eqref{20120405:eq11}, we get that
equation \eqref{20120405:eq7} is equivalent to
\begin{equation}\label{20120405:eq12}
\begin{array}{l}
\displaystyle{
\vphantom{\Big(}
-\tint (A(\partial)Y)\cdot D_{B(\partial)Z}(\partial)A(\partial)X
-\tint (A(\partial)Y)\cdot D^*_{A(\partial)X}(\partial) B(\partial)Z
} \\
\displaystyle{
\vphantom{\Big(}
+\tint (C(\partial)Y^\prime)\cdot D_{D(\partial)Z^\prime}(\partial)A(\partial)X
+\tint (C(\partial)Y^\prime)\cdot D^*_{A(\partial)X}(\partial) D(\partial)Z^\prime
} \\
\displaystyle{
\vphantom{\Big(}
+\tint (A(\partial)Y)\cdot D_{B(\partial)W}(\partial)C(\partial)X^\prime
+\tint (A(\partial)Y)\cdot D^*_{C(\partial)X^\prime}(\partial) B(\partial)W
} \\
\displaystyle{
\vphantom{\Big(}
-\tint (C(\partial)Y^\prime)\cdot D_{D(\partial)W^\prime}(\partial)C(\partial)X^\prime
-\tint (C(\partial)Y^\prime)\cdot D^*_{C(\partial)X^\prime}(\partial) D(\partial)W^\prime
=0\,.
}\end{array}
\end{equation}
Next, since by assumption $H=AB^{-1}$ is a non-local Poisson structure,
it follows by Propositions \ref{20120103:propc} and \ref{20120103:propb}
that equation \eqref{20120103:eq1} holds.
In particular,
\begin{equation}\label{20120405:eq14}
\begin{array}{l}
\displaystyle{
\vphantom{\Big(}
-\tint (A(\partial)Y)\cdot D_{B(\partial)Z}(\partial)A(\partial)X
-\tint (A(\partial)Y)\cdot D^*_{A(\partial)X}(\partial) B(\partial)Z
} \\
\displaystyle{
\vphantom{\Big(}
=-\tint (A(\partial)Y)\cdot D_{B(\partial)X}(\partial)A(\partial)Z
+\tint (A(\partial)Y)\cdot D^*_{B(\partial)X}(\partial) A(\partial)Z
} \\
\displaystyle{
\vphantom{\Big(}
+\tint (B(\partial)Y)\cdot D_{A(\partial)Z}(\partial)A(\partial)X
-\tint (B(\partial)Y)\cdot D_{A(\partial)X}(\partial) A(\partial)Z
\,.
}\end{array}
\end{equation}
Similarly, using the assumption that $K=CD^{-1}$ is a non-local Poisson structure,
we get
\begin{equation}\label{20120405:eq15}
\begin{array}{l}
\displaystyle{
\vphantom{\Big(}
-\tint (C(\partial)Y^\prime)\cdot D_{D(\partial)W^\prime}(\partial)C(\partial)X^\prime
-\tint (C(\partial)Y^\prime)\cdot D^*_{C(\partial)X^\prime}(\partial) D(\partial)W^\prime
} \\
\displaystyle{
\vphantom{\Big(}
=-\tint (C(\partial)Y^\prime)\cdot D_{D(\partial)X^\prime}(\partial)C(\partial)W^\prime
+\tint (C(\partial)Y^\prime)\cdot D^*_{D(\partial)X^\prime}(\partial) C(\partial)W^\prime
} \\
\displaystyle{
\vphantom{\Big(}
+\tint (D(\partial)Y^\prime)\cdot D_{C(\partial)W^\prime}(\partial)C(\partial)X^\prime
-\tint (D(\partial)Y^\prime)\cdot D_{C(\partial)X^\prime}(\partial) C(\partial)W^\prime
\,.
}\end{array}
\end{equation}
Combining equations \eqref{20120405:eq12}, \eqref{20120405:eq14} and \eqref{20120405:eq15},
we get \eqref{20120405:eq16}.
\end{proof}

\begin{proof}[Proof of Theorem \ref{20120126:prop2}]
By Lemma \ref{20120405:lem2},
we only need to prove that, if condition \eqref{20120405:eq3} holds,
then equation \eqref{20120405:eq16} holds
for every $X,X^\prime,Y,Y^\prime,W,Z^\prime$
satisfying the first three identities in \eqref{20120405:eq6}.
It follows by some straightforward computation that
we can rewrite each term in the LHS of \eqref{20120405:eq16} as follows
%X1
\begin{equation}\label{20120405:X1}
\begin{array}{l}
\displaystyle{
\vphantom{\Big(}
-\tint (A(\partial)Y)\cdot D_{B(\partial)X}(\partial)C(\partial)Z^\prime
=-\tint (A(\partial)Y)\cdot B(\partial)D_X(\partial)C(\partial)Z^\prime
} \\
\displaystyle{
\vphantom{\Big(}
-\int \!\!\!\!\sum_{\substack{i',j',k'\in I \\ j,k,\in I,\, n\in\mb Z_+}}\!\!\!\!
Y_{k'}A^*_{k'k}(\lambda\!+\!\mu\!+\!\partial)
\frac{\partial B_{ki'}(\lambda)}{\partial u_j^{(n)}} (\mu+\partial)^n C_{jj'}(\mu)
\big(\big|_{\lambda=\partial}X_{i'}\big)\big(\big|_{\mu=\partial}Z^\prime_{j'}\big)
\,,
}\end{array}
\end{equation}
%X2
\begin{equation}\label{20120405:X2}
\begin{array}{l}
\displaystyle{
\vphantom{\Big(}
\tint (A(\partial)Y)\cdot D^*_{D(\partial)X^\prime}(\partial) C(\partial)Z^\prime
=\tint (A(\partial)Y)\cdot D^*_{X^\prime}(\partial)D^*(\partial)C(\partial)Z^\prime
} \\
\displaystyle{
\vphantom{\Big(}
+\int \!\!\!\!\!\!\sum_{\substack{i',j',k'\in I \\ j,k,\in I, n\in\mb Z_+}}\!\!\!\!\!\!
Y_{k'}A^*_{k'k}\!(\lambda\!+\!\mu\!+\!\partial)
(\!-\!\lambda\!-\!\mu\!-\!\partial)^n
\frac{\partial D_{ji'}(\lambda)}{\partial u_k^{(n)}} C_{jj'}(\mu)
\big(\big|_{\lambda=\partial}\!X^\prime_{i'}\big)\big(\big|_{\mu=\partial}\!Z^\prime_{j'}\big)
\,,
}\end{array}
\end{equation}
%X3
\begin{equation}\label{20120405:X3}
\begin{array}{l}
\displaystyle{
\vphantom{\Big(}
\tint (D(\partial)Y^\prime)\cdot D_{C(\partial)Z^\prime}(\partial)A(\partial)X
=\tint (D(\partial)Y^\prime)\cdot C(\partial)D_{Z^\prime}(\partial)A(\partial)X
} \\
\displaystyle{
\vphantom{\Big(}
+\int \!\!\!\sum_{\substack{i',j',k'\in I \\ i,k,\in I, n\in\mb Z_+}}\!\!\!
Y^\prime_{k'}D^*_{k'k}(\lambda\!+\!\mu\!+\!\partial)
\frac{\partial C_{kj'}(\mu)}{\partial u_i^{(n)}} (\lambda+\partial)^n A_{ii'}(\lambda)
\big(\big|_{\lambda=\partial}X_{i'}\big)\big(\big|_{\mu=\partial}Z^\prime_{j'}\big)
\,,
}\end{array}
\end{equation}
%X4
\begin{equation}\label{20120405:X4}
\begin{array}{l}
\displaystyle{
\vphantom{\Big(}
-\tint (B(\partial)Y)\cdot D_{A(\partial)X}(\partial) C(\partial)Z^\prime
=-\tint (B(\partial)Y)\cdot A(\partial)D_X(\partial)C(\partial)Z^\prime
} \\
\displaystyle{
\vphantom{\Big(}
-\int \!\!\!\sum_{\substack{i',j',k'\in I \\ j,k,\in I, n\in\mb Z_+}}\!\!\!
Y_{k'}B^*_{k'k}(\lambda\!+\!\mu\!+\!\partial)
\frac{\partial A_{ki'}(\lambda)}{\partial u_j^{(n)}} (\mu+\partial)^n C_{jj'}(\mu)
\big(\big|_{\lambda=\partial}X_{i'}\big)\big(\big|_{\mu=\partial}Z^\prime_{j'}\big)
\,,
}\end{array}
\end{equation}
%X5
\begin{equation}\label{20120405:X5}
\begin{array}{l}
\displaystyle{
\vphantom{\Big(}
\tint (C(\partial)Y^\prime)\cdot D_{D(\partial)Z^\prime}(\partial)A(\partial)X
=\tint (C(\partial)Y^\prime)\cdot D(\partial)D_{Z^\prime}(\partial)A(\partial)X
} \\
\displaystyle{
\vphantom{\Big(}
+\int \!\!\!\sum_{\substack{i',j',k'\in I \\ i,k,\in I, n\in\mb Z_+}}\!\!\!
Y^\prime_{k'}C^*_{k'k}(\lambda\!+\!\mu\!+\!\partial)
\frac{\partial D_{kj'}(\mu)}{\partial u_i^{(n)}} (\lambda+\partial)^n A_{ii'}(\lambda)
\big(\big|_{\lambda=\partial}X_{i'}\big)\big(\big|_{\mu=\partial}Z^\prime_{j'}\big)
\,,
}\end{array}
\end{equation}
%X6
\begin{equation}\label{20120405:X6}
\begin{array}{l}
\displaystyle{
\vphantom{\Big(}
\tint (C(\partial)Y^\prime)\cdot D^*_{A(\partial)X}(\partial) B(\partial)W
=\tint (C(\partial)Y^\prime)\cdot D^*_{X}(\partial)A^*(\partial)B(\partial)W
} \\
\displaystyle{
\vphantom{\Big(}
+\int \!\!\!\!\!\!\sum_{\substack{i',j',k'\in I \\ j,k,\in I, n\in\mb Z_+}}\!\!\!\!\!\!\!
Y^\prime_{k'}C^*_{k'k}\!(\lambda\!+\!\mu\!+\!\partial)
(\!-\!\lambda\!-\!\mu\!-\!\partial)^n
\frac{\partial A_{ji'}(\lambda)}{\partial u_k^{(n)}} B_{jj'}(\mu)
\big(\big|_{\lambda=\partial}\!X_{i'}\big)\big(\big|_{\mu=\partial}\!W_{j'}\big)
,
}\end{array}
\end{equation}
%X7
\begin{equation}\label{20120405:X7}
\begin{array}{l}
\displaystyle{
\vphantom{\Big(}
\tint (A(\partial)Y)\cdot D_{B(\partial)W}(\partial)C(\partial)X^\prime
=\tint (A(\partial)Y)\cdot B(\partial)D_W(\partial)C(\partial)X^\prime
} \\
\displaystyle{
\vphantom{\Big(}
+\int \!\!\!\sum_{\substack{i',j',k'\in I \\ i,k,\in I, n\in\mb Z_+}}\!\!\!
Y_{k'}A^*_{k'k}(\lambda\!+\!\mu\!+\!\partial)
\frac{\partial B_{kj'}(\mu)}{\partial u_i^{(n)}} (\lambda+\partial)^n C_{ii'}(\lambda)
\big(\big|_{\lambda=\partial}X^\prime_{i'}\big)\big(\big|_{\mu=\partial}W_{j'}\big)
\,,
}\end{array}
\end{equation}
%X8
\begin{equation}\label{20120405:X8}
\begin{array}{l}
\displaystyle{
\vphantom{\Big(}
\tint (A(\partial)Y)\cdot D^*_{C(\partial)X^\prime}(\partial) D(\partial)Z^\prime
=\tint (A(\partial)Y)\cdot D^*_{X^\prime}(\partial)C^*(\partial)D(\partial)Z^\prime
} \\
\displaystyle{
\vphantom{\Big(}
+\int \!\!\!\!\!\!\sum_{\substack{i',j',k'\in I \\ j,k,\in I, n\in\mb Z_+}}\!\!\!\!\!\!
Y_{k'}A^*_{k'k}\!(\lambda\!+\!\mu\!+\!\partial)
(\!-\!\lambda\!-\!\mu\!-\!\partial)^n
\frac{\partial C_{ji'}(\lambda)}{\partial u_k^{(n)}} D_{jj'}(\mu)
\big(\big|_{\lambda=\partial}\!X^\prime_{i'}\big)\big(\big|_{\mu=\partial}\!Z^\prime_{j'}\big)
\,,
}\end{array}
\end{equation}
%X9
\begin{equation}\label{20120405:X9}
\begin{array}{l}
\displaystyle{
\vphantom{\Big(}
-\tint (C(\partial)Y^\prime)\cdot D_{D(\partial)X^\prime}(\partial)A(\partial)W
=-\tint (C(\partial)Y^\prime)\cdot D(\partial)D_{X^\prime}(\partial)A(\partial)W
} \\
\displaystyle{
\vphantom{\Big(}
-\int \!\!\!\!\sum_{\substack{i',j',k'\in I \\ j,k,\in I, n\in\mb Z_+}}\!\!\!\!
Y^\prime_{k'}C^*_{k'k}(\lambda\!+\!\mu\!+\!\partial)
\frac{\partial D_{ki'}(\lambda)}{\partial u_j^{(n)}} (\mu+\partial)^n A_{jj'}(\mu)
\big(\big|_{\lambda=\partial}X^\prime_{i'}\big)\big(\big|_{\mu=\partial}W_{j'}\big)
\,,
}\end{array}
\end{equation}
%X10
\begin{equation}\label{20120405:X10}
\begin{array}{l}
\displaystyle{
\vphantom{\Big(}
\tint (C(\partial)Y^\prime)\cdot D^*_{B(\partial)X}(\partial) A(\partial)W
=\tint (C(\partial)Y^\prime)\cdot D^*_{X}(\partial)B^*(\partial)A(\partial)W
} \\
\displaystyle{
\vphantom{\Big(}
+\int \!\!\!\!\!\!\sum_{\substack{i',j',k'\in I \\ j,k,\in I, n\in\mb Z_+}}\!\!\!\!\!\!
Y^\prime_{k'}C^*_{k'k}\!(\lambda\!+\!\mu\!+\!\partial)
(\!-\!\lambda\!-\!\mu\!-\!\partial)^n
\frac{\partial B_{ji'}(\lambda)}{\partial u_k^{(n)}} A_{jj'}(\mu)
\big(\big|_{\lambda=\partial}\!X_{i'}\big)\!\big(\big|_{\mu=\partial}\!W_{j'}\big)
,
}\end{array}
\end{equation}
%X11
\begin{equation}\label{20120405:X11}
\begin{array}{l}
\displaystyle{
\vphantom{\Big(}
\tint (B(\partial)Y)\cdot D_{A(\partial)W}(\partial)C(\partial)X^\prime
=\tint (B(\partial)Y)\cdot A(\partial)D_W(\partial)C(\partial)X^\prime
} \\
\displaystyle{
\vphantom{\Big(}
+\int \!\!\!\!\sum_{\substack{i',j',k'\in I \\ i,k,\in I, n\in\mb Z_+}}\!\!\!\!
Y_{k'}B^*_{k'k}(\lambda\!+\!\mu\!+\!\partial)
\frac{\partial A_{kj'}(\mu)}{\partial u_i^{(n)}} (\lambda+\partial)^n C_{ii'}(\lambda)
\big(\big|_{\lambda=\partial}X^\prime_{i'}\big)\big(\big|_{\mu=\partial}W_{j'}\big)
\,,
}\end{array}
\end{equation}
%X12
\begin{equation}\label{20120405:X12}
\begin{array}{l}
\displaystyle{
\vphantom{\Big(}
-\tint (D(\partial)Y^\prime)\cdot D_{C(\partial)X^\prime}(\partial) A(\partial)W
=-\tint (D(\partial)Y^\prime)\cdot C(\partial)D_{X^\prime}(\partial)A(\partial)W
} \\
\displaystyle{
\vphantom{\Big(}
-\int \!\!\!\!\sum_{\substack{i',j',k'\in I \\ j,k,\in I, n\in\mb Z_+}}\!\!\!\!
Y^\prime_{k'}D^*_{k'k}(\lambda\!+\!\mu\!+\!\partial)
\frac{\partial C_{ki'}(\lambda)}{\partial u_j^{(n)}} (\mu+\partial)^n A_{jj'}(\mu)
\big(\big|_{\lambda=\partial}X^\prime_{i'}\big)\big(\big|_{\mu=\partial}W_{j'}\big)
\,.
}\end{array}
\end{equation}
It follows from the skewadjointness conditions \eqref{20120405:skew}
that the first term in the RHS of \eqref{20120405:X1} 
cancels with the first term in the RHS of \eqref{20120405:X4},
the first term in the RHS of \eqref{20120405:X2} 
cancels with the first term in the RHS of \eqref{20120405:X8},
the first term in the RHS of \eqref{20120405:X3} 
cancels with the first term in the RHS of \eqref{20120405:X5},
the first term in the RHS of \eqref{20120405:X6} 
cancels with the first term in the RHS of \eqref{20120405:X10},
the first term in the RHS of \eqref{20120405:X7} 
cancels with the first term in the RHS of \eqref{20120405:X11},
and the first term in the RHS of \eqref{20120405:X9} 
cancels with the first term in the RHS of \eqref{20120405:X12}.
Furthermore, 
combining the second terms of the RHS's of \eqref{20120405:X7} and \eqref{20120405:X11},
we get, thanks to \eqref{20120405:A1},
$$
\begin{array}{l}
\displaystyle{
\vphantom{\Big(}
\int 
\sum_{\substack{i',j',k'\in I \\ i,j,k\in I}}
\big(B_{kk'}(\partial)Y_{k'}\big)
\{{u_i}_\lambda\{{u_j}_\mu {u_k}\}_H\}_K
\big(\big|_{\lambda=\partial}D_{ii'}(\partial)X^\prime_{i'}\big)
\big(\big|_{\mu=\partial}B_{jj'}(\partial)W_{j'}\big)
\,.
}\end{array}
$$
Combining the second terms of the RHS's of \eqref{20120405:X3} and \eqref{20120405:X5},
we get, thanks to \eqref{20120405:A2},
$$
\begin{array}{l}
\displaystyle{
\vphantom{\Big(}
\int \sum_{\substack{i',j',k'\in I \\ i,j,k,\in I}}
\big(D_{kk'}(\partial)Y^\prime_{k'}\big)
\{{u_i}_\lambda\{{u_j}_\mu {u_k}\}_K\}_H
\big(\big|_{\lambda=\partial}B_{ii'}(\partial)X_{i'}\big)
\big(\big|_{\mu=\partial}D_{jj'}(\partial)Z^\prime_{j'}\big)
\,.
}\end{array}
$$
Combining the second terms of the RHS's of \eqref{20120405:X1} and \eqref{20120405:X4},
we get, thanks to \eqref{20120405:B1},
$$
\begin{array}{l}
\displaystyle{
\vphantom{\Big(}
-\int \!\!\sum_{\substack{i',j',k'\in I \\ i,j,k,\in I}}\!\!
\big(B_{kk'}(\partial)Y_{k'}\big)
\{{u_j}_\mu\{{u_i}_\lambda {u_k}\}_H\}_K
\big(\big|_{\lambda=\partial}B_{ii'}(\partial)X_{i'}\big)
\big(\big|_{\mu=\partial}D_{jj'}(\partial)Z^\prime_{j'}\big)
\,.
}\end{array}
$$
Combining the second terms of the RHS's of \eqref{20120405:X9} and \eqref{20120405:X12},
we get, thanks to \eqref{20120405:B2},
$$
\begin{array}{l}
\displaystyle{
\vphantom{\Big(}
-\int \!\!\sum_{\substack{i',j',k'\in I \\ i,j,k,\in I}}\!\!
\big(D_{kk'}(\partial)Y^\prime_{k'}\big)
\{{u_j}_\mu\{{u_i}_\lambda {u_k}\}_K\}_H
\big(\big|_{\lambda=\partial}D_{ii'}(\partial)X^\prime_{i'}\big)
\big(\big|_{\mu=\partial}B_{jj'}(\partial)W_{j'}\big)
\,.
}\end{array}
$$
Combining the second terms of the RHS's of \eqref{20120405:X6} and \eqref{20120405:X10},
we get, thanks to \eqref{20120405:C1},
$$
\begin{array}{l}
\displaystyle{
\vphantom{\Big(}
-\int \!\!\!\sum_{\substack{i',j',k'\in I \\ i,j,k,\in I}}\!\!\!
\big(D_{kk'}(\partial)Y^\prime_{k'}\big)
\{{\{{u_i}_\lambda{u_j}\}_H}\!_{\lambda+\mu}\!{u_k}\}_K
\big(\big|_{\lambda=\partial}B_{ii'}(\partial)X_{i'}\big)
\!\big(\big|_{\mu=\partial}B_{jj'}(\partial)W_{j'}\big)
.
}\end{array}
$$
Finally, 
combining the second terms of the RHS's of \eqref{20120405:X2} and \eqref{20120405:X8},
we get, thanks to \eqref{20120405:C2},
$$
\begin{array}{l}
\displaystyle{
\vphantom{\Big(}
-\int \!\!\!\sum_{\substack{i',j',k'\in I \\ i,j,k,\in I}}\!\!\!
\big(B_{kk'}(\partial)Y_{k'}\big)
\{{\{{u_i}_\lambda{u_j}\}_K}_{\lambda+\mu}\!{u_k}\}_H
\big(\big|_{\lambda=\partial}D_{ii'}(\partial)X^\prime_{i'}\big)
\!\big(\big|_{\mu=\partial}D_{jj'}(\partial)Z^\prime_{j'}\big)
.
}\end{array}
$$
Putting together all the above results, we conclude that the LHS of \eqref{20120405:eq16}
is equal to
$$
\begin{array}{l}
\displaystyle{
\vphantom{\Big(}
\int  \sum_{i,j,k\in I}
\big(B(\partial)Y\big)_k
\Big(
\{{u_i}_\lambda\{{u_j}_\mu {u_k}\}_H\}_K
+\{{u_i}_\lambda\{{u_j}_\mu {u_k}\}_K\}_H
} \\
\displaystyle{
\vphantom{\Big(}
-\{{u_j}_\mu\{{u_i}_\lambda {u_k}\}_H\}_K
-\{{u_j}_\mu\{{u_i}_\lambda {u_k}\}_K\}_H
-\{{\{{u_i}_\lambda{u_j}\}_H}_{\lambda+\mu}{u_k}\}_K
} \\
\displaystyle{
\vphantom{\Big(}
-\{{\{{u_i}_\lambda{u_j}\}_K}_{\lambda+\mu}\!{u_k}\}_H
\Big)
\big(\big|_{\lambda=\partial}B(\partial)X\big)_i
\big(\big|_{\mu=\partial}B(\partial)W\big)_j
\,,
}\end{array}
$$
which is zero by \eqref{20120405:eq3}.
\end{proof}

In view of Theorem \ref{20120126:prop2},
we can translate Theorem \ref{mtst} in terms of compatible non-local Poisson structures.
\begin{theorem}\label{mtst-nonloc}
Let $\mc V$ be an algebra of differential functions in $u_1,\dots,u_\ell$, which is a domain.
Let $H,K\in\Mat_{\ell\times\ell}\mc V(\partial)$
be compatible non-local Poisson structures over $\mc V$.
Let $H=AB^{-1},$ and $K=CD^{-1}$ be their minimal fractional decompositions
(cf. Definition \ref{def:minimal-fraction}),
with $A,B,C,D\in\Mat_{\ell\times\ell}\mc V[\partial]$, $B$, $D$ non-degenerate
(cf. Definition \ref{def:non-deg} and Remark \ref{rem:minimal-fraction2}).
Let $F_0=B(\partial)Z,\,F_1=D(\partial)Z^\prime=B(\partial)W,\,F_2=D(\partial)W^\prime$,
with $Z,Z^\prime,W,W^\prime\in\mc V^{\oplus\ell}$, be such that
\begin{equation}\label{20130123:eq3}
D(\partial)Z^\prime=B(\partial)W
\,\,,\,\,\,\,
A(\partial)Z=C(\partial)Z^\prime
\,\,,\,\,\,\,
A(\partial)W=C(\partial)W^\prime
\,,
\end{equation}
and $F_0$ and $F_1$ are closed, i.e.
$$
D_{F_0}^*(\partial)=D_{F_0}(\partial)
\,\,,\,\,\,\,
D_{F_1}^*(\partial)=D_{F_1}(\partial)
\,.
$$
Then:
\begin{enumerate}[(a)]
\item
For all $X,\,Y\in\mc V^{\oplus\ell}$, such that
$D(\partial)X,D(\partial)Y\in\im(B)$, we have
\begin{equation}\label{20120405:eq2b}
\tint Y\cdot C^*(\partial)\big(D_{F_2}(\partial)-D_{F_2}^*(\partial)\big)C(\partial)X
=0\,.
\end{equation}
\item
If we also assume that $K$ is non-degenerate, then 
$F_2$ is closed.
\item
Moreover $F_2$ is exact in any normal differential algebra extension $\tilde{\mc V}$ of $\mc V$:
$F_2=\frac{\delta f_2}{\delta u}$, where $\tint f_2\in\tilde{\mc V}/\partial\tilde{\mc V}$.
\end{enumerate}
\end{theorem}
\begin{proof}
By Theorem \ref{20120126:prop2}, $\mc L_{A,B}(\mc K)$ and $\mc L_{C,D}(\mc K)$
are compatible Dirac structures over $\mc K$, the field of fractions of $\mc V$.
Recalling the expressions \eqref{20120405:eq5} of $\mc N_{\mc L,\mc L^\prime}\wcheck$ 
and $\mc N_{\mc L,\mc L^\prime}$ for these Dirac structures,
we get by Theorem \ref{mtst} that equation \eqref{20120405:eq2b} holds over $\mc K$,
hence over $\mc V$, proving (a).
Let us prove part (b).
It is proved in \cite{CDSK12b} that any two 
non-degenerate matrix differential operators $B(\partial)$ and $D(\partial)$
have a right common multiple 
$B(\partial)D_1(\partial)=D(\partial)B_1(\partial)$,
where $B_1(\partial),D_1(\partial)\in\Mat_{\ell\times\ell}\mc K[\partial]$
are non-degenerate.
By clearing denominators, we can assume that $B_1(\partial)$ and $D_1(\partial)$
have coefficients in $\mc V$.
Hence, if $X,Y\in\im(B_1)$, we have $D(\partial)X,D(\partial)Y\in\im(B)$.
Therefore, by part (a) we have
$$
\int G\cdot B_1^*(\partial)C^*(\partial)\big(D_{F_2}(\partial)-D_{F_2}^*(\partial)\big)C(\partial)B_1(\partial)F
=0\,,
$$
for all $F,G\in\mc V^{\oplus\ell}$.
Since, by assumption, $C(\partial)$ and $B_1(\partial)$ are non-degenerate,
it follows that $D_{F_2}^*(\partial)=D_{F_2}(\partial)$, as we wanted.
Finally, part (c) follows by the fact that, under the assumption that $\tilde{\mc V}$ is normal,
the variational complex is exact (see \cite[Thm.3.2]{BDSK09}).
\end{proof}

%%%%%%%%%%%%%%%%%%%%%%%%%%%%%%%%%%%%%%%%%%%%%%%%%%%%%%%%%%%%%%%%%%%%%%%%%%%%%%%%%%%%%%%%%%%%%%%%%%%%%%%%%%%%%%
%%%%%%%%%%%%%%% Sect 7 %%%%%%%%%%%%%%%%%%%%%%%%%%%%%%%%%%%%%%%%%%%%%%%%%%%%%%%%%%%%%%%%%%%%%%%%%%%%%%%%%%%%%%%
%%%%%%%%%%%%%%%%%%%%%%%%%%%%%%%%%%%%%%%%%%%%%%%%%%%%%%%%%%%%%%%%%%%%%%%%%%%%%%%%%%%%%%%%%%%%%%%%%%%%%%%%%%%%%%

\section{Hamiltonian equations associated to a non-local Poisson structure}
\label{sec:7}

\subsection{A simple linear algebra lemma}
\label{sec:7.0}

Let $U,V,W$ be vector spaces over $\mb F$.
\begin{definition}\label{20130104:def}
Given linear maps $A:\,U\to W$ and $B:\,U\to V$,
we say that $v\in V$ and $w\in W$ are $(A,B)$-\emph{associated} (over $U$),
and we denote this by $v\ass{(A,B)}w$ or $w\ass{(A,B)}v$,
if there exists $u\in U$ such that $v=Bu$ and $w=Au$.
\end{definition}

Let $(\cdot\,|\,\cdot):\,U\times U\to G$ be a symmetric bi-additive form 
with values in an abelian group $G$,
and, by abuse of notation, let also $(\cdot\,|\,\cdot):\,V\times W\to G$ 
be a bi-additive form with values in $G$.
Given a subspace $V_1\subset V$, we define its \emph{orthogonal complement} $V_1^\perp\subset W$ as
$$
V_1^\perp=\big\{w\in W\,\big|\,(v|w)=0\,\text{ for all } v\in V_1\big\}\,
$$
and similarly,
given a subspace $W_1\subset W$, we define its \emph{orthogonal complement} $W_1^\perp\subset V$ as
$$
W_1^\perp=\big\{v\in V\,\big|\,(v|w)=0\,\text{ for all } w\in W_1\big\}\,.
$$
We say that linear maps $A:\,U\to W$ and $A^*:\,V\to U$ are \emph{adjoint}
if we have
$$
(v|Au)=(A^*v|u)\,,
$$
for every $u\in U$ and $v\in V$,
and similarly we say that linear maps $B:\,U\to V$ and $B^*:\,W\to U$ are \emph{adjoint}
if we have
$$
(Bu|w)=(u|B^*u)\,,
$$
for every $u\in U$ and $w\in W$.
(The adjoints $A^*$ and $B^*$ are unique
if the inner product $(\cdot\,|\,\cdot):\,U\times U\to G$ is non-degenerate.)

\begin{lemma}\label{20130104:lem}
Let $A:\,U\to W$, $B:\,U\to V$, $C:\,U\to W$, $D:\,U\to V$, be linear maps,
and let 
$A^*:\,V\to U$, $B^*:\,W\to U$, $C^*:\,V\to U$, $D^*:\,W\to U$ be their adjoint.
Assume that
\begin{equation}\label{20130104:eq3}
A^*B+B^*A=0
\,\,,\,\,\,\,
C^*D+D^*C=0
\,.
\end{equation}
Let $\{v_n\}_{n=-1}^N\subset V$ and $\{w_n\}_{n=0}^N\subset W$
be finite sequences such that
\begin{equation}\label{20130104:eq1}
v_{n-1}\ass{(A,B)}w_n\ass{(C,D)}v_n
\,,
\end{equation}
holds for every $n=0,\dots,N$.
\begin{enumerate}[(a)]
\item
Then we have
$(v_m|w_n)=0$ for every $m=-1,\dots,N$, $n=0,\dots,N$.
\item
Suppose, moreover, that the the following orthogonality conditions hold:
\begin{equation}\label{20130104:eq2b}
\Big(\Span{}_{\mb F}\{v_n\}_{n=-1}^N\Big)^\perp\subset\im C
\,\,,\,\,\,\,
\Big(\Span{}_{\mb F}\{w_n\}_{n=0}^N\Big)^\perp\subset\im B\,.
\end{equation}
Then, we can extend the given finite sequences to infinite sequences $\{v_n\}_{n=-1}^\infty\subset V$,
$\{w_n\}_{n=0}^\infty\subset W$ such that the association relations \eqref{20130104:eq1} 
hold for every $n\in\mb Z_+$.
\end{enumerate}
\end{lemma}
\begin{proof}
By assumption, for every $n=0,\dots,N$ there exist $u_n,u_n^\prime\in U$ such that
$$
v_{n-1}=Bu_n
\,\,,\,\,\,\,
w_n=Au_n
\,\,,\,\,\,\,
v_n=Du_n^\prime
\,\,,\,\,\,\,
w_n=Cu_n^\prime\,.
$$
Hence, by definition of adjoint operators and assumption \eqref{20130104:eq3}, we have,
for every $m,n$,
$$
\begin{array}{l}
(v_m|w_n)=(Du_m^\prime|Cu_n^\prime)
=(u_m^\prime|D^*Cu_n^\prime)=-(u_m^\prime|C^*Du_n^\prime) \\
=-(C^*Du_n^\prime|u_m^\prime)=-(Du_n^\prime|Cu_m^\prime)=-(v_n|w_m)\,,
\end{array}
$$
and similarly
$$
\begin{array}{l}
(v_m|w_n)=(Bu_{m+1}|Au_n)
=(u_{m+1}|B^*Au_n)=-(u_{m+1}|A^*Bu_n) \\
=-(A^*Bu_n|u_{m+1})=-(Bu_n|Au_{m+1})=-(v_{n-1}|w_{m+1})
\,.
\end{array}
$$
Hence,
\begin{equation}
\label{20130104:eq4}
(v_m|w_n)=-(v_n|w_m)=-(v_{n-1}|w_{m+1})\,.
\end{equation}
Letting $m=n$ in equation \eqref{20130104:eq4} we get $(v_n|w_n)=0$ for every $n=0,\dots,N$,
while letting $m=n-1$ in \eqref{20130104:eq4} we get $(v_{n-1}|w_n)=0$ for every $n=0,\dots,N$.
Equations \eqref{20130104:eq4} imply $(v_m|w_n)=(v_{m+1}|w_{n-1})$,
and therefore, by induction on $n-m$, we get that $(v_m|w_n)=0$ for $n\geq m$.
On the other hand, by the first identity in equation \eqref{20130104:eq4} it follows that $(v_m|w_n)=0$ also for $n<m$.
This proves part (a).

By part (a), we have that $(v_N|w_n)=0$ for every $n=0,\dots,N$,
and therefore by the second orthogonality condition \eqref{20130104:eq2b}
we get that $v_N=Bu_{N+1}$ for some $u_{N+1}\in U$.
We then let $w_{N+1}=Au_{N+1}$ and we get, by construction, that $v_N\ass{(A,B)}w_{N+1}$.
By the same argument as in the proof of part (a), we have that $(v_n|w_{N+1})=0$ 
for every $n=-1,\dots,N$,
and therefore by the first orthogonality condition \eqref{20130104:eq2b}
we get that $w_{N+1}=Cu_{N+1}^\prime$ for some $u_{N+1}^\prime\in U$.
We then let $v_{N+1}=Du_{N+1}^\prime$ and we get, by construction, that $v_{N+1}\ass{(C,D)}w_{N+1}$.
Hence, we prolonged the original finite sequences $\{v_n\}_{n=-1}^N$ and $\{w_n\}_{n=0}^N$
by one step.
The claim follows by induction.
\end{proof}

\subsection{Hamiltonian functionals and vector fields, and Poisson bracket}
\label{sec:7.1}

This section serves as a motivation to the Lenard-Magri scheme of integrability,
discussed in the following sections.
%
%According to Theorem \ref{20111020:thm}, if $H$ is a non-local Poisson structure
%on the algebra of differential functions $\mc V$,
%and if $H=AB^{-1}$ is its minimal fractional decomposition, 
%with $A,B\in\Mat_{\ell\times\ell}\mc V[\partial]$ 
%and $B$ non-degenerate,
%then if we go to the field of fractions $\mc K$
%we get the following  Dirac structure:
%$\mc L_{A,B}(\mc K)\subset\mc K^{\oplus\ell}\oplus\mc K^\ell$.
%
%Therefore, having in mind applications to the theory of integrable systems,
%and in order to use statements from \cite{BDSK09},
We assume that the algebra of differential functions $\mc V$ is a domain,
and we denote by $\mc K$ its field of fractions.

Let $H\in\Mat_{\ell\times\ell}\mc V(\partial)$ be a non-local Poisson structure over $\mc V$.
Recall that, since $H$ has rational entries, it admits fractional decomposition
$H=AB^{-1}$, where 
$A,B\in\Mat_{\ell\times\ell}\mc V[\partial]$ 
and $B$ is non-degenerate
(cf. Definitions \ref{def:non-deg} and \ref{def:minimal-fraction}).

\begin{definition}\label{20120124:def}
Elements $\tint h\in\mc K/\partial\mc K$ and $P\in\mc K^\ell$ are $H$-\emph{associated},
and we denote this by
$\tint h\ass{H} P$, if
\begin{equation}\label{20121222:eq1}
\frac{\delta h}{\delta u}=B(\partial)F
\,\,,\,\,\,\,
P=A(\partial)F\,,
\end{equation}
for some fractional decomposition $H=AB^{-1}$,
with $A,B\in\Mat_{\ell\times\ell}\mc V[\partial]$ 
and $B$ non-degenerate,
and some element $F\in\mc K^{\oplus\ell}$.
In this case, we say that $\tint f$ is \emph{Hamiltonian functional} for $H$,
and $P$ is a \emph{Hamiltonian vector field} for $H$.
We denote by $\mc F(H)\subset\mc K/\partial\mc K$ the subset of all Hamiltonian functionals for $H$,
and by $\mc H(H)\subset\mc K^\ell$ the subset of all Hamiltonian vector fields.
\end{definition}
\begin{remark}\label{20120201:rem3}
Note that in the definition \eqref{20120124:def} we can fix a minimal fractional decomposition
$H=A_1B_1^{-1}$ for $H$,
with $A_1,B_1\in\Mat_{\ell\times\ell}\mc V[\partial]$ 
and $B_1$ non-degenerate of minimal possible order.
Indeed, 
by Proposition \ref{prop:minimal-fraction}(b)
if $H=AB^{-1}$ is any other fractional decomposition,
then there exists $D\in\Mat_{\ell\times\ell}\mc K[\partial]$ such that $A=A_1D,\,B=B_1D$.
Therefore, if equations \eqref{20121222:eq1} hold for some $F\in\mc K^{\oplus\ell}$,
then we have $\frac{\delta h}{\delta u}=B_1(\partial)F_1,\,P=A_1(\partial)F_1$,
where $F_1=D(\partial)F\in\mc K^{\oplus\ell}$.
It follows that $\mc F(H)$ and $\mc H(H)$ are vector spaces over $\mb F$.
In fact,
they are given by the following formulas
$$
\mc F(H)=\Big(\frac{\delta}{\delta u}\Big)^{-1}\Big(B_1\big(\mc K^\ell\big)\Big)
\subset\mc K/\partial\mc K
\,,\,\,\,
\mc H(H)=A_1\Big(B_1^{-1}\Big(\frac{\delta}{\delta u}\mc K/\partial\mc K\Big)\Big)\subset\mc K^\ell\,.
$$
\end{remark}

\begin{remark}\label{20130105:rem}
Consider Definition \ref{20130104:def} with $U=V=\mc K^{\oplus\ell}$ and $W=\mc K^\ell$.
Comparing this with Definition \ref{20120124:def},
we have that $\tint h\ass{H}P$ if and only if
$\frac{\delta h}{\delta u}\ass{(A,B)}P$
for some fractional decomposition $H=AB^{-1}$.
\end{remark}

In the local case, when $H\in\Mat_{\ell\times\ell}\mc V[\partial]$ 
is a (local) Poisson structure over $\mc V$,
then
$\tint f\in\mc F(H)=\mc K/\partial\mc K$ 
and $P\in\mc H(H)=H(\partial)\Big(\im\frac{\delta}{\delta u}\Big)\,\subset\mc K^\ell$ 
are associated if and only if
$P=H(\partial)\frac{\delta\tint f}{\delta u}$.

\begin{lemma}\label{20120907:lem1}
\begin{enumerate}[(a)]
\item
If $\tint f\ass{H}P$ and $\tint g\ass{H}Q$, then $\tint(a f+b g)\ass{H}(aP+bQ)$ for every 
$a,b\in\mc C$ (the subfield of constants in $\mc K$).
In particular, $\mc F(H)$ and $\mc H(H)$ are vector spaces over $\mc C$.
\item
If $\tint f\ass{H}P$, then
$\Big\{\tint g\in\mc F(H)\,\Big|\,\tint g\ass{H}P\Big\}=\tint f+\mc F_0(H)$,
where 
\begin{equation}\label{20120908:eq1}
\mc F_0(H)=\Big\{\tint g\in\mc F(H)\,\Big|\,\tint g\ass{H}0\Big\}\,.
\end{equation}
\item
If $\tint f\ass{H}P$, then
$\Big\{Q\in\mc H(H)\,\Big|\,\tint f\ass{H}Q\Big\}=P+\mc H_0(H)$,
where 
\begin{equation}\label{20120908:eq2}
\mc H_0(H)=\Big\{Q\in\mc H(H)\,\Big|\,0\ass{H}Q\Big\}\,.
\end{equation}
\end{enumerate}
\end{lemma}
\begin{proof}
Obvious, using Remark \ref{20120201:rem3}.
\end{proof}

\begin{lemma}\label{20120124:lem}
\begin{enumerate}[(a)]
\item
The space $\mc K^\ell$ is a Lie algebra with bracket 
\eqref{20120126:eq1},
%$$
%[P,Q]
%=D_Q(\partial)P-D_P(\partial)Q\,,
%$$
%where $D_P(\partial)$ denotes the Frechet derivative of $P$, 
%defined in \eqref{20111020:eq1},
and $\mc H(H)\subset\mc K^\ell$ is its subalgebra.
\item
We have a representation $\phi$ of the Lie algebra $\mc K^\ell$ 
on the space $\mc K/\partial\mc K$ given by
$$
\phi(P)\big(\tint h\big)=\int P\cdot\frac{\delta h}{\delta u}\,,
$$
and the subspace $\mc F(H)\subset\mc K/\partial\mc K$
is preserved by the action of the Lie subalgebra $\mc H(H)\subset\mc K^\ell$.
\item
If $\tint h\ass{H}P$ and $\tint h\ass{H}Q$ for some $\tint h\in\mc F(H)$,
then the action of $P,Q\in\mc H(H)$on $\mc F(H)$ is the same:
$$
\int P\cdot\frac{\delta g}{\delta u}
=\int Q\cdot\frac{\delta g}{\delta u}
\,\,\text{ for all } \tint g\in\mc F(H)\,.
$$
\end{enumerate}
\end{lemma}
\begin{proof}
It follows immediately from \cite[Lem.4.7-8]{BDSK09},
using the fact that $\mc L_{A,B}(\mc K)$ is a Dirac structure 
if $H=AB^{-1}$ is a minimal fractional decomposition for $H$.
\end{proof}

Thanks to Lemma \ref{20120124:lem},
we have a well-defined map 
$\{\cdot\,,\,\cdot\}_H:\,\mc F(H)\times\mc F(H)\to\mc F(H)$
given by
\begin{equation}\label{20120124:eq4}
\{\tint f,\tint g\}_H
=
\int P\cdot\frac{\delta g}{\delta u}
\quad
\Big(
=\int \frac{\delta g}{\delta u}\cdot A(\partial) B^{-1}(\partial) \frac{\delta f}{\delta u}
\,\,\Big)\,,
\end{equation}
where $P\in\mc H(H)$ is such that $\tint f\ass{H}P$.
\begin{proposition}\label{20120124:prop}
\begin{enumerate}[(a)]
\item
The bracket \eqref{20120124:eq4} is a Lie algebra bracket on the space 
of Hamiltonian functionals $\mc F(H)$.
\item
The Lie algebra action of $\mc H(H)$ on $\mc F(H)$ is by derivations
of the Lie bracket \eqref{20120124:eq4}.
\item
The subspace 
$$
\mc A(H)=\Big\{(\tint f,P)\in\mc F(H)\times\mc H(H)\,\Big|\,\tint f\ass{H}P\Big\}\,
$$
is a subalgebra of the direct product of Lie algebras $\mc F(H)\times\mc H(H)$.
\end{enumerate}
\end{proposition}
\begin{proof}
It follows immediately from \cite[Prop.4.9, Rem.4.6]{BDSK09},
using the fact that $\mc L_{A,B}(\mc K)$ is a Dirac structure.
\end{proof}

%\pecetta{
%Is it possible to define the Lie algebra $\mc F(H)$ as a subspace of $\mc V/\partial\mc V$,
%not of $\mc K/\partial\mc K$
%(and similarly, $\mc H(H)$ as subspace of $\mc V^\ell$)?
%We needed to go to the field of fractions in order to use Dirac structures and derive
%the Jacobi identity from there
%(and for the maximal isotropicity of $\mc L_{A,B}$ we need to work over a field
%of differential functions).
%But maybe one can prove Jacobi identity directly, without using Dirac structures.
%}

\subsection{Hamiltonian equations and integrability}
\label{sec:7.1b}

Let $\mc V$ be an algebra of differential functions.
We have a non-degenerate pairing 
$(\cdot\,|\,\cdot):\,\mc V^\ell\times\mc V^{\oplus\ell}\to\mc V/\partial\mc V$
given by
\begin{equation}\label{20130112:eq1}
(P|\xi)=\tint P\cdot\xi\,.
\end{equation}
(See e.g. \cite{BDSK09} for a proof of non-degeneracy of this form.)

Let $H\in\Mat_{\ell\times\ell}\mc V(\partial)$ be a non-local Poisson structure.
If $H=AB^{-1}$ is a fractional decomposition of $H$, 
with $A,B\in\Mat_{\ell\times\ell}\mc V[\partial]$ and $B$ non-degenerate,
according to Definition \ref{20130104:def} with $U=V=\mc V^{\oplus\ell}$ and $W=\mc V^\ell$,
we have that $\xi\in\mc V^{\oplus\ell}$ and $P\in\mc V^\ell$ are $(A,B)$-associated,
\begin{equation}\label{20130112:eq2}
\xi\ass{(A,B)}P\,,
\end{equation}
if there exists $F\in\mc K^{\oplus\ell}$ such that
$\xi=BF,\,P=AF$.

Let $\tint h\in\mc V/\partial\mc V$ and $P\in\mc V^\ell$
be such that $\frac{\delta h}{\delta u}\ass{(A,B)}P$
for some fractional decomposition $H=AB^{-1}$.
The corresponding \emph{Hamiltonian equation} is, by definition,
the following evolution equation on the variables $u=\big(u_i\big)_{i\in I}$:
\begin{equation}\label{20120124:eq5}
\frac{du}{dt}
=P\,.
\end{equation}

By the chain rule, any element $f\in\mc V$ evolves according to the equation
$$
\frac{df}{dt}=\sum_{i\in I}\sum_{n\in\mb Z_+}(\partial^nP_i)\frac{\partial f}{\partial u_i^{(n)}}\,,
$$
and, integrating by parts,
a local functional $\tint f\in\mc V/\partial\mc V$
evolves according to
$$
\frac{d\tint f}{dt}=\int P\cdot\frac{\delta f}{\delta u}
\quad\bigg(=\big(P\big|\frac{\delta f}{\delta u}\big)\bigg)\,.
$$

\begin{definition}\label{20130104:def2}
The Hamiltonian equation \eqref{20120124:eq5} is said to be \emph{integrable}
if there exist sequences $\{\xi_n\}_{n\in\mb Z_+}\subset\mc V^{\oplus\ell}$ 
and $\{P_n\}_{n\in\mb Z_+}\subset\mc V^\ell$
such that:
\begin{enumerate}[(i)]
\item
elements $\xi_n$'s and $P_n$'s span infinite dimensional 
(over the subalgebra $\mc C\subset\mc V$ of constants)
subspaces 
of $\mc V^{\oplus\ell}$ and $\mc V^\ell$ respectively;
\item
for every $n\in\mb Z_+$
we have the association relation $\xi_n\ass{(A,B)}P_n$,
for some fractional decomposition $H=AB^{-1}$;
\item
the elements $\xi_n$'s are closed,
i.e. they have self-adjoint Frechet derivatives: $D_{\xi_n}(\partial)=D_{\xi_n}^*(\partial)$;
\item
the elements $P_n$'s commute with respect the the Lie bracket \eqref{20120126:eq1}: 
$[P_m,P_n]=0$ for all $m,n\in\mb Z_+$;
\item
$(P_m\,|\,\xi_n)=0$ for all $m,n\in\mb Z_+$.
\end{enumerate}
In this case, we have an \emph{integrable hierarchy} of Hamiltonian equations
$$
\frac{du}{dt_n} = P_n\,,\,\,n\in\mb Z_+\,.
$$
\end{definition}

\begin{remark}\label{20130104:rem2}
Recall from \cite[Prop.1.5]{BDSK09}
that if $\xi\in\mc V^{\oplus\ell}$ is closed,
then it is exact in any normal algebra of differential function extension $\tilde{\mc V}$
of $\mc V$: $\xi=\frac{\delta h}{\delta u}$ for some $\tint h\in\tilde{\mc V}/\partial\tilde{\mc V}$.
If $\mc V$ is a domain,
then by Lemma \ref{20130112:lem}
it can be extended to a normal algebra of differential functions $\tilde{\mc V}$
which is still a domain,
and we can consider its field of fractions $\tilde{\mc K}$.
Then we have $\xi_n=\frac{\delta h_n}{\delta u}$,
where $\{\tint h_n\}_{n\in\mb Z_+}\subset\mc F(H)\subset\tilde{\mc K}/\partial\tilde{\mc K}$ 
form an infinite sequence of Hamiltonian functionals 
in involution: $\{\tint h_m,\tint h_n\}_H=0$ for every $m,n\in\mb Z_+$
(cf. Section \ref{sec:7.1}).
\end{remark}

In analogy to Liouville integrability of finite dimensional Hamiltonian systems,
we should require in addition
some completeness property 
of the span $\Xi\subset\mc V^{\oplus\ell}$ 
of the variational derivatives of the conserved densities $\xi_n=\frac{\delta h_n}{\delta u},\,n\in\mb Z_+$,
and of the span $\Pi$ of the generalized symmetries $P_n,\,n\in\mb Z_+$.
A natural condition, analogous to Liouville integrability, 
in the general setup of non-local Poisson structures, is the following.
\begin{definition}\label{def:compl-int}
A \emph{completely integrable system} associated 
to the non-local Poisson structure $H=AB^{-1}$, in its minimal fractional decomposition,
is a pair of subspaces 
$\Xi=B(U)\subset\mc V^{\oplus\ell}$
and
$\Pi=A(U)\subset\mc V^\ell$,
for some subspace $U\subset\mc V^{\oplus\ell}$,
such that
\begin{enumerate}[(i)]
\item
$\Xi$ consists of closed elements in $\mc V^{\oplus\ell}$:
$D_\xi(\partial)=D^*_\xi(\partial)$ for all $\xi\in\Xi$;
\item
$\Pi$ is an abelian subalgebra of $\mc V^\ell$ 
with respect to the Lie bracket \eqref{20120126:eq1};
\item
$\Pi^\perp=\Xi$ and $\Xi^\perp=\Pi$ with respect to the pairing \eqref{20130112:eq1}.
%%\item
%$\Xi\oplus\Pi$ is maximal isotropic in $\im B\oplus\im A$
%with respect to the inner product \eqref{20111020:eq3}.
\end{enumerate}
In this case, for every $P\in\Pi$ we get a completely integrable Hamiltonian equation
$\frac{du}{dt}=P$,
and all local functionals $\tint f\in\mc V/\partial\mc V$ such that $\frac{\delta f}{\delta u}\in\Xi$
are its integrals of motion in involution.
\end{definition}
\begin{remark}\label{20130129:rem2}
In the local case we 
let $\Xi$ be the span of  
$\xi_n=\frac{\delta h_n}{\delta u},\,n\in\mb Z_+$,
where $h_n$ are the conserved densities,
and we let $\Pi=H(\Xi)$.
Then the above condition $\Pi^\perp=\Xi$
is equivalent to the condition that $\Xi$
is a maximal isotropic subspace of $\Omega_1=\mc V^{\oplus\ell}$
with respect to the skewsymmetric bilinear form $\Omega_1\times\Omega_1\to\mc V/\partial\mc V$
given by $\langle\xi|\eta\rangle=(H\xi|\eta)$.
Indeed, $\xi\in\Omega_1$ satisfies 
$\langle\xi|\xi_n\rangle=-(\xi|H\xi_n)=0$ for all $n$
if and only if $\xi\perp H(\Xi)=\Pi$.
%Therefore, 
%to say that $\Xi\subset\Omega_1$ 
%is maximal isotropic 
%with respect to $\langle\cdot\,|\,\cdot\rangle$
%is the same as saying that $\Pi^\perp=\Xi$.
%
In this case, the $\tint h_n$'s are automatically in involution
and $\Pi$ consists of commuting higher symmetries.
\end{remark}
\begin{remark}
We can generalize Definition \ref{def:compl-int} of complete integrability
to the case of an arbitrary Dirac structure $\mc L\subset\mc V^{\oplus\ell}\oplus\mc V^\ell$
as a subspace $\Lambda\subset\mc L$ 
such that $\Xi=\pi_1(\Lambda)\subset\mc V^{\oplus\ell}$ 
and $\Pi=\pi_2(\Lambda)\subset\mc V^\ell$,
the projections of $\Lambda\subset\mc V^{\oplus\ell}\oplus\mc V^\ell$ in the first 
and second components respectively,
satisfy conditions (i)--(iii) above.
These conditions are equivalent to require that $\pi_1(\Lambda)$ consists of closed elements,
that $\Lambda_0=\pi_1(\Lambda)\oplus\pi_2(\Lambda)$
is a maximal isotropic subspace of $\mc V^{\oplus\ell}\oplus\mc V^\ell$
with respect to the symmetric bilinear form $\langle\cdot\,|\,\cdot\rangle$ 
defined in \eqref{20111020:eq3},
and the Courant-Dorfman product $\circ$ defined in \eqref{20111020:eq5} 
restricted to $\Lambda_0$ is zero.
In other words, $\Lambda_0$ is a Dirac structure with zero Courant-Dorfman product.
\end{remark}
\begin{example}
It is not hard to check, using arguments similar to those in \cite{BDSK09},
that the KdV equation is completely integrable
in the sense of Definition \ref{def:compl-int}.
\end{example}

%%%
\subsection{The Lenard-Magri scheme of integrability}
\label{sec:7.2}

\begin{theorem}\label{20130123:thm}
Let $\mc V$ be an algebra of differential functions,
%with subalgebra of constants $\mc C\subset\mc V$,
and let $H=AB^{-1}$ and $K=CD^{-1}$ be rational skewadjoint pseudodifferential operators,
with $A,B,C,D\in\Mat_{\ell\times\ell}\mc V[\partial]$ and $B,D$ non-degenerate.
Let 
$\{\xi_n\}_{n=-1}^N\subset\mc V^{\oplus\ell}$,
$\{P_n\}_{n=0}^N\subset\mc V^\ell$
be finite sequences such that
\begin{equation}\label{20130105:eq1}
\xi_{n-1}\ass{(A,B)}P_n\ass{(C,D)}\xi_n
\,,
\end{equation}
holds for every $n=0,\dots,N$.
Then:
\begin{enumerate}[(a)]
\item
We have
\begin{equation}\label{20130131:eq1}
(P_n|\xi_m)=0\,\,,\,\,\,\, m\geq-1,\,n\geq0\,.
\end{equation}
\item
Assume that $\mc V$ is a domain and 
$H$ and $K$ are compatible non-local Poisson structures with $K$ non-degenerate,
and assume that $\xi_{-1}$ and $\xi_0$ are closed,
i.e. their Frechet derivatives are selfadjoint.
Then the elements $\xi_n,\,n\geq1$, are closed as well,
and we have 
\begin{equation}\label{20130123:eq5}
[P_m,P_n]\in\ker B^*\cap\ker D^*\,\,,\,\,\,\, m,n\geq0\,.
\end{equation}
\item
Assume that the matrices $A,B,C,D$ have non-degenerate leading coefficients,
and that 
\begin{equation}\label{20130123:eq1}
\begin{array}{c}
\displaystyle{
\dord(P_n)\,>\,\max\Big\{\dord(A)-|H|+|K|,\dord(B)+|K|,
} \\
\displaystyle{
\dord(C),\dord(D)+|K|\Big\}\,,
}
\end{array}
\end{equation}
for some $n\geq0$.
Then
$$
\dord(\xi_n)=\dord(P_n)-|K|
\,\,,\,\,\,\,
\dord(P_{n+1})=\dord(P_n)+|H|-|K|\,.
$$
In particular, if $|H|\geq|K|$, then 
$$
\dord(P_{j})=\dord(P_n)+(j-n)(|H|-|K|)=\dord(\xi_{j})+|K|
\,,
$$
for every $j\geq n$.
\item
Assume that the following orthogonality conditions hold:
\begin{equation}\label{20130104:eq2}
\Big(\Span{}_{\mc C}\{\xi_m\}_{m=-1}^N\Big)^\perp\subset\im C
\,\,,\,\,\,\,
\Big(\Span{}_{\mc C}\{P_n\}_{n=0}^N\Big)^\perp\subset\im B\,,
\end{equation}
where the orthogonal complements are with respect to the pairing
defined in  \eqref{20130112:eq1}.
Then we can extend the given finite sequences to infinite sequences 
$\{\xi_m\}_{m=-1}^\infty\subset \mc V^{\oplus\ell}$,
$\{P_n\}_{n=0}^\infty\subset \mc V^\ell$ 
such that the association relations \eqref{20130105:eq1} 
hold for every $n\in\mb Z_+$.
\end{enumerate}
\end{theorem}
\begin{corollary}\label{20130123:cor}
Let $\mc V$ be an algebra of differential functions, which is a domain,
let $H=AB^{-1}$ and $K=CD^{-1}$ be compatible non-local Poisson structures,
where $A,B,C,D$ are $\ell\times\ell$ matrix differential operators with non-degenerate
leading coefficients and such that $|H|>|K|$.
Let  $\{\xi_n\}_{n=-1}^N\subset\mc V^{\oplus\ell}$,
$\{P_n\}_{n=0}^N\subset\mc V^\ell$
be finite sequences such that
$\xi_{-1}$ and $\xi_0$ are closed,
conditions \eqref{20130105:eq1} hold for every $0\leq n\leq N$,
condition \eqref{20130123:eq1} holds for some $0\leq n\leq N$,
and the orthogonality conditions \eqref{20130104:eq2} hold.
Then
the given finite sequences can be extended to infinite sequences
$\{\xi_n\}_{n=-1}^\infty\subset\mc V^{\oplus\ell}$,
$\{P_n\}_{n=0}^\infty\subset\mc V^\ell$,
such that 
the differential orders of the $\xi_n$'s and the $P_n$'s tend to infinity
as $n\to\infty$,
all $\xi_n$'s are closed,
and equations \eqref{20130105:eq1}, \eqref{20130131:eq1} and \eqref{20130123:eq5} hold.
Consequently,
\begin{equation}\label{20130123:eq2}
\frac{du}{dt_n}=P_n
\end{equation}
is an integrable bi-Hamiltonian equation for every $n\in\mb Z_+$.
If, moreover, $\ker B^*\cap\ker D^*=0$,
all equations \eqref{20130123:eq2} form a (compatible) 
integrable hierarchy of bi-Hamiltonian equations.
\end{corollary}

\begin{proof}[Proof of Theorem \ref{20130123:thm}]
Parts (a) and (d) are special cases of Lemma \ref{20130104:lem}(a) and (b) respectively,
since the assumption that $H$ and $K$ are skewadjoint
is equivalent to equations \eqref{20130104:eq3}.
Part (b) follows Lemmas \ref{20130123:lem1} and \ref{20130123:lem2} below.
Finally, part (c) follows form Lemma \ref{20130123:lem3} below,
with $\xi=\xi_n$, $P=P_n$ and $Q=P_{n+1}$.
\end{proof}
\begin{lemma}\label{20130123:lem1}
Let $H=AB^{-1}$ and $K=CD^{-1}$ be compatible non-local Poisson structures,
with $K$ non-degenerate,
over the algebra of differential functions $\mc V$, which is a domain.
Let $\xi_{0},\xi_1$ be closed elements of $\mc V^{\oplus\ell}$, $\xi_2\in\mc V^{\oplus\ell}$,
and $P_1,P_2\in\mc V^{\ell}$ be such that
\begin{equation}\label{20130123:eq4}
\xi_{0}\ass{(A,B)}P_1
\ass{(C,D)}\xi_1
\ass{(A,B)}P_2
\ass{(C,D)}
\xi_{2}\,.
\end{equation}
Then $\xi_2$ is closed.
\end{lemma}
\begin{proof}
If $H=AB^{-1}$ and $K=CD^{-1}$ are minimal fractional decompositions,
then the statement follows from Theorem \ref{mtst-nonloc}(b).
Indeed, conditions \eqref{20130123:eq4}
imply the existence of $Z,Z',W,W'\in\mc V^{\oplus\ell}$ 
such that $\xi_0=B(\partial)Z$, $\xi_1=D(\partial)Z^\prime=B(\partial)W$, 
$\xi_2=D(\partial)W^\prime$,
and solving equations \eqref{20130123:eq3}.

In general, the fractional decompositions $H=AB^{-1}$ and $K=CD^{-1}$ are 
not necessarily minimal.
Let $\mc K$ be the field of fractions of $\mc V$.
By Proposition \ref{prop:minimal-fraction}
we have
$A=A_1P$, $B=B_1P$, $C=C_1Q$, $D=D_1Q$,
with $A_1,B_1,C_1,D_1,P,Q\in\Mat_{\ell\times\ell}\mc K[\partial]$,
where 
$H=A_1B_1^{-1}$ and $K=C_1D_1^{-1}$ are minimal fractional decompositions,
and $P$, $Q$ are non-degenerate.
Obviously, by the definition \eqref{20130112:eq2} of $(A,B)$-association,
if $\xi\ass{(A,B)}P$ holds over $\mc V$,
in the sense that $\xi=BF,\,P=AF$ for some $F\in\mc V^{\oplus\ell}$,
then $\xi\ass{(A_1,B_1)}P$ holds over $\mc K$,
indeed $\xi=B_1F_1,\,P=A_1F_1$, where $F_1=PF\in\mc K^{\oplus\ell}$.
Hence, conditions \eqref{20130123:eq4} hold (over $\mc K$)
with $A,B,C,D$ replaced by $A_1,B_1,C_1,D_1$.
Then, by Theorem \ref{mtst-nonloc}(b) we get that $\xi_2$ is closed over $\mc K$,
hence over $\mc V$.
\end{proof}
\begin{lemma}\label{20130123:lem2}
Let $H=AB^{-1}$ and $K=CD^{-1}$ be compatible non-local Poisson structures,
with $K$ non-degenerate,
over the algebra of differential functions $\mc V$, which is a domain.
Let 
$\{\xi_n\}_{n=-1}^N$ be closed elements of $\mc V^{\oplus\ell}$,
and $\{P_n\}_{n=0}^N$ be elements of $\mc V^\ell$
satisfying conditions \eqref{20130105:eq1}.
Then 
$$
[P_m,P_n]\in\ker B^*\cap\ker D^*\,\,,\,\,\,\, m,n\geq0\,.
$$
\end{lemma}
\begin{proof}
Let $\mc K$ be the field of fractions of $\mc V$,
and let 
$H=A_1B_1^{-1}$, $K=C_1D_1^{-1}$ be their minimal fractional decomposition over $\mc K$.
As observed in the proof of Lemma \ref{20130123:lem1},
all association relations \eqref{20130105:eq1} hold,
over $\mc K$, after replacing $A,B,C,D$ with $A_1,B_1,C_1,D_1$ respectively.
By Theorem \ref{20111020:thm},
$\mc L_{A_1,B_1}(\mc K)$ and $\mc L_{C_1,D_1}(\mc K)$ are Dirac structures 
in $\mc K^{\oplus\ell}\oplus\mc K^\ell$,
in particular they are closed with respect to the Courant-Dorfman product \eqref{20111020:eq5}.
By the definition \eqref{20120109:eq1} of the Dirac structures 
$\mc L_{A_1,B_1}(\mc K)$ and $\mc L_{C_1,D_1}(\mc K)$,
we have
$\xi_{n-1}\oplus P_n\in\mc L_{A_1,B_1}(\mc K)$ and $\xi_n\oplus P_n\in\mc L_{C_1,D_1}(\mc K)$
for every $n\geq0$.
By the assumption that all the $\xi_n$'s are closed,
we can use formula \eqref{20120127:eq1} to deduce that
$\frac{\delta}{\delta u} (P_m\,|\,\xi_{n-1})\oplus[P_m,P_n]
\in\mc L_{A_1,B_1}(\mc K)$
and
$\frac{\delta}{\delta u} (P_m\,|\,\xi_n)\oplus[P_m,P_n]
\in\mc L_{A_1,B_1}(\mc K)$.
Hence, by Theorem \ref{20130123:thm}(a) 
we conclude that 
$0\oplus[P_m,P_n]\in\mc L_{A_1,B_1}(\mc K)\cap\mc L_{C_1,D_1}(\mc K)$.
Namely, there exist $F_1\in\ker B_1\subset\mc K^\ell$ and $G_1\in\ker D_1\subset\mc K^\ell$
such that $[P_m,P_n]=A_1(\partial)F_1=C_1(\partial)G\in A_1(\ker B_1)\cap C_1(\ker D_1)$.
By skewadjointness of $H$ and $K$,
we have $B_1^*A_1=-A_1^*B_1$ and $D_1^*C_1=-C_1^*D_1$, 
which immediately implies
$A_1(\ker B_1)\subset\ker B_1^*$ and $C_1(\ker D_1)\subset\ker D_1^*$.
Therefore, 
$[P_m,P_n]\in\ker B_1^*\cap\ker D_1^*\subset\mc K^\ell$.
On the other hand, since $B$ and $D$ are right multiples of $B_1$ and $D_1$ respectively,
we have $\ker B_1^*\cap\mc V^\ell\subset\ker B^*$ and $\ker D_1^*\cap\mc V^\ell\subset\ker D^*$.
Thererore,
$[P_m,P_n]\in\ker B^*\cap\ker D^*\subset\mc V^\ell$, as we wanted.
\end{proof}
\begin{lemma}\label{20130123:lem3}
Let $\mc V$ be an algebra of differential functions,
let $A,B,C,D\in\Mat_{\ell\times\ell}\mc V[\partial]$
be matrices with non-degenerate leading coefficients.
Denote by 
$H=AB^{-1}$ and $K=CD^{-1}$ the corresponding rational matrix 
pseudodifferential operators.
Let $P,Q\in\mc V^\ell$ and $\xi\in\mc V^{\oplus\ell}$ satisfy the following 
association relations (cf. Definition \ref{20130104:def}) 
\begin{equation}\label{20120910:eq3}
P\ass{(C,D)}\xi\ass{(A,B)}Q\,,
\end{equation}
and assume that
\begin{equation}\label{20120911:eq1}
\begin{array}{c}
\displaystyle{
\dord(P)\,>\,\max\Big\{\dord(A)-|H|+|K|,\dord(B)+|K|,
} \\
\displaystyle{
\dord(C),\dord(D)+|K|\big\}.
}
\end{array}
\end{equation}
(Here we use the notation introduced in \eqref{20120910:eq1} and \eqref{20120910:eq2}.)
Then
$$
\dord\xi=\dord(P)-|K|
\,\,\,\,\text{ and }\,\,
\dord(Q)=\dord(P)+|H|-|K|\,.
$$
\end{lemma}
\begin{proof}
By definition, the relations \eqref{20120910:eq3} amount to the existence of elements
$F,G\in\mc V^\ell$ such that
$$
CG=P
\,\,,\,\,\,\,
DG=\xi
\,\,,\,\,\,\,
BF=\xi
\,\,,\,\,\,\,
AF=Q
\,.
$$
Since $C$ has non-degenerate leading coefficient 
and $\dord(P)=\dord(CG)>\dord(C)$, 
we get by Lemma \ref{20120910:lem1}(c) that 
$$
\dord(G)=\dord(CG)-|C|=\dord(P)-|C|\,.
$$
Next, since by assumption $D$ has non-degenerate leading coefficient 
and $\dord(G)+|D|=\dord(P)-|C|+|D|=\dord(P)-|K|>\dord(D)$, 
we get by Lemma \ref{20120910:lem1}(b) that 
$$
\begin{array}{c}
\displaystyle{
\dord(\xi)=\dord(DG)=\dord(G)+|D|=\dord(P)-|C|+|D|
} \\
\displaystyle{
=\dord(P)-|K|\,.
}
\end{array}
$$
Similarly, since, by assumption, $B$ has non-degenerate leading coefficient 
and $\dord(BF)=\dord(\xi)=\dord(P)-|K|>\dord(B)$,
we get by Lemma \ref{20120910:lem1}(c) that 
$$
\dord(F)=\dord(BF)-|B|=\dord(\xi)-|B|=\dord(P)-|K|-|B|\,.
$$
Finally, since, by assumption, $A$ has non-degenerate leading coefficient 
and $\dord(F)+|A|=\dord(P)-|K|=|B|+|A|=\dord(P)-|K|+|H|>\dord(A)$, 
we get by Lemma \ref{20120910:lem1}(b) that 
$$
\begin{array}{l}
\displaystyle{
\dord(Q)=\dord(AF)=\dord(F)+|A|=\dord(P)-|K|-|B|+|A|
}\\
\displaystyle{
=\dord(P)-|K|+|H|\,.
}
\end{array}
$$
\end{proof}
\begin{proof}[Proof of Corollary \ref{20130123:cor}]
The statement of the corollary is basically a summary of parts (a)--(d) of Theorem \ref{20130123:thm},
except that we need to explain why,
by the condition \eqref{20130123:eq5},
it follows that each $P_n$ lies in an infinite dimensional abelian subalgebra 
contained in $\Span\{P_n\}_{n=0}^\infty$.
This follows from the observation that
$\ker(B^*)\cap\ker(D^*)$ is finite dimensional over $\mc C$,
and the following result.
\begin{lemma}\label{20120906:lem1}
Let $U$ be an infinite dimensional subspace of a Lie algebra
such that $[U,U]$ is finite dimensional.
Then any element of $U$ is contained in an infinite dimensional abelian subalgebra of $U$.
\end{lemma}
\begin{proof}
Let $a_1$ be a non-zero element of $U$.
The centralizer $C_1$ of $a_1$ in $U$
is the kernel of the map $\ad a:\,U\to[U,U]$,
hence, it has finite codimension in $U$.
Next, let $a_2$ be an element of $C_1$ linearly independent of $a_1$,
and let $C_2$ be its centralizer in $C_1$.
By the same argument, $C_2$ has finite codimension in $C_1$.
In this fashion we construct an infinite sequence of linearly independent
commuting elements of $U$.
\end{proof}
\end{proof}
\begin{remark}\label{20130129:rem3}
Suppose that the sequences $\{\xi_n\}_{n=-1}^\infty\subset\mc V^{\oplus\ell}$
and $\{P_n\}_{n=0}^\infty$ satisfy relations \eqref{20130105:eq1} for each $n\in\mb Z_+$
with respect to the compatible non-local Poisson structures $H=AB^{-1}$
and $K=CD^{-1}$,
and assume that their spans 
$\Xi=\Span\{\xi_n\}_{n\geq-1}$ and $\Pi=\Span\{P_n\}$
define a completely integrable system in the sense of Definition \ref{def:compl-int}.
Then the orthogonality conditions \eqref{20130104:eq2} automatically hold for some $N$
(possibly infinite).
Indeed, by the relations \eqref{20130105:eq1} we have, in particular,
that $\Xi\subset\im(B)\cap\im(D)$ and $\Pi\subset\im(A)\cap\im(C)$.
Therefore conditions \eqref{20130104:eq2}
follow by axiom (iii) in Definition \ref{def:compl-int}.
\end{remark}
\begin{remark}\label{20130129:rem1}
Unfortunately we are unable to prove a stronger form of equation \eqref{20130123:eq5},
namely that the generalized symmetries $P_n$ obtained by the Lenard-Magri scheme
commute.
The usual ``proof'' of this fact using the recursion operator
(see e.g. \cite[p.64]{Bla98} or \cite[p.317]{Olv93}) is not rigorous.
In fact, we have a counterexample in Section \ref{secb:4.2}.
\end{remark}

The recurrence relations \eqref{20130105:eq1}
are usually represented by the following diagram,
called the Lenard-Magri scheme:
%%%%%% DIAGRAM %%%%%%%%%%%%%%%%%%%%%%%%%%
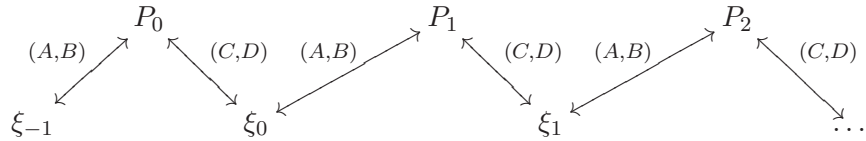
\begin{equation}\label{maxi2}
\UseTips
\xymatrix{
& P_{0} \ar@{<->}[dl]_{(A,B)} \ar@{<->}[dr]^{(C,D)} & & &
P_1 \ar@{<->}[dll]_{(A,B)} \ar@{<->}[dr]^{(C,D)} & & &
P_{2} \ar@{<->}[dll]_{(A,B)} \ar@{<->}[dr]^{(C,D)} & \\
\xi_{-1} & & \xi_{0} & & & \xi_1 & & & \dots 
}
\end{equation}
%%%%%%%%%%%%%%%%%%%%%%%%%%%%%%%%%%%%%%
%
Explicitly, diagram \eqref{maxi2} holds if there exists a sequence 
$\{F_n\}_{n\geq-1}$ in $\mc V^{\oplus\ell}$ such that the following equations hold
($n\in\mb Z_+$):
\begin{equation}\label{20120315:eq1}
B(\partial)F_{-1}=\xi_{-1}
\,,\,\,
C(\partial)F_{2n}=A(\partial)F_{2n-1}=P_n
\,,\,\,
B(\partial)F_{2n+1}=D(\partial)F_{2n}=\xi_n
\,.
\end{equation}

%%%
\subsection{Notation, terminology and assumptions}
\label{sec:7.5}

In the following sections we apply the machinery developed so far
in explicit examples.

As stated in Theorem \ref{20130123:thm}, if the algebra $\mc V$ is a domain,
the coordinates of all $\xi_n$'s and $P_n$'s solving the recurrence
equation \eqref{20130105:eq1} lie in $\mc V$.
However, for notational purposes
it is convenient to go to the field of fractions of $\mc V$,
so we will assume that $\mc V$ is a field.
Under this assumption, the notation $\tint h\ass{H}P$ introduced in Definition \ref{20120124:def}
is consistent with the notation $\xi\ass{(A,B)}P$ in \eqref{20130112:eq2}.
Namely, the following conditions are equivalent:
\begin{enumerate}[(i)]
\item
$\tint h\ass{H}P$,
\item
$\frac{\delta h}{\delta u}\ass{(A,B)}P$
for some fractional decomposition $H=AB^{-1}$,
\item
$\frac{\delta h}{\delta u}\ass{(A,B)}P$
for the minimal fractional decomposition $H=AB^{-1}$.
\end{enumerate}
Hence, in the rest of the paper we will use the more suggestive notation $\tint h\ass{H}P$.

Furthermore, in all examples we will begin the Lenard-Magri scheme with $\xi_{-1}=0$.
In this case, by Lemma \ref{20130123:lem1},
all the elements $\xi_n$ are closed,
provided that $\xi_0$ is closed.

Recall also that if all $\xi_n$'s are closed,
they are exact in some algebra of differential function extension 
$\tilde{\mc V}$ of $\mc V$:
$\xi_n=\frac{\delta h_n}{\delta u}$,
for some $h_n\in\tilde{\mc V}$.
In this case, and using the more suggestive notation above,
the diagram \eqref{maxi2} takes the form
%%%%%% DIAGRAM %%%%%%%%%%%%%%%%%%%%%%%%%%
\begin{equation}\label{maxi}
\tint 0\ass{H}P_0\ass{K}\tint h_0\ass{H}P_1\ass{K}\tint h_1\ass{H}
\dots
\end{equation}
%%%%%%%%%%%%%%%%%%%%%%%%%%%%%%%%%%%%%%
and it is equivalent to the existence of
$\{F_n\}_{n\geq-1}$ in $\mc V^{\oplus\ell}$ such that the following equations hold
($n\in\mb Z_+$):
\begin{equation}\label{20120315:eq1b}
B(\partial)F_{-1}=0
\,,\,\,
C(\partial)F_{2n}=A(\partial)F_{2n-1}=P_n
\,,\,\,
B(\partial)F_{2n+1}=D(\partial)F_{2n}=\frac{\delta h_n}{\delta u}
\,.
\end{equation}
\begin{remark}\label{20130130:rem1}
In general, if $\mc V$ is an arbitrary algebra of differential functions,
$\{\xi_n\}_{n=-1}^N\subset\mc V^{\oplus\ell}$,
$\{P_n\}_{n=0}^N\subset\mc V^\ell$,
and $A,B,C,D\in\Mat_{\ell\times\ell}\mc V[\partial]$,
then,
the whole infinite sequences $\{\xi_n\}_{n\geq-1}$
and $\{P_n\}_{n\geq0}$, constructed using Theorem \ref{20130123:thm},
have coordinates in $\mc V$.
\end{remark}

%%%
\subsubsection*{S-type vs C-type Lenard-Magri schemes}

Consider a Lenard-Magri scheme as in \eqref{maxi}.
We say that it is \emph{finite} if it can be extended indefinitely,
but in any such infinite extension 
the linear span of $\{\tint h_n\}_{n\in\mb Z_+}$ or of $\{P_n\}_{n\in\mb Z_+}$
is finite dimensional.
We say that the Lenard-Magri scheme \eqref{maxi}
is \emph{blocked} if it cannot be extended indefinitely,
namely, for some $n$, 
there is no $\tint h_n$ such that $P_n\ass{K}\tint h_n$,
or there is no $P_{n+1}$ such that $\tint h_n\ass{H}P_{n+1}$.

For an integrable Lenard-Magri scheme \eqref{maxi},
we say that it is of \emph{S-type} if the differential orders of the elements $P_n$ grow to infinity,
and it is of \emph{C-type} if the differential orders of the $P_n$'s are bounded.
It is easy to see that for an integrable Lenard-Magri scheme of S-type 
the order of the pseudodifferential $H$ should be greater than the order 
of the pseudodifferential $K$.
Indeed, since we have $P_n\ass{K}\tint h_n\ass{H}P_{n+1}$,
if $P_n$ has differential order large enough,
then, by Lemma \ref{20130123:lem3}, 
$\dord(P_{n+1})=\dord(P_n)+\ord(H)-\ord(K)$.

\begin{remark}
This terminology in inspired by the terminology of Calogero,
who calls an integrable hierarchy of ``S-type'' if the differential orders 
of the canonical conserved densities are unbounded,
and of ``C-type'' otherwise (see \cite{MSS90,MS12}).
Note that, though these two terminologies are close, they do not coincide. For example,
the linear equation $\frac{du}{dt}=u'''$ is C-integrable in Calogero's terminology,
but the corresponding Lenard-Magri scheme, considered for example in \cite{BDSK09},
is integrable of S-type.
\end{remark}

%%%%%%%%%%%%%%%%%%%%%%%%%%%%%%%%%%%%%%%
\section{Liouville type integrable systems}
\label{secb:3}

In this section $\mc V$ is a field of differential functions in $u$,
and we assume that $\mc V$ contains all the functions that we encounter
in our computations.
As before, we denote by $\mc C\subset\mc V$ the subfield of constants,
and by $\mc F\subset\mc V$ the subfield of quasiconstants.
We shall denote by $x$ an element of $\mc F$ such that $\partial x=1$.

Recall from  Example \ref{20110922:ex1} that we have the following triple
of compatible non-local Poisson structures:
$$
L_1=\partial\,,\,\, %(GFZ)
L_2=\partial^{-1}\,,\,\, %(Toda)
L_3=u'\partial^{-1}\circ u'\,. %(Sokolov)
$$
Given two non-local Poisson structures $H$ and $K$ of the form
\begin{equation}\label{20121006:eq2}
H=a_1L_1+a_2L_2+a_3L_3
\,\,,\,\,\,\,
K=b_1L_1+b_2L_2+b_3L_3\,,
\end{equation}
with $a_i,b_i\in\mc C,\,i=1,2,3$,
we want to discuss the integrability of the corresponding Lenard-Magri scheme.

%%%
\subsection{Preliminary computations}
\label{secb:3.1}

First, we find a minimal fractional decomposition for the operators $H$ and $K$.
\begin{lemma}\label{lem:frac}
For $X=x_1L_1+x_2L_2+x_3L_3$, with $x_1,x_2,x_3\in\mc C$, we have
\begin{equation}\label{frac-liouv}
X=\Big[x_1\partial^2\circ\frac1{u''}\partial+\frac{x_2+x_3(u')^2}{u''}\partial-x_3u'\Big]
\Big[\partial\circ\frac1{u''}\partial\Big]^{-1}\,.
\end{equation}
The above fractional decomposition is minimal only for $x_2x_3\neq0$.
For $x_2\neq0,\,x_3=0$, 
the minimal fractional decomposition is 
\begin{equation}\label{frac-liouv1}
X=(x_1\partial^2+x_2)\partial^{-1}\,,
\end{equation}
for $x_2=0,\,x_3\neq0$ it is 
\begin{equation}\label{frac-liouv2}
X=\Big(x_1\partial\circ\frac1{u'}\partial+x_3u'\Big)
\Big(\frac1{u'}\partial\Big)^{-1}\,,
\end{equation}
and for $x_2=x_3=0$ it is 
%\begin{equation}\label{frac-liouv1}
$X=x_1\partial$.
%\end{equation}
\end{lemma}
\begin{proof}
Straightforward.
\end{proof}

Later we will need the following simple facts concerning the numerators 
of the fractional decompositions for $X$.
\begin{lemma}\label{20120909:lem1}
\begin{enumerate}[(a)]
\item
For $x_1,x_2,x_3\in\mc C$, $x_1\neq0$, consider the equation
\begin{equation}\label{20120909:eq1}
\Big(x_1\partial^2\circ\frac1{u''}\partial+\frac{x_2+x_3(u')^2}{u''}\partial-x_3u'\Big)F=f\,,
\end{equation}
in $F\in\mc V$ and $f\in\mc F$.
If $x_3\neq0$,
then all the solutions of equation \eqref{20120909:eq1} are
$$
F=\alpha u'
\,\,,\,\,\,\,
f=x_2\alpha
\,\,\text{ for some }\,\alpha\in\mc C\,,
$$
while if $x_3=0$, 
then all the solutions of equation \eqref{20120909:eq1} are
$$
F=\alpha u'+\beta (xu'-u)+\gamma
\,\,,\,\,\,\,
f=x_2(\alpha+\beta x)
\,\,\text{ for some }\,\alpha,\beta,\gamma\in\mc C\,.
$$
\item
For $x_1,x_2\in\mc C$, $x_1\neq0$, 
an element $F\in\mc V$ satisfies
\begin{equation}\label{20120909:eq2}
(x_1\partial^2+x_2)F\in\mc F
\end{equation}
if and only if $F\in\mc F$.
\item
For $x_1,x_3\in\mc C$, $x_1\neq0$, 
an element $F\in\mc V$ satisfies
\begin{equation}\label{20120909:eq3}
\Big(x_1\partial\circ\frac1{u'}\partial+x_3u'\Big)F\in\mc V_1
\end{equation}
if and only if $F\in\mc V_0$ and $F'=\frac{\partial F}{\partial u}u'$.
\end{enumerate}
\end{lemma}
\begin{proof}
If $n\geq2$ and $F\in\mc V$ solves equation \eqref{20120909:eq1}
and has differential order less than or equal to $n$, then,
using \eqref{eq:0.4}, we have
$$
0=\frac{\partial}{\partial u^{(n+3)}}LHS\eqref{20120909:eq1}
=x_1\frac1{u''}\frac{\partial F}{\partial u^{(n)}}\,,
$$
which implies that $\frac{\partial F}{\partial u^{(n)}}=0$.
Hence $F$ must have differential order at most $1$.
Then we have
$$
0=\frac{\partial}{\partial u^{(4)}}LHS\eqref{20120909:eq1}
=x_1\Big(\frac1{u''}\frac{\partial F}{\partial u'}-\frac1{(u'')^2}F'\Big)\,,
$$
so that $F'=\frac{\partial F}{\partial u'}u''$.
But then equation \eqref{20120909:eq1} becomes
\begin{equation}\label{20120730:eq1}
x_1\Big(\frac{\partial F}{\partial u'}\Big)^{''}+(x_2+x_3(u')^2)\frac{\partial F}{\partial u'}-x_3u'F=f\,.
\end{equation}
If $\frac{\partial F}{\partial u'}$ has differential order $n\geq0$,
then applying $\frac{\partial}{\partial u^{(n+2)}}$ to both sides of equation \eqref{20120730:eq1}
we get $\frac{\partial^2 F}{\partial u^{(n)}\partial u'}=0$.
Hence, $\frac{\partial F}{\partial u'}=\varphi\in\mc F$.
In other words, $F=\varphi u'+f_0$, where $f_0\in\mc V_0$ has differential 
order less than or equal to $0$.
But then the condition $F'=\frac{\partial F}{\partial u'}u''$ becomes
$\varphi' u'+f_0'=0$.
This implies, using \eqref{20120907:eq1},
that $(f_0+\varphi'u)'=\varphi''u\in\partial\mc V\cap\mc V_0=\partial\mc F$.
So, necessarily, $\varphi''=0$.
Hence, $\varphi=\alpha+\beta x$
and $f_0=-\beta u+\gamma$, for some constants $\alpha,\beta,\gamma\in\mc C$.
Putting these results together, we have
$$
F=(\alpha+\beta x) u'-\beta u+\gamma\,,
$$
and plugging back into equation \eqref{20120730:eq1} we get
$$
x_2(\alpha+\beta x)+x_3\beta uu'-x_3\gamma u'=f\,.
$$
Since, by assumption, $f\in\mc F$,
we obtain $\beta=\gamma=0$ if $x_3\neq 0$,
completing the proof of part (a).

For part (b) we just observe that, if $F\in\mc V_n$ for some $n\geq0$
satisfies condition \eqref{20120909:eq2}, then
$$
0=\frac{\partial}{\partial u^{(n+2)}}(x_1F''+x_2F)
=x_1\frac{\partial F}{\partial u^{(n)}}\,.
$$
Hence, $F$ must be a quasiconstant.

Similarly, 
if $F\in\mc V_n$ for some $n\geq1$
satisfies condition \eqref{20120909:eq3}, then
$$
0=\frac{\partial}{\partial u^{(n+2)}}
\Big(x_1\partial\frac{F'}{u'}+x_3u'F\Big)
=\frac{x_1}{u'}\frac{\partial F}{\partial u^{(n)}}\,.
$$
Hence, $F$ must lie in $\mc V_0$.
Furthermore, 
$$
0=\frac{\partial}{\partial u''}
\Big(x_1\partial\frac{F'}{u'}+x_3u'F\Big)
=x_1\Big(-\frac{F'}{(u')^2}+\frac1{u'}
\frac{\partial F}{\partial u}\Big)\,.
$$
Hence, $F$ must be such that $F'=\frac{\partial F}{\partial u}u'$.
\end{proof}

Next, we compute the spaces $\mc F_0(X)$ and $\mc H_0(X)$
defined in \eqref{20120908:eq1} and \eqref{20120908:eq2}.
Here and further we use the following notation: given two constants $x_i,x_j\in\mc C$
such that $x_i\neq0$, we let
\begin{equation}\label{notation}
x_{ij}=\sqrt{-\frac{x_j}{x_i}}\,\in\mc C\,.
\end{equation}
(We assume that the field $\mc C$ contains all such elements.)
\begin{lemma}\label{20120908:lem1}
For $X=x_1L_1+x_2L_2+x_3L_3$, we have:
\begin{enumerate}[(a)]
\item
\begin{enumerate}[]
\item
$\mc F_0(X)=\ker\big(\frac\delta{\delta u}\big)$
if $x_1x_2x_3\neq0$;
\item
$\mc F_0(X)=\mc C\tint e^{x_{12}x}u
+\mc C\tint e^{-x_{12}x}u+\ker\big(\frac{\delta}{\delta u}\big)$
if $x_1x_2\neq0,x_3=0$;
\item
$\mc F_0(X)=\mc C\tint e^{x_{13}u}
+\mc C\tint e^{-x_{13}u}+\ker\big(\frac{\delta}{\delta u}\big)$
if $x_1x_3\neq0,x_2=0$;
\item
$\mc F_0(X)=\mc C\tint u+\ker\big(\frac\delta{\delta u}\big)$
if $x_1\neq0$ and $x_2=x_3=0$;
\item
$\mc F_0(X)=\mc C\tint\sqrt{x_2+x_3(u')^2}+\ker\big(\frac\delta{\delta u}\big)$,
if $x_1=0,x_2x_3\neq0$;
\item
$\mc F_0(X)=\ker\big(\frac\delta{\delta u}\big)$ if $x_1=0$ and $x_2=0$ or $x_3=0$.
\end{enumerate}
\item
\begin{enumerate}[]
\item
$\mc H_0(X)=\mc C\oplus \mc Cu'$ if $x_2x_3\neq0$;
\item
$\mc H_0(X)=\mc C$ if $x_2\neq0,x_3=0$;
\item
$\mc H_0(X)=\mc Cu'$ if $x_2=0,x_3\neq0$;
\item
$\mc H_0(X)=0$ if $x_2=x_3=0$.
\end{enumerate}
\end{enumerate}
\end{lemma}
\begin{proof}
First, let us find all elements $P\in\mc H_0(X)$. 
By Remark \ref{20120201:rem3},
if $X=YZ^{-1}$ is a minimal fractional decomposition,
we need to solve the following equations in $F,P\in \mc V$:
\begin{equation}\label{20120908:eq3}
ZF=0\,\,,\,\,\,\,P=YF\,.
\end{equation}
By Lemma \ref{lem:frac},
if $x_2=x_3=0$, then $Y=x_1\partial$ and $Z=1$, so the only solution 
of \eqref{20120908:eq3} is given by $F=0$, $P=0$.
If $x_2\neq0,\,x_3=0$, then $Y=x_1\partial^2+x_2$ and $Z=\partial$, so 
we get $F\in\mc C$ and $P\in\mc C$.
Similarly, if $x_2=0,\,x_3\neq0$, then $Y=x_1\partial\circ\frac1{u'}\partial+x_3u'$ 
and $Z=\frac1{u'}\partial$, so we get $F\in\mc C$ and $P\in\mc Cu'$.
Finally, if $x_2\neq0,\,x_3\neq0$, then 
$Y=x_1\partial^2\circ\frac1{u''}\partial+\frac{x_2+x_3(u')^2}{u''}\partial-x_3u'$
and $Z=\partial\circ\frac1{u''}\partial$.
Hence, the solutions of \eqref{20120908:eq3} are 
$F=\alpha+\beta u'\in\mc C\oplus\mc Cu'$,
and $P=YF=x_2\beta-x_3\alpha u'\in\mc Cu'$.
This proves part  (b).

Next, we find all elements $\tint f\in\mc F_0(X)$,
namely all solutions of the following equations in $F\in \mc V$ and $\tint f\in\mc V/\partial\mc V$:
\begin{equation}\label{20120908:eq4}
YF=0\,\,,\,\,\,\,\frac{\delta f}{\delta u}=ZF\,.
\end{equation}
If $x_1=0,x_2\neq0,x_3=0$, we have $Y=x_2$ is invertible,
and similarly, if $x_1=0,x_2=0,x_3\neq0$, we have $Y=x_3u'$ is invertible too.
In both these cases we thus have $F=0$, and hence 
$\tint f\in\ker\big(\frac\delta{\delta u}\big)$.
If $x_1=0,x_2\neq0,x_3\neq0$, then 
$Y=\frac{x_2+x_3(u')^2}{u''}\partial-x_3u'$ and $Z=\partial\circ\frac1{u''}\partial$.
The equation $YF=0$ has a one-dimensional (over $\mc C$) space of solution, spanned
by $F=\sqrt{x_2+x_3(u')^2}$.
Hence, all elements $\tint f\in\mc F_0(X)$ are obtained solving the equation
$$
\frac{\delta f}{\delta u}=\alpha\partial\circ\frac1{u''}\partial\sqrt{x_2+x_3(u')^2}
=\alpha\Big(\frac{x_3u'}{\sqrt{x_2+x_3(u')^2}}\Big)'\,,
$$
for $\alpha\in\mc C$. 
Its solutions are of the form
$\tint f=-\alpha\sqrt{x_2+x_3(u')^2}+k$,
where $k\in\ker\big(\frac\delta{\delta u}\big)$.
Next, 
if $x_1\neq0,x_2=x_3=0$, then $Y=x_1\partial$ and $Z=1$, so the equations \eqref{20120908:eq4}
give $F\in\mc C$ and $\tint f\in\mc C\tint u+\ker\big(\frac{\delta}{\delta u}\big)$.
If $x_1\neq0,x_2\neq0,\,x_3=0$, then $Y=x_1\partial^2+x_2$ and $Z=\partial$.
In this case, the first equation in \eqref{20120908:eq4} reads
$$
x_1F''+x_2F=0\,.
$$
By Lemma \ref{20120909:lem1}(b), it must be $F\in\mc F$, 
and it is easy to see that 
the space of solutions is two-dimensional over $\mc C$, 
consisting of elements of the form
$$
F=\alpha_+e^{x_{12}x}+\alpha_-e^{-x_{12}x}\,,
$$
with $\alpha_\pm\in\mc C$.
Then, the second equation in \eqref{20120908:eq4} gives
$$
\frac{\delta f}{\delta u}
=\alpha_+x_{12}e^{x_{12}x}
-\alpha_-x_{12}e^{-x_{12}x}\,,
$$
so that 
$\tint f=
\alpha_+x_{12} \tint e^{x_{12}x}u
-\alpha_-x_{12} \tint e^{-x_{12}x}u+k$,
where $k\in\ker\big(\frac{\delta}{\delta u}\big)$.
Similarly, we consider the case $x_1\neq0,x_2=0,x_3\neq0$.
In this case $Y=x_1\partial\circ\frac1{u'}\partial+x_3u'$ and $Z=\frac1{u'}\partial$.
The first equation in \eqref{20120908:eq4} reads
$$
x_1\Big(\frac{F'}{u'}\Big)'+x_3u'F=0\,.
$$
By Lemma \ref{20120909:lem1}(c),
we must have $F\in\mc V_0$ such that $F'=\frac{\partial F}{\partial u}$.
It is easy to see that
the space of solutions is two-dimensional over $\mc C$, and
it consists of elements of the form
$$
F=\alpha_+e^{x_{13}u}+\alpha_-e^{-x_{13}u}\,,
$$
with $\alpha_\pm\in\mc C$.
Then, the second equation in \eqref{20120908:eq4} gives
$$
\frac{\delta f}{\delta u}
=\alpha_+x_{13}e^{x_{13}u}
-\alpha_-x_{13}e^{-x_{13}u}\,,
$$
and its solutions for $\tint f$ are of the form
$\tint f=
\alpha_+\tint e^{x_{13}u}
+\alpha_-\tint e^{-x_{13}u}+k$,
for $k\in\ker\big(\frac{\delta}{\delta u}\big)$.
Finally, we are left to consider the case when
$x_1\neq0,x_2\neq0,\,x_3\neq0$.
In this case
$Y=x_1\partial^2\circ\frac1{u''}\partial+\frac{x_2+x_3(u')^2}{u''}\partial-x_3u'$
and $Z=\partial\circ\frac1{u''}\partial$.
The first equation in \eqref{20120908:eq4} reads
$$
x_1\Big(\frac{F'}{u''}\Big)''+(x_2+x_3(u')^2)\frac{F'}{u''}-x_3u'F=0\,.
$$
By Lemma \ref{20120909:lem1}(a),
the only solution of this equation is $F=0$.
But then the second equation in \eqref{20120908:eq4} 
gives $\tint f\in\ker\big(\frac{\delta}{\delta u}\big)$.
\end{proof}

In the statement of Lemma \ref{20120908:lem1} and further on in this section,
we assume that $\mc V$ contains all the elements which appear in the statement,
namely
$e^{x_{12}x}$, $e^{x_{13}u}$,
and $\sqrt{x_2+x_3(u')^2}$.

Next, for each element $\tint f\in\mc F_0(X)$, we want to find an element $P\in\mc H(X)$
which is $X$-associated to it,
and for each element $P\in\mc H_0(X)$, we want to find an element $\tint f\in\mc H(X)$
which is $X$-associated to it.
Recall, by Lemma \ref{20120907:lem1},
that if $\tint f\ass{X}P$,
then all elements in $\mc H(X)$ which are $X$-associated to $\tint f$
are obtained adding to $P$ an arbitrary element of $\mc H_0(X)$,
and all elements in $\mc F(X)$ which are $X$-associated to $P$
are obtained adding to $\tint f$ an arbitrary element of $\mc F_0(X)$.
\begin{lemma}\label{20120908:lem2}
Let $X=x_1L_1+x_2L_2+x_3L_3$, and let $a_2,a_3,\gamma\in\mc C\backslash\{0\}$. 
We have:
\begin{enumerate}[(i)]
\item
\begin{enumerate}[]
\item
$\nexists P\in\mc H(X)$ such that 
$\tint e^{\gamma x}u\ass{X}P$, if $x_3\neq0$;
\item
$\tint e^{\gamma x}u\ass{X}\frac1\gamma (x_1\gamma^2+x_2)e^{\gamma x}$, if $x_3=0$.
\end{enumerate}
\item
\begin{enumerate}[]
\item
$\nexists P\in\mc H(X)$ such that 
$\tint e^{\gamma u}\ass{X}P$, if $x_2\neq0$;
\item
$\tint e^{\gamma u}\ass{X}(x_1\gamma^2+x_3)e^{\gamma u}u'$, if $x_2=0$.
\end{enumerate}
\item
$\tint u\ass{X}(x_2x+x_3uu')$.
\item
$\tint\sqrt{a_2+a_3(u')^2}\ass{X}
-\Big(x_1\partial^2\circ\frac1{u''}\partial+\frac{x_2+x_3(u')^2}{u''}\partial-x_3u'\Big)\sqrt{a_2+a_3(u')^2}$.
\item
\begin{enumerate}[]
\item 
$\tint 0\ass{X}1$, if $x_2\neq0$;
\item
$\nexists \tint f\in\mc F(X)$ such that 
$\tint f\ass{X}1$, if $x_2=0,x_1x_3\neq0$;
\item
$\tint\frac1{2x_3u'}\ass{X}1$, if $x_1=0,x_2=0,x_3\neq0$;
\item
$\tint\frac{xu}{x_1}\ass{X}1$, if $x_1\neq0,x_2=0,x_3=0$.
\end{enumerate}
\item
\begin{enumerate}[]
\item 
$\tint 0\ass{X}u'$, if $x_3\neq0$;
\item
$\nexists \tint f\in\mc F(X)$ such that 
$\tint f\ass{X}u'$, if $x_3=0,x_1x_2\neq0$;
\item
$\tint\frac{-(u')^2}{2x_2}\ass{X}u'$, if $x_1=0,x_2\neq0,x_3=0$;
\item
$\tint\frac{u^2}{2x_1}\ass{X}u'$, if $x_1\neq0,x_2=0,x_3=0$.
\end{enumerate}
\end{enumerate}
\end{lemma}
\begin{proof}
The condition $\tint e^{\gamma x}u\ass{X}P$
means that,
for some fractional decomposition $X=YZ^{-1}$,
there exists $F\in\mc V$ such that $P=YF$ and $ZF=e^{\gamma x}$.
Let us consider first the case $x_3\neq0$.
In this case, by Lemma \ref{lem:frac} a minimal fractional decomposition for $X$ is
\eqref{frac-liouv} if $x_2\neq0$, and \eqref{frac-liouv2} if $x_2=0$.
In the former case
$Z=\partial\circ\frac1{u''}\partial$, hence the equation 
$ZF=e^{\gamma x}$ reads
$\partial\frac{F'}{u''}=e^{\gamma x}$,
namely
$$
F'=\frac1\gamma e^{\gamma x}u''+\alpha u''\,,
$$
for some $\alpha\in\mc C$.
This equation has no solutions since, integrating by parts, we get
$\tint \Big(\frac1\gamma e^{\gamma x}u''+\alpha u''\Big)=\gamma\tint e^{\gamma x}u$,
and this is not zero by \eqref{20120907:eq1}.
Similarly, in the case $x_2=0$ we have 
$Z=\frac1{u'}\partial$, hence the equation 
$ZF=e^{\gamma x}$ reads
$$
F'=e^{\gamma x}u'\,,
$$
which has no solutions since, integrating by parts,
$\tint e^{\gamma x}u'=-\gamma\tint e^{\gamma x}u\neq0$.
To conclude the proof of part (i), we consider the case $x_3=0$.
By Lemma \ref{lem:frac} a fractional decomposition for $X$ is
$X=YZ^{-1}$ given by \eqref{frac-liouv1}.
Hence, a solution $F\in\mc V$ to the equation $ZF=e^{\gamma x}$ is $F=\frac1\gamma e^{\gamma x}$,
and in this case we have $P=YF=(x_1\partial^2)\frac1\gamma e^{\gamma x}
=\big(x_1\gamma+\frac{x_2}{\gamma}\big)e^{\gamma x}$.

Next, let us prove part (ii).
The condition $\tint e^{\gamma u}\ass{X}P$
is equivalent to the existence of $F\in\mc V$ such that $P=YF$ and $ZF=\gamma e^{\gamma u}$,
where $X=YZ^{-1}$.
Let us consider first the case $x_2\neq0$.
In this case, by Lemma \ref{lem:frac} a minimal fractional decomposition $X=YZ^{-1}$ for $X$ 
has $Z=\partial\circ\frac1{u''}\partial$ if $x_3\neq0$, and $Z=\partial$ if $x_3=0$.
In both cases the equation $ZF=\gamma e^{\gamma u}$
would imply $\gamma e^{\gamma u}\in\partial\mc V$, which is not the case by \eqref{20120907:eq1}.
In the case $x_2=0$,
a fractional decomposition for $X$ is
$X=YZ^{-1}$ given by \eqref{frac-liouv2}.
Hence, a solution $F\in\mc V$ to the equation $ZF=\gamma e^{\gamma u}$ 
is $F=e^{\gamma u}$,
and in this case we have 
$P=YF=\big(x_1\partial\circ\frac1{u'}\partial+x_3u'\big)e^{\gamma u}
=(x_1\gamma^2+x_3)e^{\gamma u}u'$.

For part (iii) it suffices to check, 
using the fractional decomposition \eqref{frac-liouv}, 
that $F=xu'-u\in\mc V$ is a solution of the equations
$$
\begin{array}{l}
\vphantom{\Big(}
ZF=\partial\frac{F'}{u''}=\frac{\delta}{\delta u}\tint u\,, \\
YF=\Big(x_1\partial^2\circ\frac1{u''}\partial+\frac{x_2+x_3(u')^2}{u''}\partial-x_3u'\Big)F=x_2x+x_3uu'\,.
\end{array}
$$

Similarly, for part (iv),
letting $F=-\sqrt{a_2+a_3(u')^2}\in\mc V$, we have
$$
\begin{array}{l}
\vphantom{\Big(}
ZF=\partial\frac{F'}{u''}=\frac{\delta}{\delta u}\tint \sqrt{a_2+a_3(u')^2}\,, \\
YF=-\Big(x_1\partial^2\circ\frac1{u''}\partial
+\frac{x_2+x_3(u')^2}{u''}\partial-x_3u'\Big)\sqrt{a_2+a_3(u')^2}\,.
\end{array}
$$

Next, let us prove part (v).
For $x_2\neq0$, 
consider the fractional decomposition $X=YZ^{-1}$ given by \eqref{frac-liouv}.
It is easy to check that, letting $F=\frac{u'}{x_2}$, we have
$$
\begin{array}{l}
ZF=\partial\circ\frac1{u''}\partial\frac{u'}{x_2}=0\,,\\
YF=\Big(x_1\partial^2\circ\frac1{u''}\partial
+\frac{x_2+x_3(u')^2}{u''}\partial-x_3u'\Big)\frac{u'}{x_2}=1\,.
\end{array}
$$
Hence, $\tint 0\ass{X}1$, as we wanted.
If $x_1\neq0,x_2=0,x_3\neq0$, a minimal fractional decomposition for $X$ is
$X=YZ^{-1}$ given by \eqref{frac-liouv2}.
Therefore the relation $\tint f\ass{X}1$ is equivalent
to the existence of $F\in\mc V$ such that
\begin{equation}\label{20120909:eq4}
ZF=\frac{F'}{u'}=\frac{\delta f}{\delta u}
\,\,,\,\,\,\,
YF=\Big(x_1\partial\circ\frac1{u'}\partial+x_3u'\Big)F=1\,.
\end{equation}
By Lemma \ref{20120909:lem1}(c),
the second equation in \eqref{20120909:eq4} implies that $F\in\mc V_0$ 
is such that $F'=\frac{\partial F}{\partial u}u'$.
In this case, the second equation in \eqref{20120909:eq4} reads
$$
x_1\partial\frac{\partial F}{\partial u}+x_3Fu'=1\,,
$$
which, by the commutation relation \eqref{eq:0.4},  is equivalent to 
$$
\Big(x_1\frac{\partial^2 F}{\partial u^2}+x_3F\Big)u'=1\,.
$$
But obviously, the above equation is never satisfied.
If $x_1=0,x_2=0,x_3\neq0$,
it is easy to check that $F=\frac1{x_3u'}$ solves
$$
YF=x_3u'\frac1{x_3u'}=1
\,\,,\,\,\,\,
ZF=\frac{1}{u'}\partial\frac1{x_3u'}=\frac{\delta}{\delta u}\tint\frac1{2x_3u'}\,,
$$
proving that $\tint\frac1{2x_3u'}\ass{X}1$.
Finally, for $x_1\neq0,x_2=0,x_3=0$,
a minimal fractional decomposition $X=YZ^{-1}$ is given by $Y=x_1\partial$ and $Z=1$.
In this case, it is immediate to check that 
$F=\frac{x}{x_1}$ solves
$$
YF=x_1\partial\frac{x}{x_1}=1
\,\,,\,\,\,\,
ZF=\frac{x}{x_1}=\frac{\delta}{\delta u}\tint\frac{xu}{x_1}\,,
$$
proving that $\tint\frac{xu}{x_1}\ass{X}1$.

We are left to prove part (vi).
For $x_3\neq0$, 
consider the fractional decomposition $X=YZ^{-1}$ given by \eqref{frac-liouv}.
It is easy to check that, letting $F=\frac{-1}{x_3}$, we have
$$
\begin{array}{l}
ZF=\partial\circ\frac1{u''}\partial\frac{-1}{x_3}=0\,,\\
YF=\Big(x_1\partial^2\circ\frac1{u''}\partial
+\frac{x_2+x_3(u')^2}{u''}\partial-x_3u'\Big)\frac{-1}{x_3}=u'\,.
\end{array}
$$
Hence, $\tint 0\ass{X}u'$, as we wanted.
If $x_1\neq0,x_2\neq0,x_3=0$, the minimal fractional decomposition for $X$ is
$X=YZ^{-1}$ given by \eqref{frac-liouv1}.
Therefore the relation $\tint f\ass{X}u'$ is equivalent
to the existence of $F\in\mc V$ such that
\begin{equation}\label{20120909:eq5}
ZF=F'=\frac{\delta f}{\delta u}
\,\,,\,\,\,\,
YF=(x_1\partial^2+x_2)F=u'\,.
\end{equation}
If $F\in\mc V_n$, for $n\geq0$, we get, applying $\frac{\partial}{\partial u^{(n+2)}}$
to both sides of the second equation in \eqref{20120909:eq5}, 
that $\frac{\partial F}{\partial u^{(n)}}=0$.
Hence, it must be $F\in\mc F$. But in this case, the second equation in \eqref{20120909:eq5}
has clearly no solutions.
If $x_1=0,x_2\neq0,x_3=0$,
it is easy to check that $F=\frac{u'}{x_2}$ solves
$$
YF=x_2\frac{u'}{x_2}=u'
\,\,,\,\,\,\,
ZF=\partial\frac{u'}{x_2}=\frac{\delta}{\delta u}\tint\frac{-(u')^2}{2x_2}\,,
$$
proving that $\tint\frac{-(u')^2}{2x_2}\ass{X}u'$.
Finally, if $x_1\neq0,x_2=0,x_3=0$,
a minimal fractional decomposition $X=YZ^{-1}$ is given by $Y=x_1\partial$ and $Z=1$.
In this case, it is immediate to check that 
$F=\frac{u}{x_1}$ solves
$$
YF=x_1\partial\frac{u}{x_1}=u'
\,\,,\,\,\,\,
ZF=\frac{u}{x_1}=\frac{\delta}{\delta u}\tint\frac{u^2}{2x_1}\,,
$$
proving that $\tint\frac{u^2}{2x_1}\ass{X}u'$.
\end{proof}

In order the check orthogonality conditions \eqref{20130104:eq2}
for Liouville type integrable systems we will use the following results.
\begin{lemma}\label{20120909:lem2}
\begin{enumerate}[(a)]
\item
$(\mc C1)^\perp=\im(\partial)$.
\item
$(\mc Cu')^\perp=\im(\frac1{u'}\partial)$.
\item
$\big(\Span_{\mc C}\{1,u'\}\big)^\perp=\im(\partial\circ \frac1{u''}\partial)$.
\item
For $b_2,b_3\in\mc C\backslash\{0\}$, we have
$$
\big(\mc C\frac{\delta}{\delta u}\tint\sqrt{b_2+b_3(u')^2}\big)^\perp
=\im\Big(\frac{b_2+b_3(u')^2}{u''}\partial-b_3u'\Big)\,.
$$
\end{enumerate}
\end{lemma}
\begin{proof}
Parts (a) and (b) are immediate.
Let us prove part (c).
It is immediate to check, integrating by parts,
that $\tint (\alpha+\beta u')\partial\frac{f'}{u''}=0$ for every $\alpha,\beta\in\mc C$.
Hence, $\im(\partial\circ \frac1{u''}\partial)\subset\big(\Span_{\mc C}\{1,u'\}\big)^\perp$.
On the other hand, if
$f\in\big(\Span_{\mc C}\{1,u'\}\big)^\perp$,
it must be 
$$
f=\partial g=\frac1{u'}\partial h
\,\,\text{ for some }\, g,h\in\mc V\,.
$$
But then
$\partial h=u'\partial g=\partial(u'g)-u''g$,
which implies
$g=\frac1{u''}\partial(u'g-h)$.
Hence, 
$$
f=\partial \frac1{u''}\partial(u'g-h)\in\im(\partial\circ \frac1{u''}\partial)\,.
$$

We are left to prove part (d).
We have
$$
-\frac1{b_3}\frac{\delta}{\delta u}\tint\sqrt{b_2+b_3(u')^2}
=\partial\frac{u'}{\sqrt{b_2+b_3(u')^2}}\,.
$$
The inclusion
$\im\Big(\frac{b_2+b_3(u')^2}{u''}\partial-b_3u'\Big)\subset
\Big(\mc C\partial\frac{u'}{\sqrt{b_2+b_3(u')^2}}\Big)^\perp$
follows by integration by parts, and the following straightforward identity
$$
\Big(\partial\circ\frac{b_2+b_3(u')^2}{u''}+b_3u'\Big)
\partial\frac{u'}{\sqrt{b_2+b_3(u')^2}}=0\,.
$$
We are left to prove the opposite inclusion.
If $f\in\Big(\mc C\partial\frac{u'}{\sqrt{b_2+b_3(u')^2}}\Big)^\perp$,
we have
$$
f=\frac{\partial g}{\partial\frac{u'}{\sqrt{b_2+b_3(u')^2}}}
%=\frac{\partial g}{
%\frac{u''}{\sqrt{b_2+b_3(u')^2}}-\frac{b_3(u')^2u''}{(b_2+b_3(u')^2)^{\frac32}}
%}
%=(b_2+b_3(u')^2)^{\frac32}\frac{\partial g}{
%(b_2+b_3(u')^2)u''-b_3(u')^2u''
%}
=\frac{(b_2+b_3(u')^2)^{\frac32}}{b_2u''}\partial g
\,,
$$
for some $g\in\mc V$.
Letting $g=\frac{h}{\sqrt{b_2+b_3(u')^2}}$, we then get:
$$
\begin{array}{l}
\displaystyle{
f=\frac{(b_2+b_3(u')^2)^{\frac32}}{b_2u''}
\partial\frac{h}{\sqrt{b_2+b_3(u')^2}}
}\\
\displaystyle{
=\frac{(b_2+b_3(u')^2)^{\frac32}}{b_2u''}
\Big(
-\frac{b_3u'u''}{(b_2+b_3(u')^2)^{\frac32}}h
+\frac{1}{\sqrt{b_2+b_3(u')^2}}h'
\Big)
}\\
\displaystyle{
=
\Big(\frac{b_2+b_3(u')^2}{u''}\partial-b_3u'\Big)\frac{h}{b_2}\,.
}
\end{array}
$$
\end{proof}

%%%
\subsection{Integrability of the Lenard-Magri scheme: $b_1=0$}
\label{secb:3.3}

In this and the next two Sections we consider the case when $b_1=0$,
for which we get integrable Lenard-Magri schemes of S-type,
in the terminology introduced in Section \ref{sec:7.5}, in the case $a_1\neq0$ 
(described in Section \ref{secb:3.4} below),
and of C-type in the case $a_1=0$ (described in Section \ref{secb:3.5} below).
In Sections \ref{secb:3.6} we will consider the remaining case,
when $b_1\neq0$, for which we again get some integrable Lenard-Magri schemes of C-type.

According to Theorem \ref{20130123:thm} (and Remark \ref{20130104:rem2}),
in order to apply successfully the Lenard-Magri scheme of integrability,
we need to find finite sequences $\{P_n\}_{n=0}^N$, $\{\tint h_n\}_{n=0}^N$,
satisfying the relations \eqref{20130105:eq1} for $\xi_n=\frac{\delta h_n}{\delta u}$,
or equivalently the relations \eqref{maxi},
and the orthogonality conditions \eqref{20130104:eq2}.
For $b_1=0$ we display below such sequences
separately in all possibilities for the coefficients $a_2,a_3,b_2,b_3$ being zero or non-zero,
and $a_1$ arbitrary.
(Note that, since we are assuming $b_1=0$, we don't need to consider the case $b_2=b_3=0$.)
\begin{enumerate}[(i)]
\item$b_2b_3\neq0,a_2a_3\neq0$:
$\tint0\ass{H}1\ass{K}\tint 0\ass{H}u'\ass{K}\tint\sqrt{b_2+b_3(u')^2}$.
\item$b_2b_3\neq0,a_2\neq0,a_3=0$:
$\tint0\ass{H}1\ass{K}\tint\sqrt{b_2+b_3(u')^2}$.
\item$b_2b_3\neq0,a_2=0,a_3\neq0$:
$\tint0\ass{H}u'\ass{K}\tint\sqrt{b_2+b_3(u')^2}$.
\item$b_2b_3\neq0,a_2=a_3=0$:
$\tint0\ass{H}0\ass{K}\tint\sqrt{b_2+b_3(u')^2}$.
\item$b_2\neq0,b_3=0,a_2a_3\neq0$:
$\tint0\ass{H}1\ass{K}\tint 0\ass{H}u'\ass{K}\tint \frac{-(u')^2}{2b_2}$.
\item$b_2\neq0,b_3=0,a_2\neq0,a_3=0$:
$\tint0\ass{H}1\ass{K}\tint 0$.
\item$b_2\neq0,b_3=0,a_2=0,a_3\neq0$:
$\tint0\ass{H}u'\ass{K}\tint \frac{-(u')^2}{2b_2}$.
\item$b_2\neq0,b_3=0,a_2=a_3=0$:
$\tint0\ass{H}0\ass{K}\tint 0$.
\item$b_2=0,b_3\neq0,a_2a_3\neq0$:
$\tint0\ass{H}u'\ass{K}\tint 0\ass{H}1\ass{K}\tint \frac1{2b_3u'}$.
\item$b_2=0,b_3\neq0,a_2\neq0,a_3=0$:
$\tint0\ass{H}1\ass{K}\tint \frac1{2b_3u'}$.
\item$b_2=0,b_3\neq0,a_2=0,a_3\neq0$:
$\tint0\ass{H}u'\ass{K}\tint 0$.
\item$b_2=0,b_3\neq0,a_2=a_3=0$:
$\tint0\ass{H}0\ass{K}\tint 0$.
\end{enumerate}
All the above $H$- and $K$-association relations hold due 
to Lemmas \ref{20120908:lem1} and \ref{20120908:lem2}.
Moreover, 
using Lemmas \ref{lem:frac} and \ref{20120909:lem2}
we check that both orthogonality conditions \eqref{20130104:eq2} hold.
Hence, by Theorem \ref{20130123:thm} and Remark \ref{20130104:rem2}
all the above sequences can be continued indefinitely,
possibly going to a normal extension $\tilde{\mc V}$ of $\mc V$,
to an infinite sequence
$$
\tint 0\ass{H}P_0\ass{K}\tint h_0\ass{H}P_1\ass{K}\tint h_1\ass{H}\dots\,.
$$
Note that, by Lemma \ref{20120907:lem1}, 
at each step the subsequent term is unique 
up to a linear combinations of the previous steps.

Next, we want to discuss integrability of the corresponding hierarchies of Hamiltonian equations
$\frac{du}{dt_n}=P_n,\,n\in\mb Z_+$.
Namely, according to Definition \ref{20130104:def2},
we need to see when the vector spaces $\Span_{\mc C}\{\tint h_n\}\subset\mc V/\partial\mc V$
and $\Span_{\mc C}\{P_n\}\subset\mc V$ are infinite dimensional.

First, we consider the cases (vi), (viii), (xi) and (xii), where we show that
integrability does not occur (regardless of $a_1$ being zero or non-zero)
since the Lenard-Magri scheme repeats itself.
In case (vi), by Lemmas \ref{20120908:lem1} and \ref{20120908:lem2}
we have
$\mc H_0(H)=\mc C$ and, for every $\alpha\in\mc C$, 
$\{\tint f\in\mc F(K)\,|\,\tint f\ass{K}\alpha\}=\mc F_0(K)=\ker\big(\frac{\delta}{\delta u}\big)$.
Hence, any infinite sequence extending the given finite one will have
$\tint h_n\in\ker\big(\frac{\delta}{\delta u}\big)$ and $P_n\in\mc C$, for every $n\in\mb Z_+$.
Similarly, in case (xi) we have
$\mc H_0(H)=\mc Cu'$ and, for every $\alpha u'\in\mc Cu'$, 
$\{\tint f\in\mc F(K)\,|\,\tint f\ass{K}\alpha u'\}=\mc F_0(K)=\ker\big(\frac{\delta}{\delta u}\big)$.
Hence, any infinite sequence extending the given finite one will have
$\tint h_n\in\ker\big(\frac{\delta}{\delta u}\big)$ and $P_n\in\mc C u'$, for every $n\in\mb Z_+$.
In cases (viii) and (xii) we have 
$\mc H_0(H)=0$ and $\mc F_0(K)=\ker\big(\frac{\delta}{\delta u}\big)$.
Hence, 
$\tint h_n\in\ker\big(\frac{\delta}{\delta u}\big)$ and $P_n=0$, for every $n\in\mb Z_+$.
In conclusion, in all these cases $\Span_{\mc C}\{P_n\}$ is finite dimensional,
and integrability does not occur.

For the remaining 8 cases, 
we will prove in Section \ref{secb:3.4} that 
when $a_1\neq0$ we get some integrable Lenard-Magri scheme of S-type,
and we will prove in Section \ref{secb:3.5} that when $a_1=0$ 
we get some integrable Lenard-Magri scheme of C-type
(in the terminology of Section \ref{sec:7.5}).

%%%
\subsection{Integrable Lenard-Magri schemes of S-type in the case $b_1=0$ and $a_1\neq0$}
\label{secb:3.4}

\subsubsection*{Cases $(i), (ii), (iii), (iv)$}

In all the sequences (i)-(iv), after one or two steps, we arrive at (after shifting indices)
$\tint h_{-1}=\tint\sqrt{b_2+b_3(u')^2}$.
The next term in the sequence, which we denote $P_0$,
is obtained by solving the following equations for $F$ and $P_0$ in $\mc V$:
$$
\begin{array}{l}
\displaystyle{
BF=\partial\circ\frac1{u''}\partial F=\frac\delta{\delta u}\tint\tint\sqrt{b_2+b_3(u')^2}
\,,} \\
\displaystyle{
P_0=AF
=\Big(a_1\partial^2\circ\frac1{u''}\partial+\frac{a_2+a_3(u')^2}{u''}\partial-a_3u'\Big)F\,.
}
\end{array}
$$
It is easy to check that a solution is given by $F=-\sqrt{b_2+b_3(u')^2}$,
and
\begin{equation}\label{20120910:eq4}
\begin{array}{l}
\displaystyle{
P_0
=-\Big(\frac{a_1b_3u'}{\sqrt{b_2+b_3(u')^2}}\Big)^{\prime\prime}
+\frac{(a_3b_2-a_2b_3)}{\sqrt{b_2+b_3(u')^2}}
}\\
\displaystyle{
=
-\frac{a_1b_2b_3u'''}{(b_2+b_3(u')^2)^{\frac32}}
+3\frac{a_1b_2b_3^2u'(u'')^2}{(b_2+b_3(u')^2)^{\frac52}}
+\frac{(a_3b_2-a_2b_3)}{\sqrt{b_2+b_3(u')^2}}
\,.
}
\end{array}
\end{equation}

The above computation works regardless whether $a_1$ is zero or not.
But to prove that the Lenard-Magri scheme is integrable of S-type
we need to assume $a_1\neq0$,
in which case $\dord(P_0)=3$ is greater than
$\max\{\dord(A)-|H|+|K|,\,\dord(B)+|K|,\,\dord(C),\,\dord(D)+|K|\}$,
which is less than or equal to 2 (for all the cases (i)-(iv)).
%By Lemma \ref{20130123:lem3}
%it follows that $\dord(P_n)=2n+3$ for every $n\geq0$.
%In particular, the elements $P_n,\,n\in\mb Z_+$, are linearly independent.
Therefore,
by Corollary \ref{20130123:cor},
each Hamiltonian PDE $\frac{du}{dt}=P_n,\,n\in\mb Z_+$,
is integrable, associated to an integrable Lenard-Magri scheme of S-type.
(Note that, since $\ker(B^*)\cap\ker(D^*)=\mc C\oplus\mc Cu'\neq0$,
we cannot conclude that $[P_m,P_n]=0$ for every $m,n\in\mb Z_+$,
and therefore that the sequence of equations $\frac{du}{dt_n}=P_n,\,n\in\mb Z_+$,
form a compatible hierarchy.)

After rescaling $x$ and $t$ appropriately 
in the equation $\frac{du}{dt}=P_0$,
we conclude that the following bi-Hamiltonian equation
is integrable, associated to an integrable Lenard-Magri scheme of S-type:
\begin{equation}\label{beauty1}
\frac{du}{dt}
=
\frac{u'''}{(1+(u')^2)^{\frac32}}
-3\frac{u'(u'')^2}{(1+(u')^2)^{\frac52}}
+\frac{\alpha}{(1+(u')^2)^{\frac12}}
\,\,,\,\,\,\,\alpha\in\mc C\,.
\end{equation}
This is an equation of the form \cite[eq.(41.5)]{MSS90} with $a_3=(1+(u')^2)^{\frac12}$.
This particular $a_3$ does not appear in their list,
but as V. Sokolov informed us,
a simple point transformation reduces \eqref{beauty1} to an equation from their list.

\subsubsection*{Cases $(v), (vii)$}

In the sequences (v) and (vii), after one or two steps, we arrive at (after shifting indices)
$\tint h_{-1}=\tint\frac{-(u')^2}{2b_2}$.
The next term in the sequence, which we denote $P_0$,
is obtained by solving the following equations for $F$ and $P_0$ in $\mc V$:
$$
\begin{array}{l}
\displaystyle{
BF=\partial\circ\frac1{u''}\partial F=\frac\delta{\delta u}\tint\tint\frac{-(u')^2}{2b_2}
\,,} \\
\displaystyle{
P_0=AF
=\Big(a_1\partial^2\circ\frac1{u''}\partial+\frac{a_2+a_3(u')^2}{u''}\partial-a_3u'\Big)F\,.
}
\end{array}
$$
It is easy to check that a solution is given by $F=\frac{(u')^2}{2b_2}$,
and
\begin{equation}\label{20120910:eq4b}
P_0=\frac{a_1}{b_2}u'''+\frac{a_2}{b_2}u'+\frac{a_3}{2b_2}(u')^3\,.
\end{equation}

As before, 
if $a_1\neq0$, then $\dord(P_0)=3$ is greater than
$\max\{\dord(A)-|H|+|K|,\,\dord(B)+|K|,\,\dord(C),\,\dord(D)+|K|\}$,
which is at most 2.
%
%By Lemma \ref{20130123:lem3}
%it follows that $\dord(P_n)=2n+3$ for every $n\geq0$.
%In particular, the elements $P_n,\,n\in\mb Z_+$, are linearly independent.
Therefore,
by Corollary \ref{20130123:cor}
each Hamiltonian PDE $\frac{du}{dt}=P_n,\,n\in\mb Z_+$,
is integrable, associated to an integrable Lenard-Magri scheme of S-type.

After rescaling $x$ and $t$ appropriately 
in the equation $\frac{du}{dt}=P_0$,
we conclude that the following bi-Hamiltonian equation
is integrable, associated to an integrable Lenard-Magri scheme of S-type:
\begin{equation}\label{beauty2}
\frac{du}{dt}
=
u'''+\epsilon u'+\alpha(u')^3\,,
\end{equation}
where $\epsilon$ is 1 (in case (v)) or 0 (in case (vii)) and $\alpha\in\mc C$.
By a Galilean transformation we can make $\epsilon=0$.
The resulting equation is called the potential modified KdV equation
(equation (4.11) in the list of \cite{MSS90}).

\subsubsection*{Cases $(ix), (x)$}

In all the sequences (ix) and (x), after one or two steps, we arrive at (after shifting indices)
$\tint h_{-1}=\tint\frac{1}{2b_3u'}$.
The next term in the sequence, which we denote $P_0$,
is obtained by solving the following equations for $F$ and $P_0$ in $\mc V$:
$$
\begin{array}{l}
\displaystyle{
BF=\partial\circ\frac1{u''}\partial F=\frac\delta{\delta u}\tint\frac{1}{2b_3u'}
\,,} \\
\displaystyle{
P_0=AF
=\Big(a_1\partial^2\circ\frac1{u''}\partial+\frac{a_2+a_3(u')^2}{u''}\partial-a_3u'\Big)F\,.
}
\end{array}
$$
It is easy to check that a solution is given by $F=\frac{-1}{2b_3u'}$,
and
\begin{equation}\label{20120910:eq4c}
P_0=-\frac{a_1}{b_3}\frac{u'''}{(u')^3}+\frac{3a_1}{b_3}\frac{(u'')^2}{(u')^4}
+\frac{a_2}{2b_3}\frac1{(u')^2}+\frac{a_3}{b_3}
\,.
\end{equation}

If $a_1\neq0$, we have $\dord(P_0)=3$, which is greater than
$\max\{\dord(A)-|H|+|K|,\,\dord(B)+|K|,\,\dord(C),\,\dord(D)+|K|\}$,
which is at most 2.
%
%By Lemma \ref{20130123:lem3}
%it follows that $\dord(P_n)=2n+3$ for every $n\geq0$.
%In particular, the elements $P_n,\,n\in\mb Z_+$, are linearly independent.
Therefore,
by Corollary \ref{20130123:cor},
each Hamiltonian PDE $\frac{du}{dt}=P_n,\,n\in\mb Z_+$,
is integrable, associated to an integrable Lenard-Magri scheme of S-type.

After rescaling $x$ and $t$ appropriately 
in the equation $\frac{du}{dt}=P_0$,
we conclude that the following bi-Hamiltonian equation 
is integrable, associated to an integrable Lenard-Magri scheme of S-type:
\begin{equation}\label{beauty3}
\frac{du}{dt}
=\frac{u'''}{(u')^3}-3\frac{(u'')^2}{(u')^4}
+\frac1{(u')^2}+\alpha
\,\,,\,\,\,\,
\alpha\in\mc C\,.
\end{equation}
As explained in \cite{MSS90}, by a point transformation one can reduce this equation
to an equation of the form (4.1.4) in their list.

\begin{remark}\label{20120911:rem1}
Note that equation $\frac{du}{dt}=P_0$ with $P_0$ given by \eqref{20120910:eq4c}
is transformed, by the hodograph transformation $x\mapsto u,\,u\mapsto-x$, 
to the equation with $P_0$ given by \eqref{20120910:eq4b}, after exchanging $a_2$ and $b_2$ 
with $a_3$ and $b_3$ respectively.
Equivalently, equation \eqref{beauty3} can be transformed to equation \eqref{beauty2}
up to rescaling of $x$ and $t$.
\end{remark}

%%%
\subsection{Integrable Lenard-Magri schemes of C-type with $a_1=b_1=0$}
\label{secb:3.5}

\subsubsection*{Cases $(i), (ii), (iii), (iv)$}

As pointed out above,
in all the sequences (i)-(iv) we arrive,
after one or two steps (and after shifting indices), at 
$\tint h_{-1}=\tint\sqrt{b_2+b_3(u')^2}$.
In the case $a_1=0$ we can actually find an explicit solution for the sequences
$\{\tint h_n\}_{n\in\mb Z_+}$ and $\{P_n\}_{n\in\mb Z_+}$
satisfying the recursive formulas 
\begin{equation}\label{20120912:eq2}
P_n\ass{K}\tint h_n\ass{H}P_{n+1}\,\,,\,\,\,\,\,n\in\mb Z_+\,.
\end{equation}
It is given by ($n\geq0$):
\begin{equation}\label{20120912:eq1}
\begin{array}{l}
\displaystyle{
P_n=
\sum_{k=0}^n\binom{n}{k}\frac{(2n-1-2k)!!}{(2n-2k)!!} \frac{\Delta^{n+1-k}a_3^k}{b_3^{n}}
\frac{-u'}{(b_2+b_3(u')^2)^{\frac12+n-k}}\,,
} \\
\displaystyle{
h_n=
\sum_{k=0}^n\binom{n}{k}\frac{(2n-1-2k)!!}{(2n-2k+2)!!} \frac{\Delta^{n+1-k}a_3^k}{b_3^{n+1}}
\frac{-u'}{(b_2+b_3(u')^2)^{\frac12+n-k}}\,,
}
\end{array}
\end{equation}
where $\Delta=a_2b_3-a_3b_2$, which is non-zero unless the operators $H$ and $K$ are proportional.
Here and further we let $(-1)!!=1$.

First, note that for $n=0$, the above expression for $P_0$ is the same 
as the one in \eqref{20120910:eq4} with $a_1=0$. 
Hence $\tint h_{-1}\ass{H}P_0$.
Next, we check that indeed the sequences $\{\tint h_n\}_{n\in\mb Z_+}$ and $\{P_n\}_{n\in\mb Z_+}$
solve the recursive relations \eqref{20120912:eq2}.
For this, we fix the fractional decompositions $H=AB^{-1}$ and $K=CD^{-1}$ given by 
$$
A=\frac{a_2+a_3(u')^2}{u''}\partial-a_3u'
\,\,,\,\,\,\,
C=\frac{b_2+b_3(u')^2}{u''}\partial-b_3u'
\,\,,
B=D=\partial\circ\frac1{u''}\partial\,.
$$
Since $B=D$, the relations \eqref{20120912:eq2} hold
if there exists an element $F_n$ such that
\begin{equation}\label{20120912:eq3}
CF_n=P_n
\,\,,\,\,\,\,
BF_n=\frac{\delta h_n}{\delta u}
\,\,,\,\,\,\,
P_{n+1}=AF_n\,.
\end{equation}
A solution $F_n$ to equations \eqref{20120912:eq3} is given by $F_n=-h_n$.
Since $h_n$ depends only on $u'$,
we have 
$$
B(-h_n)=-\partial\circ\frac1{u''}\partial h_n
=(-\partial)\frac{\partial h_n}{\partial u'}=\frac{\delta h_n}{\delta u}\,,
$$
hence $F_n=-h_n$ satisfies the second equation in \eqref{20120912:eq3}.
The first and third equations in \eqref{20120912:eq3} can be easily checked
using the following straightforward identities ($m\in\mb Z_+$):
$$
\begin{array}{l}
\displaystyle{
C\frac{1}{(b_2+b_3(u')^2)^{\frac m2}}=\frac{-b_3 m u'}{(b_2+b_3(u')^2)^{\frac m2}}
\,, } \\
\displaystyle{
A\frac{1}{(b_2+b_3(u')^2)^{\frac m2}}
=m\Delta\frac{-u'}{(b_2+b_3(u')^2)^{\frac m2 +1}}
+(m+1)a_3\frac{-u'}{(b_2+b_3(u')^2)^{\frac m2}}
\,. }
\end{array}
$$

If $\Delta\neq0$, all the elements $P_n$ are linearly independent,
therefore, by Theorem \ref{20130123:thm} and Lemma \ref{20120906:lem1},
each Hamiltonian PDE $\frac{du}{dt}=P_n,\,n\in\mb Z_+$,
is integrable, associated to an integrable Lenard-Magri scheme of C-type.

\subsubsection*{Cases $(v), (vii)$}

In the sequences (v) and (vii) we arrive,
after one or two steps (and after shifting indices), at 
$\tint h_{-1}=\tint\frac{-(u')^2}{2b_2}$.
In the case $a_1=0$ we can actually find an explicit solution for the sequences
$\{\tint h_n\}_{n\in\mb Z_+}$ and $\{P_n\}_{n\in\mb Z_+}$
satisfying the recursive formulas \eqref{20120912:eq2}.
It is given by ($n\in\mb Z_+$):
\begin{equation}\label{20120912:eq1b}
\begin{array}{l}
\displaystyle{
P_{n-1}=
\sum_{k=0}^{n}\binom{n}{k}\frac{(2k-1)!!}{(2k)!!} \frac{a_2^{n-k}a_3^k}{b_2^{n}}
(u')^{2k+1}\,,
} \\
\displaystyle{
h_{n-1}=
-\sum_{k=0}^{n}\binom{n}{k}\frac{(2k-1)!!}{(2k+2)!!} \frac{a_2^{n-k}a_3^k}{b_2^{n}}
(u')^{2k+2}\,.
}
\end{array}
\end{equation}

First, note that the above expression for $P_0$ is the same 
as the one in \eqref{20120910:eq4b} with $a_1=0$. 
Hence $\tint h_{-1}\ass{H}P_0$.
Next, we check that indeed the sequences $\{\tint h_n\}_{n\in\mb Z_+}$ and $\{P_n\}_{n\in\mb Z_+}$
solve the recursive relations \eqref{20120912:eq2}.
For this, we fix the fractional decompositions $H=AB^{-1}$ and $K=CD^{-1}$ given by 
$$
A=\frac{a_2+a_3(u')^2}{u''}\partial-a_3u'
\,\,,\,\,\,\,
B=\partial\circ\frac1{u''}\partial
\,\,,\,\,\,\,
C=b_2
\,\,,\,\,\,\,
D=\partial\,.
$$

The relations \eqref{20120912:eq2} hold
if there exist elements $F_n,G_n\in\mc V$ such that
\begin{equation}\label{20120912:eq3b}
CF_n=P_n
\,\,,\,\,\,\,
DF_n=\frac{\delta h_n}{\delta u}
\,\,,\,\,\,\,
BG_n=\frac{\delta h_n}{\delta u}
\,\,,\,\,\,\,
P_{n+1}=AG_n\,.
\end{equation}
Solutions $F_n,G_n$ to equations \eqref{20120912:eq3b} are given by 
$F_n=\frac1{b_2}P_n$ and $G_n=-h_n$.
The first and third equations in \eqref{20120912:eq3b} are immediate.
The third equation follows from the immediate identity 
$\frac{P_n}{b_2}=-\frac{\partial h_n}{\partial u'}$.
Finally, the fourth identity in \eqref{20120912:eq3b} is easily checked
using the Tartaglia-Pascal triangle.

Clearly, if $a_3\neq0$, all the elements $P_n$ are linearly independent,
therefore, by Theorem \ref{20130123:thm} and Lemma \ref{20120906:lem1},
each Hamiltonian PDE $\frac{du}{dt}=P_n,\,n\in\mb Z_+$,
is integrable, associated to an integrable Lenard-Magri scheme of C-type.

\subsubsection*{Cases $(ix), (x)$}

In the sequences (ix) and (x) we arrive,
after one or two steps (and shift of indices), at 
$\tint h_{-1}=\tint\frac{1}{2b_3u'}$.
If $a_1=0$ we can find an explicit solution for the sequences
$\{\tint h_n\}_{n\in\mb Z_+}$ and $\{P_n\}_{n\in\mb Z_+}$
satisfying the recursive formulas \eqref{20120912:eq2}.
It is given by ($n\in\mb Z_+$):
\begin{equation}\label{20120912:eq1c}
\begin{array}{l}
\displaystyle{
P_{n-1}=
\sum_{k=0}^{n}\binom{n}{k}\frac{(2k-1)!!}{(2k)!!} \frac{a_3^{n-k}a_2^k}{b_3^{n}}\frac1{(u')^{2k}}
\,, } \\
\displaystyle{
h_{n-1}=
\sum_{k=0}^{n}\binom{n}{k}\frac{(2k-1)!!}{(2k+2)!!} \frac{a_3^{n-k}a_2^k}{b_3^{n+1}}
\frac1{(u')^{2k+1}}
\,.}
\end{array}
\end{equation}

First, note that the above expression for $P_0$ is the same 
as the one in \eqref{20120910:eq4c} with $a_1=0$. 
Hence $\tint h_{-1}\ass{H}P_0$.
Next, we check that indeed the sequences $\{\tint h_n\}_{n\in\mb Z_+}$ and $\{P_n\}_{n\in\mb Z_+}$
solve the recursive relations \eqref{20120912:eq2}.
For this, we fix the fractional decompositions $H=AB^{-1}$ and $K=CD^{-1}$ given by 
$$
A=\frac{a_2+a_3(u')^2}{u''}\partial-a_3u'
\,\,,\,\,\,\,
B=\partial\circ\frac1{u''}\partial
\,\,,\,\,\,\,
C=b_3 u'
\,\,,\,\,\,\,
D=\frac1{u'}\partial\,.
$$
Since equations \eqref{20120912:eq3b} hold with
$F_n=\frac1{b_3u'}P_n$ and $G_n=-h_n$
(a fact that can be easily checked directly),
it follows that the recursive relations \eqref{20120912:eq2} hold.

Clearly, if $a_2\neq0$, all the elements $P_n$ are linearly independent,
therefore, by Theorem \ref{20130123:thm} and Lemma \ref{20120906:lem1},
each Hamiltonian PDE $\frac{du}{dt}=P_n,\,n\in\mb Z_+$,
is integrable, associated to an integrable Lenard-Magri scheme of C-type.

%%%
\subsection{Integrable Lenard-Magri scheme of C-type with $b_1\neq0$}
\label{secb:3.6}

As we did in the previous sections, we study here the integrability of the Lenard-Magri scheme
when $b_1\neq0$.
We will consider separately the various cases, depending on the parameters $b_2,b_3,a_2,a_3$
being zero or non-zero.

\subsubsection*{Case 1: $b_2b_3\neq0$}

Let us consider first the case when $b_2$ and $b_3$ are both non-zero.
If $\tint h_0\in\mc V/\partial\mc V$ and $P_0\in\mc V$ satisfy
the relations $\tint 0\ass{H}P_0\ass{K}\tint h_0$,
then, by Lemmas \ref{20120908:lem1} and \ref{20120908:lem2},
we necessarily have $P_0\in\mc C\oplus\mc Cu'$
and $\frac{\delta h_0}{\delta u}=0$.
Hence, any infinite sequence extending the given finite one will have
$\tint h_n\in\ker\big(\frac{\delta}{\delta u}\big)$ and $P_n\in\mc C\oplus\mc Cu'$, 
for every $n\in\mb Z_+$.
In other words, 
the Lenard-Magri scheme repeats itself
and integrability does not occur.

\subsubsection*{Case 2: $b_2\neq0,b_3=0,a_3=0$}

In the case when $b_1b_2\neq0$, $b_3=0$ and $a_3=0$,
we can find explicitly all possible solutions for the sequences
$\{\tint h_n\}_{n\in\mb Z_+}$ and $\{P_n\}_{n\in\mb Z_+}$
satisfying the Lenard-Magri recursive relations \eqref{maxi}.

In order to describe such solutions, we need to introduce some polynomials.
We let $p_n(x;A,\epsilon),\,q_n(x;A,\epsilon)$, $n\in\mb Z_+$,
be the sequences of polynomials, depending on the 
$2\times2$ matrix 
$A=\Big(\begin{array}{ll} a_1 & a_2 \\ b_1 & b_2 \end{array}\Big)$,
and on the sequence of constant parameters $\epsilon=(\epsilon_0,\epsilon_1,\dots)$,
defined by the following recursive relations:
$p_0(x;A,\epsilon)=0$, and
\begin{equation}\label{20121003:eq1}
\begin{array}{l}
\displaystyle{
p_{n+1}(x;A,\epsilon)
=\frac{a_1}{b_1}p_n(x;A,\epsilon)
+\frac{a_2b_1-a_1b_2}{b_1^2}q_n(x;A,\epsilon)
} \\
\displaystyle{
\Big(\frac{d^2}{dx^2}+2b_{12}\frac{d}{dx}\Big)q_n(x;A,\epsilon)=p_n(x;A,\epsilon)
\,.}
\end{array}
\end{equation}
Here and further, as before, we use the notation \eqref{notation} with $x_i$ and $x_j$ 
replaced by $b_i$ and $b_j$.
It is easy to see that, if $p(x)$ is a polynomial of degree $n$,
then a solution $q(x)$ of the differential equation
$q''(x)+2b_{12}q'(x)=p(x)$
is a polynomial of degree $n+1$,
defined uniquely up to an additive constant $\epsilon_0$.
Hence, at each step in the recursion \eqref{20121003:eq1},
the resulting polynomial $p_{n+1}(x)$ depends on the previous step $p_n(x)$
and on the choice of a constant parameter $\epsilon_{n+1}$.

With the above notation, all sequences 
$\{\tint h_n\}_{n\in\mb Z_+}$, $\{P_n\}_{n\in\mb Z_+}$,
satisfying the Lenard-Magri recursive relations \eqref{maxi},
are as follows:
\begin{equation}\label{20121003:eq2}
\begin{array}{l}
\displaystyle{
\vphantom{\Bigg(}
P_n=
p_n(x;A,\epsilon^+)e^{b_{12}x}
+p_n(-x;A,\epsilon^-)e^{-b_{12}x}+a_2\delta_n\,,
} \\
\displaystyle{
\vphantom{\Bigg(}
h_n=
\frac1{b_1} 
\Big(q_n^\prime(x;A,\epsilon^+)+b_{12}q_n(x;A,\epsilon^+)\Big)
e^{b_{12}x} u
} \\
\displaystyle{
\vphantom{\Bigg(}
\,\,\,\,\,\,\,\,\,
-\frac1{b_1} 
\Big(q_n^\prime(-x;A,\epsilon^-)+b_{12}q_n(-x;A,\epsilon^-)\Big)
e^{-b_{12}x} u
\,,
}
\end{array}
\end{equation}
where $\epsilon^\pm=(\epsilon^\pm_0,\epsilon^\pm_1,\dots)$
and $\delta=(\delta_0,\delta_1,\dots)$
are arbitrary sequences of constant parameters.

It is not hard to check that 
the sequences $\{\tint h_n\}_{n\in\mb Z_+}$ and $\{P_n\}_{n\in\mb Z_+}$
indeed solve the recursive relations \eqref{maxi},
and any solution of the recursive relations \eqref{maxi} is obtained in this way.
To conclude, we observe that,
since $\Delta=a_2b_1-a_1b_2$ is non-zero
(unless the operators $H$ and $K$ are proportional),
all the elements $P_n$ are linearly independent,
therefore, by Theorem \ref{20130123:thm} and Lemma \ref{20120906:lem1},
each Hamiltonian PDE $\frac{du}{dt}=P_n,\,n\in\mb Z_+$,
is integrable, associated to an integrable Lenard-Magri scheme of C-type.

\subsubsection*{Case 3: $b_2\neq0,b_3=0,a_3\neq0$}

In the case when $b_1b_2\neq0$, $b_3=0$ and $a_3\neq0$,
we have by Lemma \ref{20120908:lem1}(b) that 
$\mc H_0(H)=a_2\mc C\oplus\mc Cu'$.
On the other hand, by Lemma \ref{20120908:lem2}(vi) there is no
element $\tint f\in\mc F(K)$ such that $\tint f\ass{K}u'$.
Similarly, 
by Lemma \ref{20120908:lem1}(a) we have 
$\mc F_0(K)=\mc C\tint e^{b_{12}x}u
+\mc C\tint e^{-b_{12}x}u
+\ker\big(\frac\delta{\delta u}\big)$.
On the other hand, by Lemma \ref{20120908:lem2}(i) there is no
element $P\in\mc F(H)$ such that $\tint e^{\pm b_{12}x}u\ass{H}P$.

In conclusion,
the Lenard-Magri recursion scheme, in this case,
cannot be applied,
since the following finite sequences cannot be extended
to infinite sequences satisfying \eqref{maxi}:
$$
\tint 0\ass{H} u'\ass{K}\not\exists\tint f
\,\,,\,\,\,\,
\tint 0\ass{H} 0\ass{K}\tint e^{\pm b_{12}x}u
\ass{H}\not\exists P
\,.
$$

\subsubsection*{Case 4: $b_2=0,b_3\neq0,a_2=0$}

In the case when $b_1b_3\neq0$, $b_2=0$ and $a_2=0$,
we can find explicitly all possible solutions for the sequences
$\{\tint h_n\}_{n\in\mb Z_+}$ and $\{P_n\}_{n\in\mb Z_+}$
satisfying the Lenard-Magri recursive relations \eqref{maxi}.
They are as follows:
\begin{equation}\label{20121003:eq3}
\begin{array}{l}
\displaystyle{
\vphantom{\Bigg(}
P_n=
p_n(u;A,\epsilon^+)e^{b_{13}u}
+p_n(-u;A,\epsilon^-)e^{-b_{13}u}+a_3\delta_nu'\,,
} \\
\displaystyle{
\vphantom{\Bigg(}
\frac{\delta h_n}{\delta u}=
\frac1{b_1} 
\Big(q_n^\prime(u;A,\epsilon^+)+b_{13}q_n(u;A,\epsilon^+)\Big)
e^{b_{13}u}
} \\
\displaystyle{
\vphantom{\Bigg(}
\,\,\,\,\,\,\,\,\,
-\frac1{b_1} 
\Big(q_n^\prime(-u;A,\epsilon^-)+b_{13}q_n(-u;A,\epsilon^-)\Big)
e^{-b_{12}u}
\,,
}
\end{array}
\end{equation}
where 
$p_n(u;A,\epsilon^+)$ and $q_n(u;A,\epsilon^+)$
are the polynomials defined in \eqref{20121003:eq1},
depending on the 
matrix $A=\Big(\begin{array}{ll} a_1 & a_3 \\ b_1 & b_3 \end{array}\Big)$,
and on the sequences of constant parameters
$\epsilon^\pm=(\epsilon^\pm_0,\epsilon^\pm_1,\dots)$
and $\delta=(\delta_0,\delta_1,\dots)$.

It is not hard to check, as in case 1, that 
the sequences $\{\tint h_n\}_{n\in\mb Z_+}$ and $\{P_n\}_{n\in\mb Z_+}$
solve the recursive relations \eqref{maxi},
and any solution of the recursive relations \eqref{maxi} is obtained in this way.
We also observe that,
since $\Delta=a_3b_1-a_1b_3$ is non-zero
(unless the operators $H$ and $K$ are proportional),
all the elements $P_n$ are linearly independent,
therefore, by Theorem \ref{20130123:thm} and Lemma \ref{20120906:lem1},
each Hamiltonian PDE $\frac{du}{dt}=P_n,\,n\in\mb Z_+$,
is integrable, associated to an integrable Lenard-Magri scheme of C-type.

\subsubsection*{Case 5: $b_2=0,b_3\neq0,a_2\neq0$}

In the case when $b_1b_3\neq0$, $b_2=0$ and $a_3=0$,
we have by Lemma \ref{20120908:lem1}(b) that 
$\mc H_0(H)=\mc C\oplus a_3\mc Cu'$.
On the other hand, by Lemma \ref{20120908:lem2}(v) there is no
element $\tint f\in\mc F(K)$ such that $\tint f\ass{K}1$.
Similarly, 
by Lemma \ref{20120908:lem1}(a) we have 
$\mc F_0(K)=\mc C\tint e^{b_{13}u}
+\mc C\tint e^{-b_{13}u}
+\ker\big(\frac\delta{\delta u}\big)$.
On the other hand, by Lemma \ref{20120908:lem2}(ii) there is no
element $P\in\mc F(H)$ such that $\tint e^{\pm b_{13}x}u\ass{H}P$.

In conclusion,
the Lenard-Magri recursion scheme, in this case,
cannot be applied,
since the following finite sequences cannot be extended
to infinite sequences satisfying the relations \eqref{maxi}:
$$
\tint 0\ass{H} 1\ass{K}\not\exists\tint f
\,\,,\,\,\,\,
\tint 0\ass{H} 0\ass{K}\tint e^{\pm b_{13}u}
\ass{H}\not\exists P
\,.
$$

\subsubsection*{Case 6: $b_2=b_3=0$}

In the case when $b_1\neq0$, which we set equal to $1$, and $b_2=b_3=0$,
we have different possibilities according to the constants $a_2$ and $a_3$ being zero or not.

If $a_2a_3\neq0$,
the Lenard-Magri recursion scheme cannot be applied.
Indeed, by Lemma \ref{20120908:lem1} we have 
$\mc H_0(H)=\mc C\oplus\mc Cu'$
and
$\mc F_0(K)=\mc C\tint u\oplus\ker\big(\frac{\delta}{\delta u}\big)$,
and, whichever way we start the finite sequences
$\{\tint h_n\}_{n=0}^N$, $\{P_n\}_{n=0}^N$ as in \eqref{maxi},
there is no way to extend them to non-trivial infinite sequences:
$$
\begin{array}{l}
\displaystyle{
\vphantom{\bigg(}
\tint 0\ass{H} 1\ass{K}\tint xu\ass{H}\not\exists P_1
\,,} \\
\displaystyle{
\vphantom{\bigg(}
\tint 0\ass{H} u'\ass{K}\tint\frac12 u^2\ass{H}\not\exists P_1
\,,} \\
\displaystyle{
\vphantom{\bigg(}
\tint 0\ass{H} 0\ass{K}\tint u\ass{H} a_2x+a_3uu' \ass{K} \not\exists \tint h_1
\,.} 
\end{array}
$$

Next, we consider the cases when exactly one of the elements $a_2$ and $a_3$ is zero.
Recall the sequence of polynomials $p_n(x;A,\epsilon)$ defined by the
recursive equations \eqref{20121003:eq1}.
In the case $b_1=1,b_2=0$, such equations reduce to
$p_0(x;a_1,a_2,\epsilon)=0$ and
\begin{equation}\label{20121003:eq1b}
p_{n+1}(x;a_1,a_2,\epsilon)
=\Big(\frac{a_1}{b_1}+\frac{a_2}{b_1}\Big(\frac{d}{dx}\Big)^{-2}\Big)p_n(x;a_1,a_2,\epsilon)
\,.
\end{equation}
Here $\Big(\frac{d}{dx}\Big)^{-2}$ means integrating twice with respect to $x$,
which is defined uniquely up to adding a linear term $\epsilon_{2n}+\epsilon_{2n+1}x$.
In particular, at each step the degree increases by two.

In the case $a_2\neq0,a_3=0$, 
it is not hard to prove that all the sequences 
$\{\tint h_n\}_{n\in\mb Z_+}$, $\{P_n\}_{n\in\mb Z_+}$,
satisfying the Lenard-Magri recursive relations \eqref{maxi},
are as follows:
\begin{equation}\label{20121003:eq2b}
P_n=p_n(x;a_1,a_2,\epsilon)+\delta_n
\,\,,\,\,\,\,
\tint h_n=\tint \Big(\frac{d}{dx}\Big)^{-1}p_n(x;a_1,a_2,\epsilon)u\,,
\end{equation}
where $\epsilon=(\epsilon_0,\epsilon_1,\dots)$ and $(\delta_0,\delta_1,\dots)$
are arbitrary sequence of constant parameters.
Since, obviously,
all the elements $P_n$ are linearly independent,
we conclude that each Hamiltonian PDE $\frac{du}{dt}=P_n,\,n\in\mb Z_+$,
is integrable, associated to an integrable Lenard-Magri scheme of C-type.

Similarly, in the case $a_2=0,a_3\neq0$, all the sequences 
$\{\tint h_n\}_{n\in\mb Z_+}$, $\{P_n\}_{n\in\mb Z_+}$,
satisfying the Lenard-Magri recursive relations \eqref{maxi},
are as follows:
\begin{equation}\label{20121003:eq2c}
P_n=p_n(u;a_1,a_2,\epsilon)u'+\delta_nu'
\,\,,\,\,\,\,
\tint h_n=\tint \Big(\frac{d}{du}\Big)^{-2}p_n(u;a_1,a_2,\epsilon)\,,
\end{equation}
where $\epsilon=(\epsilon_0,\epsilon_1,\dots)$ and $(\delta_0,\delta_1,\dots)$
are arbitrary sequences of constant parameters.
Again, we conclude that each Hamiltonian PDE $\frac{du}{dt}=P_n,\,n\in\mb Z_+$,
is integrable, associated to an integrable Lenard-Magri scheme of C-type.

%%%
\subsection{Summary}
\label{secb:3.7}

Let us summarize the results from the previous sections 
by listing all the possibilities for the pairs $H$ and $K$ as in \eqref{20121006:eq2},
and specifying,
using the terminology of Section \ref{sec:7.5},
whether the corresponding Lenard-Magri sequence \eqref{maxi}
is \emph{integrable of S-type},
i.e. the orders of the elements $P_n$'s and $\frac{\delta h_n}{\delta u}$'s tend to infinity
($b_1=0,a_1\neq0$),
whether it is \emph{integrable of C$_1$-type},
i.e. the $\mc C$-span of the elements $P_n$'s and $\tint h_n$'s is infinite dimensional
and the orders of $H$ and $K$ are both equal to $-1$ ($b_1=a_1=0$),
whether it is \emph{integrable of C$_2$-type},
i.e. the $\mc C$-span of the elements $P_n$'s and $\tint h_n$'s is infinite dimensional
and $H$ has order less than or equal to $K$ and $K$ has order 1 ($b_1\neq0$),
whether it is of \emph{finite type},
i.e. the $\mc C$-span of the elements $P_n$'s or $\tint h_n$'s is necessarily finite dimensional,
or whether it is \emph{blocked},
i.e. there are choices of $\tint h_n$ or $P_n$ for which the scheme
cannot be continued.
This is the list of all possibilities:
\begin{itemize}
\item
\emph{integrable of S-type}:
\begin{enumerate}[(a)]
%\item
%$b_1=0,b_2b_3\neq0,a_1a_2a_3\neq0$;
%\item
%$b_1=0,b_2b_3\neq0,a_1a_2\neq0,a_3=0$;
%\item
%$b_1=0,b_2b_3\neq0,a_1a_3\neq0,a_2=0$;
%\item
%$b_1=0,b_2b_3\neq0,a_1\neq0,a_2=a_3=0$;
%\item
%$b_1=b_3=0,b_2\neq0,a_1a_2a_3\neq0$;
%\item
%$b_1=b_3=0,b_2\neq0,a_1a_3\neq0,a_2=0$;
%\item
%$b_1=b_2=0,b_3\neq0,a_1a_2a_3\neq0$;
%\item
%$b_1=b_2=0,b_3\neq0,a_1a_2\neq0,a_3=0$.
\item
$b_1=0$, $(b_2,b_3)\neq(0,0)$, $a_1a_2a_3\neq0$;
\item
$b_1=0$, $a_1\neq0$, 
and either $b_2\neq0$, $a_2=0$ and $(b_3,a_3)\neq(0,0)$,
or $b_3\neq0$, $a_3=0$ and $(b_2,a_2)\neq(0,0)$.
\end{enumerate}
\item
\emph{integrable of C$_1$-type}:
\begin{enumerate}[]
%\item
%$b_1=0,b_2b_3\neq0,a_1=0,a_2a_3\neq0$;
%\item
%$b_1=0,b_2b_3\neq0,a_1=a_3=0,a_2\neq0$;
%\item
%$b_1=0,b_2b_3\neq0,a_1=a_2=0,a_3\neq0$;
%\item
%$b_1=b_3=0,b_2\neq0,a_1=0,a_2a_3\neq0$;
%\item
%$b_1=b_3=0,b_2\neq0,a_1=a_2=0,a_3\neq0$;
%\item
%$b_1=b_2=0,b_3\neq0,a_1=0,a_2a_3\neq0$;
%\item
%$b_1=b_2=0,b_3\neq0,a_1=a_3=0,a_2\neq0$.
\item
$b_1=a_1=0$, 
and either $b_2a_3\neq0$ and $(b_3,a_2)$ arbitrary,
or $b_3a_2\neq0$ and $(b_2,a_3)$ arbitrary.
\end{enumerate}
\item
\emph{integrable of C$_2$-type}:
\begin{enumerate}[(a)]
%\item
%$b_1b_2\neq0,b_3=0,a_1a_2\neq0,a_3=0$;
%\item
%$b_1b_2\neq0,b_3=0,a_1\neq0,a_2=a_3=0$;
%\item
%$b_1b_2\neq0,b_3=0,a_1=a_3=0,a_2\neq0$;
%\item
%$b_1b_3\neq0,b_2=0,a_1a_3\neq0,a_2=0$;
%\item
%$b_1b_3\neq0,b_2=0,a_1\neq0,a_2=a_3=0$;
%\item
%$b_1b_3\neq0,b_2=0,a_1=a_2=0,a_3\neq0$;
%\item
%$b_1\neq0,b_2=b_3=0,a_1a_2\neq0,a_3=0$;
%\item
%$b_1\neq0,b_2=b_3=0,a_1=a_3=0,a_2\neq0$;
%\item
%$b_1\neq0,b_2=b_3=0,a_1a_3\neq0,a_2=0$;
%\item
%$b_1\neq0,b_2=b_3=0,a_1=a_2=0,a_3\neq0$.
\item
$b_1a_1\neq0$, 
and either $b_2=a_2=0$ and $(b_3,a_3)\neq(0,0)$,
or $b_3=a_3=0$ and $(b_2,a_2)\neq(0,0)$;
\item
$b_1\neq0$, $a_1=0$,
and either $b_2=a_2=0$ and $a_3\neq0$ (with $b_3$ arbitrary),
or $b_3=a_3=0$ and $a_2\neq0$ (with $b_2$ arbitrary).
\end{enumerate}
\item
\emph{finite type}:
\begin{enumerate}[(a)]
%\item
%$b_1=b_3=0,b_2\neq0,a_1a_2\neq0,a_3=0$;
%\item
%$b_1=b_3=0,b_2\neq0,a_1\neq0,a_2=a_3=0$;
%\item
%$b_1=b_3=0,b_2\neq0,a_1=a_3=0,a_2\neq0$;
%\item
%$b_1=b_2=0,b_3\neq0,a_1a_3\neq0,a_2=0$;
%\item
%$b_1=b_2=0,b_3\neq0,a_1\neq0,a_2=a_3=0$;
%\item
%$b_1=b_2=0,b_3\neq0,a_1=a_2=0,a_3\neq0$;
%\item
%$b_1b_2b_3\neq0,a_1a_2a_3\neq0$;
%\item
%$b_1b_2b_3\neq0,a_1a_2\neq0,a_3=0$;
%\item
%$b_1b_2b_3\neq0,a_1a_3\neq0,a_2=0$;
%\item
%$b_1b_2b_3\neq0,a_1\neq0,a_2=a_3=0$;
%\item
%$b_1b_2b_3\neq0,a_1=0,a_2a_3\neq0$;
%\item
%$b_1b_2b_3\neq0,a_1=a_3=0,a_2\neq0$;
%\item
%$b_1b_2b_3\neq0,a_1=a_2=0,a_3\neq0$.
\item
$b_1b_2b_3\neq0$, $a_1=0$, $(a_2,a_3)\neq(0,0)$;
\item
$b_1=0$, $a_1\neq0$,
and either $b_2=a_2=0$ and $b_3\neq0$ (with $a_3$ arbitrary),
or $b_3=a_3=0$ and $b_2\neq0$ (with $a_2$ arbitrary).
\item[(c1)]
$b_1=a_1=0$,
and either $b_2=a_2=0$ and $b_3a_3\neq0$,
or $b_2a_2\neq0$ and $b_3=a_3=0$;
\item[(c2)]
$b_1b_2b_3\neq0, a_1a_2a_3\neq0$;
\item[(d)]
$b_1b_2b_3\neq0, a_1\neq0, a_2a_3=0$.
\end{enumerate}
\item
\emph{blocked}:
\begin{enumerate}[(a)]
%\item
%$b_1b_2\neq0,b_3=0,a_1a_2a_3\neq0$;
%\item
%$b_1b_2\neq0,b_3=0,a_1a_3\neq0,a_2=0$;
%\item
%$b_1b_2\neq0,b_3=0,a_1=0,a_2a_3\neq0$;
%\item
%$b_1b_2\neq0,b_3=0,a_1=a_2=0,a_3\neq0$;
%\item
%$b_1b_3\neq0,b_2=0,a_1a_2a_3\neq0$;
%\item
%$b_1b_3\neq0,b_2=0,a_1a_2\neq0,a_3=0$;
%\item
%$b_1b_3\neq0,b_2=0,a_1=0,a_2a_3\neq0$;
%\item
%$b_1b_3\neq0,b_2=0,a_1=a_3=0,a_2\neq0$;
%\item
%$b_1\neq0,b_2=b_3=0,a_1a_2a_3\neq0$;
%\item
%$b_1\neq0,b_2=b_3=0,a_1=0,a_2a_3\neq0$.
\item
$b_1\neq0$, $a_1=0$, 
and either $b_2=0$, $a_2\neq0$ and $(b_3,a_3)\neq(0,0)$,
or $b_3=0$, $a_3\neq0$ and $(b_2,a_2)\neq(0,0)$;
\item
$b_1\neq0, b_2b_3=0, a_1a_2a_3\neq0$;
\item
$b_1a_1\neq0$, 
and either $b_2a_3\neq0$ and $b_3=a_2=0$,
or $b_2=a_3=0$ and $b_3a_2\neq0$.
\end{enumerate}
\end{itemize}

%%%
\subsection{Going to the left}
\label{secb:3.8}

Suppose we have an integrable Lenard-Magri sequence \eqref{maxi}
(of S, C$_1$ or $C_2$-type).
A natural question is whether this sequence can be continued to the left:
\begin{equation}\label{20121006:eq1}
\dots\ass{H}P_{-1}\ass{K}\tint0\ass{H}P_0\ass{K}\tint h_0\ass{H} P_1\ass{K}\tint h_1\ass{H}
P_2\ass{K}\dots\,.
\end{equation}
In this way we get some additional equations
compatible with the given hierarchy $\frac{du}{dt_n}=P_n,\,n\in\mb Z_+$,
and additional integrals of motion $\tint h_n,\,n=-1,-2,\dots$,
in involution with the given $\tint h_n$'s, with $n\geq0$.

Clearly, trying to extend the Lenard-Magri scheme \eqref{20121006:eq1} to the left
amounts to switching the roles of the non-local Poisson structures $H$ and $K$,
and to constructing the ``dual'' Lenard-Magri sequence
\begin{equation}\label{20121006:eq3}
\tint0\ass{K}P_{-1}\ass{H}\tint h_{-1}\ass{K} P_{-2}\ass{H}\tint h_{-2}\ass{K}
P_{-3}\ass{H}\dots\,.
\end{equation}
So, we need to study, for each possible choice of the parameters $a_i,b_i,i=1,2,3$,
what type of Lenard-Magri scheme we get when we switch all the 
coefficients $a_i$'s with the $b_i$'s.

By looking at the list of all possibilities in the previous section,
after switching the roles of $H$ and $K$ we have the following:
\begin{itemize}
\item
integrable of S-type (a) $\stackrel{H\leftrightarrow K}{\longleftrightarrow}$ finite type (a);
\item
integrable of S-type (b) $\stackrel{H\leftrightarrow K}{\longleftrightarrow}$ blocked (a);
\item
integrable of C$_1$-type $\stackrel{H\leftrightarrow K}{\longleftrightarrow}$ integrable of C$_1$-type;
\item
integrable of C$_2$-type (a) $\stackrel{H\leftrightarrow K}{\longleftrightarrow}$ 
integrable of C$_2$-type (a);
\item
integrable of C$_2$-type (b) $\stackrel{H\leftrightarrow K}{\longleftrightarrow}$ 
finite-type (b);
\item
finite-type (c) $\stackrel{H\leftrightarrow K}{\longleftrightarrow}$ finite-type (c);
\item
finite-type (d) $\stackrel{H\leftrightarrow K}{\longleftrightarrow}$ blocked (b);
\item
blocked (c) $\stackrel{H\leftrightarrow K}{\longleftrightarrow}$ blocked (c).
\end{itemize}

% dual finite

We are only interested in the integrable (S or C-type) Lenard-Magri schemes.
We see from the above list that, after exchanging the roles of $H$ and $K$,
three things can happen.
The ``dual'' Lenard-Magri scheme \eqref{20121006:eq3} can be of \emph{finite}-type
(this happens in the cases S(a) and C$_2$(b)).
In this situation continuing the Lenard-Magri scheme to the left 
we never get any new interesting integrals of motion or equations.

% dual integrable

The second possibility is that the ``dual'' Lenard-Magri scheme \eqref{20121006:eq3} 
is of \emph{integrable}-type
(this happens in the cases C$_1$ and C$_2$(a)).
In this situation we can continue the Lenard-Magri scheme to the left indefinitely.
In other words, in each of these cases we can merge two integrable systems,
``dual'' to each other, to get one integrable system with twice as many integrals of motion
and equations.

% dual blocked

The most interesting situation is when the ``dual'' Lenard-Magri scheme \eqref{20121006:eq3} 
is \emph{blocked} (which happens in the case S(b)).
In this case,
if the sequence \eqref{20121006:eq3} is blocked at $P_{-k},\,k\geq0$,
we obtain an integrable PDE which is not of evolutionary type.

%%%
\subsection{Non-evolutionary integrable equations}
\label{secb:3.9}

According to the previous discussion, 
we need to consider the case of integrable Lenard-Magri scheme
of S-type (b), which means the following 5 cases:
\begin{enumerate}
\item
$b_1=0, b_2\neq0, b_3\neq0, a_1\neq0, a_2\neq0, a_3=0$; % S(ii) <--> C(3)
\item
$b_1=0, b_2\neq0, b_3\neq0, a_1\neq0, a_2=0, a_3\neq0$; % S(iii) <--> C(5)
\item
$b_1=0, b_2\neq0, b_3\neq0, a_1\neq0, a_2=0, a_3=0$; % S(iv) <--> C(6)
\item
$b_1=0, b_2\neq0, b_3=0, a_1\neq0, a_2=0, a_3\neq0$; % S(vii) <--> C(5)
\item
$b_1=0, b_2=0, b_3\neq0, a_1\neq0, a_2\neq0, a_3=0$; % S(x) <--> C(3)
\end{enumerate}

\subsubsection*{Case 1: $b_1=0, b_2\neq0, b_3\neq0, a_1\neq0, a_2\neq0, a_3=0$} % S(ii) <--> C(3)

This case gives, to the right, the Lenard-Magri scheme listed as case (ii) in Section \ref{secb:3.4},
while, after exchanging the roles of $H$ and $K$ we get, to the left,
the ``blocked'' Lenard-Magri scheme listed as case 3 in Section \ref{secb:3.6}.
Hence, overall, we get the following scheme:
$$
\begin{array}{l}
\vphantom{\Bigg(}
\displaystyle{
\not\exists P\ass{K}
\tint e^{\pm a_{12}x}u
\ass{H}0 \ass{K}\tint 0
\ass{H}1\ass{K}\tint\sqrt{b_2+b_3(u')^2}\ass{H}
} \\
\displaystyle{
\ass{H}
-\frac{a_1b_2b_3u'''}{(b_2+b_3(u')^2)^{\frac32}}
+3\frac{a_1b_2b_3^2u'(u'')^2}{(b_2+b_3(u')^2)^{\frac52}}
-\frac{a_2b_3}{\sqrt{b_2+b_3(u')^2}}
\ass{K}\dots
\,.
}
\end{array}
$$
Here and further $\pm$ means that we take arbitrary linear combination of the above expressions
with $+$ and with $-$.
Trying to solve naively for $P$ in the above scheme,
we get the following expression
$$
P=
\pm \frac{b_2}{a_{12}} e^{\pm a_{12}x}
+b_3u' \partial^{-1} \big(e^{\pm a_{12}x}u'\big)
\,.
$$
The meaning of the above expression for $P$
is that the following partial differential equation
is a member of the integrable hierarchy associated to the Lenard-Magri scheme of S-type (ii):
\begin{equation}\label{20121020:eq3}
\Big(\frac{u_t}{u_x}\Big)_x
=
\pm\frac{b_2}{a_{12}}
\Big(\frac1{u_x}e^{\pm a_{12}x}\Big)_x
+b_3 e^{\pm a_{12}x} u_x
\,.
\end{equation}

\subsubsection*{Case 2: $b_1=0, b_2\neq0, b_3\neq0, a_1\neq0, a_2=0, a_3\neq0$} % S(iii) <--> C(5)

This case gives, to the right, the Lenard-Magri scheme listed as case (iii) in Section \ref{secb:3.4},
while, after exchanging the roles of $H$ and $K$ we get, to the left,
the ``blocked'' Lenard-Magri scheme listed as case 5 in Section \ref{secb:3.6}.
Hence, overall, we get the following scheme:
$$
\begin{array}{l}
\vphantom{\Bigg(}
\displaystyle{
\not\exists P\ass{K}
\tint e^{\pm a_{13}u}
\ass{H}0 \ass{K}\tint 0
\ass{H}u'\ass{K}\tint\sqrt{b_2+b_3(u')^2}\ass{H}
} \\
\displaystyle{
\ass{H}
-\frac{a_1b_2b_3u'''}{(b_2+b_3(u')^2)^{\frac32}}
+3\frac{a_1b_2b_3^2u'(u'')^2}{(b_2+b_3(u')^2)^{\frac52}}
+\frac{a_3b_2}{\sqrt{b_2+b_3(u')^2}}
\ass{K}\dots
\,.
}
\end{array}
$$
Trying to solve naively for $P$ we get
$$
P=
\pm a_{13}b_2\partial^{-1}e^{\pm a_{13}u}
+b_3u'e^{\pm a_{13}u}
\,.
$$
This means that the following hyperbolic partial differential equation
is a member of the integrable hierarchy associated to the Lenard-Magri scheme of S-type (iii):
\begin{equation}\label{20121020:eq4}
u_{tx}=
\pm a_{13}b_2 e^{\pm a_{13}u}\pm\frac{b_3}{a_{13}} \big(e^{\pm a_{13}u}\big)_{xx}
\,.
\end{equation}

\subsubsection*{Case 3: $b_1=0, b_2\neq0, b_3\neq0, a_1\neq0, a_2=0, a_3=0$} % S(iv) <--> C(6)

This case gives, to the right, the Lenard-Magri scheme listed as case (iv) in Section \ref{secb:3.4},
while, after exchanging the roles of $H$ and $K$ we get, to the left,
the ``blocked'' Lenard-Magri scheme listed as case 6 in Section \ref{secb:3.6}.
Hence, overall, we get, depending on how we choose to continue the scheme to the left,
the following two possibilities (or any their linear combination):
$$
\begin{array}{l}
\vphantom{\Bigg(}
\displaystyle{
\not\exists P
\ass{K}
\frac{1}{a_1}\tint xu
\ass{H}
1
\ass{K}\tint 0\ass{H}0
\ass{K}\tint\sqrt{b_2+b_3(u')^2}\ass{H}
} \\
\displaystyle{
\ass{H}
-\frac{a_1b_2b_3u'''}{(b_2+b_3(u')^2)^{\frac32}}
+3\frac{a_1b_2b_3^2u'(u'')^2}{(b_2+b_3(u')^2)^{\frac52}}
\ass{K}\dots
\,,
}
\end{array}
$$
or
$$
\begin{array}{l}
\vphantom{\Bigg(}
\displaystyle{
\not\exists P
\ass{K}
\frac{1}{a_1}\tint \frac12 u^2
\ass{H}
u'
\ass{K}\tint 0\ass{H}0
\ass{K}\tint\sqrt{b_2+b_3(u')^2}\ass{H}
} \\
\displaystyle{
\ass{H}
-\frac{a_1b_2b_3u'''}{(b_2+b_3(u')^2)^{\frac32}}
+3\frac{a_1b_2b_3^2u'(u'')^2}{(b_2+b_3(u')^2)^{\frac52}}
\ass{K}\dots
\,.
}
\end{array}
$$
Trying to solve naively for $P$ we get, in the first case
$$
P=
\frac{b_2}{2a_1}x^2+\frac{b_3}{a1}xuu'-\frac{b_3}{a_1}u'\partial^{-1}u
\,,
$$
which corresponds to the following integrable non-evolutionary partial differential equation:
\begin{equation}\label{20121020:eq1}
\Big(\frac{u_t}{u_x}\Big)_x
=
\frac{b_2}{2a_1}\Big(\frac{x^2}{u_x}\Big)_x+\frac{b_3}{a_1}xu_x
\,.
\end{equation}
In the second case we get
$$
P=
\frac{b_2}{a_1}\partial^{-1}u+\frac{b_3}{2a_1}u^2u'
\,,
$$
which corresponds to the following integrable hyperbolic partial differential equation:
\begin{equation}\label{20121020:eq2}
u_{tx}
=
\frac{b_2}{a_1}u+\frac{b_3}{6a_1}(u^3)_{xx}
\,.
\end{equation}
In conclusion, both equations \eqref{20121020:eq1} and \eqref{20121020:eq2}
are members of the integrable hierarchy associated to the Lenard-Magri scheme of S-type (iv).

\subsubsection*{Case 4: $b_1=0, b_2\neq0, b_3=0, a_1\neq0, a_2=0, a_3\neq0$} % S(vii) <--> C(5)

This case gives, to the right, the Lenard-Magri scheme listed as case (vii) in Section \ref{secb:3.4},
while, after exchanging the roles of $H$ and $K$ we get, to the left,
the ``blocked'' Lenard-Magri scheme listed as case 5 in Section \ref{secb:3.6}.
Hence, overall, we get the following scheme:
$$
\begin{array}{l}
\vphantom{\Bigg(}
\displaystyle{
\not\exists P\ass{K}
\tint e^{\pm a_{13}u}
\ass{H}0 \ass{K}\tint 0
\ass{H}u'\ass{K}\tint\frac{-(u')^2}{2b_2}\ass{H}
} \\
\displaystyle{
\ass{H}
\frac{a_1}{b_2}u'''+\frac{a_3}{2b_2}(u')^3
\ass{K}\dots
\,.
}
\end{array}
$$
Trying to solve naively for $P$ we get
$$
P=
\pm a_{13}b_2\partial^{-1}e^{\pm a_{13}u}
\,.
$$
This means that the following hyperbolic partial differential equation
is a member of the integrable hierarchy associated to the Lenard-magri scheme of S-type (vii):
\begin{equation}\label{20121020:eq6}
u_{tx}=
\pm a_{13}b_2 e^{\pm a_{13}u}
\,.
\end{equation}
As expected, equation \eqref{20121020:eq6} is obtained by \eqref{20121020:eq4}
letting $b_3=0$.

\subsubsection*{Case 5: $b_1=0, b_2=0, b_3\neq0, a_1\neq0, a_2\neq0, a_3=0$} % S(x) <--> C(3)

This case gives, to the right, the Lenard-Magri scheme listed as case (x) in Section \ref{secb:3.4},
while, after exchanging the roles of $H$ and $K$ we get, to the left,
the ``blocked'' Lenard-Magri scheme listed as case 3 in Section \ref{secb:3.6}.
Hence, overall, we get the following scheme:
$$
\begin{array}{l}
\vphantom{\Bigg(}
\displaystyle{
\not\exists P\ass{K}
\tint e^{\pm a_{12}x}u
\ass{H}0 \ass{K}\tint 0
\ass{H}1\ass{K}\tint\frac1{2b_3u'}\ass{H}
} \\
\displaystyle{
\ass{H}
-\frac{a_1}{b_3}\frac{u'''}{(u')^3}+\frac{3a_1}{b_3}\frac{(u'')^2}{(u')^4}
+\frac{a_2}{2b_3}\frac1{(u')^2}
\ass{K}\dots
\,.
}
\end{array}
$$
Trying to solve naively for $P$ in the above scheme,
we get the following expression
$$
P=
b_3u' \partial^{-1} \big(e^{\pm a_{12}x}u'\big)
\,,
$$
and the associated non-evolutionary partial differential equation is
\begin{equation}\label{20121020:eq5}
\Big(\frac{u_t}{u_x}\Big)_x
=
b_3 e^{\pm a_{12}x} u_x
\,.
\end{equation}
In conclusion, equation \eqref{20121020:eq5}
is a member of the integrable hierarchy associated to the Lenard-Magri scheme of S-type (x).
Note that this equation is obtained letting $b_2=0$ in equation \eqref{20121020:eq3}.

\subsubsection*{Conclusion} 

After rescaling the variables $u$, $x$ and $t$, or replacing $x$ by $x+$ const., or $u$ by $u+$ const.,
in equations \eqref{20121020:eq3}-\eqref{20121020:eq5},
we conclude that the following are all the integrable non-evolutionary partial differential equations
which are members of some integrable hierarchy of bi-Hamiltonian equations, 
with $H$ and $K$ as in \eqref{20121006:eq2}:
\begin{eqnarray}
&& u_{tx}
=
e^{u}-\alpha e^{-u}
+\epsilon(e^{u}-\alpha e^{-u})_{xx}
\,,\label{20121020:eq8}\\
&& \Big(\frac{u_t}{u_x}\Big)_x
=
(e^{x}-\alpha e^{-x}) u_x
+\epsilon \Big(\frac{e^{x}-\alpha e^{-x}}{u_x}\Big)_x
\,, \label{20121020:eq7}\\
&& u_{tx}
=
u+(u^3)_{xx}
\,, \label{20121020:eq10}\\
&& \Big(\frac{u_t}{u_x}\Big)_x
=
\Big(\frac{x^2}{u_x}\Big)_x+xu_x
\,,\label{20121020:eq9}
\end{eqnarray}
where $\alpha$ and $\epsilon$ are $0$ or $1$.

Recall that the case when $\epsilon=0$ equation \eqref{20121020:eq8}
is the Liouville equation when $\alpha=0$,
and the sinh-Gordon equation when $\alpha=1$, cf. \cite{Dor93}.
Equation \eqref{20121020:eq7} (respectively \eqref{20121020:eq9}) can be obtained from equation \eqref{20121020:eq8} (resp. \eqref{20121020:eq10})
by the hodograph transformation $u\mapsto x$, $x\mapsto -u$.
Equation \eqref{20121020:eq10} is called the ``short pulse equation'' \cite{SW02},
and its integrability was proved in \cite{SS04}.
Equations \eqref{20121020:eq8} with $\epsilon=1$
was studied in \cite{Fok95}.

%%%%%%%%%%%%%%%%%%%%%%%%%%%%%%%%%%%%%%%
\section{KN type integrable systems}
\label{secb:4}

In this section $\mc V$ is a field of differential functions in $u$,
and, as usual, we assume that $\mc V$ contains all the functions that we encounter
in our computations.

Recall from Example \ref{20110922:ex2} that the following
is a pair of compatible non-local Poisson structures:
$$
L_1=
u'\partial^{-1}\circ u' 
\,\,\text{ (Sokolov) }
\,\,,\,\,\,\,
L_2=
\partial^{-1}\circ u'\partial^{-1}\circ u'\partial^{-1}
\,\,\text{ (Dorfman) }
\,.
$$
We consider two non-local Poisson structures $H$ and $K$
which are linear combinations of $L_1$ and $L_2$:
$H=a_1L_1+a_2L_2$ and $K=b_1L_1+b_2L_2$.
As we have seen in the example of Liouville type integrable systems, 
discussed in Section \ref{secb:3}, 
integrable hierarchies associated to Lenard-Magri schemes of C type are usually 
not very interesting (cf. Sections \ref{secb:3.5} and \ref{secb:3.6}).
Hence, in this section, we will only consider integrable Lenard-Magri schemes of S-type
(in the terminology of Section \ref{sec:7.5}),
which is possible only when the order of the pseudodifferential operator $H$
is greater than the order of $K$,
namely when $a_1\neq0$ and $b_1=0$.
Therefore, we consider the following compatible pair of non-local structures:
\begin{equation}\label{20121006:eq2kn}
H=u'\partial^{-1}\circ u' + a\partial^{-1}\circ u'\partial^{-1}\circ u'\partial^{-1}
\,\,,\,\,\,\,
K=\partial^{-1}\circ u'\partial^{-1}\circ u'\partial^{-1}\,,
\end{equation}
with $a\in\mc C$.
We want to discuss the integrability of the corresponding Lenard-Magri scheme.

%%%
\subsection{Preliminary computations}
\label{secb:4.1}

Note that $K$ is the inverse of a differential operator, hence its minimal fractional decomposition is
$K=1D^{-1}$, where
\begin{equation}\label{frac-D}
D=\partial\circ\frac1{u'}\partial\circ\frac1{u'}\partial\,.
\end{equation}
We next find a minimal fractional decomposition for $H$.
It is given by the following
\begin{lemma}\label{lem:frac-kn}
For every $a\in\mc C$, we have $H=AB^{-1}$, where
\begin{equation}\label{frac-kn}
\begin{array}{l}
\displaystyle{
A=
\bigg(
\partial^2-2\frac{u''}{u'}\partial+\Big(\frac{u''}{u'}\Big)^\prime+a
\bigg)\circ
\frac{1}{D(u')}
\partial-u'
\,,}\\
\displaystyle{
B=
\partial\circ\frac1{u'}\partial\circ\frac1{u'}\partial\circ
\frac{1}{D(u')}
\partial
\,.} 
\end{array}
\end{equation}
Here and further, we have, recalling \eqref{frac-D},
\begin{equation}\label{20121015:eq3}
D(u')=\bigg(\frac1{u'}\Big(\frac{u''}{u'}\Big)^\prime\bigg)^\prime\,.
\end{equation}
The above fractional decomposition is minimal only for $a\neq0$.
For $a=0$, the minimal fractional decomposition for $H$ is 
$H=1S^{-1}$, where
\begin{equation}\label{frac-kn1}
S=\frac1{u'}\partial\circ\frac1{u'}\,.
\end{equation}
\end{lemma}
\begin{proof}
We need to prove that $AB^{-1}=S^{-1}+aD^{-1}$.
By looking at the coefficient of $a$ in $AB^{-1}$, we get
$$
\frac{1}{D(u')}
\partial
\Bigg(
\partial\circ\frac1{u'}\partial\circ\frac1{u'}\partial\circ
\frac{1}{D(u')}
\partial
\Bigg)^{-1}
=
\partial^{-1}
\circ u'
\partial^{-1}
\circ u'
\partial^{-1}
=D^{-1}\,.
$$
Letting $a=0$ in $AB^{-1}$, we have
$$
\begin{array}{l}
\displaystyle{
\Bigg(
\bigg(
\partial^2-2\frac{u''}{u'}\partial+\Big(\frac{u''}{u'}\Big)^\prime
\bigg)\circ
\frac{1}{D(u')}
\partial-u'
\Bigg)
\Bigg(
\partial\circ\frac1{u'}\partial\circ\frac1{u'}\partial\circ
\frac{1}{D(u')}
\partial
\Bigg)^{-1}
} \\
\displaystyle{
=\bigg(
\partial^2-2\frac{u''}{u'}\partial+\Big(\frac{u''}{u'}\Big)^\prime
-u'\partial^{-1}\circ D(u')
\bigg)\circ
\partial^{-1}u'\partial^{-1}\circ u'\partial^{-1}
} \\
\displaystyle{
=
\partial\circ u'\partial^{-1}\circ u'\partial^{-1}
-2u''\partial^{-1}\circ u'\partial^{-1}
+\Big(\frac{u''}{u'}\Big)^\prime\partial^{-1}u'\partial^{-1}\circ u'\partial^{-1}
} \\
\displaystyle{
-u'\partial^{-1}\circ D(u')\partial^{-1}u'\partial^{-1}\circ u'\partial^{-1}
=
u'\partial^{-1}\circ u'
+\bigg(
u'\partial^{-1}\circ \frac{u''}{u'}
} \\
\displaystyle{
-u''\partial^{-1}
+\Big(\frac{u''}{u'}\Big)^\prime\partial^{-1}u'\partial^{-1}
-u'\partial^{-1}\circ D(u')\partial^{-1}u'\partial^{-1}
\bigg)\circ u'\partial^{-1}
\,.}
\end{array}
$$
In the last identity we used the Leibniz rule for $\partial$:
$\partial\circ f=f\partial+f'$.
To conclude the proof, we need to check that the expression in parenthesis
in the RHS is zero:
\begin{equation}\label{20121025:eq1}
u'\partial^{-1}\circ \frac{u''}{u'}
-u''\partial^{-1}
+\Big(\frac{u''}{u'}\Big)^\prime\partial^{-1}u'\partial^{-1}
-u'\partial^{-1}\circ D(u')\partial^{-1}u'\partial^{-1}
=0\,.
\end{equation}
This identity is obtained applying repeatedly the commutation relation ($f\in\mc V$),
\begin{equation}\label{20121025:eq2}
\partial^{-1}\circ f=f\partial^{-1}-\partial^{-1}\circ f'\partial^{-1}\,,
\end{equation}
which is a consequence of the Leibniz rule for $\partial$,
and using the expression \eqref{20121015:eq3} for $D(u')$.
\end{proof}

In order to apply successfully the Lenard-Magri scheme of integrability we need to compute
the kernel of the operator $B$.
\begin{lemma}\label{20121025:lem1}
The kernel of the operator $B$ in \eqref{frac-kn}
is a 4-dimensional vector space over $\mc C$,
spanned by
$$
\begin{array}{l}
\displaystyle{
f_1=1
\,\,,\,\,\,\,
f_2=\frac1{u'}\Big(\frac{u''}{u'}\Big)^\prime
\,\,,\,\,\,\,
f_3=\frac{u}{u'}\Big(\frac{u''}{u'}\Big)^\prime-\frac{u''}{u'}
\,,} \\
\displaystyle{
f_4=
\frac{u^2}{u'}\Big(\frac{u''}{u'}\Big)^\prime
-2u\frac{u''}{u'}+2u'
\,.}
\end{array}
$$
\end{lemma}
\begin{proof}
It is immediate to check that all the elements $f_i$ are indeed in the kernel of $B$.
On the other hand, since $B$ has order 4, its kernel has dimension at most 4.
\end{proof}

%%%
\subsection{Applying the Lemard-Magri scheme for $a\neq0$}
\label{secb:4.2}

According to the Lenard-Magri scheme of integrability,
starting with $\tint h_{-1}=\tint0$,
we need to find sequences $\{\tint h_n\}_{n=0}^N$
and $\{P_n\}_{n=0}^N$
solving the recursion relations \eqref{maxi}.

% finite scheme

Since $C=1$, in order to find solutions of the scheme \eqref{maxi} up to $N=3$,
we need to find elements $F_n,h_n,P_n\in\mc V,\,n=0,\dots,3$,
such that
$$
BF_{n}=\frac{\delta h_{n-1}}{\delta u}
\,\,,\,\,\,\,
P_n=AF_{n}
\,\,,\,\,\,\,
\frac{\delta h_n}{\delta u}=DP_n
\,,
$$
for all $n=0,1,2,3$ (we let, as usual, $\tint h_{-1}=\tint 0$).
Recalling the expressions \eqref{frac-D} and \eqref{frac-kn} of $A,B,D$,
and using Lemma \ref{20121025:lem1},
it is a straightforward but lengthy calculation to find solutions:
$$
\begin{array}{lll}
\displaystyle{
F_0=\frac{f_2}a
=\frac1{au'}\Big(\frac{u''}{u'}\Big)^\prime
\,\,,}&
P_0=1
\,\,,&
\tint h_0=\tint0
\,, \\
\displaystyle{
F_1=\frac{f_3}a 
= \frac{u}{au'}\Big(\frac{u''}{u'}\Big)^\prime-\frac{u''}{u'}
\,\,,}&
P_1=u
\,\,,&
\tint h_1=\tint0
\,, \\
\displaystyle{
F_2=\frac{f_4}a 
=\frac{u^2}{au'}\Big(\frac{u''}{u'}\Big)^\prime-2\frac{uu''}{au'}+\frac{2}{a}u'
\,\,,}&
P_2=u^2
\,\,,&
\tint h_2=\tint0
\,, \\
\displaystyle{
F_3=-f_1=-1
\,\,,}&
P_3=u'
\,\,,&
\displaystyle{
\tint h_3=\frac12\int\Big(\frac{u''}{u'}\Big)^2
\,.} 
\end{array}
$$
Hence, we get the following Lenard-Magri scheme
\begin{equation}\label{20121102:eq1}
\begin{array}{l}
\displaystyle{
\tint 0\ass{H}1\ass{K}\tint 0
\ass{H} u\ass{K}\tint 0
\ass{H}u^2\ass{K}\tint 0\ass{H}
} \\
\displaystyle{
\ass{H}u'\ass{K}\frac12\tint\Big(\frac{u''}{u'}\Big)^2
\ass{K}\dots\,.
}
\end{array}
\end{equation}

% orthogonality conditions

We next prove that the scheme \eqref{20121102:eq1} can be extended indefinitely,
possibly going to a normal extension $\tilde{\mc V}$ of $\mc V$.
According to Theorem \ref{20130123:thm} and Remark \ref{20130104:rem2},
this is the case, provided that the orthogonality conditions \eqref{20130104:eq2} hold.
Since $C=1$, the first condition in \eqref{20130104:eq2} is trivial.
As for the second orthogonality condition,
let $\varphi\in\big(\Span_{\mc C}\{P_0,P_1,P_2,P_3\}\big)^\perp$.
Since $\varphi\perp P_0$, we have that $\varphi=\varphi_1^\prime$, for some $\varphi_1\in\mc V$.
Since $\varphi\perp P_1$, we have that $\varphi_1=\frac{\varphi_2^\prime}{u'}$, 
for some $\varphi_2\in\mc V$.
Since $\varphi\perp P_2$, we have that $\varphi_2=\frac{\varphi_3^\prime}{u'}$, 
for some $\varphi_3\in\mc V$.
And, finally,
since $\varphi\perp P_3$, we have that $\varphi_3=\frac{\varphi_4^\prime}{D(u')}$, 
for some $\varphi_4\in\mc V$.
In conclusion, $\varphi=B\varphi_4$, proving the second orthogonality condition \eqref{20130104:eq2}.

% first non-trivial equation

We compute explicitly the next element $P_4$ in the Lenard-Magri scheme,
which gives the first non-trivial equation of the corresponding bi-Hamiltonian hierarchy.
For this, we need to solve, for $F_4,P_4\in\mc V$, the following equations
$$
BF_{4}=\frac{\delta h_{3}}{\delta u}
=D(u')
\,\,,\,\,\,\,
P_4=AF_{4}
\,.
$$
The general solution is:
$$
F_{4}=
\Big(\frac{u''}{u'}\Big)^\prime-\frac12\Big(\frac{u''}{u'}\Big)^2
+(a-\alpha_1)f_1+\frac{\alpha_2}{a}f_2
%\frac1{u'}\Big(\frac{u''}{u'}\Big)^\prime
+\frac{\alpha_3}{a} f_3
%\Big(\frac{u}{u'}\Big(\frac{u''}{u'}\Big)^\prime-\frac{u''}{u'}\Big)
+\frac{\alpha_4}{a} f_4
%\Big(\frac{u^2}{u'}\Big(\frac{u''}{u'}\Big)^\prime-2u\frac{u''}{u'}+2u'\Big)
\,,
$$
where $\alpha_i,\,i=1,\dots,4$, are arbitrary constants.
Hence, the first non-trivial integrable equation in the hierarchy has the form:
\begin{equation}\label{20121102:eq2}
\frac{du}{dt}=P_4=
u'''-\frac32\frac{(u'')^2}{u'}+\alpha_1 u'
+\alpha_2+\alpha_3u+\alpha_4u^2\,.
\end{equation}

% differential orders and linear independence

In order to prove that equation \eqref{20121102:eq2} is indeed integrable,
we are left to prove that the sequences $\{\tint h_n\}_{n\in\mb Z_+}$
and $\{P_n\}_{n\in\mb Z_+}$ are linearly independent.
For this, we use Lemma \ref{20130123:lem3}.

Since $\dord(D(u'))=4$,
we have $\dord(A)=6$, $\dord(B)=7$, $\dord(C)=-\infty$ and $\dord(D)=3$.
Moreover, $|H|=-1$ and $|K|=-3$.
Hence, the RHS of inequality \eqref{20120911:eq1} is $4$.

Next, we compute the differential order of the next element $P_5$ in the Lenard-Magri scheme.
It is obtained by solving, for $\xi_4=\frac{\delta h_4}{\delta u}, F_5, P_5\in\mc V$,
the following equations:
\begin{equation}\label{20121102:eq3}
\xi_4=DP_4
\,\,,\,\,\,\,
BF_5=\xi_4
\,\,,\,\,\,\,
P_5=AF_5\,.
\end{equation}
From the first equation in \eqref{20121102:eq3} we get
$$
\xi_4=\partial\frac1{u'}\partial\frac1{u'}\partial(u'''+\rho)
\,,
$$
where $\rho\in\mc V$ has $\dord(\rho)=2$.
Hence, the second equation in \eqref{20121102:eq3} gives
$$
\frac1{D(u')}\partial F_5=u'''+\rho_1\,,
$$
where $\dord(\rho_1-\rho)=0$.
In particular, $F_5$ has differential order less than or equal to $3$.
It follows by the third equation in \eqref{20121102:eq3} that
$$
P_5=
\bigg(
\partial^2-2\frac{u''}{u'}\partial+\Big(\frac{u''}{u'}\Big)^\prime+a
\bigg)(u'''+\rho_1)
-u'F_5\,.
$$
Hence, $\frac{\partial P_5}{\partial u^{(5)}}=1$, and $\frac{\partial P_5}{\partial u^{(n)}}=0$ 
for every $n>5$.
In particular, $\dord(P_5)=5$.

According to Lemma \ref{20130123:lem3},
since we have $\dord(P_5)=5>4$,
we obtain:
$\dord\big(\frac{\delta h_n}{\delta u}\big)=2n-2$,
and $\dord(P_n)=2n-5$,
for every $n\geq3$.
In particular,
all the elements $\{\tint h_n\}_{n\in\mb Z_+}$
and $\{P_n\}_{n\in\mb Z_+}$ are linearly independent.
As a consequence, every equation of the hierarchy $\frac{du}{dt_n}=P_n,\,n\in\mb Z_+$,
including equation \eqref{20121102:eq2}, is integrable of S-type.

Note that, since the kernels of $B^*$ and $D^*$ have non-zero intersections,
we cannot conclude that $[P_m,P_n]$ is zero for every $m,n\in\mb Z_+$.
In fact, we have $[P_0,P_1]=P_0,\,[P_0,P_2]=2P_1,\,[P_1,P_2]=P_2$,
and $\ker(B^*)\cap\ker(D^*)=\Span\{P_0,P_1,P_2\}$
(which is isomorphic to sl$_2$),
in complete agreement with our Theorem \ref{20130123:thm}
(and in disagreement, for example, with \cite[Thm.5.36]{Olv93} and \cite[Thm.3.12]{Bla98}).

When all constants $\alpha_i$ are equal to zero,
equation \eqref{20121102:eq2} is usually called the Schwarz KdV equation,
see e.g. \cite{MS12}
(in \cite{Dor93} it is called the Krichever-Novikov (KN) equation,
since it is a degeneration of the KN equation).
As explained in \cite{MS12}, equation \eqref{20121102:eq2} can be reduced to
the Schwarz KdV equation by some point transformation.

\begin{remark}\label{20130130:rem3}
By Remark \ref{20130130:rem1}, all $\xi_n$'s and $P_n$'s constructed in this section
have coordinates in $\mc V=\mb F[u,{u'}^{\pm1},u'',u''',\dots]$.
By Example \ref{20130112:ex2},
this algebra is contained in a normal extension $\tilde{\mc V}=\mc V[\log u']$,
and all conserved densities $h_n$'s can be chosen in $\tilde{\mc V}$.
\end{remark}

%%%
\subsection{The case $a=0$}
\label{secb:4.3}

In the case when $a=0$ all the computations are much easier.
Since $A=C=1$, 
the recursive conditions $\tint h_{n-1}\ass{H}P_n\ass{K}\tint h_n,\,n\in\mb Z_+$,
are equivalent to the equations
$$
BP_n=\frac{\delta h_{n-1}}{\delta u}
\,\,,\,\,\,\,
\frac{\delta h_n}{\delta u}=DP_n
\,.
$$
It is easy to find the first few steps of the Lenard-Magri scheme:
\begin{equation}\label{20121102:eq1b}
\tint 0\ass{H}P_0=u'\ass{K}\tint h_0=\frac12\tint\Big(\frac{u''}{u'}\Big)^2
\ass{H}P_1=u'''-\frac{3}{2}\frac{(u'')^2}{u'}+\alpha_1 u'
\ass{K}\dots
\end{equation}
for arbitrary $\alpha_1\in\mc C$.

% orthogonality conditions

As before, the scheme \eqref{20121102:eq1b} can be extended indefinitely.
Indeed, since $C=1$, the first orthogonality condition in \eqref{20130104:eq2} is trivial,
while the second one holds since $P_0^\perp=\im B$.

% differential orders and linear independence

Moreover, in this case $\dord(A)=\dord(C)=-\infty$, $\dord(B)=2$, and $\dord(D)=3$,
so the RHS of inequality \eqref{20120911:eq1} is $0$.
Since $\dord(P_0)=1>0$,
we can apply Lemma \ref{20130123:lem3}
to deduce that all the elements $\{\tint h_n\}_{n\in\mb Z_+}$
and $\{P_n\}_{n\in\mb Z_+}$ are linearly independent.

In conclusion, 
every equation of the hierarchy $\frac{du}{dt_n}=P_n,\,n\in\mb Z_+$,
is integrable of S-type.
Note that the first non-trivial equation is $\frac{du}{dt}=P_1$,
which is the same as equation \eqref{20121102:eq2} with $\alpha_2=\alpha_3=\alpha_4=0$.
Note also that, since $\ker B^*\cap\ker D^*=0$ in this case,
we have $[P_m,P_n]=0$ for all $m,n\in\mb Z_+$.

%%%
\subsection{One step back}
\label{secb:4.4}

As we did in the example of Liouville type integrable systems,
we can ask whether the Lenard-Magri schemes \eqref{20121102:eq1} and \eqref{20121102:eq1b}
can be continued to the left.
This amounts to finding $P_{n}\in\mc V$ and $\tint h_{n}\in\mc V/\partial\mc V$, with $n\leq-1$,
such that
\begin{equation}\label{20121103:eq1}
\dots\ass{K}\tint h_{-2}\ass{H}P_{-1}\ass{K}\tint0
\end{equation}

% a\neq0

We consider separately the cases $a\neq0$ and $a=0$.
When $a\neq0$, the conditions \eqref{20121103:eq1}
give the following equations
for $P_{-1}$, $F$ and $\xi_{-2}=\frac{\delta h_{-2}}{\delta u}$:
\begin{equation}\label{20121103:eq2}
DP_{-1}=0
\,\,,\,\,\,\,
AF=P_{-1}
\,\,,\,\,\,\,
\xi_{-2}=BF\,,
\end{equation}
where $A$, $B$, and $D$ are as in \eqref{frac-D} and \eqref{frac-kn}.
All solutions $P_{-1}$ of the first equation in \eqref{20121103:eq2} are
$$
P_{-1}=c_0+c_1u+c_2u^2\,,
$$
with $c_0,c_1,c_2\in\mc C$.
Next, we want to find all solutions $F$ of the second equation in \eqref{20121103:eq2}.
Applying $\frac{\partial}{\partial u^{(n)}}$, with $n\geq4$, to both sides of the equation $AF=P_{-1}$
we immediately get that $\dord(F)\leq3$ and $\partial F=fD(u')$, with $\dord(f)\leq1$.
Hence, the second equation in \eqref{20121103:eq2} can be rewritten as
the following system of equations,
\begin{equation}\label{20121103:eq3}
\begin{array}{l}
\displaystyle{
\partial^2f-2\frac{u''}{u'}\partial f+\Big(\frac{u''}{u'}\Big)^\prime f+af-u'F
=c_0+c_1u+c_2u^2\,,
} \\
\displaystyle{
\partial F=fD(u')
\,,}
\end{array}
\end{equation}
for $F,f\in\mc V$ with $\dord(F)\leq3$ and $\dord(f)\leq1$.
Applying $\frac{\partial}{\partial u^{(3)}}$ to both sides of the first equation in \eqref{20121103:eq3}
and $\frac{\partial}{\partial u^{(4)}}$ to both sides of the second equation in \eqref{20121103:eq3},
we get
$$
\frac{\partial F}{\partial u'''}=\frac{f}{(u')^2}
\,\,,\,\,\,\,
\frac{\partial f}{\partial u'}=0\,.
$$
Hence, $\dord(f)\leq0$.
Next, 
applying $\frac{\partial}{\partial u^{(2)}}$ to the first equation in \eqref{20121103:eq3}
and $\frac{\partial}{\partial u^{(3)}}$ to the second equation in \eqref{20121103:eq3},
we get
$$
\frac{\partial F}{\partial u''}=-2\frac{u''}{(u')^3}f-\frac{1}{(u')^2}\partial f
\,\,,\,\,\,\,
\partial f=\frac{\partial f}{\partial u}u'\,.
$$
Hence, $f$ is a function of $u$ only.
Using the above result, 
we can rewrite the second equation in \eqref{20121103:eq3},
after integrating by parts twice, as
$$
\partial F=
\partial\Big(
\frac{f}{u'}\Big(\frac{u''}{u'}\Big)^\prime
-\frac{\partial f}{\partial u}\frac{u''}{u'}
+\frac{\partial^2 f}{\partial u^2} u'
\Big)
-\frac{\partial^3 f}{\partial u^3}(u')^2\,.
$$
In particular, it must be $\frac{\partial^3 f}{\partial u^3}(u')^2\in\partial\mc V$,
which is possible only if $\frac{\partial^3 f}{\partial u^3}(u')^2=0$,
see \cite{BDSK09}.
In conclusion, $f$ must be a quadratic polynomial in $u$ with constant coefficients,
and 
$F=\frac{f}{u'}\Big(\frac{u''}{u'}\Big)^\prime
-\frac{\partial f}{\partial u}\frac{u''}{u'}
+\frac{\partial^2 f}{\partial u^2} u'+$ const.
Plugging these results back into equation \eqref{20121103:eq3}
we finally get that
$$
\frac{\partial F}{D(u')}=f=\frac{c_0}a+\frac{c_1}au+\frac{c_2}au^2\,.
$$
Hence, the third equation in \eqref{20121103:eq2} gives $\xi_{-2}=0$.
In conclusion, in this case
the ``dual'' Lenard-Magri sequence, obtained by exchanging the roles of $H$ and $K$,
is of \emph{finite} type,
namely it repeats itself with $\tint h_n\in\ker\big(\frac\delta{\delta u}\big)$
and $P_n\in\ker(D)$ for every $n\leq-1$,
and we don't get any new interesting integrals of motion or equations.

% a=0

Next, we consider the case $a=0$.
In this case, relations \eqref{20121103:eq1}
give the following equations
for $P_{-1}$, $P_{-2}$, and $\tint h_{-2}$:
\begin{equation}\label{20121103:eq4}
DP_{-1}=0
\,\,,\,\,\,\,
DP_{-2}=\frac{\delta h_{-2}}{\delta u}=SP_{-1}\,,
\end{equation}
where $S$, and $D$ are as in \eqref{frac-D} and \eqref{frac-kn1}.
As before, 
$P_{-1}=c_0+c_1u+c_2u^2$, with $c_0,c_1,c_2\in\mc C$.
Hence, the second equation in \eqref{20121103:eq4} reads
\begin{equation}\label{20121103:eq5}
\bigg(\frac1{u'}
\Big(\frac{\partial P_{-2}}{u'}\Big)^\prime\bigg)^\prime
=
\frac1{u'}\Big(
\frac{c_0+c_1u+c_2u^2}{u'}
\Big)^\prime
\,.
\end{equation}
For every $n\in\mb Z_+$, we have the identity
$$
\frac1{u'}\Big(
\frac{u^n}{u'}
\Big)^\prime
=\frac12\partial\frac{u^n}{(u')^2}
+\frac n2\frac{u^{n-1}}{u'}\,.
$$
It follows that the RHS of \eqref{20121103:eq5} cannot be a total derivative
unless $c_1=c_2=0$ (cf. \cite{BDSK09}).
Moreover, if $c_1=c_2=0$ equation \eqref{20121103:eq5} reduces to
$$
\Big(\frac{\partial P_{-2}}{u'}\Big)^\prime
=
\frac{c_0}{2u'}+\text{ const.}u'
\,,
$$
which, for the same reason as before, has no solutions unless $c_0=0$.
In conclusion,
for every non-zero $P_{-1}\in\ker D$,
the ``dual'' Lenard-Magri scheme
is \emph{blocked} at $P_{-2}$.
In this case, as we saw in Section \ref{secb:3.9},
we obtain integrable PDE's which are not of evolutionary type.

In particular, for $(c_1,c_2)\neq(0,0)$ we get the following non-evolutionary
integrable PDE:
\begin{equation}\label{20121103:eq6}
\bigg(\frac1{u_x}
\Big(\frac{u_{tx}}{u_x}\Big)_x\bigg)_x
=
\frac1{u_x}\Big(
\frac{c_0+c_1u+c_2u^2}{u_x}
\Big)_x
\,,
\end{equation}
while for $c_0=1$ and $c_1=c_2=0$, we obtain the following integrable equation
\begin{equation}\label{20121103:eq7}
\Big(\frac{u_{tx}}{u_x}\Big)_x
=
\frac1{2u_x}+\gamma u_x
\,,
\end{equation}
where $\gamma$ is a constant.
Note that, if we apply the differential substitution $v=\log u'$ to equation \eqref{20121020:eq8} 
with $\epsilon=0$, we get equation \eqref{20121103:eq7}.

%%%%%%%%%%%%%%%%%%%%%%%%%%%%%%%%%%%%%%%
\section{NLS type integrable systems}
\label{secb:5}

Recall from Example \ref{20110922:ex5} that the following
is a triple of compatible non-local Poisson structures in two differential variables $u,v$:
$$
L_1=\partial\id
\,\,,\,\,\,\, 
L_2=\left(\begin{array}{cc} 0 & -1 \\ 1 & 0 \end{array}\right)
\,\,,\,\,\,\,
L_3=\left(\begin{array}{cc} 
v\partial^{-1}\circ v & -v\partial^{-1}\circ u \\
-u\partial^{-1}\circ v & u\partial^{-1}\circ u
\end{array}\right)
\,.
$$
We want to consider two non-local Poisson structures $H$ and $K$
which are linear combinations of them:
$H=a_1L_1+a_2L_2+a_3L_3$ and $K=b_1L_1+b_2L_2+b_3L_3$,
where $a_i$'s and $b_i$'s are constants.
As in the previous section, we are only interested in integrable Lenard-Magri schemes of S-type.
In particular, we assume
that the order of the pseudodifferential operator $H$
is greater that the order of $K$,
and so we consider only the case when $b_1=0$.
Note that when $b_2=0$, we get $K=b_3L_3$,
which is a degenerate pseudodifferential operator.
In this case, we cannot apply Theorem \ref{20130123:thm}
and we do not know how to prove integrability.
Hence, we assume that $b_2=1$.
And, since we want that the order of $H$ is greater that the order of $K$,
we assume also that $a_1=1$.

In conclusion, we consider the following compatible pair of non-local Hamiltonian structures:
\begin{equation}\label{20121109:eq1}
\begin{array}{l}
\displaystyle{
H=
\partial\id
+a_2\left(\begin{array}{cc} 0 & -1 \\ 1 & 0 \end{array}\right)
+a_3\left(\begin{array}{cc} 
v\partial^{-1}\circ v & -v\partial^{-1}\circ u \\
-u\partial^{-1}\circ v & u\partial^{-1}\circ u
\end{array}\right)
\,}\\
\displaystyle{
K=
\left(\begin{array}{cc} 0 & -1 \\ 1 & 0 \end{array}\right)
+b_3\left(\begin{array}{cc} 
v\partial^{-1}\circ v & -v\partial^{-1}\circ u \\
-u\partial^{-1}\circ v & u\partial^{-1}\circ u
\end{array}\right)
\,.}
\end{array}
\end{equation}

Note that if $a_3=b_3=0$, the above pair is such that $H$ is ``local'' differential operator,
and $K$ is invertible.
In this case, the Lenard-Magri recursion relations give
$\tint h_{-1}=0$ and $H\frac{\delta h_{n-1}}{\delta u}=K\frac{\delta h_n}{\delta u}$ for every $n\geq0$,
hence $\frac{\delta h_n}{\delta u}=0$ for every $n$.
Therefore, in this case,
the corresponding Lenard-Magri scheme is of finite type, and we don't get any integrable system.
Hence, we assume that $(a_3,b_3)\neq(0,0)$.

Next, we need to find minimal fractional decompositions for $H$ and $K$.
This is given by the following
\begin{lemma}\label{20121108:lem1}
We have the following minimal fractional decomposition for the operator $L_3$:
$$
\left(\begin{array}{cc} 
v\partial^{-1}\circ v & -v\partial^{-1}\circ u \\
-u\partial^{-1}\circ v & u\partial^{-1}\circ u
\end{array}\right)
=
\left(\begin{array}{cc} 
0 & -uv \\
0 & u^2
\end{array}\right)
\left(\begin{array}{cc} 
1 & 0 \\
\frac{v}{u} & \frac1u\partial\circ u
\end{array}\right)^{-1}
$$
The rational matrix pseudodifferential operator $H$ admits the fractional decomposition 
$H=AB^{-1}$ given by
\begin{equation}\label{20121109:eq2}
A=
\left(\begin{array}{cc} 
\partial-a_2\frac{v}{u} & -a_2\frac{1}{u}\partial\circ u-a_3uv \\
\partial\circ\frac{v}{u}+a_2 & \partial\circ\frac1u\partial\circ u+a_3u^2
\end{array}\right)
\,\,,\,\,\,\,
B=
\left(\begin{array}{cc} 
1 & 0 \\
\frac{v}{u} & \frac1u\partial\circ u
\end{array}\right)
\,,
\end{equation}
which is minimal for $a_3\neq0$,
while, for $a_3=0$, $H$ is a matrix differential operator.
The rational matrix pseudodifferential operator $K$ admits the fractional decomposition 
$K=CB^{-1}$
with $B$ as in \eqref{20121109:eq2} and
\begin{equation}\label{20121109:eq3}
C=
\left(\begin{array}{cc} 
-\frac{v}{u} & -\frac{1}{u}\partial\circ u-b_3uv \\
1 & b_3u^2
\end{array}\right)
\,.
\end{equation}
This decomposition is minimal for $b_3\neq0$,
while, for $b_3=0$, $K=L_2$ is an invertible matrix.
\end{lemma}
\begin{proof}
Straightforward.
\end{proof}

As usual, in order to apply the Lenard-Magri scheme it is convenient to find the kernels 
of the operators $B$ and $C$:
$$
\ker B=\mc C \left(\begin{array}{c} 0 \\ \frac1u \end{array}\right)
\,\,,\,\,\,\,
\ker C=\mc C \left(\begin{array}{c} -b_3u \\ \frac1u \end{array}\right)\,.
$$

We next compute the first few steps in the Lenard-Magri scheme.
We have the following $H$ and $K$-associations:
$\tint 0\ass{H}P_0\ass{K}\tint h_0\ass{H}P_1\ass{K}\tint h_1\ass{H}P_2$,
where ($\alpha\in\mc C$)
\begin{equation}\label{20121109:eq4}
\begin{array}{l}
\displaystyle{
P_0=\alpha a_3\left(\begin{array}{c} -v \\ u \end{array}\right)
\,\,,\,\,\,\,
\tint h_0=\frac{1}{2}\tint (u^2+v^2)
\,,} \\
\displaystyle{
P_1= \left(\begin{array}{c} u'-a_2v \\ v'+a_2u \end{array}\right)
\,\,,\,\,\,\,
\tint h_1=
\int\Big(
uv'+\frac{a_2}{2}(u^2+v^2)+\frac{b_3}{8}(u^2+v^2)^2
\Big)
\,,} \\
\displaystyle{
P_2=
\left(\begin{array}{c} 
v''+2a_2u'-a_2^2v+\frac{b_3}{2}\big(u(u^2+v^2)\big)^\prime +\frac{a_3-a_2b_3}{2}v(u^2+v^2) \\ 
-u''+2a_2v'+a_2^2u+\frac{b_3}{2}\big(v(u^2+v^2)\big)^\prime -\frac{a_3-a_2b_3}{2}u(u^2+v^2)
\end{array}\right)
\,.}
\end{array}
\end{equation}
Indeed, we have
$P_0=AF_0$, $BF_0=0$, for $F_0=\alpha\left(\begin{array}{c} 0 \\ \frac1u \end{array}\right)$.
We have
$P_0=CF_1$, $\frac{\delta h_0}{\delta u}=BF_1$, 
for $F_1=\left(\begin{array}{c} u \\ \frac{\beta}u \end{array}\right)$,
where $\alpha,\beta\in\mc C$ are chosen so that $\alpha a_3-\beta b_3=1$
(we can always do so, since, by assumption, $(a_3,b_3)\neq(0,0)$).
We have
$P_1=AF_2$, $\frac{\delta h_0}{\delta u}=BF_2$, 
for $F_2=\left(\begin{array}{c} u \\ 0 \end{array}\right)$.
We have
$P_1=CF_3$, $\frac{\delta h_1}{\delta u}=BF_3$, 
for $F_3=\left(\begin{array}{c} 
v'+a_2u+\frac{b_3}{2}u(u^2+v^2) \\ 
-\frac12\frac{u^2+v^2}{u} 
\end{array}\right)$.
And, finally, we have
$P_2=AF_4$, $\frac{\delta h_1}{\delta u}=BF_4$, 
for $F_4=F_3$.

Next, we check that the orthogonality conditions \eqref{20130104:eq2} hold for $N=0$.
We have
$F=\left(\begin{array}{c} f \\ g \end{array}\right)\in\frac{\delta h_0}{\delta u}^\perp$
if and only if $\tint (uf+vg)=0$,
namely if $f=-\frac{v}{u} g+\frac{h'}{u}$, for some $h\in\mc V$.
But in this case
$$
F=\left(\begin{array}{c} -\frac{v}{u} g+\frac{h'}{u} \\ g \end{array}\right)
=C\left(\begin{array}{c} g+b_3uh \\ -\frac{h}{u} \end{array}\right)\in\im C\,,
$$
proving the first orthogonality condition \eqref{20130104:eq2}.
As for the first orthogonality condition, if $a_3=0$ there is nothing to prove since $H$ is a matrix 
differential operator (i.e. the denominator is $\id$ in its minimal fractional decomposition).
If $a_3\neq0$, we can choose $\alpha=\frac1{a_3}$, and we have
$F=\left(\begin{array}{c} f \\ g \end{array}\right)\in P_0^\perp$
if and only if $\tint (-vf+ug)=0$,
namely if $g=\frac{v}{u} f+\frac{h'}{u}$, for some $h\in\mc V$.
But in this case
$$
F=\left(\begin{array}{c} f \\ \frac{v}{u} f+\frac{h'}{u} \end{array}\right)
=B\left(\begin{array}{c} f \\ \frac{h}{u} \end{array}\right)\in\im B\,,
$$
proving the second orthogonality condition \eqref{20130104:eq2}.
Therefore, by Theorem \ref{20130123:thm} and Remark \ref{20130104:rem2},
we deduce that the elements \eqref{20121109:eq4}
can be extended, possibly going to a normal extension $\tilde{\mc V}$ of $\mc V$,
to infinite sequences $\{\tint h_n\}_{n\in\mb Z_+}$, $\{P_n\}_{n\in\mb Z_+}$,
such that 
$\tint h_{n-1}\ass{H}P_n\ass{K}\tint h_n$.

Finally, 
we have $|H|=1$, $|K|=0$, $\dord(A)=2$, $\dord(B)=\dord(C)=\dord(D)=1$, and $\dord(P_2)=2$.
Hence, the inequality \eqref{20120911:eq1} holds.
Therefore, by Lemma \ref{20130123:lem3}
we have $\dord(P_n)=\dord(\frac{\delta h_n}{\delta u})=n$ for every $n\in\mb Z_+$.
In particular, all the elements $\tint h_n$'s and $P_n$'s are linearly independent.

In conclusion, each equation of the hierarchy $\frac{du}{dt_n}=P_n$
is integrable, and the local functionals $\tint h_n$'s are their integrals of motion.
The first ``non-trivial'' equation of this hierarchy is for $n=2$.
Letting $a_2=0$, $a_3=2\alpha$, and $b_3=2\beta$, it takes the form
\begin{equation}\label{20121109:eq5}
\left\{\begin{array}{l}
\displaystyle{
\frac{du}{dt}=
v'' +\alpha v(u^2+v^2) +\beta \big(u(u^2+v^2)\big)^\prime
} \\ 
\displaystyle{
\frac{dv}{dt}=
-u'' -\alpha u(u^2+v^2) +\beta \big(v(u^2+v^2)\big)^\prime
}
\end{array}\right.
\end{equation}
If we view $u$ and $v$ as real valued functions, and we consider the complex valued function
$\psi=u+iv$, the system \eqref{20121109:eq5} can be written as the following PDE:
$$
i\frac{d\psi}{dt}=\psi''+\alpha\psi|\psi|^2+i\beta(\psi|\psi|^2)^\prime\,,
$$
which, for $\beta=0$, is the well-known Non-Linear Schroedinger equation
(see e.g. \cite{TF86,Dor93,BDSK09}).
The case $\beta\neq0$ has been studied by many authors as well, 
see \cite{KN78,CLL79,WKI79,CC87}.

It is not difficult to show that, when going back,
the Lenard-Magri scheme is ``blocked'' at $\tint h_{-2}$ when $a_3=0$,
and it is of finite type when $a_3\neq0$.
Hence, we don't get any non-evolutionary PDE in this case.

\begin{remark}\label{20130130:rem4}
By Remark \ref{20130130:rem1}, all $\xi_n$'s and $P_n$'s constructed above
have coordinates in $\mc V_u=\mb F[u^{\pm1},v,u',v',u'',v'',\dots]$.
Moreover, using a different fractional decomposition,
this time over $\mc V_v=\mb F[u,v^{\pm1},u',v',u'',v'',\dots]$,
we can show that all coordinates of the $\xi_n$'s and $P_n$'s lie in $\mc V_v$,
hence they actually lie in the algebra of differential polynomials $\mc V=\mb F[u,v,u',v',u'',v'',\dots]$.
This is a normal algebra of differential functions, therefore all
conserved densities $h_n$'s can be chosen in $\mc V$.
\end{remark}

%%%%%%%%%%%%%%%%%%%%%%%%%%%%%%%%%%%%%%%%%%%%%%%%%%%%%%%%%%%%%%%%%%%%%%%%%%%%%%%%%%%%%%%%%%%%%%%%%%%%%%%%%%%%%%
%%%%%%%%%%%%%%% Bibliography %%%%%%%%%%%%%%%%%%%%%%%%%%%%%%%%%%%%%%%%%%%%%%%%%%%%%%%%%%%%%%%%%%%%%%%%%%%%%%%%%
%%%%%%%%%%%%%%%%%%%%%%%%%%%%%%%%%%%%%%%%%%%%%%%%%%%%%%%%%%%%%%%%%%%%%%%%%%%%%%%%%%%%%%%%%%%%%%%%%%%%%%%%%%%%%%

% Non-BibTeX users please use

% LocalWords:  sesquilinearity Leibniz Jacobi RHS Leibniz's polynomially unital
% LocalWords:  eigenspace variational bilinearity integrability

%%%%%%%%%%%%%%%%%%%%%%%%%%%%%%%%%%%%%%%%%%%%%%%%%%%%%%%%%%%%%%%%%%%%%%%%%%%%%%%%%%%%%%%%%%%%%%%%%%%%%%%%%%%%%%
%%%%%%%%%%%%%%% Bibliography %%%%%%%%%%%%%%%%%%%%%%%%%%%%%%%%%%%%%%%%%%%%%%%%%%%%%%%%%%%%%%%%%%%%%%%%%%%%%%%%%
%%%%%%%%%%%%%%%%%%%%%%%%%%%%%%%%%%%%%%%%%%%%%%%%%%%%%%%%%%%%%%%%%%%%%%%%%%%%%%%%%%%%%%%%%%%%%%%%%%%%%%%%%%%%%%

\end{document}